\documentclass[10pt]{article}

\usepackage[margin=1.0in]{geometry}
\usepackage{amsmath}
\usepackage{amsthm, amsfonts,amssymb,euscript}
\usepackage{mathtools}
\usepackage{enumitem}
\usepackage{bm}
\usepackage{amssymb}
\usepackage{graphicx}
\usepackage{caption}
\usepackage{subcaption}
\usepackage{amsfonts}
\usepackage{float}
\usepackage{color}
\usepackage{slashed}
\usepackage[affil-it]{authblk}
\usepackage{mathrsfs}
\usepackage{upgreek}
\usepackage{pifont}
\usepackage{bbm}

\usepackage[usenames,dvipsnames,svgnames,table]{xcolor}
\usepackage[colorlinks=true]{hyperref}
\hypersetup{linkcolor=BrickRed, urlcolor=green, citecolor=blue, linktoc=page}

\numberwithin{equation}{section}

\newtheorem*{proposition*}{Proposition}
\newtheorem*{theorem*}{Theorem}
\newtheorem*{conjecture*}{Conjecture}
\newtheorem*{claim*}{Claim}
\newtheorem*{lemma*}{Lemma}
\newtheorem*{corollary*}{Corollary}
\newtheorem{theorem}{Theorem}[section]
\newtheorem{proposition}[theorem]{Proposition}

\newtheorem{lemma}[theorem]{Lemma}
\newtheorem{corollary}[theorem]{Corollary}

\newtheorem*{definition*}{Definition}

\newtheorem*{assumption*}{\mathcal{A}ssumption}

\newtheorem*{remark*}{Remark}
\newtheorem{remark}{Remark}[section]

\newcommand{\la}{\langle}
\newcommand{\ra}{\rangle}
\newcommand{\R}{\mathbb{R}}
\newcommand{\s}{\mathbb{S}}
\newcommand{\C}{\mathbb{C}}

\newcommand{\Z}{\mathbb{Z}}
\newcommand{\N}{\mathbb{N}}

\newcommand{\snabla}{\slashed{\nabla}}

\newcommand{\Sc}{\textrm{\ding{34}}}

\newcommand{\Lbar}{\underline{L}}

\makeatletter
\newcommand{\biggg}{\bBigg@{4}}
\newcommand{\Biggg}{\bBigg@{5}}
\makeatother

\allowdisplaybreaks

\begin{document}

\title{Late-time asymptotics for geometric wave equations with inverse-square potentials}
\author{Dejan Gajic\thanks{dejan.gajic@uni-leipzig.de}}
\affil{\small Institut f\"ur Theoretische Physik, Universit\"at Leipzig, Br\"uderstrasse 16, 04103 Leipzig, Deutschland}
\date{}
\renewcommand\Affilfont{\itshape\small}
\maketitle
\begin{abstract}
We introduce a new, physical-space-based method for deriving the precise leading-order late-time behaviour of solutions to geometric wave equations on asymptotically flat spacetime backgrounds and apply it to the setting of wave equations with asymptotically inverse-square potentials on Schwarzschild black holes. This provides a useful toy model setting for introducing methods that are applicable to more general linear and nonlinear geometric wave equations, such as wave equations for electromagnetically charged scalar fields, wave equations on extremal Kerr black holes and geometric wave equations in even space dimensions, where existing proofs for deriving precise late-time asymptotics might not apply. The method we introduce relies on exploiting the spatial decay properties of time integrals of solutions to derive the existence and precise genericity properties of asymptotic late-time tails and obtain sharp, uniform decay estimates in time.
\end{abstract}
\tableofcontents

\section{Introduction}
The leading-order late-time behaviour of solutions to geometric wave equations of the form:
\begin{equation}
\label{eq:wavewithoutpot}
\square_{g_M}\phi=0,
\end{equation}
with $\square_{g_M}$ denoting the geometric wave operator with respect to the Schwarzschild spacetime metric with mass $M$:
\begin{equation}
\label{eq:schwmetrictr}
g_M=-\left(1-\frac{2M}{r}\right)dt^2+\left(1-\frac{2M}{r}\right)^{-1}dr^2+r^2(d\theta^2+\sin^2\theta d\varphi^2),
\end{equation}
is dictated by inverse-polynomially decaying tails with coefficients and rates that are connected to the existence of conserved charges \emph{at infinity}. These conserved charges are generalizations of the \emph{Newman--Penrose constants} introduced in \cite{NP1,np2}.\footnote{More precisely, the conserved charges are quantities that are evaluated at future null infinity and they are linear in the radiation field of $\phi$ and appropriate rescaled derivatives; see \S\ref{sec:compzeropot}.} In \cite{paper1,paper2, aagprice}, a physical-space based mechanism was developed for exploiting the conserved charges to obtain precise leading-order late-time asymptotics and demonstrate the existence of inverse-polynomially decaying tails (the corresponding decay rates are referred to as ``Price's law'', as they were first suggested in \cite{Price1972}). This mechanism has been shown to be robust under the replacement of $g_M$ with more general asymptotically flat spacetimes, like the Kerr black hole spacetimes \cite{aagkerr} and more general equations, like the Teukolsky equations of linearized gravity \cite{mazhang21, mazhang21b}. 

Besides elucidating the late-time properties of gravitational radiation at infinity in the setting of the Einstein equations of general relativity, understanding the existence and decay properties of late-time tails is necessary for determining the singularity properties in dynamical black hole interiors and addressing the \emph{strong cosmic censorship conjecture} in general relativity; see for example \cite{LukSbierski2016, LukOh2016a, LukOh2016b, sbierski22} and references therein.

The mechanism described above, which relies on the existence of conserved charges at infinity, still applies when we consider geometric wave equations with a potential $V$,
\begin{equation}
\label{eq:wavewithpot}
\square_{g_M}\phi=V\phi,
\end{equation}
where $V$ denotes a time-independent function satisfying $|V|\lesssim r^{-p}$ for $p>2$ with respect to a suitable radial coordinate $r$.

In the present paper, we consider instead potentials with $p=2$:
\begin{align}
\label{eq:wavewithinvsqpot}
\square_{g_M}\phi=&\:V_{\alpha}\phi,\\
\label{eq:scalecrit}
V_{\alpha}(r)\sim&\: \alpha r^{-2}\quad (r\to \infty),
\end{align}
with the additional global assumption that $V_{\alpha}(r)\geq -\frac{1}{4}r^{-2}$,\footnote{In fact, the methods of the present paper can be extended to a larger class of potentials $V$ with $\alpha>-\frac{1}{4}$ that do not need to satisfy $V(r)\geq -\frac{1}{4}r^{-2}$ globally; see assumption \ref{assm:B} in \S \ref{sec:classeq} for more details.} which requires in particular the restriction to $\alpha>-\frac{1}{4}$. Potentials satisfying \eqref{eq:scalecrit} are \emph{asymptotically scale-critical}.\footnote{Here, scale-critical means that the potential has the property that for $\phi$ solutions to  \eqref{eq:wavewithpot} the rescaled function $\phi_{\lambda}(t,r,\theta,\varphi)~=\lambda^{-2}\phi(\lambda t, \lambda r, \theta,\varphi)$ with $\lambda\in \R\setminus \{0\}$, is also a solution to \eqref{eq:wavewithpot}. The potential $V(r)=\alpha r^{-2}$ would therefore be exactly scale-critical. Inverse-square asymptotics for $V$ can also be thought of as ``critical'' in the following sense: when $\alpha<0$, the operator $\Delta-V$ on $\R^3$ has an infinite discrete spectrum if $V(r)\sim \alpha r^{-p}$ with $p<2$, whereas the discrete spectrum is finite if $p=2$ and $\alpha>-\frac{1}{4}$ or if $p>2$; see Theorem XIII.6 of \cite{reedsimoniv}.} We will moreover assume that $\alpha\neq l(l+1)$ for any $l\in \N_0$ (the case where $\alpha=l(l+1)$ is discussed in  \S \ref{sec:compzeropot}).

\textbf{For wave equations with asymptotically inverse-square potentials $V_{\alpha}$, the mechanism relying on the existence of conserved charges along future null infinity \emph{breaks down}.} The aim of the present paper is to develop a new physical-space based mechanism for deriving late-time asymptotics in a general setting of geometric wave equations on asymptotically flat spacetimes with inverse-square potentials that avoids the use of conserved charges at infinity. 

The study of late-time asymptotics for \eqref{eq:wavewithinvsqpot} is motivated by the wave equation for electromagnetically charged scalar fields (on black hole spacetimes), as well as the wave equation on extremal Kerr black hole spacetimes and wave equations on asymptotically flat spacetimes in even space dimensions. In all of these settings, we encounter analogous difficulties to those faced in the study of  \eqref{eq:wavewithinvsqpot} when trying to derive the leading-order term in the late-time asymptotics of solutions to the corresponding equations. This motivates the need for the development of sufficiently robust mathematical tools, which are introduced in the present paper; see \S \ref{sec:motivation} for more details.

To simplify the exposition, we will restrict the consideration in the present paper to Schwarzschild spacetime backgrounds, but the main methods and philosophy apply also when considering more general asymptotically flat spacetimes. The proofs in the paper will moreover be entirely self-contained. We state now the first version of the main result:

\begin{theorem}[First version]
\label{thm:rough}
Let $\phi$ be a solution to \eqref{eq:wavewithinvsqpot} arising from suitably regular and decaying initial data. Then:
\begin{enumerate}[label=\emph{(\roman*)}]
\item\emph{(Late-time asymptotics)} There exist constants $C=C(V_{\alpha},M)>0$ and $\nu=\nu(\alpha)>0$ such that for a suitable hyperboloidal\footnote{The level sets of $\tau$ intersect the event horizon to the future of the bifurcation sphere and are asymptotically hyperboloidal, i.e. they can be thought of as intersecting future null infinity appropriately.}time coordinate $\tau\geq 0$, a radial coordinate $r$ and spherical coordinates $(\theta,\varphi)$, we can estimate:
\begin{align}
\label{eq:mainrough}\nonumber
\left|\phi(\tau,r,\theta,\varphi)-Q_{\alpha}[\phi]w_0(r) f(\tau,r) \right|\leq &\: C D[\phi] \cdot r^{-\frac{1}{2}+\frac{\beta_0}{2}}\cdot f(\tau,r) \cdot (\tau+M)^{-\nu}, \quad \textnormal{where} \\ 
f(\tau,r)=&\: (\tau+2r)^{-\frac{1}{2}-\frac{\beta_0}{2}}(\tau+M)^{-\frac{1}{2}-\frac{\beta_0}{2}},\\ \nonumber
\beta_0=&\: \sqrt{1+4\alpha},
\end{align}
with $Q_{\alpha}[\phi]$ a constant that can be explicitly determined from initial data, $w_0$ denoting the stationary, spherically symmetric solution to \eqref{eq:wavewithinvsqpot} that has the asymptotic behaviour $w_0(r)\sim r^{-\frac{1}{2}+\frac{\beta_0}{2}}$ as $r\to \infty$ and $D[\phi]$ an initial data norm.
\item\emph{(Genericity of late-time tails)} The constant $Q_{\alpha}[\phi]$ is non-vanishing for an open and dense subset $\mathcal{G}$ of initial data, with respect to an appropriate weighted $L^1$-type initial data norm, even when restricting to compactly supported initial data, and the complement $\mathcal{G}^c$ has codimension 1.
\end{enumerate}
\end{theorem}
We refer to Theorem \ref{thm:precise} for a more precise version of the above theorem, including precise expressions for the initial data quantities $Q_{\alpha}[\phi]$ and $D[\phi]$, as well as the exponent $\nu$, and moreover, refined late-time asymptotics when projecting $\phi$ to spherical harmonics of degree $\geq \ell$ with $\ell\in \N$. The quantities $Q_{\alpha}[\phi]$ can \emph{a posteriori} be interpreted as new conserved charges at infinity at the level of time integrals of $\phi$; see Remark \ref{rm:consvlaws}.

Theorem \ref{thm:rough} implies in particular that solutions to \eqref{eq:wavewithinvsqpot} have late-time asymptotics with rates that depend on the coefficient $\alpha$ in the following way in terms of standard Schwarzschild coordinates $(t,r,\theta,\varphi)$: for $r_0\geq 2M$,
\begin{align*}
 \phi|_{\{r=r_0\}}\sim &\: (t+r_*(r_0))^{-(1+\sqrt{1+4\alpha})}\quad  (t+r_*(r_0)\to \infty),\\
r\phi|_{\mathcal{I}^+=\{t+r_*=\infty, t-r_*<\infty\}}\sim &\: (t-r_*)^{-\frac{1}{2}(1+\sqrt{1+4\alpha})}\quad (t-r_*\to \infty),
\end{align*}
with $t$ the standard Schwarzschild time coordinate and $r_*(r)$ the tortoise coordinate on Schwarzschild, which satisfies $\frac{dr_*}{dr}=(1-\frac{2M}{r})^{-1}$. 
\subsection{A comparison with more rapidly decaying potentials and Price's law for late-time tails}
\label{sec:compzeropot}
It was first shown in  \cite{paper2,paper1} that in the setting of \eqref{eq:wavewithoutpot} (i.e.\ the wave equation on Schwarzschild without a potential), \eqref{eq:mainrough} holds for conformally regular initial data, with $\alpha=0$ (in which case $w_0\equiv 1$) and $Q_{\alpha}[\phi]=4 I_0[\phi]$, where $I_0[\phi]$ is the Newman--Penrose charge of $\phi$, defined as follows:
\begin{equation*}
 I_0[\phi]=\lim_{r\to \infty}\int_{\s^2}r^2(\partial_t+\partial_{r_*})(r\phi)(\tau=0,r,\theta,\varphi)\,\sin \theta d\theta d\varphi.
\end{equation*}
The above limit is in fact well-defined for all $\tau\in \R$ and is moreover \emph{conserved in $\tau$}.

In the case $I_0[\phi]=0$; for example, when considering compactly supported initial data on $\tau=0$, Theorem~\ref{thm:rough} still holds with $\alpha=0$ if we replace $f(\tau,r)$ with $\partial_{\tau}f(\tau,r)$. In this case, $Q_{\alpha}[\phi]=4 I_0^{(1)}[\phi]$, with $I^{(1)}_0[\phi]$ the \emph{time inverted Newman--Penrose charge}, which was first defined in \cite{paper2} and is in fact equal to $I_0[\partial_{\tau}^{-1}\phi]$, where $\partial_{\tau}^{-1}\phi$ is the \emph{time integral of $\phi$}, defined as follows:
\begin{equation*}
\partial_{\tau}^{-1}\phi(\tau,r,\theta,\varphi)=-\int^{\infty}_{\tau} \phi(s,r,\theta,\varphi)\,ds.
\end{equation*}

In both the $I_0[\phi]\neq 0$ and $I_0[\phi]=0$ cases, the mechanism developed in \cite{paper2,paper1} for deriving late-time asymptotics relies on the conservation property of a Newman--Penrose charge at infinity, which is then propagated to the rest of the spacetime using appropriate decay estimates in time for $\phi$.

When considering projections of $\phi$ onto all spherical harmonics $Y_{\ell m}(\theta,\varphi)$ of degree $\ell$, denoted $\phi_{ \ell}$, the Laplacian on the unit round sphere $\slashed{\Delta}_{\s^2}$, which appears in the wave operator $\square_g$, satisfies the following property:
\begin{equation*}
\slashed{\Delta}_{\s^2}\phi_{\ell }=-\ell(\ell+1)\phi_{\ell },
\end{equation*}
so we can interpret the term $r^{-2}\slashed{\Delta}_{\s^2} \phi_{\ell}$ appearing in $\square_{g_{M}}\phi_{\ell}$ as a potential term $-V_{\alpha}\phi_{\ell}$, where $V_{\alpha}(r)=\alpha r^{-2}$ with $\alpha=\ell(\ell+1)$ and $\ell\in \N_0$.

In the setting of \eqref{eq:wavewithoutpot}, the late-time asymptotics for $\phi_{\ell}$ (and in fact for $\phi_{\geq \ell}$, the projections of $\phi$ onto all spherical harmonics of degree $\geq \ell$), were obtained in \cite{aagprice}; see also earlier work on $\ell$-dependent decay estimates \cite{dssprice} and on late-time asymptotics in compact regions in space \cite{hintzprice}. 

These late-time asymptotics agree with Theorem \ref{thm:rough} if we take $\alpha=\ell(\ell+1)$ and we replace $w_0$ with $P_{\ell}(r-M)$, the $\ell$-th Legendre polynomial. In this case, $Q_{\alpha}[\phi]$ will be proportional to $I_{\ell}[\phi](\theta,\varphi)$, with $I_{\ell}[\phi](\theta,\varphi)$ the $\ell$-th Newman--Penrose charge, which is defined as a limit along $\tau=0$ as $r\to \infty$ of appropriate linear combinations of higher-order powers of derivatives $r^2(\partial_t+\partial_{r_*})$ acting on $r\phi$; see \cite{gk22} for explicit expressions for $I_{\ell}[\phi]$. Moreover, for compactly supported data ($I_{\ell}[\phi]=0$), $f(\tau,r)$ is replaced with $\partial_{\tau}f(\tau,r)$ and $Q_{\alpha}[\phi]$ is then proportional to $I_{\ell}[\partial_{\tau}^{-1}\phi]$. 

When considering \eqref{eq:wavewithinvsqpot} with $\alpha=\ell(\ell+1)$, where $\ell\in \N_0$, the conservation of $\ell$-th Newman--Penrose charges at infinity can therefore still be exploited to derive late-time asymptotics.

The numerology appearing in the decay rates of asymptotic inverse-polynomial tails in the setting of \eqref{eq:wavewithoutpot} (Price's law) dates back to heuristic work in \cite{Price1972}. The presence of inverse-polynomial late-time tails for \eqref{eq:wavewithoutpot} with sufficiently rapidly decaying data can be attributed to the non-vanishing of the mass $M$, which appears as a factor in the expressions for $I_{\ell}[\partial_{\tau}^{-1}\phi]$.\footnote{On a $3+1$-dimensional Minkowski spacetime background, which is locally isometric to Schwarzschild with $M=0$, the \emph{strong Huygens principle} holds and there are no late-time tails in the case of compactly supported initial data; instead, $\phi$ is compactly supported in $\tau$.}

We note also the late-time asymptotics derived in \cite{keh21b,keh21c,gk22}, where initial data are considered with various different $r$-asymptotics that arise from more physically motivated settings  and are connected to the failure of peeling of curvature components for the Einstein equations. Such data require the use of \emph{modified} Newman--Penrose charges to obtain late-time asymptotics.

In the setting of the present paper, where we consider \eqref{eq:wavewithinvsqpot} with $\alpha\neq l(l+1)$ for any $l\in \N_0$, we can no longer exploit conserved charges for $\phi$ at infinity to obtain late-time asymptotics, so the methods involving Newman--Penrose charges do not apply, thus motivating the need for a new, more robust method. We refer to \S \ref{sec:proofideas} for an outline of this method.

When considering wave equations with non-trivial potentials, an important difference between potentials satisfying $V(r)\sim \alpha r^{-2-\delta}$, with $\delta>0$ and asymptotically inverse-square potentials $V(r)\sim \alpha r^{-2}$, with $\alpha\neq l(l+1)$ for any $l\in \N_0$, is that in the former case, the methods involving Newman--Penrose charges \emph{can} still be applied as sketched above in the $V\equiv 0$ setting to obtain late-time asymptotics with decay rates that are independent of $\alpha$. For example, in the former case we would expect $\phi|_{r=r_0}\sim \tau^{-\min\{3,2+\delta\}}$ for compactly supported initial data. When $V(r)\sim \alpha r^{-2}$, we instead see:  $\phi|_{r=r_0}\sim \tau^{-1-\sqrt{1+4\alpha}}$, so the leading-order coefficient $\alpha$ appears in the decay rate, in contrast with the case where $V(r)\sim \alpha r^{-2-\delta}$.

Note moreover that in contrast with the $V\equiv 0$ setting, both $V(r)\sim \alpha r^{-2}$ with $\alpha\neq l(l+1)$, $l\in \N_0$, and $V(r)\sim \alpha r^{-2-\delta}$ with $\delta>0$ and $\delta\in \R_+\setminus \N_0$, will lead to time integrals $\partial_\tau^{-1}\phi$ that are \emph{not} conformally smooth at null infinity, even if we start with smooth, compactly supported initial data. That is to say, time integrals will only have \emph{finite regularity} with respect to the coordinate $x=\frac{1}{r}$ at $x=0$ ($r=\infty$). While this is perhaps not so surprising for $V(r)\sim \alpha r^{-2-\delta}=\alpha x^{2+\delta}$, which itself has only finite regularity at $x=0$, it is \emph{a priori} perhaps not so obvious that there should also be a breakdown in conformal regularity when, for example, $V(r)= \alpha r^{-2}=\alpha x^2$, with $\alpha\neq l(l+1)$, $l\in \N_0$, which \emph{is} conformally smooth at $x=0$.

\subsection{Inverse-square potentials as a toy model}
\label{sec:motivation}
The main motivation for considering geometric wave equations with inverse-square potentials in the present paper is to provide a simplified setting for tackling some of the difficulties that show up when studying the late-time behaviour and stability properties of solutions to the charged scalar field equation on black hole backgrounds and solutions to the (uncharged) wave equation on extremal Kerr backgrounds. The methods of the present paper can moreover be used to determine the late-time asymptotics of geometric wave equations corresponding to asymptotically flat spacetimes with even space dimensions. We briefly outline these different settings below.
\paragraph{Charged scalar fields}
A natural generalization of the equation $\square_{g_{M}}\phi=0$ to the setting of the Einstein--Maxwell equations is the equation:
	\begin{align*}
	\sum_{0\leq \alpha,\beta \leq 3}(g_{M,e}^{-1})^{\alpha \beta} D_{\alpha} D_{\beta} \phi=&\:0,\quad \textnormal{with}\\ \nonumber
	D_{\mu} =&\: \nabla_{\mu}+ iq A_{\mu},\\
	dA=&\:F_{RN}=\frac{e}{r^2}dt\wedge dr,
	\end{align*}
	where $e\in \R$ is the electromagnetic charge appearing in the subextremal Reissner--Nordstr\"om black hole metric $g_{M,e}$, $q\in \R\setminus \{0\}$ is the electromagnetic charge carried by the scalar field $\phi$ and  $F_{RN}$ is a solution to the vacuum Maxwell equations $d ^\star F_{RN}=0$ with respect to  $g_{M,e}$. 
	
	The above equations show up in the linearization of the coupled system of Einstein--Maxwell-charged scalar field equations around the Reissner--Nordstr\"om solution, and they model the gravitational properties of electromagnetic waves coupled to charged scalar field matter when keeping the Maxwell field and background spacetime metric fixed. After choosing an appropriate electromagnetic gauge $A$, the equation for $\phi$ will resemble \eqref{eq:wavewithpot} with an inverse-square potential $V(r)\sim \alpha r^{-2}$ where $
	\alpha<0$ depends on the scalar field charge $q$, and hence, when determining the late-time asymptotics of $\phi$ for the charged scalar field equation, we also encounter the difficulties faced in the setting of the wave equation on Schwarzschild with inverse-square potential. In particular, one cannot exploit conserved charges at null infinity to obtain late-time asymptotics. 
	
	In upcoming work \cite{gvdm22b}, the methods developed in the present paper are applied to obtain the precise late-time asymptotics of charged scalar fields on Reissner--Nordstr\"om spacetimes under the assumption $q^2e^2<\frac{1}{4}$, which is analogous to the assumption $\alpha>-\frac{1}{4}$, and, in particular, the decay rates first suggested in \cite{hp98b} are validated for this range of charge parameters. In \cite{gvdm22}, a different method is applied for deriving results on late-time asymptotics for charged scalar fields on Minkowski spacetime backgrounds that relies on the special conformal geometric properties of the Minkowski spacetime. We refer also to earlier work \cite{vdm18} on decay estimates for the coupled Maxwell--charged scalar field system with small scalar field charge on a fixed Reissner--Nordstr\"om background.
	
	Determining the late-time asymptotics of charged scalar fields on Reissner--Nordstr\"om is an important first step towards determining the late-time behaviour of gravitational radiation along the event horizon of dynamical black holes with charged matter, which is closely connected to the singularity properties in the corresponding black hole interiors and features some interesting new phenomena that are not seen in the uncharged setting; see for example \cite{vdm19,kvdm21}.
	
	\paragraph{Extremal black holes} The properties of solutions to the geometric wave equation on an extremal Kerr spacetime,
	\begin{equation*}
	\square_{g_{EK}}\phi=0,
	\end{equation*}
	near the event horizon are closely related to the properties near null infinity, provided $\phi$ is assumed to be axisymmetric. In particular, axisymmetric solutions also satisfy conservation laws along the event horizon of extremal Kerr black holes, which are directly related to the conservation laws for Newman--Penrose charges along future null infinity; see \cite{couch,hm2012, bizon2012}. The presence of such conservation laws can not only be used to derive late-time asymptotics, see \cite{paper4}, but it is also the origin of the existence of axisymmetric asymptotic \emph{instabilities} on extremal Kerr black holes \cite{aretakis3, aretakis4}.
	
	When considering \underline{non}-axisymmetric solutions, however, both the above relation between null infinity and the event horizon, and the corresponding conservation laws, break down. If we restrict to azimuthal modes, which are eigenfunctions of the axial vector field $\partial_{\varphi}$ with eigenvalues $im$, where $m\in \Z$ is the azimuthal number, the wave equation near the event horizon instead resembles the equation for a charged scalar field with charge $q\sim m$ near infinity. Therefore, the philosophy underlying the methods developed in the present paper, which relates (conformal) regularity of time integrals at $\tau=0$ to decay rates in time, can also be applied to the setting of extremal Kerr. In upcoming work \cite{gajEK}, we derive the late-time asymptotics for azimuthal modes on extremal Kerr and obtain as a corollary the existence of \emph{azimuthal instabilities}, which are stronger than the axisymmetric instabilities mentioned above. We refer to \cite{zimmerman1} for earlier heuristic work on azimuthal instabilities that is consistent with the results in \cite{gajEK}, and \cite{costa20} for a proof of mode stability on extremal Kerr.
	
	\paragraph{Late-time tails of waves in even space dimensions}
	The wave equation on $n+1$-dimensional Schwarzschild spacetimes with $n\geq 3$, which are also known as the Schwarzschild--Tangherlini spacetimes, see \cite{tan63}, takes the following form in standard Schwarzschild-type coordinates:
	\begin{equation*}
	0=-D^{-1}\partial_t^2\phi+r^{-(n-1)}\partial_r(Dr^{n-1}\partial_r\phi)+r^{-2}\slashed{\Delta}_{\s^{n-1}}\phi,
	\end{equation*}
	with $D=1-\frac{2M}{r^{n-2}}$ and $\slashed{\Delta}_{\s^{n-1}}$ the Laplacian with respect to the unit round $n-1$-sphere metric; see for example \cite{volker1, moschidis1}. Introducing the rescaled variable $\psi=r^{\frac{n-1}{2}}\phi$,\footnote{When evaluated at future null infinity, $\psi$ corresponds to the \emph{Friedlander radiation field} in $n$ space dimensions.} the above equation is equivalent to:
	\begin{equation*}
	0=-D^{-1}\partial_t^2\psi+ \partial_r(D\partial_r\psi)-\frac{D}{4r^2}(n-1)(n-3)\psi-\frac{n-1}{2r}\frac{dD}{dr}\psi+r^{-2}\slashed{\Delta}_{\s^{n-1}}\psi.
	\end{equation*}
	 The term $-\frac{D}{4r^2}(n-1)(n-3)\psi$ can be interpreted as a potential term:
	\begin{equation*}
	-V_{\alpha}(r)\psi,
	\end{equation*}
 with $V_{\alpha}(r)=\alpha r^{-2}+O_{\infty}(r^{-n})$. When $n=2k+3$, with $k\in \N_0$, we have that $\alpha=k(k+1)$. Hence, one can obtain late-time asymptotics and sharp decay estimates for $\phi$ by exploiting the existence of conserved charges at infinity that are analogous to those for fixed spherical harmonic projections $\phi_{\ell}$ in the $n=3$ case; see the discussion in \S \ref{sec:compzeropot}.
	
	When $n=2k+4$, with $k\in \N_0$, we have that $\alpha\neq l(l+1)$ for any $l\in \N_0$, and one can no longer make use of conservation laws at infinity. In the case of the $n+1$-dimensional wave equation on Minkowski ($M=0$), this is reflected in the fact that the strong Huygens principle does not hold in even space dimensions $n$. The method for deriving late-time asymptotics introduced in the present paper, however, does remain applicable in the case of even space dimensions, and it can be used to derive a higher-dimensional analogue of Price's law (in both even and odd space dimensions).
	
	Note that when considering the $n+1$-dimensional wave equation on Minkowski with $n=2$, and taking $\psi=r^{\frac{1}{2}}\phi$, we see an effective potential term with potential $V_{\alpha}(r)=\alpha r^{-2}$ with $\alpha=-\frac{1}{4}$. As the methods in the present paper rely on the restriction $\alpha>-\frac{1}{4}$, they do not apply directly in this setting. It would therefore be an interesting problem to extend the methods to include the case $\alpha=-\frac{1}{4}$. See \cite{dk16} for related work on sharp decay estimates for wave equations with potentials on the real (half-)line corresponding to the choice $\alpha=-\frac{1}{4}$, which are derived by using Fourier transforms that are adapted to the potential.\\
	\\
	
	We finally note that wave equations with inverse-square potentials also show up in the study of wave equations on manifolds with asymptotically conical ends; see \cite{sss10a,sss10b}.
\subsection{Outline of paper}
Here, we give an outline of the structure of the remainder of the paper.

In \S\S\ref{sec:geometry}--\ref{sec:main}, we give precise formulations of the main results, after introducing the preliminary geometric objects and notations needed to understand the formulations. We moreover introduce the key ideas and the main steps of the proofs here. More precisely:
\begin{itemize}
\item In \S \ref{sec:geometry} we derive preliminary geometric properties of the Schwarzschild spacetime backgrounds that will play a role in the remainder of the paper. We moreover introduce the main notational conventions.
\item Then in \S \ref{sec:prelimwave}, we introduce precisely the class of wave equations with potentials that will be analyzed.
\item We give precise formulations of the main theorems of the paper in \S \ref{sec:main} and already give a sketch of their proofs, emphasizing the key new ideas introduced in the present paper.
\end{itemize}

The full proofs of the results stated in \S \ref{sec:main} are then carried out in \S\S \ref{sec:boundedness}--\ref{sec:latetimeasymell}. These sections are structured as follows:

\begin{itemize}
\item In \S \ref{sec:boundedness}, we derive energy boundedness estimates.
\item We derive integrated energy decay estimates in \S \ref{sec:iled}.
\item In \S \ref{sec:rp}, we derive (commuted) $r$-weighted energy estimates.
\item Then in \S \ref{sec:decaytimeder}, we combine the estimates in \S\S \ref{sec:boundedness}--\ref{sec:rp} to derive energy decay-in-time estimates for time derivatives.
\item We construct time integrals and develop an elliptic theory of time inversion in \S \ref{sec:timinv}.
\item The energy decay-in-time estimates in \S \ref{sec:decaytimeder} are combined with the time-inversion theory in \S \ref{sec:timinv} to derive the main result of the paper in \S \ref{sec:latetimeasymp}: uniform, pointwise, decay-in-time estimates from which we can read off the precise late-time asymptotics.
\item We conclude in \S \ref{sec:latetimeasymell} by showing how the results in \S \ref{sec:latetimeasymp} also give late-time asymptotics for solutions supported on spherical harmonics with higher angular momenta $\ell$.
\end{itemize}

\subsection{Acknowledgements}
The author thanks Leonhard Kehrberger for comments on an earlier version of the manuscript.

\section{Preliminaries: geometry and foliations}
\label{sec:geometry}
In this section, we introduce the main geometric concepts and notation that will be used throughout the paper.
\subsection{Spacetime metric and foliations}
\label{sec:metric}

Let $M>0$ and consider the manifold-with-boundary $(\mathcal{M}_M,g_M)$ with:
\begin{align*}
 \mathcal{M}_M=&\:\R_v\times [2M,\infty)_r\times \s^2_{\theta,\varphi},
  \end{align*}
equipped with a coordinate chart $(v,r,\theta,\varphi)$ and the metric:
\begin{align*}
g_M=&\: - Ddv^2+2dvdr +r^2(d\theta^2+\sin^2\theta d\varphi^2),\\
D(r)=&\: 1-\frac{2M}{r}.
\end{align*}
The coordinates $(\theta,\varphi)$ are standard spherical coordinates on the 2-sphere $\s^2$ that degenerate at a line connecting the two poles. The spacetimes $(\mathcal{M}_M,g_M)$ with $M>0$ are called \emph{Schwarzschild black hole exteriors}.

We define the functions
\begin{align*}
t=&\:v-r_*,\\
u=&\:v-2r_*,\quad \textnormal{with}\\
r_*=&\:\int_{3M}^r D^{-1}(r')\,dr'.
\end{align*}
Note that the level sets of the coordinate function $v$ define ingoing null hypersurfaces, whereas the level sets of $u$ define outgoing null hypersurfaces. Furthermore, by considering coordinates $(t,r,\theta,\varphi)$ coordinates we recover the standard form of the Schwarzschild metric in the subset $\R_t\times (2M,\infty)_r\times \s^2_{\theta,\varphi}$; see \eqref{eq:schwmetrictr}.

Let $\tilde{h}:[2M,\infty)\to  (0,2]$ be a smooth function such that
\begin{equation*}
2-\tilde{h}D> 0.
\end{equation*}
Define moreover $h:[2M,\infty)\to \R_{+}$ as follows:
\begin{equation*}
h(r):=2-\tilde{h}(r) D(r).
\end{equation*}
Consider the time function $\tau$:
\begin{equation*}
\tau=v-v_0-\int_{2M}^r \tilde{h}\,dr',
\end{equation*}
where $v_0>0$.

We then have that
\begin{equation*}
(g_M)^{-1}(d\tau,d\tau)=-\tilde{h}h<0,
\end{equation*}
so the level sets of $\tau$,
\begin{equation*}
\Sigma_{\tau'}=\{\tau=\tau'\},
\end{equation*}
are spacelike hypersurfaces.

We further restrict to an \emph{asymptotically hyperboloidal} foliation $\Sigma_{\tau}$ by imposing
\begin{align*}
h(r)=&\:h_0r^{-2}+O_{\infty}(r^{-3}),
\end{align*}
with $h_0$ a constant, which is necessarily positive, by the previous assumption on $\tilde{h}$.

Here we applied the following variation of ``big-O notation'': we write $f(r)= O_k(r^{-q})$ for a suitably regular function $f: [2M,\infty)\to \R$, $q\in \R$ and $k\in \N_0$, if there exists a constant $C>0$ such that $|\partial_r^l f|\leq C r^{-q-l}$ for all $0\leq l\leq k$. If $f=O_k(r^{-q})$ for all $k\in \N_0$, we write $f=O_{\infty}(r^{-q})$ and we write $O(r^{-q})=~O_0(r^{-q})$.

We denote with
\begin{equation*}
\mathcal{R}=\bigcup_{\tau\in [0,\infty)}\Sigma_{\tau}
\end{equation*}
the spacetime subset of $\mathcal{M}_M$ that will be relevant for the analysis in this paper.

Note that $\Sigma_{\tau}=F_{\tau}(\Sigma_0)$, where $F_{\tau}$ denotes the flow along the vector field $T=\partial_v$, which is the standard timelike Killing vector field in Schwarzschild. From the Killing property of $T$, it follows that the hypersurfaces $\Sigma_{\tau}$ are all isometric. We will moreover denote $\Sigma:=\Sigma_0$.

The intersections $\Sigma_{\tau}\cap \{r=r'\}$ are 2-spheres, which we will denote as follows:
\begin{equation*}
S^2_{\tau,r'}:=\Sigma_{\tau}\cap \{r=r'\}.
\end{equation*}

The null coordinate $u$ takes on the following values along $\Sigma_{\tau}$,
\begin{equation*}
u|_{\Sigma_{0}}(r)=v|_{\Sigma_{0}}(r)-2r_*=\tau+\left[v_0+\int_{2M}^{3M}\tilde{h}\,dr'\right]-\int_{3M}^r h\,D^{-1}dr',
\end{equation*}
so
\begin{align*}
\lim_{r\to \infty}u|_{\Sigma_{\tau}}(r)=&\:\tau+u_0, \quad \textnormal{with}\\
u_0=&\:\left[v_0+\int_{2M}^{3M}\tilde{h}\,dr'\right]-\int_{3M}^{\infty} h\,D^{-1}dr',
\end{align*}
where the above expression for $u_0$ is well-defined because $h$ is integrable in $r$, by the assumptions on the $r$-decay of $h$ (asymptotical hyperbolicity).

The triple $(r,\theta,\varphi)$ defines a coordinate chart on the hypersurfaces $\Sigma_{\tau}$, which can be complemented with the time function $\tau$ to obtain the coordinate chart $(\tau, r,\theta,\varphi)$ on $\mathcal{R}$. We can express $g_M$ as follows in these coordinates:
\begin{equation*}
g_M=-Dd\tau^2-2(1-h)d\tau dr+ h\tilde{h} dr^2+r^2(d\theta^2+\sin^2\theta d\varphi^2).
\end{equation*}
We consider now the rescaled metric $\widehat{g}_M=r^{-2}g_M$ and employ the coordinate $x=\frac{1}{r}$ to express:
\begin{equation*}
\widehat{g}_M=- x^2 d\tau^2+2(1-h)d\tau dx+r^2h \tilde{h}dx^2+d\theta^2+\sin^2\theta d\varphi^2.
\end{equation*}
By assumption, $r^2h: [0,(2M)^{-1})_x \to\R_{+}$ defines a smooth function of $x$, so we can smoothly\footnote{The smoothness of the extension is actually not necessary to be able to apply the methods in this paper, since we only make use of \underline{finite} regularity properties.} extend the Lorentzian manifold-with-boundary $([0,\infty)_{\tau}\times (0,\frac{1}{2M}]_x\times \s^2,\widehat{g}_M)$ to $([0,\infty)_{\tau}\times [0,\frac{1}{2M}]_x\times \s^2,\widehat{g}_M)$, where we denote:
\begin{equation*}
\widehat{\mathcal{R}}=\left[0,\frac{1}{2M}\right]_x\times \s^2.
\end{equation*}
We similarly denote with $\widehat{\Sigma}_{\tau}$ the union of $\Sigma_{\tau}$ with the appropriate level set of $\tau$ in the extended manifold-with-boundary $\widehat{\mathcal{R}}$.

We denote 
\begin{equation*}
\mathcal{I}^+=\{x=0\}\subset \widehat{\mathcal{R}}
\end{equation*}
and refer to this hypersurface as \emph{future null infinity}.

Note that we can express $T=\partial_{\tau}$. It will be convenient to denote moreover 
\begin{equation*}
X=\partial_{r}
\end{equation*}
in $(\tau,r,\theta,\varphi)$ coordinates. In $(v,r,\theta,\varphi)$ coordinates, on the other hand, we can express: $T=\partial_v$, $X=\partial_r+\tilde{h}\partial_v$.

Note that
\begin{align}
\label{eq:LbarXT}
\underline{L}=&\:-\frac{D}{2}X+\frac{D}{2}\tilde{h}T\quad \textnormal{and}\\
\label{eq:LXT}
L=&\:T-\underline{L}=\frac{D}{2}X+\frac{1}{2}hT
\end{align}
define ingoing and outgoing null vector fields, respectively.\footnote{In $(u,v,\theta,\varphi)$ coordinates, we can express $\underline{L}=\partial_u$ and $L=\partial_v$.}

We also denote 
\begin{equation}
\label{eq:defZ}
Z=L-\underline{L}=DX+(h-1)T.
\end{equation}

Finally, we let $\mathbf{N}$ be a smooth vector field, such that $\mathbf{N}=T$ for $r\geq 4M$ and $\mathbf{N}=\mathbf{n}_{\tau}$ for $r\leq 3M$,  where $\mathbf{n}_{\tau}$ denotes the normal vector field to $\Sigma_{\tau}$ 

\begin{figure}[H]
	\begin{center}
\includegraphics[scale=0.75]{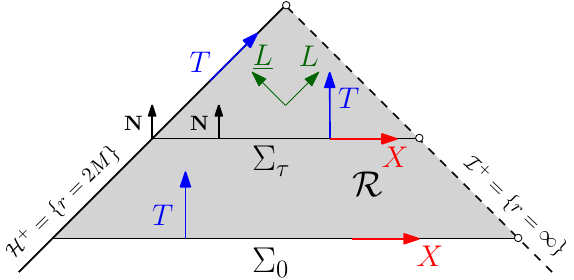}
\end{center}
\caption{A 2-dimensional representation of the spacetime $\mathcal{M}_{M}$, with the hypersurfaces $\Sigma_{\tau}$ and the shaded region depicting $\mathcal{R}$. Each point in the picture represents a round 2-sphere of radius $r$ and the hypersurface $\mathcal{I}^+$, which represents the points $(\tau,\infty,\theta,\varphi)$, is depicted at a finite distance.}
	\label{fig:penrose}
\end{figure}

\subsection{Properties of functions on $\s^2$}
\label{sec:sphere}
Let $Y_{\ell m}(\theta,\varphi)$ denote spherical harmonics on $\s^2$, labelled by an angular momentum $\ell\in \N_0$ and an azimuthal number $m\in \Z$, $|m|\leq \ell$. We can define projection operators $\pi_{\ell}: C^{\infty}(\Sigma_{\tau})\to C^{\infty}(\Sigma_{\tau})$ as follows:
\begin{equation*}
(\pi_{\ell}f)(r,\theta,\varphi):=\sum_{m=-\ell}^{\ell} \left[\int_{\s^2}f (r,\theta',\varphi')Y_{\ell m}(\theta',\varphi')\,\sin \theta' d\theta' d\varphi'\right] Y_{\ell m}(\theta,\varphi)
\end{equation*}
and extend $\pi_{\ell}$ as a bounded linear operator to $L^2(\Sigma_{\tau})$ (with respect to the standard volume form on $\Sigma_{\tau}$ given the induced metric on $\Sigma_{\tau}$).

We denote the $f_{\ell}=\pi_{\ell}f$ and refer to $f_{\ell}$ as a spherical harmonic mode. Since the functions $Y_{\ell m}$ make up an orthonormal basis on $L^2(\s^2)$, we can express:
\begin{equation*}
f=\sum_{\ell=0}^{\infty}f_{\ell}.
\end{equation*}
We moreover denote for $\ell_0\in \N_0$, $f_{\geq \ell_0}=\sum_{\ell=\ell_0}^{\infty}f_{\ell}$.

Define the angular momentum operators $\Omega_i$, $i=1,2,3$ as follows:
\begin{align*}
\Omega_1=&-\sin \varphi \partial_{\theta}-\cot\theta \cos\varphi \partial_{\varphi},\\
\Omega_2=&\:\cos \varphi\partial_{\theta}-\cot\theta \sin \varphi \partial_{\varphi},\\
\Omega_3=&\:\partial_{\varphi}.
\end{align*}
It is straightforward to show that $[\Omega_i,\slashed{\Delta}_{\s^2}]=0$, where $\slashed{\Delta}_{\s^2}(\cdot)=\frac{1}{\sin \theta}\partial_{\theta}(\sin \theta \partial_{\theta}(\cdot))+\frac{1}{\sin^2 \theta}\partial_{\varphi}^2(\cdot)$ denotes the Laplacian on the unit round sphere $\s^2$.

We denote with $H^k(\s^2)$ the $k$-th order $L^2$-based Sobolev spaces on $\s^2$ with respect to the volume form $d\sigma:=\sin \theta d\theta d\varphi$.
\begin{lemma}
Let $k\in \N$, $\ell\in \N_0$ and $f\in H^k(\s^2)$. Then:
\begin{align}
\label{eq:sphere1}
\int_{\s^2} |\snabla_{\s^2}f|^2\,d\sigma=&\: \sum_{|I|=1} \int_{\s^2} (\Omega^{I}f)^2\,d\sigma,\\
\label{eq:sphere1b}
\frac{1}{k!}  \sum_{|I|=k} \int_{\s^2} (\Omega^{I}f)^2\,d\sigma\leq \int_{\s^2} |\snabla_{\s^2}^kf|^2\,d\sigma\leq &\: \sum_{|I|=k} \int_{\s^2} (\Omega^{I}f)^2\,d\sigma,\\
\label{eq:sphere2}
\int_{\s^2} |\snabla_{\s^2}^kf_{\geq \ell}|^2\,d\sigma\geq&\: (\ell(\ell+1))^k \int_{\s^2}f_{\geq \ell}^2\,d\sigma,\\
\label{eq:sphere3}
\int_{\s^2} |\snabla_{\s^2}^kf_{\ell}|^2\,d\sigma=&\: (\ell(\ell+1))^k \int_{\s^2}f_{\ell}^2\,d\sigma,
\end{align}
where we used multi-index notation $\Omega^{I}=\Omega^{I_1}\ldots \Omega^{I_n}$ for $I=(I_1,\ldots,I_n)$.

Furthermore, if $f\in H^{2k}(\s^2)$, then
\begin{equation}
\label{eq:sphere4}
 \int_{\s^2} (\slashed{\Delta}_{\s^2}^kf)^2\,d\sigma=\int_{\s^2} |\snabla_{\s^2}^{2k}f|^2\,d\sigma.
\end{equation}
and if $f\in H^2(\s^2)$, then there exist a numerical constant $C>0$, such that
\begin{equation}
\label{eq:sphere5}
||f||_{L^{\infty}(\s^2)}\leq C||f||_{H^2(\s^2)}.
\end{equation}
\end{lemma}
\begin{proof}
The estimates \eqref{eq:sphere2}--\eqref{eq:sphere4} follow simply by using that $f=\sum_{\ell\in \N_0}f_{\ell}$, restricting to $f_{\ell}$, making use of the identity $\slashed{\Delta}_{\s^2}f_{\ell}=-\ell(\ell+1)f_{\ell}$ and integrating by parts.

We obtain \eqref{eq:sphere1} by applying the expressions for $\Omega_i$ above. For \eqref{eq:sphere1b}, it is convenient to restrict to $f_{\ell}$. Using \eqref{eq:sphere1}, we obtain:
\begin{equation*}
\sum_{|I|= k} \int_{\s^2} (\Omega^{I}f_{\ell})^2\,d\sigma\leq \sum_{|J_k|=1}\ldots \sum_{|J_1|= 1} \int_{\s^2} (\Omega^{J_k}\ldots \Omega^{J_1}f_{\ell})^2\,d\sigma=(\ell(\ell+1))^k \int_{\s^2}f_{\ell}^2\,d\sigma.
\end{equation*}
Similarly,
\begin{equation*}
(\ell(\ell+1))^k \int_{\s^2}f_{\ell}^2\,d\sigma= \sum_{|\beta_k|=1}\ldots \sum_{|J_1|= k} \int_{\s^2} (\Omega^{J_k}\Omega^{J_1}f_{\ell})^2\,d\sigma\leq k! \sum_{|I|= k} \int_{\s^2} (\Omega^{I}f_{\ell})^2\,d\sigma.
 \end{equation*}
We then apply \eqref{eq:sphere3} and sum over $\ell$  to obtain \eqref{eq:sphere1b}. The inequality \eqref{eq:sphere5} is a standard Sobolev inequality on $\s^2$, so the proof will be omitted. 
\end{proof}

\subsection{A Hardy inequality}
Throughout the paper, we will frequently make use of the following Hardy inequality:
\begin{lemma}[Hardy inequality]
Let $f\in C^1([a,b])$ with $a,b\in \R_{\geq 0}$ such that $a<b$. Then for $q\in \R\setminus \{-1\}$:
\begin{align}
\label{eq:hardy}
\int_a^b x^q |f|^2(x)\,dx\leq \frac{2}{q+1}\left[b^{q+1}|f|^2(b)-a^{q+1}|f|^2(a)\right]+\frac{4}{(q+1)^2}\int_a^b x^{q+2}\left|\frac{df}{dx}\right|^2(x)\,dx.
\end{align}
\end{lemma}
\begin{proof}
Integrate $\frac{d}{dx}(x^{q+1}|f|^2)$, apply the fundamental theorem of calculus and apply a (weighted) Young's inequality to estimate:
\begin{equation*}
\left|2 x^{q+1}f \frac{df}{dx}\right|\leq \frac{|q+1|}{2} x^q|f|^2+\frac{2}{|q+1|}x^{q+2}\left|\frac{df}{dx}\right|^2.
\end{equation*}
\end{proof}
\subsection{Additional notation}
We denote with $d\mu_{\tau}$ the standard volume form corresponding to the induced metric $i^*_{\Sigma_{\tau}}g_M$ along $\Sigma_{\tau}$, with $i_{\Sigma_{\tau}}: \Sigma_{\tau} \xhookrightarrow{}\mathcal{R}$ the inclusion map.

Throughout the paper we will frequently use $C$ or $c$ to denote positive constants depending only on $M,h$ and $V_{\alpha}$. In order to make the notation more compact, we will not keep precise track of constants and we will instead make use of the following ``algebra of constants'' when constants $c$ or $C$ appear in an equality:
\begin{equation*}
C+C=C\cdot C=C, c+c=c\cdot c=c.
\end{equation*}
We will moreover use the notations $f\lesssim g$ and $f\sim g$, with $f,g$ non-negative functions, to mean respectively:
\begin{equation*}
f\leq Cg,\quad c g\leq f \leq C g.
\end{equation*}

\section{Preliminaries: wave equations}
\label{sec:prelimwave}
In this section, we introduce the wave equations that are studied in the present paper, and we derive their global existence, uniqueness and regularity properties.
\subsection{Class of wave equations with potentials}
\label{sec:classeq}
Let $V_{\alpha}: [2M,\infty)\to \R$ be a smooth function, such that:
\begin{align*}
r^2V_{\alpha}(r)=&\:\alpha+O_{\infty}(r^{-1}),
\end{align*}
with $\alpha>-\frac{1}{4}$.

We additionally make the following assumption:
\begin{enumerate}[label=\Alph*)]
\item
\label{assm:A}
$V_{\alpha}$ satisfies the condition:
\begin{equation*}
V_{\alpha}(r)>-\frac{1}{4r^2}\quad \textnormal{for all $r\geq 2M$}.
\end{equation*}
\end{enumerate}
In fact, it is possible to consider the alternative assumption:
\begin{enumerate}[label=\Alph*')]
\item 
\label{assm:B}
The operator 
\begin{equation*}
\mathcal{L}_{\ell}=\frac{d}{dr}\left(Dr^2 \frac{d}{dr} (\cdot)\right)-[r^2V_{\alpha}+\ell(\ell+1)](\cdot)
\end{equation*}
has \emph{no zero-energy resonances} for any $\ell\in \N_0$: there do not exist smooth solutions $w_{\ell}$ to the ODE $\mathcal{L}_{\ell}w_{\ell}=0$ that satisfy outgoing boundary conditions, i.e.\ that behave asymptotically like the most rapidly decaying branch of ODE solutions as $r\to \infty$.
\end{enumerate}

Assumption \ref{assm:A} ensures coercivity of the energy corresponding to the vector field $T$; see Proposition \ref{prop:ebound}. In Lemma \ref{lm:propertiesw} below, it is shown that assumption \ref{assm:A} implies assumption \ref{assm:B}. For the sake of convenience, we will assume the stronger assumption \ref{assm:A} throughout the rest of the paper, as it allows us to employ more straightforward energy estimates.

In this article, we will focus on the equations:
\begin{equation}
\label{eq:waveeq}
\square_{g_M} \phi=V_{\alpha}\phi
\end{equation}
in $\mathcal{R}$, with $\square_{g_M}$ denoting the Laplace--Beltrami operator with respect to $g_M$.

We introduce the following rescaled function that will play an important role in the remainder of the article:
\begin{equation*}
\psi:=r \phi.
\end{equation*}
The limit $\psi|_{\mathcal{I}^+}$ is known as the \emph{(Friedlander) radiation field}.

In the lemma below we denote various different forms of the wave equation with potential \eqref{eq:waveeq}.
\begin{lemma}
Let $\phi$ be a solution to \eqref{eq:waveeq}. Then:
\begin{equation}
\label{eq:waveeq1}
0=X(Dr^2 X\phi)-2 r^2(1-h)XT\phi-r^2h\tilde{h}T^2\phi+\slashed{\Delta}_{\s^2}\phi-r^2V_{\alpha}\phi+\left(\frac{d}{dr}\left(r^2h\right)-2r\right)T\phi.
\end{equation}
Furthermore, \eqref{eq:waveeq} is equivalent to the following equations for $\psi$:
\begin{align}
\label{eq:waveeq2}
0=&\: X(DX\psi)-2(1-h)XT\psi-h\tilde{h}T^2\psi+r^{-2}\slashed{\Delta}_{\s^2}\psi+\frac{dh}{dr}T\psi-V_{\alpha}\psi-D' r^{-1}\psi,\\
\label{eq:waveeq2a}
0=&\: -4\underline{L} L\psi+Dr^{-2}\slashed{\Delta}_{\s^2}\psi-D(V_{\alpha}+D'r^{-1})\psi,
\end{align}
and we can express in $(\tau,x,\theta,\varphi)$ coordinates:
\begin{equation}
\label{eq:waveeq3}
0=\partial_x(Dx^2 \partial_x\psi)+2(1-h)\partial_x \partial_{\tau}\psi-r^2D^{-1}h(2-h)\partial_{\tau}^2\psi+\slashed{\Delta}_{\s^2}\psi+\frac{dh}{dr}r^2\partial_{\tau}\psi-r^2V_{\alpha}\psi-D' r \psi.
\end{equation}
\end{lemma}
\begin{proof}
The equations are straightforward computations, which follow from writing:
\begin{equation*}
\square_{g_M}(\cdot)=\frac{1}{\sqrt{-\det g_M}}\sum_{0\leq \alpha,\beta\leq 3}\partial_{\alpha}\left(\sqrt{-\det g_M}(g^{-1}_M)^{\alpha \beta} \partial_{\beta}(\cdot)\right)
\end{equation*}
and making use of the relations between the different vector fields in \eqref{eq:LbarXT}--\eqref{eq:defZ}. See for example\\ \cite{paper1}[Appendix A.1] for similar computations.
\end{proof}

\subsection{Global existence, uniqueness and regularity at null infinity}
We establish below global existence and uniqueness for the initial value problem for \eqref{eq:waveeq} with initial data on $\Sigma_0$. We moreover show that regularity of the initial data with respect to the conformal manifold $\widehat{\mathcal{R}}$ is propagated.
\begin{proposition}
Let $k\in \N_0$. 
\begin{enumerate}[label=\emph{(\roman*)}]
\item \label{item:existence1}
Consider the initial data functions $(\phi_i,\phi'_i)\in H^{k+1}_{\rm loc}(\Sigma_0)\times H^{k}_{\rm loc}(\Sigma_0)$.  There exists a unique solution solution $\phi$ to \eqref{eq:waveeq} in $\mathcal{R}$ (in a weak sense), such that for all $T_*\geq 0$:
\begin{align*}
\phi\in &\:C^0([0,T_*], H^{k+1}_{\rm loc}(\Sigma_0)) \cap C^1([0,T_*], H^{k}_{\rm loc}(\Sigma_0)),\\
\phi|_{\Sigma_0}=&\:\phi_i,\\
T\phi|_{\Sigma_0}=&\:\phi'_i.
\end{align*}
Here the Sobolev spaces are defined with respect to the standard volume form on $\Sigma_0$.
\item \label{item:existence2} If we moreover assume that $(r\phi_i,r\phi'_i)\in H^{k+1}(\widehat{\Sigma}_0)\times H^{k}(\widehat{\Sigma}_0)$. Then, for all $\tau\geq 0$,
\begin{equation}
\label{eq:reghoderphi}
(\psi|_{\Sigma_{\tau}},T\psi|_{\Sigma_{\tau}})\in  H^{k+1}(\widehat{\Sigma}_{\tau})\times H^{k}(\widehat{\Sigma}_{\tau}).
\end{equation}
Here the Sobolev spaces are defined with respect to the standard volume form on $\widehat{\Sigma}_0$.
\end{enumerate}
\end{proposition}
\begin{proof}
Part \emph{\ref{item:existence1}} follows from standard global existence and uniqueness arguments for linear wave equations, so we will only sketch the proof: 
\begin{enumerate}[label=\arabic*.]
\item We establish first an energy estimate in $\bigcup_{\tau \in [0,T_*]} \Sigma_{\tau}$, with constants depending on $T_*$, by integrating by parts $\mathbf{N}\phi\square_g\phi$ and then repeating the argument with higher-order energies, by integrating by parts $(-1)^k\mathbf{N}\slashed{\Delta}_{\s^2}^k X^j T^l\phi \square_g  X^j T^l\phi$. 

\item We apply a functional analytic argument, using the above energy estimates to establish the existence of weak solutions; see for example \S I.3 of \cite{sogge},
\item We improve the regularity of the solutions; see for example \S I.3 of \cite{sogge}.
\end{enumerate}

In order to obtain part \emph{\ref{item:existence2}}, we instead work with the equation \eqref{eq:waveeq3} and apply the arguments for global existence and uniqueness as above, using that  the coefficients in \eqref{eq:waveeq3} are sufficiently regular with respect to the differentiable structure on the manifold-with-boundary $\widehat{\mathcal{R}}$.
\end{proof}

\subsection{Definition of stationary weight functions}
\label{sec:weightfunctions}
The following number will play a fundamental role in the decay rates and late-time asymptotics of $\phi$:
\begin{equation*}
\beta_{\ell}=\sqrt{1+4\alpha+4\ell(\ell+1)}.
\end{equation*}
The number $\beta_{\ell}$ arises from the asymptotic behaviour of stationary solutions to \eqref{eq:waveeq}, as demonstrated in Lemma \ref{lm:propertiesw} below, where we derive the existence, uniqueneness and $r\to \infty$ asymptotics of stationary solutions.
\begin{lemma}
\label{lm:propertiesw}
Let $V_{\alpha}$ satisfy assumption \ref{assm:A}. There exists a unique smooth solution $w_{\ell}: [2M,\infty) \to (0,\infty)$ to
\begin{equation}
\label{eq:odew}
(Dr^2 w_{\ell}')'-[r^2V_{\alpha}+\ell(\ell+1)]w_{\ell}=0,
\end{equation}
such that 
\begin{equation}
\label{eq:asymptw}
w_{\ell}(r)= r^{-\frac{1}{2}+\frac{1}{2}\beta_{\ell}}+O_{\infty}(r^{-\frac{1}{2}-\frac{1}{2}\beta_{\ell}})+O_{\infty}(r^{-\frac{3}{2}+\frac{1}{2}\beta_{\ell}}).
\end{equation}
The functions $w_{\ell}\cdot Y_{\ell m}$ with $|m|\leq \ell$ are stationary solutions to \eqref{eq:waveeq} (i.e.\ solutions with vanishing $T$-derivative).
\end{lemma}
\begin{proof}
See Appendix \ref{app:weightfunction}.
\end{proof}
\begin{remark}
When $V_{\alpha}\equiv 0$, we have that $w_{\ell}(r)=P_{\ell}(r-M)$, with $P_{\ell}(y)$ the $\ell$-th Legendre polynomial in $y$. Furthermore, in this case $w_{\ell}(r)\sim r^{\ell}$ as $r\to \infty$.
\end{remark}
We introduce the following \emph{renormalization} of the spherical harmonic modes $\phi_{\ell}$:
\begin{equation*}
\check{\phi}_{\ell}= w_{\ell}^{-1} \phi_{\ell},
\end{equation*}
In the lemma below, we obtain equations for the renormalized functions $\check{\phi}_{\ell}$
\begin{lemma}
\label{lm:eqrenorm}
Let $\phi$ be a solution to \eqref{eq:waveeq}. Then $\check{\phi}_{\ell}$ satisfies the equations:
\begin{align}
\label{eq:checkpsi1}
0=&\:Z^2\check{\phi}_{\ell}+2(r^{-1}+Dw_{\ell}'w_{\ell}^{-1}) Z\check{\phi}_{\ell}-T^2\check{\phi}_{\ell},\\
\label{eq:checkpsi2}
0=&\:X(Dr^2w_{\ell}^2X\check{\phi}_{\ell})-2r^2w_{\ell}^2(1-h)XT\check{\phi}_{\ell}-h\tilde{h}r^2w_{\ell}^2T^2\check{\phi}_{\ell}-(w_{\ell}^2r^2(1-h))'T\check{\phi}_{\ell}.
\end{align}
\end{lemma}
\begin{proof}
Apply \eqref{eq:waveeq1} in combination with \eqref{eq:odew} to obtain \eqref{eq:checkpsi2}. Then \eqref{eq:checkpsi1} follows by using \eqref{eq:defZ}. \end{proof}

\begin{remark}
The choice of renormalization $\check{\phi}_{\ell}$ is made so that the equations satisfied by $\check{\phi}_{\ell}$ contain only $T$- and $X$-derivatives of $\check{\phi}_{\ell}$, i.e.\ they do not feature zeroth-order terms in $\check{\phi}_{\ell}$.
\end{remark}

\section{The main results and ideas}
\label{sec:main}
In this section, we give a precise formulation of the main results, and we outline the key new ideas that play a role in the corresponding proofs.
\subsection{Precise statements of the main results}
\label{sec:precise}

\paragraph{Relevant initial data quantities} In order to give precise statements of the main results, we introduce the following initial data quantities: recall that
\begin{equation*}
\beta_{\ell}=\sqrt{1+4\alpha+4\ell(\ell+1)}.
\end{equation*}
 Then the assumption that $\alpha>-\frac{1}{4}$ and $\alpha \neq l(l+1)$ for any $l\in \N_0$, implies that $\beta_{\ell} \neq 2l+1$ for any $l\in \N_0$. We moreover introduce:
 \begin{equation*}
 n_{\beta_{\ell}}=\left\lfloor \frac{\beta_{\ell}+1}{2}\right\rfloor, 
 \end{equation*}
and denote for $N\in \N_0$ and $\delta>0$:\footnote{Here, $d\tau^{\sharp}$ is the vector field dual to the 1-form $d\tau$, with respect to the metric $g_M$.}
\begin{align*}
\mathfrak{I}_{\beta_{0}}[\phi]=&\frac{1-\beta_{0}}{2\beta_{0}} \int_{\Sigma_0} \left[r(d\tau)^{\sharp}(r\phi)+\frac{dh}{dr} r^2\phi-r(1-h)X(r\phi) \right]w_{0} d\sigma dr,\\
\mathfrak{I}_{\beta_{\ell}}[\phi](\theta,\varphi)=&\frac{1-\beta_{\ell}}{2\beta_{\ell}} \sum_{|m|\leq \ell}\int_{\Sigma_0} \left[r (d\tau)^{\sharp }(r\phi)+\frac{dh}{dr} r^2\phi-r(1-h)X(r\phi) \right]\cdot Y_{\ell m}w_{\ell} d\sigma dr\cdot  Y_{\ell m}(\theta,\varphi),\\
\Phi_{\ell}[\phi]=&\:\frac{(-1)^{n_{\beta_{\ell}}}n_{\beta_{\ell}}!}{2^{n_{\beta_{\ell}}-\beta_{\ell}} \pi} \frac{\Gamma\left(\frac{\beta_{\ell}+1}{2}\right) \Gamma\left(\frac{\beta_{\ell}-1}{2}-n_{\beta_{\ell}}\right) }{\Gamma(\beta_{\ell}-n_{\beta_{\ell}})}\mathfrak{I}_{\beta_{\ell}}[\phi],\\
E_{\mathbf{N}}[\phi]=&\: \int_{\Sigma_{\tau}}r^2|X\phi|^2+h\tilde{h} r^2|T\phi|^2+|\snabla_{\s^2}\phi|^2+|\phi|^2\,d\sigma dr,\\
\mathbf{D}_{\beta_0,N,\delta}[\phi]= &\:\sum_{n\leq N+2}\left(||((r+1)^2(rX)^nX(r\phi_0)||_{L^{\infty}(\Sigma_0)}^2+||(rX)^n\phi_0||_{L^{\infty}(\Sigma_0)}^2+||r(rX)^nT\phi_0||_{L^{\infty}(\Sigma_0)}^2\right)\\
&+\sum_{n_1+n_2\leq N+2}\int_{\Sigma_0\cap\{r\geq R\}} r^{3-\delta} |\snabla_{\s^2}^{n_1}X(rX)^{n_2} \widehat{\psi}_{\geq 1}|^2+r^{-1-\delta} |\snabla_{\s^2}^{n_1}(rX)^{n_2} T\widehat{\psi}_{\geq 1}|^2\,d\sigma dr\\
&+ \sum_{\substack{0\leq n_1+n_2+n_3\leq 2N+3\\ n_1+n_2\leq N+2}}\int_{\Sigma_{0}\cap \{r\geq R\}}  r^{\min\{3\beta_0,1\}-\delta}|L(rL)^{n_2}T^{n_3}\widehat{\psi}_0|^2+r^{1-\delta}|\snabla_{\s^2}^{n_1}L(rL)^{n_2}T^{n_3-1}{\psi}_{\geq 1}|^2\,d\sigma dr\\
&+ \sum_{\substack{0\leq n_1+n_2\leq 3N+5\\ n_1\leq N+2}}E_{\mathbf{N}}[\snabla_{\s^2}^{n_1}T^{n_2}\widehat{\phi}]+\mathfrak{I}_{\beta_0}^2[\phi]\qquad (\beta_0<1),\\
\mathbf{D}_{\beta_{\ell},N,\delta}[\phi_{\geq \ell}]=&\: \sum_{n\leq N+3n_{\beta_{\ell}}+2}\Biggl( ||(r+1)^{\frac{3}{2}+\frac{\beta_\ell}{2}}(rX)^nX(r\phi_{\ell})||_{L^{\infty}(\Sigma_0)}^2\\
&+||r^{\frac{1}{2}+\frac{\beta_{\ell}}{2}}(rX)^n\phi_{\ell}||_{L^{\infty}(\Sigma_0)}^2+||r^{\frac{3}{2}+\frac{\beta_{\ell}}{2}}(rX)^nT\phi_{\ell}||_{L^{\infty}(\Sigma_0)}^2\Biggr)\\
&+\sum_{n_1+n_2\leq N+2n_{\beta}+2}\int_{\Sigma_{0}\cap \{r\geq R\}} r^{2+\min\{2n_{\beta_{\ell}}+1,\beta_{\ell+1}\}-\delta}|\snabla_{\s^2}^{n_1}X(rX)^{n_2}\widehat{\psi}_{\geq \ell+1}|^2\\
 &+r^{-2+\min\{2n_{\beta_{\ell}}+1,\beta_{\ell+1}\}-\delta}|\snabla_{\s^2}^{n_1}X(rX)^{n_2}T\widehat{\psi}_{\geq \ell+1}|^2\,d\sigma dr\\
&+\sum_{\substack{0\leq n_1+n_2+n_3\leq 2N+3n_{\beta_{\ell}}+5\\ n_1+n_2\leq N+2n_{\beta_{\ell}}+2}}\int_{\Sigma_{0}\cap \{r\geq R\}} r|\snabla_{\s^2}^{n_1}L(rL)^{n_2}T^{n_3}\widehat{\psi}_{\ell}|^2\\
&+ r^{\min\{\beta_{\ell+1}-2n_{\beta_{\ell}},1\}-\delta}|\snabla_{\s^2}^{n_1}L(rL)^{n_2}T^{n_3}{\psi}_{\geq \ell+1}|^2+r^{2+\min\{\beta_{\ell+1}-2n_{\beta_{\ell}},0\}-\delta}|\snabla_{\s^2}^{n_1}L(rL)^{n_2}T^{n_3}{\psi}_{\geq \ell+1}|^2\,d\sigma dr\\
&+ \sum_{\substack{0\leq n_1+n_2\leq 3N+5n_{\beta_{\ell}}+5\\ n_1\leq N+2n_{\beta_{\ell}}+2}}E_{\mathbf{N}}[\snabla_{\s^2}^{n_1}T^{n_2}\widehat{\phi}_{\geq \ell}]+\mathfrak{I}_{\beta_{\ell}}^2[\phi_{\geq \ell}]\qquad (\beta_0>1\quad \textnormal{or}\quad \ell\geq 1)\quad \textnormal{where},\\
\widehat{\phi}_{\geq \ell}(\tau,r,\theta,\varphi)=&\:\phi_{\geq \ell}(\tau,r,\theta,\varphi)-\Phi^{(\beta_{\ell})}(\tau,r,\theta,\varphi)v\quad \textnormal{with}\\
\Phi^{(\beta_{\ell})}(\tau,r,\theta,\varphi)=&\:\Phi_{\ell} [\phi](\theta,\varphi)w_{\ell} (r)(\tau+2r)^{-\frac{1}{2}-\frac{\beta_{\ell}}{2}}(\tau+M)^{-\frac{1}{2}-\frac{\beta_{\ell}}{2}}.
\end{align*}

\paragraph{Main theorem} We can now give precise formulations of the main results.
\begin{theorem}
\label{thm:precise}
Let $\delta>0$ be arbitrarily small, $N\in \N_0$ and let $\phi$ be a solution to \eqref{eq:waveeq} with initial data satisfying $\sum_{0\leq k\leq 2}\mathbf{D}_{\beta_0,N,\delta}[\snabla_{\s^2}^k\phi]<~\infty$.
\begin{enumerate}[label=\emph{(\roman*)}]
\item \label{item:mainthm1} \emph{(Leading-order late-time asymptotics)} There exists a constant $C=C(M,h,V_{\alpha},N,\delta)>0$ such that:
\begin{align}
\label{eq:thmasympbg1}
&||T^N\phi- \Phi_0[\phi] w_0(r)(T^Nf)(\tau,r)||_{L^{\infty}(S^2_{\tau,r})}\\ \nonumber
\leq&\: Cr^{-\frac{1}{2}+\frac{1}{2}\beta_0}f(\tau,r)(\tau+M)^{-\frac{1}{2}\min\{2\beta_0, 1-\beta_0+n_{\beta_0},\beta_1-\beta_0\}-N+\delta}\sqrt{\sum_{0\leq k\leq 2}\mathbf{D}_{\beta_0,N,\delta}[\snabla_{\s^2}^k\phi]}, \\
&f(\tau,r)=(\tau+2r)^{-\frac{1}{2}-\frac{1}{2}\beta_0}(\tau+M)^{-\frac{1}{2}-\frac{1}{2}\beta_0}.
\end{align}
\item \label{item:mainthm2} \emph{(Refined late-time asymptotics for higher spherical harmonics)} Let $\phi_{\geq \ell}$ be a solution to \eqref{eq:waveeq} supported on spherical harmonics of degree $\geq \ell$, with initial data satisfying $\sum_{0\leq k\leq 2}\mathbf{D}_{\beta_{\ell},N,\delta}[\snabla_{\s^2}^k\phi_{\geq \ell}]<\infty$. Then there exists a constant $C=C(M,h,V_{\alpha},N,\delta,\ell)>0$ such that:
\begin{align}
\label{eq:thmasympbgell}
&||T^N\phi_{\geq \ell}- \Phi_{\ell}[\phi_{\geq \ell}] w_{\ell}(r)(T^N f_{\ell})(\tau,r)||_{L^{\infty}(S^2_{\tau,r})}\\ \nonumber
\leq&\: Cr^{-\frac{1}{2}+\frac{1}{2}\beta_{\ell}}f_{\ell}(\tau,r)(\tau+M)^{-\min\{n_{\beta_{\ell}}-\frac{1}{2}(\beta_{\ell}-1),\frac{1}{2}(\beta_{\ell+1}-\beta_{\ell})\}-N+\delta}\sqrt{\sum_{0\leq k\leq 2}\mathbf{D}_{\beta_{\ell},N,\delta}[\snabla_{\s^2}^k\phi_{\geq \ell}]},\\
&f_{\ell}(\tau,r)=(\tau+2r)^{-\frac{1}{2}-\frac{1}{2}\beta_{\ell}}(\tau+M)^{-\frac{1}{2}-\frac{1}{2}\beta_{\ell}}.
\end{align}

\item \label{item:mainthm3} \emph{(Genericity of late-time tails)} Let $(\phi_i,\phi'_i)=(\phi,T\phi)|_{\Sigma_0}$ denote initial data for $\phi$ on $\Sigma_0$ and consider:
\begin{equation*}
\mathcal{D}=\left\{((\phi_i)_{\geq \ell},(\phi_i)_{\geq \ell}')\: \textnormal{measurable}\,|\, \sum_{0\leq k\leq 2}\mathbf{D}_{\beta_{\ell},0,\delta}[\snabla_{\s^2}^k\phi_{\geq \ell}]<\infty\right\}.
\end{equation*}
Then the subset
\begin{equation*}
\mathcal{G}=\left\{((\phi_i)_{\geq \ell},(\phi_i)_{\geq \ell}') \in \mathcal{D}\,|\, \mathfrak{I}_{\beta_{\ell}}[\phi] \neq 0\right\}
\end{equation*}
is open and dense with respect to the weighted $L^1$-type norm:
\begin{equation*}
||((\phi_i)_{\geq \ell},(\phi_i)_{\geq \ell}')||_{\mathcal{D}}=\int_{\Sigma_0}\left[r|X(r(\phi_i)_{\geq \ell})|+r^{-1}|(\phi_i)_{\geq \ell}|+|(\phi'_i)_{\geq \ell}|\right]w_{\ell}\,d\sigma dr
\end{equation*}
and the complement $\mathcal{G}^c$ has codimension $2\ell+1$. The above statements holds also if we replace $\mathcal{D}$ above with the smaller data set $C_c^{\infty}(\Sigma_0)\times C_c^{\infty}(\Sigma_0)$.

The estimates \eqref{eq:thmasympbg1} and \eqref{eq:thmasympbgell} therefore give the precise leading-order late-time behaviour of $\phi_{\geq \ell}$ with $\ell\in \N_0$, proving existence of inverse-polynomial late-time tails, if the initial data lie in $\mathcal{G}$, which is generic in the above sense.
\item \label{item:mainthm4} \emph{(Data of time integrals)} For each $\ell\in \N_0$, the $(n_{\beta_{\ell}}+1)$-th order time integrals
\begin{equation*}
\begin{split}
T^{-1-n_{\beta_{\ell}}}\phi_{\geq \ell}(\tau,r,\theta,\varphi)=&(-1)^{1+n_{\beta_{\ell}}}\int_{\tau}^{\infty}\int_{\tau_1}^{\infty}\ldots \int_{\tau_{n_{\beta_{\ell}}}}^{\infty}\phi_{\geq \ell}(s,r,\theta,\varphi)\,ds d\tau_{n_{\beta_{\ell}}}\ldots d\tau_1\\
=&\:- \frac{1}{n_{\beta_{\ell}}!}\int_{\tau}^{\infty} (\tau-s)^{n_{\beta_{\ell}}}\phi(s,r,\theta,\varphi)\,ds.
\end{split}
\end{equation*}
are well-defined and satisfy:
\begin{equation*}
\lim_{r \to \infty} r^{\frac{1}{2}+\frac{1}{2}\beta_{\ell}-n_{\beta_{\ell}}}X(rT^{-1-n_{\beta_{\ell}}}\phi_{\ell})(0,r,\theta,\varphi)=\mathfrak{I}_{\beta_{\ell}}[\phi](\theta,\varphi).
\end{equation*}
\end{enumerate}
\end{theorem}

Part \emph{\ref{item:mainthm1}} of Theorem \ref{thm:precise} is proved as part of Proposition \ref{prop:mainlinftyest} and part \emph{\ref{item:mainthm2}} is proved as part of Proposition \ref{prop:mainlinftyestell}. We conclude part \emph{\ref{item:mainthm3}} by using that $\mathfrak{I}_{\beta_{\ell}}[\phi_{\ell m}]$ with $|m|\leq \ell$, defined via integrals on $\Sigma_0$, are well-defined as linearly independent bounded linear functionals on the initial data space $\mathcal{D}$ with respect to the norm $||(\cdot,\cdot)||_{\mathcal{D}}$. Their vanishing immediately implies that $\mathcal{G}^c$ has codimension $2\ell+1$ and is closed. Density follows easily by letting $(\phi_i,\phi_i')\in \mathcal{G}^c$ and taking a sequence of smooth compactly supported initial data $(\phi_i^{(k)},{\phi'}_i^{(k)})$ such that $\mathfrak{I}_{\beta_{\ell}}[\phi^{(k)}_{\geq \ell}]\neq 0$ and $||((\phi_i^{(k)})_{\geq \ell},({\phi'}_i^{(k)})_{\geq \ell})||_{\mathcal{D}}\to 0$, and then considering $(\phi_i+\phi_i^{(k)},\phi_i+{\phi'}_i^{(k)})$ with $(\phi_i,\phi'_i)\in \mathcal{G}^c$. Part \emph{\ref{item:mainthm4}} follows from a combination of Proposition \ref{prop:preciselargertimeint} and Corollary \ref{cor:equivdeftimeintell}.

\begin{remark}
\label{rm:consvlaws}
Although the proof of Theorem \ref{thm:precise} does not make use of conserved charges along $\mathcal{I}^+$, one can \emph{a posteriori} identify the following conserved quantities in $\tau$ along $\mathcal{I}^+$ at the level of (multiple) time integrals:
\begin{align*}
r^{\frac{1}{2}+\frac{1}{2}\beta_{\ell}}X(rT^{-1-n_{\beta_{\ell}}}\phi_{\ell})|_{\mathcal{I}^+}(\tau,\theta,\varphi)=&\:\mathfrak{I}_{\beta_{\ell}}[\phi](\theta,\varphi)\quad \textnormal{if $\beta_{\ell}<3$},\\
r^{\frac{1}{2}+\frac{1}{2}\beta_{\ell}- n_{\beta_{\ell}}}X((r^2X)^{n_{\beta_{\ell}}}rT^{-1-n_{\beta_{\ell}}}\phi_{\ell})|_{\mathcal{I}^+}(\tau,\theta,\varphi)=&\:2^{-n_{\beta_{\ell}} +1}(3-\beta_{\ell})\cdot \ldots\cdot (2n_{\beta_{\ell}}-1-\beta_{\ell})\mathfrak{I}_{\beta_{\ell}}[\phi](\theta,\varphi)\quad \textnormal{if $\beta_{\ell}>3$}.
\end{align*}
The proof of the above conservation property is omitted in the present paper, but it follows in a straightforward manner by taking the $T$-derivative of the expressions on the left-hand side above and applying \eqref{eq:waveeq} together with the decay-in-time estimates of Theorem \ref{thm:precise}.
\end{remark}

\begin{remark}
The methods leading to a proof of Theorem \ref{thm:precise} also apply to the setting of initial data with slower decay towards future null infinity, which can lower the conformal regularity of time integrals and lead to late-time tails with \underline{slower} decay. For examples illustrating this fact in the setting of \eqref{eq:wavewithoutpot}, see \cite{keh21b,keh21c,gk22}. 

Conversely, for initial data in $\mathcal{G}^c$, the time integrals have higher conformal regularity at infinity and the methods will still apply, provided we consider additional time integrals, which translate into inverse polynomial late-time tails with \underline{faster} decay. It is therefore expected that the subset of initial data which does not lead to \emph{any} inverse polynomial late-time tail in the late-time asymptotics is in fact of infinite codimension.
\end{remark}

\subsection{Overview of proof and summary of new ideas}
\label{sec:proofideas}
In this section, we list the main steps of the proof of Theorem \ref{thm:precise} and we give a sketch of the relevant new ideas and tools that are developed in the present paper.
\paragraph{\emph{\underline{Step 0:} Energy boundedness}}
Given vector fields $V,W$, we denote with
\begin{equation*}
\mathbb{T}[\phi](V,W)=V(\phi)W(\phi)-\frac{1}{2}g_M(V,W)\left[(g^{-1}_M)^{\alpha \beta}\partial_{\alpha}\phi\partial_{\beta}\phi+V_{\alpha}\phi^2\right]
\end{equation*}
the energy-momentum tensor corresponding to the equation \eqref{eq:waveeq}.

It follows from integrating $T\phi(\square_g-V_{\alpha})\phi$ by parts in a suitable spacetime slab and using assumption \ref{assm:A} on the potential that the following energy quantity that degenerates at the event horizon $r=2M$ is coercive and uniformly bounded in time:
\begin{equation*}
\int_{\Sigma_{\tau}} r^2\mathbb{T}(T,\mathbf{n}_{\tau})[\phi]\,d\sigma dr \sim \int_{\Sigma_{\tau}}\left(1-\frac{2M}{r}\right)(X(r\phi))^2+(T\phi)^2+|\snabla_{\s^2}\phi|^2 +\phi^2\,d\sigma dr.
\end{equation*}
To remove the degenerate factor $1-\frac{2M}{r}$ on the right-hand side above, we apply the red-shift estimates introduced in \cite{redshift}. \textbf{Step 0} is carried out in detail in \S \ref{sec:boundedness}.

\paragraph{\emph{\underline{Step 1:} Integrated energy decay estimates}}
The first significant step towards passing from energy boundedness to energy decay estimates in $\tau$ is to establish appropriate control of the integral in time of appropriate energy quantities, i.e.\ we wish to control
\begin{equation*}
\int_{0}^{\infty} \int_{\Sigma_{\tau}}r^2\mathbb{T}(\mathbf{N},\mathbf{n}_{\tau})[\phi]+\phi^2\,d\sigma dr d\tau\sim \int_{0}^{\infty} \int_{\Sigma_{\tau}} (X(r\phi))^2+(T\phi)^2+|\snabla_{\s^2}\phi|^2 +\phi^2\,d\sigma dr d\tau.
\end{equation*}
\begin{itemize}
\item \textbf{Differences with $\square_{g_M}\phi=0$}

When taking $V_{\alpha}\equiv 0$ in \eqref{eq:waveeq}, it is possible to control the above quantity locally as follows for any $R>2M$:
\begin{equation}
\label{eq:ileda0}
\int_{0}^{\infty} \int_{\Sigma_{\tau}\cap \{r\leq R\}} r^2\mathbb{T}(\mathbf{N},\mathbf{n}_{\tau})[\phi]+\phi^2\,d\sigma dr d\tau \lesssim \int_{\Sigma_0} r^2\mathbb{T}(\mathbf{N},\mathbf{n}_{0})[\phi]+r^2\mathbb{T}(\mathbf{N},\mathbf{n}_{0})[T\phi]\,d\sigma dr,
\end{equation}
where the extra $T$-derivative on the right-hand side is necessary due to the presence of trapped null geodesics at the photon sphere at $r=3M$ (see for example \cite{janpaper}). While it is possible to prove \eqref{eq:ileda0} for $\alpha\geq 0$, the $\alpha<0$ is significantly different. In fact, it follows from Theorem \ref{thm:precise} that a uniform estimate of the form \eqref{eq:ileda0} cannot possibly hold for $-\frac{1}{4}<\alpha\leq -\frac{3}{16}$, or equivalently, $0<\beta_0\leq \frac{1}{2}$.

To see this, we consider for the sake of simplicity spherically symmetric, smooth and compactly supported data on $\Sigma_0$ and take $-\frac{1}{4}<\alpha\leq -\frac{3}{16}$. Then, by \emph{\ref{item:mainthm4}} of Theorem \ref{thm:precise}, $r^{\frac{1}{2}+\frac{1}{2}\beta_{0}}X(rT^{-1}\phi)|_{\Sigma_0}$ attains a finite, non-zero limit as $r\to \infty$, which is equal to $\mathfrak{I}_{\beta_0}[\phi]$. In particular,
\begin{align}
\label{eq:contradictioniniten}
\int_{\Sigma_0\cap \{r\leq R\}} r^2\mathbb{T}(\mathbf{N},\mathbf{n}_{0})[T^{-1}\phi]+r^2\mathbb{T}(\mathbf{N},\mathbf{n}_{0})[T(T^{-1}\phi)]\,d\sigma dr<&\:\infty.
\end{align}

Suppose now that \eqref{eq:ileda0} holds also for $\phi$ arising from initial data in the closure of smooth compactly supported functions under the energy norm on the right-hand side of \eqref{eq:ileda0}. Then, by a density argument, using \eqref{eq:contradictioniniten}, the estimate \eqref{eq:ileda0} would also hold with $\phi$ replaced by $T^{-1}\phi$, so we would in particular have that
\begin{equation*}
\int_{0}^{\infty}(T^{-1}\phi)^2|_{r=r_0}\,d\tau<\infty
\end{equation*}
for any $2M\leq r_0<\infty$.\footnote{We control the integral of $\phi^2$ along $r=r_0$ by integrating $\partial_r(\chi \phi^2)$ from $r=r_0$ to $r=r_0+M$, with $\chi$ a smooth cut-off function which vanishes for $r\geq r_0+M$ such that $\chi(r_0)=1$, and applying \eqref{eq:ileda0}. } However, integrating the late-time asymptotics in \emph{\ref{item:mainthm1}} of Theorem \ref{thm:precise} in time to obtain late-time asymptotics for $T^{-1}\phi$, we also have that
\begin{equation*}
 T^{-1}\phi|_{r=r_0}= -\frac{1}{\beta_0}\Phi_0[\phi]w_0(r_0) (\tau+M)^{-\beta_0}+O((\tau+M)^{-\beta_0-\nu})
 \end{equation*}
 for some $\nu>0$ and therefore
\begin{equation*}
\int_{0}^{\infty}(T^{-1}\phi)^2|_{r=r_0}\,d\tau\gtrsim \mathfrak{I}_{\beta_0}[\phi] \int_{0}^{\infty}(\tau+M)^{-2\beta_0}\,d\tau=\infty \quad \textnormal{when $\beta_0\leq \frac{1}{2}$},
\end{equation*}
which is a contradiction. The above failure of \eqref{eq:ileda0} to hold is related to the fact that the $V_{\alpha}$-term in \eqref{eq:waveeq} cannot be dealt with by simply treating it as an inhomogeneity in the derivation of \eqref{eq:ileda0} for $V_{\alpha}\equiv 0$.

\textbf{We avoid the contradiction above by adding growing $r$-weights to the right-hand side of \eqref{eq:ileda0}:} when $\alpha<0$ (or equivalently, $\beta_0<1$), we prove in particular the inequality:
\begin{equation}
\label{eq:iledaneg}
\int_{0}^{\infty} \int_{\Sigma_{\tau}\cap \{r\leq R\}} \phi^2+r^2\mathbb{T}(\mathbf{N},\mathbf{n}_{\tau})[\phi]\,d\sigma dr d\tau\lesssim \int_{\Sigma_0} r^{p}(X(r\phi))^2+r^{2}\mathbb{T}(\mathbf{N},\mathbf{n}_{0})[\phi]+r^2\mathbb{T}(\mathbf{N},\mathbf{n}_{0})[T\phi]\,d\sigma dr,
\end{equation}
with $p>1-\beta_0$. One may easily verify that the right-hand side of \eqref{eq:iledaneg} blows up for the time integral $T^{-1}\phi$ that was introduced in the contradiction argument above, so the numerology of the decay rates in the late-time tails for $T^{-1}\phi$ that follows Theorem \ref{thm:precise} does not lead to a contradiction with \eqref{eq:iledaneg}. For the decay mechanism in \textbf{Step 2} below, it will in fact be sufficient to apply the above estimate with $p=1$.

\item \textbf{A robust physical space method for deriving integrated energy decay}

To prove \eqref{eq:iledaneg}, we use a new argument that is sufficiently robust so as to apply to the general class of potentials $V_{\alpha}$. This argument provides moreover an alternative to existing methods of deriving \eqref{eq:ileda0} when $V_{\alpha}\equiv 0$. We decompose $\phi$ into spherical harmonic modes $\phi_{\ell}$ and consider two cases:
\begin{enumerate}[label=\Alph*)]
\item \label{item:introalphacaseA} $\alpha\leq 0$ and $\ell=0$,
\item \label{item:introalphacaseB} $\alpha>0$ or $\ell\geq 1$.
\end{enumerate}
We obtain an integrated energy estimate in the case \ref{item:introalphacaseA} by considering the renormalized quantity
\begin{equation*}
\check{\phi}_0=w_0^{-1}\phi_0,
\end{equation*}
which as shown in Lemma \ref{lm:eqrenorm} satisfies a wave equation \underline{with no zeroth order terms}. 

In view of the late-time asymptotics of Theorem \ref{thm:precise}, this quantity will have the property that its derivatives, both in the $T$ and $X$ directions, will decay \emph{faster} than $\check{\phi}_0$ itself. This is reflected in the fact that we are able to obtain the following integrated energy estimate:
\begin{equation}
\label{eq:iledcheck0}
\int_{0}^{\infty} \int_{\Sigma_{\tau}\cap \{r\leq R\}} (T\check{\phi}_0)^2+(X\check{\phi}_0)^2\,d\sigma dr d\tau \lesssim \int_{\Sigma_0} r^2\mathbb{T}(N,\mathbf{n}_{0})[\phi]\,d\sigma dr,
\end{equation}
which, in contrast with \eqref{eq:iledaneg}, has \underline{no} extra term with growing $r$-weights on the right-hand side. We obtain \eqref{eq:iledcheck0} by multiplying both sides of \eqref{eq:checkpsi1} with the vector field multiplier $r^{-q}Z\check{\phi}_0$, where $q>0$ is taken suitably large. The estimate \eqref{eq:iledcheck0} can then be applied in combination with $r$-weighted energy estimates (see Step 2 below) to conclude \eqref{eq:iledaneg}.

In case \ref{item:introalphacaseB}, there is no need to \emph{couple} the integrated energy estimates for $\check{\psi}_{0}$ with the $r$-weighted energy estimates; it is in fact possible to obtain integrated energy estimates involving only the $T$-energy norm at the level of initial data. 

We first split up the spherical harmonics as follows:
\begin{enumerate}[label=(\roman*)]
\item \label{item:introalphacaseB1} $\ell\leq \ell_0$,
\item \label{item:introalphacaseB2} $\ell>\ell_0$,
\end{enumerate}
with $\ell_0$ sufficiently large. In case \ref{item:introalphacaseB2}, the angular derivatives dominate the terms involving $V_{\alpha}$ by the Poincar\'e inequality on $\s^2$, so standard Morawetz estimates can be used with a straightforward vector field multiplier: we multiply both sides of \eqref{eq:waveeq2a} with $f(r)Z\psi+\frac{D}{2}\frac{df}{dr}\psi$, where $f(r)=-r^{-1}(r-3M)$ and integrate by parts. It is only in this case that the presence of trapped null geodesics at $r=3M$ plays a role. 

In case \ref{item:introalphacaseB1}, we consider again the renormalized quantities
\begin{equation*}
\check{\phi}_{\ell}=w_{\ell}^{-1}\phi_{\ell}.
\end{equation*} 
As in case \ref{item:introalphacaseA} above, these are quantities that will have both $T$ and $X$ derivatives that decay \emph{faster} in time than the quantity itself. In particular, we can show that
\begin{equation}
\label{eq:iledcheckel}
\int_{0}^{\infty} \int_{\Sigma_{\tau}\cap \{r\leq R\}} (T\check{\phi}_{\ell})^2+(X\check{\phi}_{\ell})^2\,d\sigma dr d\tau\leq C_{\ell}\int_{\Sigma_0} r^2\mathbb{T}(T,\mathbf{n}_{0})[\phi_{\ell}]\,d\sigma dr
\end{equation}
by multiplying both sides of \eqref{eq:checkpsi1} with the vector field multiplier $r^{-q}Z\check{\phi}_{\ell}$, with $q>0$ sufficiently large.

Finally, we combine \eqref{eq:iledcheckel} with standard Morawetz/integrated energy decay estimates involving straightforward multipliers in a region $r\geq R$, where $R>2M$ is sufficiently large: that is to say, we simply multiply \eqref{eq:waveeq2a} with $-Z\psi$ and integrate by parts in $\{r\geq R\}$, using that the boundary terms along $r=R$ are already controlled via \eqref{eq:iledcheckel}.

We note that in all cases, \textbf{the consideration of the renormalized quantities $\check{\phi}_{\ell}$ allows us to circumvent the need for delicately chosen vector field multipliers}. So even in the $V_{\alpha}\equiv 0$ setting, this approach provides a cleaner alternative to previous proofs of integrated energy decay estimates on spherically symmetric black hole spacetimes.

In previous, standard approaches for deriving integrated energy decay estimates (in the $V_{\alpha}\equiv 0$ setting), the case $0<\ell\leq \ell_0$ was treated by considering $\psi_{\ell}$ in combination with more delicate, $\ell$-dependent choices of multipliers $f_{\ell}(r)Z\psi_{\ell}$. One may contrast this with the approach in the present paper outlined above, where we instead consider $\check{\phi}_{\ell}$ in combination with the simple multipliers of the form $r^{-q}Z\check{\phi}_{\ell}$ and where the $\ell$-dependence, as well as the properties of the potential $V_{\alpha}$, are already dealt with via the estimates for $w_{\ell}$.

\textbf{Step 1} is carried out in detail in \S \ref{sec:iled}.
\end{itemize}

\paragraph{\emph{\underline{Step 2:} A decay mechanism for time derivatives via a hierarchy of $r$-weighted energy estimates}}
Equipped with an integrated energy estimate (restricted to $r\leq R$), we will obtain energy decay in time by considering a hierarchy of $r$-weighted energy estimates in the region $r\geq R$, for suitably large $R$. These take the schematic form:

\begin{equation}
\label{eq:rpschem}
\begin{split}
\int_{\Sigma_{\tau_2}\cap \{r\geq R\}} r^{p} (L\psi)^2\,d\sigma dr+\:&\int_{\tau_1}^{\tau_2}\left[ \int_{\Sigma_{\tau}\cap \{r\geq R\}} r^{p-1} (L\psi)^2+(2-p)r^{p-3}(\phi^2+|\snabla_{\s^2}\phi|^2)\,d\sigma dr\right] d\tau\\
\lesssim &\: \int_{\Sigma_{\tau_1}\cap \{r\geq R\}} r^{p} (L\psi)^2\,d\sigma dr+\ldots .
\end{split}
\end{equation}

In the case $\alpha\geq 0$, \eqref{eq:rpschem} holds for $0<p\leq 2$ and it first appeared in \cite{newmethod}. These $r$-weighted energy estimates can be translated into energy decay in time via the following principle: restrict to $p\geq 1$, then
\begin{equation*}
\textbf{energy decay rate in time}\:=\:\mathbf{1}+\textbf{range of $p$}.
\end{equation*}
For example, if $p\in [1,2]$, then the energy decay rate will be 2. This principle follows from an application of the mean-value theorem along an appropriate dyadic sequence of times (``the pigeonhole principle'').

When $\alpha<0$, however, the $V_{\alpha}$-term contributes to leading-order with a bad sign and, after applying a Hardy inequality, \eqref{eq:rpschem} is only valid for the smaller range $1-\beta_0<p<1+\beta_0$.  As $\beta_0\downarrow 0$, the range of allowed $p$ therefore decreases to 0, so according to the principle above we are only above to prove energy decay with a rate smaller than $1+\beta_0$. To obtain the desired sharp decay in time, \textbf{it will turn out to be sufficient to consider just \eqref{eq:rpschem} with $p=1$}, provided we also commute with $T$, $rL$ and angular derivatives.

Indeed, we apply \eqref{eq:waveeq2a} together with $T=L+\underline{L}$ to relate the $T$-derivative to the remaining derivatives as follows: for $r>R$
\begin{equation*}
r^{p+1}(LT\psi)^2\lesssim \sum_{n_1+n_2\leq 1}r^{p-1}|\snabla_{\s^2}^{n_1}L (rL)^{n_2}\psi|^2+r^{p-1}|\snabla_{\s^2}^{n_1+1}(rL)^{n_2}\psi|^2.
\end{equation*}

Via the above, we can in fact control
\begin{equation*}
\int_{\tau_1}^{\tau_2}\int_{\Sigma_{\tau}\cap \{r\geq R\}} r^{p+1} (LT\psi)^2\,d\sigma dr d\tau,
\end{equation*}
using that \eqref{eq:rpschem} also holds under commutation with angular derivatives and $rL$. \emph{We therefore effectively gain of two powers in $r$, at the expense of replacing $\psi$ by $T\psi$, which translates into two extra powers in the energy decay rate in time according to the principle above.} The existence of longer $r$-weighted hierarchies for $T\psi$ after commutation with $rL$ were first exploited in \cite{volker1,moschidis1} and played an important role in \cite{paper1,aagkerr,aagprice}. In the present paper, we additionally need to apply an interpolation argument to compensate for the fact that for every commutation with $T$ we only gain \underline{one} extra estimate in the $r$-weighted hierarchy, but we want to gain \underline{two} powers in the corresponding energy decay rate. We derive energy decay with the rate $\tau^{-3}$ for $T\phi$ and, by repeating this strategy, energy decay of $T^N\phi$ with the rate $\tau^{-1-2N}$, and we need to consider energy norms involving in particular higher-order $rL$-derivatives of $\psi$ itself.

\textbf{Step 2} is carried out in detail in \S \S\ref{sec:rp}--\ref{sec:decaytimeder}.

\paragraph{\emph{\underline{Step 3:} Sharp decay via time integrals and elliptic estimates}}
While \textbf{Step 2} allows us to obtain arbitrarily high polynomial decay rates in time (given sufficiently regular initial data), we are only able to make such statements for $T^N\phi$ with appropriately high $N$ and not for the original solution $\phi$. To remedy this, we will act (formally) with the operator $T^{-N}$. In other words, we consider the time integrals:
\begin{equation*}
\begin{split}
T^{-N}\phi(\tau,r,\theta,\varphi)=&(-1)^{N}\int_{\tau}^{\infty}\int_{\tau_1}^{\infty}\ldots \int_{\tau_{N-1}}^{\infty}\phi(s,r,\theta,\varphi)\,ds d\tau_{N-1}\ldots d\tau_1\\
=&\:- \frac{1}{(N-1)!}\int_{\tau}^{\infty} (\tau-s)^{N-1}\phi(s,r,\theta,\varphi)\,ds,
\end{split}
\end{equation*}
which appear also in Theorem \ref{thm:precise}. The objects $T^{-N}\phi$ are well-defined if $\phi$ decays suitably fast in time. Furthermore, $T(T^{-N}\phi)=T^{-N+1}\phi$ and, by the fact that $\square_{g_M}$ commutes with $T$ (due to the stationarity property of the Schwarzschild metric), $T^{-N}\phi$ are (formally) solutions to \eqref{eq:waveeq} if $\phi$ is a solution.\footnote{Note that in the setting of time-dependent wave operators or nonlinear wave equations, where one cannot appeal to stationarity, the definition of $T^{-N}\phi$ still makes sense, provided $\phi$ decays suitably fast in time. However, in this case, $T^{-N}\phi$ will not satisfy the same wave equation as $\phi$; it will instead satisfy a modified wave equation.}

Rather than working directly with the above definition of $T^{-N}\phi$, however, which requires \emph{a priori} sufficient time decay in the definition, we instead determine the initial data $(T^{-N}\phi|_{\Sigma_0}, T(T^{-N}\phi)|_{\Sigma_0})$ that evolve in time to form the solution $T^{-N}\phi$. We consider first the case $N=1$. Then $T(T^{-1}\phi)|_{\Sigma_0}=\phi|_{\Sigma_0}$, so it remains only to find $T^{-1}\phi|_{\Sigma_0}$. We find $T^{-1}\phi|_{\Sigma_0}$ by rearranging \eqref{eq:waveeq} to obtain:
\begin{align*}
\mathcal{L}\phi=&\:F[T\phi],\quad \textnormal{with}\\
\mathcal{L}(\cdot)=&\:X(Dr^2X(\cdot))+\slashed{\Delta}_{\s^2}(\cdot)-r^2V_{\alpha}(\cdot),\\
F[(\cdot) ]=&\:2r(1-h)X(r (\cdot))-\frac{dh}{dr}r^2 (\cdot)+r^2h\tilde{h}T (\cdot).
\end{align*}
We therefore construct $T^{-1}\phi|_{\Sigma_0}$ by solving the equation $\mathcal{L}(T^{-1}\phi|_{\Sigma_0})=F[\phi|_{\Sigma_0}]$, where the right-hand side is known, as it is determined by the initial data for $\phi$. This is a degenerate elliptic equation (with a degeneracy at the horizon).  We obtain a solution by constructing $\mathcal{L}^{-1}$ on appropriate energy spaces. Given sufficiently rapidly decaying initial data we can repeat this procedure to obtain $T^{-n}\phi$, for $n\leq N$, where the exponent $N$ is determined by the requirement that $T^{-N}\phi|_{\Sigma_0}$ still has finite $r$-weighted energies with appropriate powers in $r$ so as to be able to apply Step 2.\footnote{Strictly speaking, we will encounter the situation where $T^{-N}\phi|_{\Sigma_0}$ has an infinite $p=1$ $r$-weighted energy. To obtain an sharp energy decay rate, we will then need to approach $T^{-N}\phi|_{\Sigma_0}$ with a sequence of compactly supported functions and apply an interpolation argument to convert the extra growth in $r$ of the initial data to extra growth in $\tau$ compared to the case of a finite $p=1$ energy.} It turns out that we will need to take $N=n_{\beta_0}+1$ here. For $\beta_0<1$, we therefore take $N=1$.

We see that \textbf{the total number of $r$-weighted estimates in the hierarchy is determined by the $r$-asymptotics of $T^{-n_{\beta_0}-1}\phi|_{\Sigma_0}$.} In other words, the $r$-asymptotics of sufficiently many time integrals form the bottleneck for extending further the hierarchy of $r$-weighted estimates from \textbf{Step 2} and obtaining faster decay in time. These $r$-asymptotics in turn reflect conformal regularity properties at null infinity. Even if the initial data for $\phi$ are smooth and compactly supported, the initial data for its time integrals will generically have finite polynomial decay in $r$ as $r\to \infty$, which then determines the decay rates appearing in the late-time asymptotics. Starting with initial data that decay less rapidly in $r$, simply means that fewer time integrals might be needed to form a bottleneck, which implies slower decay in time.

In summary, with the above method involving $r$-weighted energy estimates for time derivatives, we are able to make the following principle \emph{quantitative}:
\begin{equation}
\label{timeintprincip}
\textbf{faster $r$-decay of (multiple) time integrals at infinity} \implies \textbf{faster $\tau$-decay}.
\end{equation}
\textbf{Step 3} is carried out in \S \ref{sec:timinv}.

From \textbf{Step 2} and \textbf{Step 3} it is already possible to obtain \emph{almost-sharp decay in time} i.e.\ uniform decay-in-time estimates with a rate that is $\epsilon$ smaller than the sharp rate, with $\epsilon>0$ arbitrarily small. However, this is not the approach taken in the present paper, where we instead establish sharp decay-in-time estimates directly and moreover identify the leading-order term in the late-time asymptotics \emph{at the same time}; see \textbf{Step 4} below.

The ($r$-weighted) energy decay can be used to obtain pointwise decay for $r \phi$ and $\sqrt{r}\phi$. To obtain faster decay for $\phi$ and $r^{-q}\phi$, with appropriate $q>0$, i.e.\ to exchange powers of $r$ for powers of $\tau^{-1}$, we prove faster decay for energies with negative $r$-weights by applying elliptic estimates for the equation $\mathcal{L}[\phi]=F[T\phi]$ that also played a role in constructing time integral data in \textbf{Step 3}. See also \cite{aagkerr, aagprice} for details on the use of elliptic estimates to control energies with negative $r$-weights for wave equations without a potential.

\paragraph{\emph{\underline{Step 4:} Exact leading-order asymptotics and late-time tails via subtraction}}
Rather than deriving the almost-sharp decay estimates of $\phi$, we directly derive sharp decay, and in fact, we directly derive the precise leading-order asymptotics in $\tau^{-1}$ by proving uniform decay-in-time estimates for the difference function:
\begin{equation*}
\widehat{\phi}=\phi-\Phi^{(\beta_0)},
\end{equation*}
where $\Phi^{(\beta_0)}(\tau,r)$ corresponds to the leading-order late-time behaviour of $\phi$ that we want to prove. We will show that $\widehat{\phi}$ decays \emph{faster} than $\Phi^{(\beta_0)}$, which implies that the function $\Phi^{(\beta_0)}$ indeed determines completely the leading-order late-time behaviour. 

More precisely, we take
\begin{align*}
\Phi^{(\beta_{0})}(\tau,r)=\Phi_{0}w_{0} (r)(\tau+1+2r)^{-\frac{1}{2}-\frac{\beta_{0}}{2}}(\tau+1)^{-\frac{1}{2}-\frac{\beta_{0}}{2}},
\end{align*}
with $\Phi_0\in \R$. If $\Phi^{(\beta_0)}(\tau,r)$ has the above form, then the expression $(\square_{g_M}-V_{\alpha})\widehat{\phi}$, while not zero, will decay sufficiently fast in $\tau$ and $r$, so that the methods of \textbf{Step 2} still apply even with $\phi$ replaced by $\widehat{\phi}$. The favourable decay properties of $(\square_{g_M}-V_{\alpha})\Phi^{(\beta_{0})}$ are related to the fact that when $M=0$, $\Phi^{(\beta_{0})}$ is a (spherically symmetric) solution to
\begin{equation*}
(\square-\alpha r^{-2}) \Phi^{(\beta_{0})}=0.
\end{equation*}
with $\square$ the standard wave operator on Minkowski spacetime. This observation also plays an important role in the analysis in \cite{gvdm22}.

In order to show that $\widehat{\phi}$ decays faster than the expected sharp behaviour, we need to fix an appropriate choice for the constant $\Phi_0\in \R$. It turns out that by fixing
\begin{equation*}
\Phi_0=\frac{(-1)^{n_{\beta_{0}}}n_{\beta_{0}}!}{2^{n_{\beta_{0}}-\beta_{0}} \pi} \frac{\Gamma\left(\frac{\beta_{0}+1}{2}\right) \Gamma\left(\frac{\beta_{0}-1}{2}-n_{\beta_{0}}\right) }{\Gamma(\beta_{0}-n_{\beta_{0}})}\mathfrak{I}_{\beta_{0}}[\phi]
\end{equation*}
and considering the time integral $T^{-1-n_{\beta_0}}\Phi^{(\beta_{0})}$, we obtain that
\begin{equation*}
X(rT^{-1-n_{\beta_0}}\Phi^{(\beta_{0})})|_{\Sigma_0}\sim \mathfrak{I}_{\beta_{0}}[\phi] r^{-\frac{1}{2}-\frac{1}{2}\beta_{0}+n_{\beta_0}} \quad (r\to \infty),
\end{equation*}
which is the same asymptotic behaviour as that of $X(rT^{-1-n_{\beta_0}}\phi)|_{\Sigma_0}$ (see \emph{\ref{item:mainthm4}} of Theorem \ref{thm:precise}). Therefore, the difference $X(rT^{-1-n_{\beta_0}}\widehat{\phi})|_{\Sigma_0}$ \underline{must decay faster in $r$} than $X(rT^{-1-n_{\beta_0}}\phi)|_{\Sigma_0}$ and we are able to, via \textbf{Step 2} and \textbf{Step 3}, get more decay in $\tau$ for $\widehat{\phi}$ compared to the sharp decay of $\phi$, adhering to the principle \eqref{timeintprincip}.

\textbf{Step 4} is carried out in \S \ref{sec:latetimeasymp}.

\paragraph{\emph{\underline{Step 5:} Refined estimates for higher-order spherical modes}}
In order to obtain the leading-order late-time behaviour of $\phi_{\geq \ell}$, with faster decay in $\tau$, we note the following rough correspondence:
\begin{equation*}
\square_{g_M}\phi_{\geq \ell}=V_{\alpha} \phi_{\geq \ell} \quad \longleftrightarrow \quad \square_{g_M}\phi_{\geq 0}=V_{\alpha+\ell(\ell+1)} \phi_{\geq 0}.
\end{equation*}
We obtain this by simply rearranging:
\begin{equation*}
r^{-2}\slashed{\Delta}_{\s^2}-V_{\alpha}= r^{-2}(\slashed{\Delta}_{\s^2}+\ell(\ell+1))-(V_{\alpha}+\ell(\ell+1)r^{-2})
\end{equation*}
and noting that the angular operator $-(\slashed{\Delta}_{\s^2}+\ell(\ell+1))$ has non-negative eigenvalues when restricted to functions projected on spherical harmonics with degree $\geq \ell$, just like the operator $-\slashed{\Delta}_{\s^2}$ acting on functions without any restriction to spherical harmonics. From this, it follows that the estimates described in \textbf{Step 4} apply, where $\phi_{\geq \ell}$ with the above modified angular operator takes on the role of $\phi_{\geq 0}$ with the standard operator $\slashed{\Delta}_{\s^2}$ and $\alpha+\ell(\ell+1)$ replaces $\alpha$, which is equivalent to $\beta_{\ell}$ replacing $\beta_0$ everywhere.

\textbf{Step 5} is carried out in more detail in \S \ref{sec:latetimeasymell}.

\subsubsection{Comparison with previous physical space methods}
While not included in the present paper, the strategy described in Steps 1--5 above applies also to the setting of the equation
\begin{equation*}
\square_{g_M}\phi=V\phi,
\end{equation*}
where $|V|\leq r^{-2-\delta}$, $\delta>0$ (including $V\equiv 0$) and to potentials $V_{\alpha}$ with $\alpha=l(l+1)$, $l\in \N$. Additionally, \textbf{Steps} \textbf{2}--\textbf{5} apply also when $g_M$ is replaced by suitable, more general asymptotically flat metrics $g$ and an inverse-square (or more rapidly decaying potential), assuming that analogues of energy boundedness and integrated energy estimates (see \textbf{Steps 0} and \textbf{1}) hold. We outline here briefly the key ideas in the case $V\equiv 0$ and $g=g_M$.

One can obtain the leading-order late-time behaviour of $\phi$ for sufficiently rapidly decaying data by considering
\begin{equation*}
T^{-1}\widehat{\phi}=T^{-1}\phi-\Phi_{0}[\phi](\tau+1+2r)^{-1}(\tau+1)^{-1},
\end{equation*}
where
\begin{equation*}
\Phi_0[\phi]=4 I_0^{(1)}[\phi]
\end{equation*}
with $I_0^{(1)}[\phi]$ the \emph{time-inverted Newman--Penrose constant} of $\phi$ or equivalently, the (standard) Newman--Penrose constant of $T^{-1}\phi$; see \S \ref{sec:compzeropot}.\footnote{In particular, $I_0^{(1)}[\phi]$ is well-defined for initial data with vanishing Newman--Penrose constant. If the Newman--Penrose constant $I_0[\phi]$ is instead non-vanishing, we instead consider $\widehat{\phi}=\phi-\Phi_0(\tau+1+2r)^{-1}(\tau+1)^{-1}$ and take $\Phi_0=4I_0[\phi]$.} 

One can show that $T^{-1}\widehat{\phi}$ decays with a better rate than the (expected) sharp decay rate of $T^{-1}\phi$. For this, we construct $T^{-2}{\phi}$ and we show that the $r$-weighted energy estimates from \textbf{Step 2} with $N=2$ apply when we replace $\phi$ with $T^{-2}\phi$ and the bottleneck for the decay rate arises from the initial data asymptotics of $T^{-2}\phi$, which feature $I_0^{(1)}[\phi]$ in the leading-order asymptotic term as $r\to \infty$. Hence, by taking $\Phi_0[\phi]$ to be an appropriate multiple of $I_0^{(1)}[\phi]$, $T^{-2}\widehat{\phi}$ will decay faster in $r$ at $\tau=0$ than $T^{-2}\phi$ and, consequently, we can show that $\widehat{\phi}$ decays with a faster rate than the sharp decay rate for $\phi$.

The above procedure stands in contrast with the method introduced in \cite{paper2} where late-time asymptotics are instead obtained by propagating the conservation law for the quantity
\begin{equation*}
r^2X(rT^{-1}\phi_0)|_{\mathcal{I}^+}(\tau)
\end{equation*}
in $\tau$ to a suitable region near infinity and then extending the late-time asymptotics towards the horizon. The latter method requires in particular commuting the $r$-weighted estimates of \textbf{Step 2} with vector fields that behave to leading-order like $r^2L$, whereas the method presented in the present paper involves instead an additional application of $T^{-1}$.  The application of $T^{-1}$ and $r^2L$ each result in a decrease of the sharp decay rate in time of the $L^{\infty}(\Sigma_{\tau})$-norm of $r\phi$ by exactly one power. However, as the method developed in the present paper does not require the existence of conservation laws along $\mathcal{I}^+$, it is more widely applicable.

\subsection{Other related works}
In this section, we mention further works that are related to the present paper, but did not feature in the discussion in \S \ref{sec:compzeropot}.

\paragraph{Wave equations on asymptotically flat black hole spacetimes}
The physical-space method for proving decay estimates in the present paper involves important techniques that were introduced in the literature in the setting of wave equations on asymptotically flat black hole spacetimes, most notably, the red-shift estimates of \cite{redshift}, the method of $r$-weighted estimates introduced in \cite{newmethod}, the integrated energy estimates or Morawetz estimates in \cite{blu1, redshift} and the hierarchy of elliptic estimates from \cite{aagkerr, aagprice}. We note that many of these ideas are sufficiently robust so that they can been applied in much more general settings: in the study of wave equations on Kerr spacetimes, see \cite{part3, aagkerr} and references therein, as well as in the study of the Teukolsky equations and linearized gravity; see for example the recent developments \cite{Dafermos2016,dhr-teukolsky-kerr, ma20, shlcosta20, mazhang21b}, and even in the setting of the nonlinear Einstein equations \cite{klainerman17,DHRT21}.

\paragraph{Fourier-analytic methods}
	The method outlined in \S \ref{sec:proofideas} is based entirely in physical space. One can alternatively approach the question of late-time asymptotics and the existence of inverse-polynomial late-time tails in Fourier space. 
	
	The corresponding inverse-polynomial decay rates are related to the regularity of the resolvent operator at the zero time-frequency. The resolvent operator can be viewed as the inverse of the operator $\square_{g_S}-V_{\alpha}$ with time derivatives $T$ (with respect to an asymptotically hyperboloidal time slicing) replaced by multiplication with $-i \omega$ with $\omega\in \C$ denoting the frequency with respect to the choice of time function. 
	
	This fact was exploited in a heuristic setting in \cite{leaver} to arrive at late-time tails for the wave equation without a potential on Schwarzschild. In \cite{hintzprice}, a mathematically rigorous argument was provided to infer the late-time leading-order behaviour on Schwarzschild and Kerr spacetimes, as well as on the Minkowski spacetime with a potential satisfying $|V|\lesssim r^{-3}$, from a careful analysis of the resolvent operator in a neighbourhood of $\omega=0$. We note that for the special case $\omega=0$, the resolvent operator coincides with the inverse of $\mathcal{L}$ appearing in the construction of time-integral data in \textbf{Step 3} from \S \ref{sec:proofideas}. See also earlier work \cite{tataru3} where sharp decay estimates are obtained via the same philosophy in a general setting of stationary spacetimes, assuming a priori appropriate integrated energy estimates. We point the reader to  \cite{hintzprice} for a detailed overview of further references to results pertaining to low-energy resolvent estimates.
		
	A similar point of view was taken in \cite{other1,dssprice} where the behaviour near the zero time-frequency of a resolvent operator is studied using techniques adapted to the framework of the wave equation in one space dimension with appropriate potential, in order to derive refined decay-in-time estimates for fixed spherical harmonic modes $\phi_{\ell}$ on Schwarzschild with $V_{\alpha}\equiv 0$. Within a similar one-dimensional setting, let us also mention \cite{ds10,ch15}, where sharp decay estimates are obtained for wave equations on the real line with potential $V$ satisfying $|V|\lesssim |x|^{-p}$, where $p>2$, and \cite{csst08,dk16}, where the $p=2$ case is considered. See also \cite{sch21}, where the low-energy behaviour of Schr\"odinger operators derived from \cite{csst08} is used to obtain sharp decay-in-time estimates for the one-dimensional wave equation with a model inverse-square potential satisfying assumption \ref{assm:B} in \S \ref{sec:classeq}, featuring decay rates that are analogous to the rates appearing in the precise leading-order asymptotics in Theorem \ref{thm:precise}. We refer to \cite{sch07,sch21} for a further overview of literature pertaining to time-decay estimates for one-dimensional wave equations with inverse-polynomially decaying potentials.
	
	\paragraph{Exactly inverse-square potentials on Minkowski}
	In \cite{pst03,bpst03,bpst04}, the time-decay properties for wave equations with the singular potentials $V_{\alpha}(r)=\alpha r^{-2}$, where $\alpha>-\frac{1}{4}$, on $3+1$-dimensional Minkowski spacetime backgrounds (i.e.\ with respect to the standard wave operator), as well as natural higher-dimensional analogues, were studied extensively. In these works, weighted $L^2$ estimates and Strichartz estimates were derived and applied to the analysis of several nonlinear problems. A different approach for obtaining decay-in-time estimates is provided in \cite{gvdm22}, exploiting instead a connection of the conformal structure of Minkowski with $1+1$-dimensional anti-de Sitter spacetimes to derive the leading-order late-time asymptotics from local existence theory for wave equations on anti-de Sitter spacetimes and obtaining the analogue of Theorem \ref{thm:precise} for the wave equation with singular potential $V_{\alpha}(r)=\alpha r^{-2}$ on the $3+1$-dimensional Minkowski spacetime.

\section{Uniform energy boundedness}
\label{sec:boundedness}
We start by deriving uniform boundedness for energy quantities that degenerate at the event horizon at $r=2M$. We define the \emph{$T$-energy} of $\phi$ as follows:
\begin{equation*}
E_T[\phi](\tau):=\int_{\Sigma_{\tau}}Dr^2(X\phi)^2+h\tilde{h} r^2(T\phi)^2+|\snabla_{\s^2}\phi|^2+\phi^2\,d\sigma dr.
\end{equation*}
Observe that $E_T[\phi]$ degenerates at $r=2M$ due to the presence of $D(r)=1-\frac{2M}{r}$.

We moreover define the \emph{renormalized $T$-energy} for $\phi_{\ell}$ as follows:
\begin{equation*}
\check{E}_T[\phi_{\ell}](\tau):=\int_{\Sigma_{\tau}}Dw_{\ell}^2r^2(X\check{\phi}_{\ell})^2+h\tilde{h} r^2w^2_{\ell}(T\check{\phi}_{\ell})^2\,d\sigma dr,
\end{equation*}
with the weight functions $w_{\ell}$ and the renormalizations $\check{\phi}_{\ell}$ defined in \S \ref{sec:weightfunctions}. We will moreover denote $E_T[\phi]:=~E_T[\phi](0)$ and $\check{E}_T[\phi_{\ell}]:=\check{E}_T[\phi_{\ell}](0)$.
\begin{proposition}
\label{prop:ebound}
Let $\psi\in C^{\infty}(\mathcal{R})\cap C^1 (\widehat{\mathcal{R}})$  be a solution to \eqref{eq:waveeq}. Then there exists a constant $C=C(V_{\alpha})>0$ such that for all $0\leq \tau_1<\tau_2$:
\begin{align}
\label{eq:ebound}
E_T[\phi](\tau_2)+\int_{\mathcal{H}^+\cup \mathcal{I}^+\cap\{\tau_1\leq \tau\leq \tau_2\}}(T\phi)^2 r^2\,d\sigma d\tau\leq &\: C E_T[\phi](\tau_1),\\
\label{eq:checkebound}
\check{E}_T[\phi_{\ell}](\tau_2)+\int_{\mathcal{H}^+\cup \mathcal{I}^+\cap\{\tau_1\leq \tau\leq \tau_2\}}(T\check{\phi}_{\ell})^2 w_{\ell}^2r^2\,d\sigma d\tau= &\: \check{E}_T[\phi_{\ell}](\tau_1).
\end{align}
Furthermore, there exists a constant $C=C(V_{\alpha})>0$ such that for all $\tau\geq 0$:
\begin{equation}
\label{eq:Tenrel}
\check{E}_T[\phi_{\ell}](\tau)\leq C E_T[\phi_{\ell}](\tau).
\end{equation}
\end{proposition}
\begin{proof}
We first multiply both sides of \eqref{eq:waveeq1} with $-T\phi$ and integrate by parts to obtain:
\begin{equation*}
0=\int_{\tau_1}^{\tau_2}\int_{\Sigma_{\tau}}X\left(-Dr^2 X\phi T\phi+r^2(1-h)(T\phi)^2\right)+T\left(\frac{1}{2}h\tilde{h}r^2(T\phi)^2+\frac{1}{2}Dr^2(X\phi)^2+\frac{1}{2} r^2V_{\alpha}\phi^2+\frac{1}{2}|\snabla_{\s^2}\phi|^2\right)\,d\sigma dr d\tau.
\end{equation*}
If we do not have the non-negativity property $V_{\alpha}\geq 0$, we apply \eqref{eq:hardy} with $q=0$, together with $\psi\in C^{1}(\widehat{\mathcal{R}})$ to estimate
\begin{equation*}
\int_{\Sigma_{\tau}} \frac{1}{2}r^2V_{\alpha}\phi^2\,d\sigma dr\geq 4 \min\{\inf_{r'\in [2M,\infty)} r'^2V_{\alpha}(r'),0\}\int_{\Sigma_{\tau}}\frac{1}{2}D^2r^2(X\phi)^2\,d\sigma dr,
\end{equation*}
so coercivity of $E_T[\phi](\tau)$ follows from assumption \ref{assm:A}, which allow us to absorb the right-hand side above, and we obtain \eqref{eq:ebound}.

In order to obtain \eqref{eq:checkebound}, we instead multiply both sides of \eqref{eq:checkpsi2} with $-T\check{\phi}_{\ell}$ and integrate by parts to obtain:
\begin{equation*}
0=\int_{\tau_1}^{\tau_2}\int_{\Sigma_{\tau}}X\left(-Dw_{\ell}^2r^2 X\check{\phi}_{\ell} T\check{\phi}_{\ell}+r^2(1-h)w_{\ell}^2(T\check{\phi}_{\ell})^2\right)+T\left(Dw_{\ell}^2r^2(X\check{\phi}_{\ell})^2+\frac{1}{2}h\tilde{h}r^2w_{\ell}^2(T\check{\phi}_{\ell})^2\right)\,d\sigma dr d\tau.
\end{equation*}
Finally, \eqref{eq:Tenrel} follows straightforwardly from the definition of $\check{\phi}_{\ell}$.
\end{proof}

\begin{remark}
We only invoked the global assumption $r^2V_{\alpha}>-\frac{1}{4}$ (assumption \ref{assm:A}) to obtain \eqref{eq:ebound}. Energy boundedness of the renomalized energy \eqref{eq:checkebound} does \underline{not} require this assumption.
\end{remark}
\section{Integrated energy estimates}
\label{sec:iled}
In this section, we derive the integrated energy estimates that are necessary for establishing pointwise energy decay estimates in time. Though the main results of this paper invoke the assumption that $\alpha\neq 0$, the estimates in this section hold also for $\alpha=0$ and they provide an alternative way of obtaining integrated energy estimates for \eqref{eq:wavewithoutpot}, compared to existing methods in the literature. We start by deriving \emph{local} integrated energy estimates for the renormalized quantities $\check{\phi}_{\ell}$.
\begin{proposition}
\label{prop:intencheckpsi}
Let $\phi\in C^{\infty}(\mathcal{R})\cap C^1 (\widehat{\mathcal{R}})$ be a solution to \eqref{eq:waveeq}. Let $R>2M$ be arbitrarily large. Then there exists a constant $C=C(\ell,V_{\alpha},h,R)>0$ such that for all $0\leq \tau_1<\tau_2\leq \infty$:
\begin{align}
\label{eq:checkiled}
\int_{\tau_1}^{\tau_2}&\int_{\Sigma_{\tau}\cap\{r\leq R\}} (DX\check{\phi}_{\ell})^2+(T\check{\phi}_{\ell})^2\,d\sigma drd\tau+ \int_{\tau_1}^{\tau_2}\int_{\s^2}(X\check{\phi}_{\ell})^2+(T\check{\phi}_{\ell})^2|_{r=R}\,d\sigma d\tau\\ \nonumber
\leq &\: C \check{E}_T[\phi_{\ell}](\tau_1),\\
\label{eq:checkiled2}
\int_{\tau_1}^{\tau_2}&\int_{\Sigma_{\tau}} r^{-1-\delta}(DX\check{\phi}_{\ell})^2+r^{-1-\delta}(T\check{\phi}_{\ell})^2\,d\sigma drd\tau\\ \nonumber
\leq &\: C \check{E}_T[\phi_{\ell}](\tau_1)
\end{align}
\end{proposition}
\begin{proof}
Let $f:[2M,\infty)\to \R$ be a differentiable function, to be determined later. We multiply both sides of \eqref{eq:checkpsi1} with $f(r)Z\check{\phi}_{\ell}$ and apply the Leibniz rule to obtain:
\begin{equation*}
\begin{split}
0=&\: \frac{1}{2} f Z((Z\check{\phi}_{\ell})^2)+2f(r^{-1}+Dw_{\ell}'w_{\ell}^{-1})(Z\check{\phi}_{\ell})^2-f T^2\check{\phi}_{\ell}Z\check{\phi}_{\ell}\\
=&\:Z\left(\frac{f}{2}(Z\check{\phi}_{\ell})^2+\frac{f}{2}(T\check{\phi}_{\ell})^2\right)-T\left(f T\check{\phi}_{\ell}Z\check{\phi}_{\ell}\right)+\left[2fDr^{-1}+2fDw_{\ell}'w_{\ell}^{-1}-\frac{D}{2}\frac{df}{dr}\right](Z\check{\phi}_{\ell})^2\\
&-\frac{1}{2}D\frac{df}{dr}(T\check{\phi}_{\ell})^2.
\end{split}
\end{equation*}
The above equation is equivalent to:
\begin{equation*}
\begin{split}
0=&\:X\left(\frac{f}{2}(Z\check{\phi}_{\ell})^2+\frac{f}{2}(T\check{\phi}_{\ell})^2\right)-T\left(\frac{f}{2}(1-h)D^{-1}(Z\check{\phi}_{\ell})^2+\frac{f}{2}(1-h)D^{-1}(T\check{\phi}_{\ell})^2+fD^{-1}T\check{\phi}_{\ell}Z\check{\phi}_{\ell}\right)\\
&+\left[2fD^{-1}r^{-1}+2fw_{\ell}'w_{\ell}^{-1}-\frac{1}{2}\frac{df}{dr}\right](Z\check{\phi}_{\ell})^2-\frac{1}{2}\frac{df}{dr}(T\check{\phi}_{\ell})^2.
\end{split}
\end{equation*}
By \eqref{eq:asymptw}, it follows that
\begin{equation*}
2D^{-1}r^{-1}+2w_{\ell}'w_{\ell}^{-1}=(1+\beta_{\ell})r^{-1}+O_{\infty}(r^{-1-\beta_{\ell}})+O_{\infty}(r^{-2}),
\end{equation*}
so by taking $f(r)=r^{-q}$ with $q>0$ suitably large, depending on $w_{\ell}$, we obtain the following inequality for all $r\in [2M,\infty)$:
\begin{equation*}
2f D^{-1}r^{-1}+2fw_{\ell}'w_{\ell}^{-1}-\frac{1}{2}\frac{df}{dr}>0.
\end{equation*}
After integrating over $\bigcup_{\tau\in [\tau_1,\tau_2]}\Sigma_{\tau}$, the contribution of the total $T$-derivative results moreover in boundary terms at $\Sigma_{\tau_i}$, $i=1,2$, which can be estimated via Proposition \ref{prop:ebound} after applying the identity:
\begin{multline*}
D^{-1}(1-h)(Z\check{\phi}_{\ell})^2+(1-h)D^{-1}(T\check{\phi}_{\ell})^2+2D^{-1}T\check{\phi}_{\ell}Z\check{\phi}_{\ell}\\
=(1-h)D(X\check{\phi}_{\ell})^2-(2-h)(1-h)hD^{-1}(T\check{\phi}_{\ell})^2+2(2-h)h X\check{\phi}_{\ell}T\check{\phi}_{\ell}\\
=(1-h)D(X\check{\phi}_{\ell})^2-(1-h)h\tilde{h}(T\check{\phi}_{\ell})^2+2h \tilde{h}DX\check{\phi}_{\ell}T\check{\phi}_{\ell},
\end{multline*}
which follows from \eqref{eq:defZ}. We conclude that \eqref{eq:checkiled} holds.

In order to obtain \eqref{eq:checkiled2}, we subsequently take $q=\delta$, with $\delta>0$ arbitrarily small and combine the resulting estimate with the already established local integrated energy estimate \eqref{eq:checkiled}.
\end{proof}

In order to pass from the integrated energy estimates for $\check{\phi}_{\ell}$ in Proposition \ref{prop:intencheckpsi} to an integrated energy estimate for $\phi$, we need to consider separately the cases:
\begin{enumerate}[label=\Alph*)]
\item \label{item:iledcase1}
$\alpha\leq 0$ \underline{and} $\ell=0$\quad \textnormal{and}
\item  \label{item:iledcase2}
$\alpha>0$ \underline{or} $\ell\geq 1$.
 \end{enumerate}

\subsubsection{$\alpha\leq 0$ and $\ell=0$}

In this section, we consider the case $\alpha\leq 0$ and restrict to zeroth angular frequency modes $\phi_0$.
\begin{proposition}
\label{prop:intestaneg}
Let $1-\beta_{0}<p<1+\beta_{0}$. Then, for $R>2M$ suitably large depending on $M$ and $V_{\alpha}$, there exists a constant $C=C(M,h,R,V_{\alpha})>0$ such that
\begin{equation}
\label{eq:rweightest1}
\begin{split}
\int_{\Sigma_{\tau_2}\cap\{r\geq R\}}& r^p(L\psi_0)^2\,d\sigma dr+\int_{\tau_1}^{\tau_2}\int_{\Sigma_{\tau}\cap\{r\geq R\}} r^{p-1}(L\psi_0)^2+r^{p-3}\psi_0^2+r^{p-3}h_0(T\psi_0)^2\,d\sigma dr\\
&+\int_{\tau_1}^{\tau_2}\int_{\s^2} r^{p-2}\psi_0^2\Big|_{r=R}\,d\sigma d\tau\\
\leq &\:C\int_{\Sigma_{\tau_1}\cap\{r\geq R\}}r^p(L\psi_0)^2\,d\sigma dr+CE_T[\phi_0](\tau_1).
\end{split}
\end{equation}
\end{proposition}
\begin{proof}
We will multiply both sides of \eqref{eq:waveeq2a} with $-r^pL\psi$:
\begin{equation}
\label{eq:rpid}
\begin{split}
0=&\: 2r^p\underline{L} ((L\psi)^2)-Dr^{p-2}L\psi\slashed{\Delta}_{\s^2}\psi+\frac{1}{2}D(V_{\alpha}r^p+D'r^{p-1})L(\psi^2)\\
=&\:\underline{L}\left[2r^p(L\psi)^2\right]+L\left[\frac{1}{2}D(V_{\alpha}r^p+D'r^{p-1})\psi^2+\frac{1}{2}Dr^{p-2}|\snabla_{\s^2}\psi|^2\right]\\
&+p D r^{p-1}(L\psi)^2-\frac{1}{4}D\frac{d}{dr}\left(DV_{\alpha}r^p\right)\psi^2-\frac{1}{4}D\frac{d}{dr}\left(D r^{p-2}\right)|\snabla_{\s^2}\psi|^2\\
=&\:D X\left[-r^p(L\psi)^2+\frac{1}{4}D(V_{\alpha}r^p+D'r^{p-1})\psi^2+\frac{1}{4}Dr^{p-2}|\snabla_{\s^2}\psi|^2\right]\\
&+DT\left[\tilde{h}r^p(L\psi)^2+\frac{h}{4}D(V_{\alpha}r^p+D'r^{p-1})\psi^2+\frac{h}{4}Dr^{p-2}|\snabla_{\s^2}\psi|^2\right]\\
&+p D r^{p-1}(L\psi)^2-\frac{1}{4}D\frac{d}{dr}\left(DV_{\alpha}r^p\right)\psi^2-\frac{1}{4}D\frac{d}{dr}\left(D r^{p-2}\right)|\snabla_{\s^2}\psi|^2.
\end{split}
\end{equation}
We integrate the above equality over $\bigcup_{\tau\in [\tau_1,\tau_2]}\Sigma_{\tau}\cap\{r\geq R\}$, using that $\psi=\psi_0$, and we apply \eqref{eq:ebound} in combination with \eqref{eq:checkiled} and \eqref{eq:Tenrel} to obtain:
\begin{equation*}
\begin{split}
\int_{\Sigma_{\tau_2}\cap\{r\geq R\}}& \tilde{h}r^p(L\psi_0)^2\,d\sigma dr+\int_{\tau_1}^{\tau_2}\int_{\s^2}r^p(L\psi_0)^2-\frac{1}{4}(\alpha +O_{\infty}(r^{-1}))r^{p-2}\psi_0^2\Big|_{r=R}\,d\sigma d\tau\\
&+\int_{\tau_1}^{\tau_2}\int_{\Sigma_{\tau}\cap\{r\geq R\}} p  r^{p-1}(L\psi_0)^2+\frac{1}{4}(2-p) r^{p-3}[\alpha+O_{\infty}(r^{-1})] \psi_0^2\,d\sigma dr d\tau\\
\leq &\:C\int_{\Sigma_{\tau_1}\cap\{r\geq R\}}\tilde{h}r^p(L \psi_0)^2\,d\sigma dr+CE_T[\phi_0](\tau_1).
\end{split}
\end{equation*}

Note that by \eqref{eq:hardy}, we can estimate
\begin{equation}
\label{eq:auxestal0}
\begin{split}
\int_{\tau_1}^{\tau_2}&\int_{\Sigma_{\tau}\cap\{r\geq R\}}  \frac{1}{4}(2-p) r^{p-3}(|\alpha|+O_{\infty}(r^{-1})) \psi_0^2\,d\sigma dr d\tau\leq  \frac{|\alpha|}{2-p}\int_{\tau_1}^{\tau_2}\int_{\Sigma_{\tau}\cap\{r\geq R\}}  r^{p-3}(1+C r^{-1})(X\psi_0)^2\,d\sigma dr d\tau\\
&+\frac{1}{2}\int_{\tau_1}^{\tau_2}\int_{\s^2} (|\alpha|r^{p-2}+O_{\infty}(r^{p-3}))\psi_0^2\,d\sigma d\tau\Big|_{r=R},
\end{split}
\end{equation}
for some constant $C=C(M,p,V_{\alpha})>0$.

We can further apply Young's inequality to estimate for $\epsilon>0$ arbitrarily small:
\begin{equation*}
	\int_{\tau_1}^{\tau_2}\int_{\Sigma_{\tau}\cap\{r\geq R\}} \frac{1}{4}p r^{p-1}(D X\psi_0)^2\,d\sigma dr d\tau \leq \int_{\tau_1}^{\tau_2}\int_{\Sigma_{\tau}\cap\{r\leq R\}} (1+\epsilon) pr^{p-1}(L\psi_0)^2+\frac{1}{4}(1+\epsilon^{-1})h^2r^{p-1}(T\psi_0)^2\,d\sigma dr d\tau.
\end{equation*}
By \eqref{eq:checkiled2} and \eqref{eq:Tenrel}, we can estimate
\begin{equation*}
\int_{\tau_1}^{\tau_2}\int_{\Sigma_{\tau}\cap\{r\geq R\}} h^2r^{p-1}(T\psi_0)^2\,d\sigma dr d\tau \leq C E_T[\phi_0](\tau_1),
\end{equation*}
provided we restrict to $p-5<-2-\beta_0$, or equivalently $p<3-\beta_0$.

We can therefore absorb the spacetime integral on the right-hand side of \eqref{eq:auxestal0} into the left-hand side if we take
\begin{equation*}
\frac{|\alpha|}{2-p}<\frac{1}{4}p,
\end{equation*}
which is equivalent to $1-\beta_0<p<1+\beta_0$. Note that $3-\beta_0\leq 1+\beta_0$ for $\alpha\leq 0$.

%Recall that, $L=\frac{1}{2}DX+\frac{1}{2}h T$. We first multiply both sides of \eqref{eq:waveeq2} with $-Dr^pX\psi$ to obtain:
%\begin{equation*}
%\begin{split}
%0=&-\frac{1}{2}r^pX((DX\psi)^2)+(1-h)Dr^pT((X\psi)^2)+h\tilde{h}Dr^pT^2\psi X\psi\\
%& -\frac{dh}{dr}Dr^p X\psi T\psi+\frac{1}{2}[V_{\alpha}-D' %r^{-1}]D r^pX(\psi^2)\\
%=&\: \frac{1}{2}p D^2r^{p-1}(X\psi)^2+\frac{1}{2}(2-p)Dr^{p-3}\alpha\psi^2+\frac{1}{2}\frac{d}{dr}(h\tilde{h}D r^p)(T\psi)^2-\frac{dh}{dr}Dr^p X\psi T\psi+O_{\infty}(r^{p-4})\psi^2\\
%&+X\left[-\frac{1}{2}r^p(DX\psi)^2-\frac{1}{2}h\tilde{h}Dr^p (T\psi)^2+\frac{1}{2}\alpha r^{p-2}\psi^2+O_{\infty}(r^{p-3})\psi^2\right]\\
%&+T\left[(1-h)Dr^p (X\psi)^2+h\tilde{h}D r^p T\psi X\psi\right].
%\end{split}
%\end{equation*}
%Now, we multiply both sides of \eqref{eq:waveeq2} with $-r^phT\psi$ to obtain
%\begin{equation*}
%\begin{split}
%0=&\:-r^p h T\psi X(DX\psi)+(1-h)hr^pX((T\psi)^2)+\frac{1}{2}r^ph^2\tilde{h}T((T\psi)^2)-r^ph\frac{dh}{dr}(T\psi)^2+\frac{1}{2}hr^p(V_{\alpha}+D' r^{-1})T(\psi^2)\\
%=&\:X\left[-Dr^p h T\psi X\psi+(1-h)h r^p(T\psi)^2\right]+T\left[\frac{1}{2}D r^p h (X\psi)^2+\frac{1}{2}r^p h^2\tilde{h}(T\psi)^2+\frac{1}{2}hr^{p-2} (\alpha+O_{\infty}(r^{-1})\psi^2)\right]\\
%&+\frac{d}{dr}(hr^p)DT\psi X\psi-\left[\frac{d}{dr}((1-h)hr^p)+r^ph \frac{dh}{dr}\right](T\psi)^2
%\end{split}
%\end{equation*}

It remains to group and estimate the various boundary terms on $\{r=R\}$. It is here that we will make use of Proposition \ref{prop:intencheckpsi}.

We apply Young's inequality with $\epsilon>0$ arbitrarily small together with \eqref{eq:asymptw} to obtain:
\begin{equation*}
\begin{split}
(L\psi)^2=&\:2(L(r w_0 \check{\phi}_0))^2\geq 2(1-\epsilon)\frac{D^2}{4}\left(\frac{d(rw_0)}{dr}\right)^2\check{\phi}_0^2-2(\epsilon^{-1}-1)r^2w_0^2(L\check{\phi}_0)^2\\
\geq &\: \left[\frac{1}{8}(1+\beta_0)^2+O_{\infty}(r^{-1})\right]r^{-2}\psi^2- 2(\epsilon^{-1}-1) r^2w_0^2(L\check{\phi}_0)^2,
\end{split}
\end{equation*}
so, for $\epsilon>0$ arbitrarily small, we can take $R>2M$ suitably large so that we can estimate
\begin{equation}
\label{eq:auxrpaneg}
\begin{split}
\int_{\Sigma_{\tau_2}\cap\{r\geq R\}}& r^p(L\psi_0)^2\,d\sigma dr+\int_{\tau_1}^{\tau_2}\int_{\s^2}\frac{1}{8}\left((1-\epsilon)(1+\beta_0)^2-2|\alpha|\right)r^{p-2}\psi_0^2\Big|_{r=R}\,d\sigma d\tau\\
&+\int_{\tau_1}^{\tau_2}\int_{\Sigma_{\tau}\cap\{r\geq R\}} \left(\frac{1}{4}p  -\frac{|\alpha|}{2-p}+O_{\infty}(r^{-1})\right)r^{p-1}(L\psi_0)^2\,d\sigma dr d\tau\\
\leq &\:C\int_{\Sigma_{\tau_1}\cap\{r\geq R\}}r^p(L\psi_0)^2\,d\sigma dr+C\epsilon^{-1}\int_{\tau_1}^{\tau_2}\int_{\s^2} w_0^2r^{p-2}\left[(X\check{\phi}_0)^2+(T\check{\phi}_0)^2\right]\,d\sigma d\tau+CE_T[\phi_0](\tau_1).
\end{split}
\end{equation}
Note that
\begin{equation*}
(1+\beta_0)^2-2|\alpha|=1+\beta_0^2-2|\alpha|+2\beta_0>2(1-3|\alpha|)>0,
\end{equation*}
so for $1-\beta_0<p<1+\beta_0$ and $\epsilon>0$ suitably small, the terms on the left-hand side of \eqref{eq:auxrpaneg} are non-negative definite.

To obtain also control of a weighted spacetime integral of $\underline{L}\psi_0$, we multiply both sides of \eqref{eq:waveeq2a} with $-r^{p-2}\underline{L}\psi$ to obtain:
\begin{equation}
\label{eq:rpidLbar}
\begin{split}
0=&\:L\left[2r^{p-2}(\underline{L}\psi)^2\right]+\underline{L}\left[\frac{1}{2}D(V_{\alpha}r^{p-2}+D'r^{p-3})\psi^2+\frac{1}{2}Dr^{p-4}|\snabla_{\s^2}\psi|^2\right]\\
&+(2-p) D r^{p-3}(\underline{L}\psi)^2+\frac{1}{4}D\frac{d}{dr}\left(DV_{\alpha}r^{p-2}\right)\psi^2+\frac{1}{4}D\frac{d}{dr}\left(D r^{p-4}\right)|\snabla_{\s^2}\psi|^2\\
=&\:D X\left[r^{p-2}(\underline{L}\psi)^2-\frac{1}{4}D(V_{\alpha}r^{p-2}+D'r^{p-3})\psi^2-\frac{1}{4}Dr^{p-4}|\snabla_{\s^2}\psi|^2\right]\\
&+DT\left[hr^{p-2}(\underline{L}\psi)^2+\frac{\tilde{h}}{4}D(V_{\alpha}r^{p-2}+D'r^{p-3})\psi^2+\frac{\tilde{h}}{4}Dr^{p-4}|\snabla_{\s^2}\psi|^2\right]\\
&+(2-p) D r^{p-3}(\underline{L}\psi)^2+\frac{1}{4}D\frac{d}{dr}\left(DV_{\alpha}r^{p-2}\right)\psi^2+\frac{1}{4}D\frac{d}{dr}\left(D r^{p-4}\right)|\snabla_{\s^2}\psi|^2.
\end{split}
\end{equation}
We integrate the above equality with $\psi$ replaced by $\psi_0$ in spacetime, combine it with \eqref{eq:auxrpaneg} and \eqref{eq:auxestal0}, using that $\alpha>-\frac{1}{4}$ and $1-\beta_0<p<1+\beta_0$, and we apply \eqref{eq:checkiled} and \eqref{eq:Tenrel} to conclude that \eqref{eq:rweightest1} holds.
\end{proof}

\begin{remark}
Note that the range $(1-\beta_{0},1+\beta_0)$ for $p$ in Proposition \ref{prop:intestaneg} can be smaller than the range $(0,2]$ appearing in standard $r$-weighted energy estimates of \cite{newmethod}. This smaller range is caused by the fact that when $\beta_0<1$ (or equivalently, $\alpha<1$), there are additional zeroth order terms for $\psi$ that come with a bad sign and have to be absorbed by the terms with a good sign via the Hardy inequality \eqref{eq:hardy} (which breaks down when $\beta_0=0$).
\end{remark}
\begin{corollary}
\label{cor:iedaneg}
Let $1-\beta_{0}<p<1+\beta_{0}$. Then there exists a constant $C=C(M,V_{\alpha},h,R,p)>0$ such that for all $0\leq \tau_1<\tau_2$:
\begin{equation}
\label{eq:ied1}
\begin{split}
\int_{\Sigma_{\tau_2}}& Dr^p(X(r\phi_0))^2\,d\sigma dr+\int_{\tau_1}^{\tau_2}\int_{\Sigma_{\tau}} r^{p-1} \left[(L(r\phi_0))^2+   h_0 (T\phi_0)^2+\phi_0^2 \right]\,d\sigma drd\tau\\
\leq &\:CE[\phi_0](\tau_1	)+C\int_{\Sigma_{\tau_1}} Dr^p(X(r\phi_0))^2\,d\sigma dr.
\end{split}
\end{equation}
\end{corollary}
\begin{proof}
By \eqref{eq:hardy}, we can estimate
\begin{equation*}
\int_{\tau_1}^{\tau_2}\int_{\Sigma_{\tau}\cap\{r\leq R\}}\check{\phi}_0^2\,d\sigma dr d\tau\leq 4\int_{\tau_1}^{\tau_2}\int_{\Sigma_{\tau}\cap\{r\leq R\}}r^2(DX\check{\phi}_0)^2\,d\sigma drd\tau+2 \int_{\tau_1}^{\tau_2}\int_{\s^2}(r-2M) \check{\phi}_0^2\Big|_{r=R}\,d\sigma d\tau.
\end{equation*}
Furthermore, there exists a constant $C=C(M,V_{\alpha},h,R)>0$ such that
\begin{equation*}
\int_{\Sigma_{\tau}\cap\{r\leq R\}}(DX\psi_0)^2+(T\psi_0)^2+\psi^2\,d\sigma dr d\tau\leq C \int_{\Sigma_{\tau}\cap\{r\leq R\}}(DX\check{\phi}_0)^2+(T\check{\phi}_0)^2+\check{\phi}_0^2\,d\sigma dr d\tau.
\end{equation*}
By combining the above equation with \eqref{eq:ebound}, \eqref{eq:checkiled} and \eqref{eq:rweightest1} we obtain \eqref{eq:ied1}.
\end{proof}

\subsubsection{$\alpha > 0$ or $\ell\geq 1$: integrated energy estimates for $\phi_{\leq \ell_0}$}
In this section we consider simultaneously the case $\alpha>0$ and the case $\alpha\leq 0$ with $\phi$ restricted to angular frequencies $\ell\geq 1$. We can then make use of the following positivity property:
\begin{equation*}
\alpha+\ell(\ell+1)>0.
\end{equation*}
We will additionally split $\phi$ as follows:
\begin{equation*}
\phi=\phi_{\leq \ell_0}+\phi_{>\ell_0},
\end{equation*}
where $\ell_0$ will be taken suitably large (depending on $V_{\alpha}$). When considering $\phi_{\leq \ell_0}$, we will make use of Proposition \ref{prop:intencheckpsi}.
\begin{proposition}
\label{prop:aposlbound}
Let $\alpha>0$, or let $\phi=\phi_{\geq 1}$. Then there exists a constant $C=C(M,h,\alpha,R,\ell_0)>0$ such that for all $0\leq \tau_1\leq \tau_2\leq \infty$:
\begin{equation}
\label{eq:iedapos}
\begin{split}
\int_{\tau_1}^{\tau_2}& \int_{\Sigma_{\tau}\cap \{r\leq R\}} (X\phi_{\leq \ell_0})^2+(T\phi_{\leq \ell_0})^2+\phi_{\leq \ell_0}^2\,d\sigma dr d\tau\\
\leq &\: C E_T[\phi](\tau_1).
\end{split}
\end{equation}
\end{proposition}
\begin{proof}
Let $f:[2M,\infty) \to \R$ be a differentiable function, which will be specified later. We multiply both sides of \eqref{eq:waveeq2a} with $fZ\psi= fL\psi-f\underline{L}\psi$:
\begin{equation}
\label{eq:auxiledapos}
\begin{split}
0=&\: -4fL\psi \underline{L} L\psi+4f\underline{L}\psi L\underline{L} \psi+Dr^{-2}fZ\psi\slashed{\Delta}_{\s^2}\psi-fZ\psi D(V_{\alpha}+D'r^{-1})\psi\\
=&\:\underline{L}(-2f(L\psi)^2)+L(2f(\underline{L}\psi)^2)+Z\left(-\frac{f}{2}Dr^{-2}|\snabla_{\s^2}\psi|^2-\frac{f}{2}D(V_{\alpha}+D'r^{-1})\psi^2\right)\\
&-D\frac{df}{dr}((L\psi)^2+(\underline{L}\psi)^2)+\frac{1}{2}D\frac{d}{dr}(fDr^{-2})|\snabla_{\s^2}\psi|^2+\frac{1}{2}D\frac{d}{dr}(f(V_{\alpha}+D'r^{-1}))\psi^2+\textnormal{div}_{\s^2}(\ldots)\\
=&\:DX\left[f(\underline{L}\psi)^2+f(L\psi)^2-\frac{f}{2}Dr^{-2}|\snabla_{\s^2}\psi|^2-\frac{f}{2}D(V_{\alpha}+D'r^{-1})\psi^2\right]\\
&+DT\left[-f\tilde{h}(L\psi)^2+fhD(D^{-1}\underline{L}\psi)^2+(1-h)\frac{f}{2}\left(r^{-2}|\snabla_{\s^2}\psi|^2+(V_{\alpha}+D'r^{-1})\psi^2\right)\right]\\
&-D\frac{df}{dr}((L\psi)^2+(\underline{L}\psi)^2)+\frac{1}{2}D\frac{d}{dr}(fD(V_{\alpha}+D'r^{-1}))\psi^2+\frac{1}{2}D\frac{d}{dr}(fDr^{-2})|\snabla_{\s^2}\psi|^2+\textnormal{div}_{\s^2}(\ldots),
\end{split}
\end{equation}
where we use the schematic notation $\textnormal{div}_{\s^2}(\ldots)$ to denote divergences with respect to $\s^2$, which vanish after integration over $\s^2$.

We take $f\equiv -1$, integrate over $\{r\geq R\}$, where $R>2M$ will be chosen suitably large and apply \eqref{eq:ebound} to obtain:
\begin{equation*}
\begin{split}
\int_{\tau_1}^{\tau_2}& \int_{\Sigma_{\tau}\cap \{r\geq R\}} (\alpha r^{-3}+O_{\infty}(r^{-4})) \psi^2+(1+O_{\infty}(r^{-1}))r^{-3}|\snabla_{\s^2}\psi|^2\,d\sigma dr d\tau\\
&+\int_{\tau_1}^{\tau_2}\int_{\s^2}(\underline{L}\psi)^2+(L\psi)^2-\frac{1}{2}Dr^{-2}|\snabla_{\s^2}\psi|^2-\frac{1}{2}Dr^{-2}(\alpha+O_{\infty}(r^{-1}))\psi^2\Big|_{r=R}\,d\sigma d\tau\\
\leq &\: C E_T[\phi](\tau_1)+C \eta \sum_{\ell\leq \ell_0}\int_{\tau_1}^{\tau_2}\int_{\s^2} r^{-2}\psi_{\ell}^2\Big|_{r=R}\,d\sigma d\tau .
\end{split}
\end{equation*}
We will now restrict to $\phi=\phi_{\ell}$ and observe that the spacetime integral on the left-hand side above is non-negative definite, using \eqref{eq:sphere2} in combination with $\alpha+\ell(\ell+1)>0$. We rewrite the $\{r=R\}$ boundary terms by first considering separately each $\phi_{\ell}$, with $\ell\leq \ell_0$: by Young's inequality with $\epsilon>0$ arbitrarily small together with \eqref{eq:asymptw}, we obtain:
\begin{equation*}
\begin{split}
(L\psi_{\ell})^2+(\underline{L}\psi_{\ell})^2=&\:(L	(r w_{\ell} \check{\phi}_{\ell}))^2+(\Lbar (r w_{\ell}\check{\phi}_{\ell}))^2\geq\frac{1-\epsilon}{2}D^2\left(\frac{d(rw_{\ell})}{dr}\right)^2\check{\phi}_{\ell}^2+(1-\epsilon^{-1})r^2w_{\ell}^2\left((L\check{\phi}_{\ell})^2+(\underline{L}\check{\phi}_{\ell})^2\right)\\
\geq &\: \frac{1-\epsilon}{8}\left[(1+\beta_{\ell})^2 +O_{\infty}(r^{-1})\right]r^{-2}\psi_{\ell}^2-(\epsilon^{-1}-1)r^2w_{\ell}^2\left((L\check{\phi}_{\ell})^2+(\underline{L}\check{\phi}_{\ell})^2\right).
\end{split}
\end{equation*}
Note that
\begin{equation*}
	\frac{1}{8}(1+\beta_{\ell})^2-\frac{1}{2}\ell(\ell+1)-\frac{1}{2}\alpha=\frac{1}{4}(1+\beta_{\ell}),
\end{equation*}
so there exists a constant $C=C(M,h,\alpha,R)>0$ such that
\begin{equation*}
\begin{split}
\int_{\tau_1}^{\tau_2}& \int_{\Sigma_{\tau}} (\alpha+\ell(\ell+1)) r^{-3} \psi_{\ell}^2\,d\sigma dr d\tau+\int_{\tau_1}^{\tau_2}\int_{\s^2}\frac{1-\epsilon}{4}(1+\beta_{\ell})r^{-2}\psi^2\Big|_{r=R}\,d\sigma d\tau\\
\leq &\: C E_T[\phi](\tau_1)+C \epsilon^{-1}\int_{\tau_1}^{\tau_2}\int_{\s^2}w_{\ell}^2(X\check{\phi}_{\ell})^2+w_{\ell}^2(T\check{\phi}_{\ell})^2\Big|_{r=R}\,d\sigma d\tau.
\end{split}
\end{equation*}
We can estimate the remaining $\{r=R\}$ term on the right-hand side above and sum $\ell$ from $\ell=0$ (if $\alpha>0$) or $\ell=1$ (if $\alpha\leq 0$) to $\ell=\ell_0$ by applying \eqref{eq:checkiled}.
\end{proof}
\begin{remark}
In contrast with the $\alpha<0$, $\ell=0$ setting, we are able to obtain via \eqref{eq:iedapos} a local integrated energy estimate with only the energy $E_T[\phi]$ appearing on the right-hand side of \eqref{eq:iedapos}, i.e.\\ \underline{without any growing weights in $r$} in front the the terms in the integral defining $E_T[\phi]$; see also the discussion in Step 1 of \S \ref{sec:proofideas}.
\end{remark}

\begin{remark}
The combination of \eqref{eq:checkiled} with the above proposition provides also a straightforward way of obtaining a local integrated energy (Morawetz) estimate in the setting of \eqref{eq:wavewithoutpot} in a way that is uniform for bounded angular frequencies $0\leq \ell\leq \ell_0$ and does not require any delicate choices of multiplier functions.
\end{remark}

\subsubsection{$\alpha > 0$ or $\ell\geq 1$: integrated energy estimates for $\phi_{> \ell_0}$}
In this section, we consider the high angular frequency part $\phi_{> \ell_0}$, where $\ell_0\in \N_0$ will be taken sufficiently large. Since $\ell_0$ is large, the presence of the $V_{\alpha}$ potential in \eqref{eq:waveeq} does not play an important role and we can apply standard arguments from the setting of the geometric wave equation on Schwarzschild without potential (i.e.\ $V_{\alpha}\equiv 0$) to obtain an integrated local energy decay estimate (see for example \S 4.11 of \cite{lecturesMD}). For the sake of completeness, we include them below.
\begin{proposition}
\label{prop:aposllarge}
Let $R>2M$ be arbitrarily large. Then there exists a suitably large $\ell_0\in \N$ and a constant $C=C(M,V_{\alpha},h,\ell_0,R)>0$, such that for all $0\leq \tau_1\leq \tau_2\leq \infty$:
	\begin{equation}
\label{eq:iedaposhighfreq}
\begin{split}
\int_{\tau_1}^{\tau_2}& \int_{\Sigma_{\tau}\cap \{r\leq R\}} D(D^{-1}Z\phi_{\geq \ell_0})^2+\left(1-3Mr^{-1}\right)^2\left(|\snabla_{\s^2}\phi_{\geq \ell_0}|^2+(T\phi_{\geq \ell_0})^2\right)\,d\sigma dr d\tau\\
\leq &\: C E_T[\phi](\tau_1).
\end{split}
\end{equation}
\end{proposition}
\begin{proof}
Let $f:[2M,\infty) \to \R$ be a differentiable function, which will be specified later.  We  multiply both sides of \eqref{eq:waveeq2a} with $\frac{D}{2}\frac{df}{dr}\psi$ to obtain:
\begin{equation}
\label{eq:zerothoderILED}
\begin{split}
0=&\: -2D\frac{df}{dr}\psi \underline{L} L\psi+\frac{1}{2}D^2r^{-2}\frac{df}{dr}\psi\slashed{\Delta}_{\s^2}\psi-\frac{1}{2}D^2\frac{df}{dr}(V_{\alpha}+D'r^{-1})\psi^2\\
=&\: \underline{L}\left( -2 \frac{df}{dr_*}\psi L\psi\right)+2\frac{df}{dr_*}\underline{L}\psi L\psi-\frac{1}{2}\frac{d^2f}{dr_*^2}L(\psi^2)-\frac{1}{2}Dr^{-2}\frac{df}{dr_*}\left[(r^2V_{\alpha}+D'r)\psi^2+|\snabla_{\s^2}\psi|^2\right]+\textnormal{div}_{\s^2}(\ldots)\\
=&\:-\frac{D}{2}X\left[ -2 \frac{df}{dr_*}\psi L\psi-\frac{1}{2}\frac{d^2f}{dr_*^2}\psi^2\right]+\frac{D}{2}T\left[ -2\tilde{h}\frac{df}{dr_*}\psi L\psi-\frac{h}{2D}\frac{d^2f}{dr_*^2}\psi^2\right]\\
&-\frac{D}{2}r^{-2}D\frac{df}{dr}\left[(r^2V_{\alpha}+D'r)\psi^2+|\snabla_{\s^2}\psi|^2\right]+\frac{1}{4}\frac{d^3f}{dr^3_*}\psi^2+2D\frac{df}{dr}L\psi \underline{L}\psi+\textnormal{div}_{\s^2}(\ldots).
\end{split}
\end{equation}
By combining with \eqref{eq:auxiledapos}, we obtain
\begin{equation*}
\begin{split}
0=&\:X\left[f(\underline{L}\psi)^2+f(L\psi)^2-\frac{f}{2}Dr^{-2}|\snabla_{\s^2}\psi|^2-\frac{f}{2}D(V_{\alpha}+D'r^{-1})\psi^2   + \frac{df}{dr_*}\psi L\psi+\frac{1}{4}\frac{d^2f}{dr_*^2}\psi^2\right]\\
&+T\left[-f\tilde{h}(L\psi)^2+fhD(D^{-1}\underline{L}\psi)^2+(1-h)\frac{f}{2}\left(r^{-2}|\snabla_{\s^2}\psi|^2+(V_{\alpha}+D'r^{-1})\psi^2\right) -\tilde{h}\frac{df}{dr_*}\psi L\psi-\frac{h}{4D}\frac{d^2f}{dr_*^2}\psi^2\right]\\
&-\frac{df}{dr}(Z\psi)^2+\frac{f}{2}\frac{d}{dr}(D(V_{\alpha}+D'r^{-1}))\psi^2+\frac{1}{4D}\frac{d^3f}{dr_*^3}\psi^2+\frac{f}{2}\frac{d}{dr}(Dr^{-2})|\snabla_{\s^2}\psi|^2+\textnormal{div}_{\s^2}(\ldots).
\end{split}
\end{equation*}
Note that
\begin{equation*}
	\frac{d}{dr}(Dr^{-2})=-2r^{-4}(r-3M),
\end{equation*}
so by taking $f(r_*)=-\arctan(r_*-r_*(3M))$, using that $ \frac{df}{dr}(r)=-\frac{1}{D(r)}(1+(r_*(r)-r_*(3M))^2))^{-1}$, $\frac{1}{4D}\frac{d^3f}{dr_*^3}|_{r=3M}=\frac{1}{2D(3M)}>0$ and $\frac{f}{2}\frac{d}{dr}(Dr^{-2})\geq cr^{-3}(1-3Mr^{-1})^2$, and by applying \eqref{eq:sphere2}, we can estimate for $\phi_{\geq \ell_0}$, with $\ell_0$ suitably large:
\begin{multline*}
	\int_{\s^2}-\frac{df}{dr}(Z\psi_{\geq \ell_0})^2+\left[\frac{f}{2}\frac{d}{dr}(D(V_{\alpha}+D'r^{-1}))+\frac{1}{4D}\frac{d^3f}{dr_*^3}\right]\psi_{\geq \ell_0}^2+\frac{f}{2}\frac{d}{dr}(Dr^{-2})|\snabla_{\s^2}\psi_{\geq \ell_0}|^2\,d\sigma \\
	\geq c \int_{\s^2}D^{-1}r^{-2}(Z\psi_{\geq \ell_0})^2+r^{-3}\left(1-\frac{3M}{r}\right)^2 |\snabla_{\s^2}\psi_{\geq \ell_0}|^2+r^{-3}|\psi_{\geq \ell_0}|^2\,d\sigma.
\end{multline*}

We obtain \eqref{eq:iedaposhighfreq} after integrating in spacetime, applying \eqref{eq:ebound} to estimate the boundary term and then applying \eqref{eq:auxiledapos} with $f(r)=-(1-3Mr^{-1})^3$ to control also a spacetime integral of $(1-3Mr^{-1})^2(T\psi_{\geq \ell_0})^2$.
\end{proof}

\subsection{Non-degenerate (integrated) energy estimates}
In this section, we  will derive a non-degenerate integrated energy estimate for all $\alpha>-\frac{1}{4}$. We combine first the integrated energy estimates from the previous section to state integrated energy estimates that are valid for all $\alpha>-\frac{1}{4}$, but which degenerate at $r=2M$.
\begin{corollary}[Degenerate integrated energy estimates]
\label{cor:degiled}
For $R>2M$ suitably large, there exists a constant $C=C(M,h,V_{\alpha},R)>0$ such that for all $0\leq \tau_1\leq \tau_2\leq \infty$:
\begin{align}
\label{eq:fullied}
\int_{\Sigma_{\tau_2}\cap\{r\geq R\}}&r(X(r\phi))^2\,d\sigma dr+\int_{\tau_1}^{\tau_2}\int_{\Sigma_{\tau}} [\phi^2+(1-3Mr^{-1})^2\left(r^2(DX\phi)^2+|\snabla_{\s^2}\phi|^2+h_0(T\phi)^2\right)]\,d\sigma dr d\tau\\ \nonumber
\leq &\:C\int_{\Sigma_{\tau_1}\cap\{r\geq R\}}r(X(r\phi))^2\,d\sigma dr+CE_T[\phi](\tau_1),\\
\label{eq:fulliedlossder}
\int_{\Sigma_{\tau_2}\cap\{r\geq R\}}& r(X(r\phi))^2\,d\sigma dr+\int_{\tau_1}^{\tau_2}\int_{\Sigma_{\tau}} \left[r^2(DX\phi)^2+\phi^2+|\snabla_{\s^2}\phi|^2+h_0(T\phi)^2\right]\,d\sigma dr d\tau\\\nonumber
\leq &\:C\int_{\Sigma_{\tau_1}\cap\{r\geq R\}}r(X(r\phi))^2+r(X(rT\phi))^2\,d\sigma dr+CE_T[\phi](\tau_1)+C E_T[T\phi](\tau_1)
\end{align}
\end{corollary}
\begin{proof}
We consider separately the cases 1) $\alpha\leq 0$ and $\phi=\phi_0$, and 2) $\alpha>0$ or $\phi=\phi_{\geq 1}$.

In case 1), \eqref{eq:fullied} and \eqref{eq:fulliedlossder} follow directly from Corollary \ref{cor:iedaneg} with $p=1$. In fact,  \eqref{eq:fullied} holds without a degeneracy at $r=3M$.

In case 2), we need to complement the estimates Proposition \ref{prop:aposlbound} and \ref{prop:aposllarge} with additional $r$-weighted estimates in the region $r\geq R$, with $R>3M$ arbitrarily large. We therefore integrate \eqref{eq:rpid} with $p=1$ and combine it with \eqref{eq:ebound} to estimate terms without a good sign on $\Sigma_{\tau_2}$ and Proposition \ref{prop:aposlbound} and \ref{prop:aposllarge}, together with a standard averaging argument, to estimate boundary terms at $r=R$ and obtain \eqref{eq:fullied}.

We now turn to \eqref{eq:fulliedlossder}. By the Killing property of $T$, $[\square_{g_M}-V_{\alpha},T]=0$, so \eqref{eq:fullied} also applies with $\phi$ replaced by $T\phi$. We therefore obtain immediately:
\begin{multline}
\label{eq:auxremovedegiled}
\int_{\Sigma_{\tau_2}\cap\{r\geq R\}} r(X(r\phi))^2\,d\sigma dr+\int_{\tau_1}^{\tau_2}\int_{\Sigma_{\tau}} \left[r^2(DX\phi)^2+\phi^2+h_0(T\phi)^2+(1-3Mr^{-1})^2|\snabla_{\s^2}\phi|^2\right]\,d\sigma dr d\tau\\
\leq C\int_{\Sigma_{\tau_1}\cap\{r\geq R\}}r(X(r\phi))^2+r(X(rT\phi))^2\,d\sigma dr+CE_T[\phi](\tau_1)+C E_T[T\phi](\tau_1).
\end{multline}
It remains to remove the degenerate factor in front of $|\snabla_{\s^2}\phi|^2$. We apply \eqref{eq:zerothoderILED} with $f(r)=-\chi(r)r$, where $\chi(r)$ is a smooth cut-off function satisfying $\chi(r)=1$ for $(3-\frac{1}{4})M\leq r\leq (3+\frac{1}{4})M$ and $\chi(r)=0$ for $r\leq (3-\frac{1}{2})M$ and $r\geq (3+\frac{1}{2})M$. We then obtain:
\begin{multline*}
\int_{\tau_1}^{\tau_2}\int_{\Sigma_{\tau}\cap \{(3-\frac{1}{4})M\leq r\leq (3+\frac{1}{4})M\} } |\snabla_{\s^2}\phi|^2\,d\sigma dr d\tau\\
\leq C\int_{\tau_1}^{\tau_2}\int_{\Sigma_{\tau}} \left[r^2(DX\phi)^2+\phi^2+h_0(T\phi)^2+(1-3Mr^{-1})^2|\snabla_{\s^2}\phi|^2\right]\,d\sigma dr d\tau+CE_T[\phi](\tau_2)+C E_T[\phi](\tau_1).
\end{multline*}
We conclude \eqref{eq:fulliedlossder} by applying \eqref{eq:auxremovedegiled} and \eqref{eq:ebound}.
\end{proof}

In order to remove the factors $D$ in the (integrated) energies that appear in \eqref{eq:fullied} and which correspond to degeneracies at $r=2M$, we apply the \emph{red-shift mechanism}, developed in \cite{redshift, lecturesMD}. We first introduce the \emph{non-degenerate energy} $E_{\mathbf{N}}[\phi](\tau)$:
\begin{equation*}
E_{\mathbf{N}}[\phi](\tau):=\int_{\Sigma_{\tau}}r^2(X\phi)^2+h\tilde{h} r^2(T\phi)^2+|\snabla_{\s^2}\phi|^2+\phi^2\,d\sigma dr.
\end{equation*}
We will moreover denote $E_{\mathbf{N}}[\phi]:=E_{\mathbf{N}}[\phi](0)$.
\begin{proposition}[Red-shift estimate]
\label{prop:redshift}
For $\delta>0$, there exists a constant $C=C(M,h,V_{\alpha},\delta)>0$ such that
\begin{equation}
\label{eq:redshift}
\begin{split}
\int_{\Sigma_{\tau_2}\cap\{r\leq (2+\delta)M\}}& \left(D^{-1}\underline{L}\psi\right)^2+r^{-2}|\snabla_{\s^2}\psi|^2+r^{-2}\psi^2\,d\sigma dr+ \int_{\tau_1}^{\tau_2}\int_{\Sigma_{\tau}\cap\{r\leq (2+\delta)M\}}\left(D^{-1}\underline{L}\psi\right)^2 \,d\sigma dr d\tau\\
\leq&\: C\int_{\Sigma_{\tau_1}\cap\{r\geq R\}}r(X(r\phi))^2\,d\sigma dr+CE_{\mathbf{N}}[\phi](\tau_1)
\end{split}
\end{equation}
\begin{proof}
We multiply both sides of \eqref{eq:waveeq2a} with $D^{-2}\underline{L}\psi$ to obtain
\begin{equation*}
\begin{split}
0=&\: 4D^{-1}\underline{L} \psi  D^{-1}L\underline{L}\psi-r^{-2} D^{-1}\underline{L}\psi \slashed{\Delta}_{\s^2}\psi+ D^{-1}\underline{L}\psi(V_{\alpha}+D'r^{-1})\psi\\
=&\: 2L\left[\left(D^{-1}\underline{L}\psi\right)^2\right]+2\frac{dD}{dr}\left(D^{-1}\underline{L}\psi\right)^2+\frac{1}{2}D^{-1}\underline{L}\left[r^{-2}|\snabla_{\s^2}\psi|^2+(V_{\alpha}+r^{-1} D')\psi^2\right]-\frac{1}{2}r^{-3}|\snabla_{\s^2}\psi|^2\\
&+\frac{1}{4}\frac{d}{dr}\left(V_{\alpha}+D'r^{-1}\right)\psi^2+\textnormal{div}_{\s^2}(\ldots)\\
=&\: DX \left[\left(D^{-1}\underline{L}\psi\right)^2 \right]+T\left[h \left(D^{-1}\underline{L}\psi\right)^2+\frac{1}{4}\tilde{h}r^{-2}|\snabla_{\s^2}\psi|^2+ \frac{1}{4}\tilde{h}(V_{\alpha}+r^{-1} D')\psi^2\right]-\frac{1}{4}X\left[r^{-2}|\snabla_{\s^2}\psi|^2+(V_{\alpha}+r^{-1} D')\psi^2\right]\\
&+2\frac{dD}{dr}\left(D^{-1}\underline{L}\psi\right)^2-\frac{1}{2}r^{-3}|\snabla_{\s^2}\psi|^2+\frac{1}{4}\frac{d}{dr}\left(V_{\alpha}+D'r^{-1}\right)\psi^2+\textnormal{div}_{\s^2}(\ldots)\\
=&\: X \left[D \left(D^{-1}\underline{L}\psi\right)^2-\frac{1}{4}r^{-2}|\snabla_{\s^2}\psi|^2-\frac{1}{4}\left(V_{\alpha}+r^{-1} D'\right)\psi^2 \right]+T\left[h \left(D^{-1}\underline{L}\psi\right)^2+\frac{1}{4}\tilde{h}r^{-2}|\snabla_{\s^2}\psi|^2+ \frac{1}{4}\tilde{h}(V_{\alpha}+r^{-1} D')\psi^2\right]\\
&+2\frac{dD}{dr}\left(D^{-1}\underline{L}\psi\right)^2-\frac{1}{2}r^{-3}|\snabla_{\s^2}\psi|^2+\frac{1}{4}\frac{d}{dr}\left(V_{\alpha}+D'r^{-1}\right)\psi^2+\textnormal{div}_{\s^2}(\ldots).
\end{split}
\end{equation*}
Now, we integrate the above equation in $\{2M\leq r\leq (2+\delta)M\}$, with $\delta>0$ suitably small to guarantee $V_{\alpha}+r^{-1} D' >0$. All the resulting boundary terms are then non-negative definite. Hence, we obtain for $\kappa=\frac{dD}{dr}|_{r=2M}=\frac{1}{2M}$:
\begin{equation}
\label{eq:redshiftest}
\begin{split}
&\int_{\Sigma_{\tau_2}\cap\{r\leq (2+\delta)M\}}h \left(D^{-1}\underline{L}\psi\right)^2+\frac{1}{4}\tilde{h}r^{-2}|\snabla_{\s^2}\psi|^2+ \frac{1}{4}\tilde{h}(V_{\alpha}+r^{-1} D')\psi^2\,d\sigma dr\\
&+\frac{2}{1+\frac{1}{2}\delta}\kappa \int_{\tau_1}^{\tau_2}\int_{\Sigma_{\tau}\cap\{r\leq (2+\delta)M\}}\left(D^{-1}\underline{L}\psi\right)^2 \,d\sigma dr d\tau\\
\leq &\: C\int_{\tau_1}^{\tau_2}\int_{\s^2} r^{-2}|\snabla_{\s^2}\psi|^2+r^{-2}\psi^2\Bigg|_{r=(2+\delta)M}\,d\sigma d\tau+C \int_{\tau_1}^{\tau_2}\int_{\Sigma_{\tau}\cap\{r\leq (2+\delta)M\}} r^{-3}|\snabla_{\s^2}\psi|^2+r^{-3}\psi^2\,d\sigma dr d\tau\\
&+\int_{\Sigma_{\tau_1}\cap\{r\leq (2+\delta)M\}}h \left(D^{-1}\underline{L}\psi\right)^2+\frac{1}{4}\tilde{h}r^{-2}|\snabla_{\s^2}\psi|^2+ \frac{1}{4}\tilde{h}(V_{\alpha}+r^{-1} D')\psi^2\,d\sigma dr.
\end{split}
\end{equation}
We conclude that \eqref{eq:redshift} holds by combining the above equation with \eqref{eq:ebound}, \eqref{eq:hardy}, \eqref{eq:fullied} and a standard averaging argument to estimate the $r=(2+\delta)M$ boundary terms.
\end{proof}

\begin{corollary}[Non-degenerate integrated energy estimates]
\label{cor:nondegiled}
For $R>2M$ suitably large, there exists a constant $C=C(M,h,V_{\alpha},R)>0$ such that
\begin{align}
\label{eq:fulliednondeg}
\int_{\Sigma_{\tau_2}\cap\{r\geq R\}}& r(X(r\phi))^2\,d\sigma dr+\int_{\tau_1}^{\tau_2}\int_{\Sigma_{\tau}} (1-3Mr^{-1})^2\left[r^2(X\phi)^2+|\snabla_{\s^2}\phi|^2+h_0(T\phi)^2\right]+\phi^2\,d\sigma dr d\tau\\ \nonumber
\leq &\:C\int_{\Sigma_{\tau_1}\cap\{r\geq R\}}r(X(r\phi))^2\,d\sigma dr+CE_{\mathbf{N}}[\phi](\tau_1),\\
\label{eq:fulliednondeglossder}
\int_{\Sigma_{\tau_2}\cap\{r\geq R\}}& r(X(r\phi))^2\,d\sigma dr+\int_{\tau_1}^{\tau_2}\int_{\Sigma_{\tau}} \left[r^2(X\phi)^2+\phi^2+|\snabla_{\s^2}\phi|^2+h_0(T\phi)^2\right]\,d\sigma dr d\tau\\ \nonumber
\leq &\:C\int_{\Sigma_{\tau_1}\cap\{r\geq R\}}r(X(r\phi))^2+r(X(rT\phi))^2\,d\sigma dr+CE_{\mathbf{N}}[\phi](\tau_1)+CE_T[T\phi](\tau_1).
\end{align}
Furthermore, for $N\in \N_0$ and $3M<R_0<R_1$, there exists a constant $C=C(M,h,V_{\alpha},R_0,R_1,N)>0$ such that
\begin{equation}
\label{eq:hoiled}
\begin{split}
\sum_{n_1+n_2+n_3\leq N+1} \int_{\tau_1}^{\tau_2}&\int_{\Sigma_{\tau}\cap\{R_0\leq r\leq R_1\}} |\snabla_{\s^2}^{n_1}X^{n_2}T^{n_3}\phi|^2\,d\sigma dr d\tau\leq C\sum_{n\leq N}\int_{\Sigma_{\tau_1}\cap\{r\geq R\}}r(X(rT^n\phi))^2\,d\sigma dr+CE_T[T^n\phi](\tau_1).
\end{split}
\end{equation}
\end{corollary}
\end{proposition}
\begin{proof}
We obtain \eqref{eq:fulliednondeg} and \eqref{eq:fulliednondeglossder} by combining \eqref{eq:fullied} and \eqref{eq:fulliedlossder} with \eqref{eq:redshift}.

The inequality \eqref{eq:hoiled} follows by applying \eqref{eq:fullied}, restricted to $R_0\leq r\leq R_1$ and then apply standard elliptic estimates to express higher-order $X$ and $\snabla_{\s^2}$ derivatives in terms of higher-order $T$ derivatives.
\end{proof}

Applying the red-shift estimate, we can moreover derive a non-degenerate energy boundedness estimate:
\begin{corollary}[Non-degenerate energy estimate]
\label{cor:redshifebound}
There exists a constant $C=C(M,h,V_{\alpha})>0$ such that for all $0<\tau_1\leq \tau_2$
\begin{equation}
\label{eq:nondegebound}
E_{\mathbf{N}} [\phi](\tau_2)\leq CE_{\mathbf{N}} [\phi](\tau_1).
\end{equation}
\end{corollary}
\begin{proof}
By adding to both sides of \eqref{eq:redshiftest} the spacetime integral $\int_{\tau_1}^{\tau_2}E_T[\phi](\tau)\,d\tau$, we obtain that for some $b>0$ we have for all $0\leq \tau_1\leq \tau_2$:
\begin{equation*}
E_{\mathbf{N}}[\phi](\tau_2)+b \int_{\tau_1}^{\tau_2}E_{\mathbf{N}}(\tau)\,d\tau\leq C \int_{\tau_1}^{\tau_2}E_T[\phi](\tau)\,d\tau+E_{\mathbf{N}}[\phi](\tau_1).
\end{equation*}

We have that $\int_{\tau_1}^{\tau_2}E_T[\phi](\tau)\,d\tau \leq C(\tau_2-\tau_1)E_T[\phi](\tau_1)$ as a result of \eqref{eq:ebound}, so after: 1) rearranging terms, 2) dividing by $\tau_2-\tau_1$ and 3) taking the limit $\tau_2\downarrow \tau_1$, we obtain for all $\tau\geq 0$ the following inequality: for all $\tau\geq \tau_1$
\begin{equation*}
\frac{d}{d\tau}E_{\mathbf{N}}[\phi](\tau)+bE(\tau)\leq CE_T[\phi](\tau_1),
\end{equation*}
which is equivalent to
\begin{equation*}
\frac{d}{d\tau}(e^{b(\tau-\tau_1)}E_{\mathbf{N}}[\phi](\tau))\leq Ce^{b(\tau-\tau_1)}E_T[\phi](\tau_1)
\end{equation*}
from which \eqref{eq:nondegebound} follows after integrating in $\tau$.
\end{proof}
\section{Higher-order $r$-weighted energy estimates}
\label{sec:rp}
In this section, we will derive $r$-weighted energy estimates for weighted higher-order derivatives of $\psi$ of the form $(rL)^k\psi$, with $k\in \N_0$. These will form the main tool for establishing time decay of higher-order $T$-derivatives in the next section.
\begin{lemma}
Let $N\in \N_0$ and let $\phi\in C^{\infty}(\mathcal{R})\cap C^1 (\widehat{\mathcal{R}})$ be a solution to \eqref{eq:waveeq}. Then
\begin{equation}
\label{eq:waveqcommrL}
\begin{split}
 0=&\:-4\underline{L} L(rD^{-1}L)^N\psi+Dr^{-2}[\slashed{\Delta}_{\s^2}-\alpha+O_{\infty}(r^{-1})](rD^{-1}L)^N\psi-2N[r^{-1}+O_{\infty}(r^{-2})]L(rD^{-1}L)^N\psi\\
 &-\frac{N}{2}Dr^{-2}\slashed{\Delta}_{\s^2}(rD^{-1}L)^{N-1}\psi+N\sum_{n=0}^N O_{\infty}(r^{-2})(rD^{-1}L)^n\psi+N(N-1)\sum_{n=0}^{N-2} O_{\infty}(r^{-2})\slashed{\Delta}_{\s^2}(rD^{-1}L)^n\psi.
 \end{split}
\end{equation}
\end{lemma}
\begin{proof}
We will prove the identity via induction. Note first of all that the $N=0$ case follows from \eqref{eq:waveeq2a}. Now suppose \eqref{eq:waveqcommrL} holds. Then we can multiply both sides of the equation by $rD^{-1}$ and rearrange some terms to obtain:
\begin{equation*}
\begin{split}
 0=&\:-4\underline{L} (rD^{-1}L)^{N+1}\psi+r^{-1}[\slashed{\Delta}_{\s^2}-\alpha+O_{\infty}(r^{-1})](rD^{-1}L)^N\psi-2(N+1)[r^{-1}+O_{\infty}(r^{-2})](rD^{-1}L)^{N+1}\psi\\
 &-\frac{N}{2}r^{-1}\slashed{\Delta}_{\s^2}(rD^{-1}L)^{N-1}\psi+N\sum_{n=0}^N O_{\infty}(r^{-1})(rD^{-1}L)^n\psi+N(N-1)\sum_{n=0}^{N-2} O_{\infty}(r^{-1})\slashed{\Delta}_{\s^2}(rD^{-1}L)^n\psi.
 \end{split}
\end{equation*}
Then, act with $L$ on both sides to obtain:
\begin{equation*}
\begin{split}
 0=&\:-4\underline{L} L(rD^{-1}L)^{N+1}\psi+Dr^{-2}[\slashed{\Delta}_{\s^2}-\alpha+O_{\infty}(r^{-1})](rD^{-1}L)^{N+1}\psi-[2(N+1)r^{-1}+O_{\infty}(r^{-2})](rD^{-1}L)^{N+1}\psi\\
 &-\frac{N+1}{2}Dr^{-2}\slashed{\Delta}_{\s^2}(rD^{-1}L)^{N}\psi+\sum_{n=0}^{N+1} O_{\infty}(r^{-2})(rD^{-1}L)^n\psi+N\sum_{n=0}^{N-1} O_{\infty}(r^{-2})\slashed{\Delta}_{\s^2}(rD^{-1}L)^n\psi.
 \end{split}
\end{equation*}
Hence, \eqref{eq:waveqcommrL} holds also with $N$ replaced by $N+1$ and the induction argument is complete.
\end{proof}

\begin{proposition}[Higher-order $r$-weighted energy estimates]
\label{prop:horp}
Let $N\in \N_0$ and let $\phi\in C^{\infty}(\mathcal{R})\cap C^1 (\widehat{\mathcal{R}})$ be a solution to \eqref{eq:waveeq}. Let $\max\{1-\beta_0,0\}<p<\min\{1+\beta_0,2\}$. Then, for $R>3M$ suitably large, there exists a constant $C=C(M,h,N,R,\alpha,p)>0$ such that
\begin{equation}
\label{eq:horp}
\begin{split}
\sum_{0\leq n_1+n_2\leq N}&\int_{\Sigma_{\tau_2}\cap\{r\geq R\}} r^{p} (L(rL)^{n_1}T^{n_2}\psi)^2\,d\sigma dr\\
&+\int_{\tau_1}^{\tau_2} \int_{\Sigma_{\tau}\cap\{r\geq R\}} r^{p-1} (L(rL)^{n_1}T^{n_2}\psi)^2+r^{p-3} |\snabla_{\s^2}(rL)^{n_1}T^{n_2}\psi|^2+r^{p-3} ((rL)^{n_1}T^{n_2}\psi)^2\,d\sigma dr d\tau\\
\leq & C \sum_{0\leq n_1+n_2\leq N}\int_{\Sigma_{\tau_1}\cap\{r\geq R\}} r^{p} (L(rL)^{n_1}T^{n_2}\psi)^2\,d\sigma dr+C\sum_{n\leq N}\int_{\Sigma_{\tau_1}\cap\{r\geq R\}}r(XT^n\psi)^2\,d\sigma dr\\
&+C\sum_{0\leq n\leq N}E_T[T^n\phi](\tau_1).
\end{split}
\end{equation}
\end{proposition}
\begin{proof}
Let $\chi: [M,\infty)\to \R_{\geq 0}$ be a smooth cut-off function, such that $\chi(r)=1$ for $r\geq R$ and $\chi(r)=0$ for $r\leq 3M+\frac{R-3M}{2}$.

We multiply both sides of \eqref{eq:waveqcommrL} with $-\chi r^p L(rD^{-1}L)^{N}\psi$ and apply the Leibniz rule to obtain:
\begin{equation}
\begin{split}
\label{eq:horpid}
0=&\:2\chi r^p \underline{L} ((L(rD^{-1}L)^N\psi)^2)-D\chi r^{p-2}[\slashed{\Delta}_{\s^2}-\alpha+O_{\infty}(r^{-1})](rD^{-1}L)^N\psi L(rD^{-1}L)^{N}\psi\\
&+2N[r^{p-1}+O_{\infty}(r^{p-2})](L(rD^{-1}L)^N\psi)^2+\frac{N}{2}D\chi r^{p-2}\slashed{\Delta}_{\s^2}(rD^{-1}L)^{N-1}\psi L(rD^{-1}L)^{N}\psi\\
&+N\sum_{n=0}^N \chi O_{\infty}(r^{p-2})(rD^{-1}L)^n\psi L(rD^{-1}L)^{N}\psi+N(N-1)\sum_{n=0}^{N-2} \chi O_{\infty}(r^{p-2})\slashed{\Delta}_{\s^2}(rD^{-1}L)^n\psi L(rD^{-1}L)^{N}\psi\\
=&\:\underline{L} \left[2\chi r^p (L(rD^{-1}L)^N\psi)^2\right]+L\left[\frac{1}{2}D \chi r^{p-2}\left(|\snabla_{\s^2}(rD^{-1}L)^{N}\psi|^2+(\alpha+O_{\infty}(r^{-1}))((rD^{-1}L)^{N}\psi)^2\right) \right]\\
&+ \chi r^{p-1}\left[p+2N+O_{\infty}(r^{-1})\right](L(rD^{-1}L)^N\psi)^2+\frac{1}{4}(2-p)\chi  r^{p-3}(1+O_{\infty}(r^{-1})) |\snabla_{\s^2}(rD^{-1}L)^{N}\psi|^2\\
&+\frac{1}{4}(2-p)\chi  r^{p-3}(\alpha+O_{\infty}(r^{-1}))((rD^{-1}L)^{N}\psi)^2++\frac{N}{2}D\chi r^{p-2}\slashed{\Delta}_{\s^2}(rD^{-1}L)^{N-1}\psi L(rD^{-1}L)^{N}\psi\\
&+N\sum_{n=0}^N \chi O_{\infty}(r^{p-2})(rD^{-1}L)^n\psi L(rD^{-1}L)^{N}\psi+N(N-1)\sum_{n=0}^{N-2} \chi O_{\infty}(r^{p-2})\slashed{\Delta}_{\s^2}(rD^{-1}L)^n\psi L(rD^{-1}L)^{N}\psi\\
&+\textnormal{Err}_{\chi,N}+\textnormal{div}_{\s^2}(\ldots).
\end{split}
\end{equation}
where $\textnormal{Err}_{\chi,N+1}$ denotes all terms multiplied by derivatives of $\chi$ and squares of derivatives of $\phi$ up to order $N+1$. We can further rewrite \eqref{eq:horpid} as follows:

\begin{equation}
\label{eq:horpidv2}
\begin{split}
0=&\:  DX\left[\frac{1}{4} \chi r^{p-2}\left(|\snabla_{\s^2}(rD^{-1}L)^{N}\psi|^2+(\alpha+O_{\infty}(r^{-1}))((rD^{-1}L)^{N}\psi)^2\right) -\chi r^p (L(rD^{-1}L)^N\psi)^2\right]\\
&+DT\left[\chi \tilde{h} r^p (L(rD^{-1}L)^N\psi)^2+\frac{h}{4D} \chi r^{p-2}\left(|\snabla_{\s^2}(rD^{-1}L)^{N}\psi|^2+(\alpha+O_{\infty}(r^{-1}))((rD^{-1}L)^{N}\psi)^2\right) \right]\\
&+ \chi Dr^{p-1}\left[p+2N+O_{\infty}(r^{-1})\right](L(rD^{-1}L)^N\psi)^2+\frac{1}{4}(2-p)\chi  r^{p-3}(1+O_{\infty}(r^{-1})) |\snabla_{\s^2}(rD^{-1}L)^{N}\psi|^2\\
&+\frac{1}{4}(2-p)D\chi  r^{p-3}(\alpha+O_{\infty}(r^{-1}))((rD^{-1}L)^{N}\psi)^2+\frac{N}{2}D\chi r^{p-2}\slashed{\Delta}_{\s^2}(rD^{-1}L)^{N-1}\psi L(rD^{-1}L)^{N}\psi\\
&+ND\sum_{n=0}^N \chi O_{\infty}(r^{p-2})(rD^{-1}L)^n\psi L(rD^{-1}L)^{N}\psi+N(N-1)D\sum_{n=0}^{N-2} \chi O_{\infty}(r^{p-2})\slashed{\Delta}_{\s^2}(rD^{-1}L)^n\psi L(rD^{-1}L)^{N}\psi\\
&+\textnormal{Err}_{\chi,N+1}+\textnormal{div}_{\s^2}(\ldots).
\end{split}
\end{equation}

We will prove \eqref{eq:horp} by induction. We first show that \eqref{eq:horp} holds for $N=0$. We integrate \eqref{eq:horpidv2} with $N=0$. We apply \eqref{eq:hardy} to estimate
\begin{equation*}
\begin{split}
\int_{\tau_1}^{\tau_2}\int_{\Sigma_{\tau}}& \frac{1}{4}(2-p)D\chi  r^{p-3}(\max\{-\alpha,0\}+O_{\infty}(r^{-1}))((rD^{-1}L)^{N}\psi)^2 \,d\sigma dr d\tau\\
\leq &\: \int_{\tau_1}^{\tau_2}\int_{\Sigma_{\tau}}\chi  r^{p-3}((2-p)^{-1}\max\{-\alpha,0\}+C r^{-1})r^{p-1}(X(rD^{-1}L)^{N}\psi)^2 \,d\sigma dr d\tau\\
+& \int_{\tau_1}^{\tau_2}\int_{\Sigma_{\tau}}\textnormal{Err}_{\chi,N}\,d\sigma dr d\tau\\
\leq &\: \int_{\tau_1}^{\tau_2}\int_{\Sigma_{\tau}}4\chi  ((2-p)^{-1}\max\{-\alpha,0\}+C r^{-1})r^{p-1}(L(rD^{-1}L)^{N}\psi)^2+C r^{p-1}h^2 (T (rD^{-1}L)^{N}\psi)^2\,d\sigma dr d\tau\\
&+ \int_{\tau_1}^{\tau_2}\int_{\Sigma_{\tau}}\textnormal{Err}_{\chi,N}\,d\sigma dr d\tau.
\end{split}
\end{equation*}
We can absorb the first term on the very right-hand side above into the spacetime integral of $\chi Dp r^{p-1} (L(rD^{-1}L)^{N}\psi)^2$ provided that $R>3M$ is chosen suitably large and moreover, when $\alpha<0$, $1-\beta_0<p<1+\beta_0$. The second term on the very right-hand side, which is only present in the $M>0$ case, can be estimated using \eqref{eq:fullied} when $N=0$. We moreover use \eqref{eq:hoiled} to estimate $\textnormal{Err}_{\chi,N+1}$, so that we are left with  \eqref{eq:horp} for $N=0$.

Now suppose we have already established \eqref{eq:horp} for $N$ replaced with $N-1$. We will show that \eqref{eq:horp} holds also for $N$. We integrate \eqref{eq:horpidv2} again and apply \eqref{eq:hardy} as above. In this case, however, we estimate the integral of $r^{p-1}h^2 (T (rD^{-1}L)^{N}\psi)^2$ by observing that
\begin{equation*}
r^{p-1}h^2 (T (rD^{-1}L)^{N}\psi)^2\leq C r^{p-3}(L (rD^{-1})^{N-1}T\psi)^2
\end{equation*}
and applying the induction hypothesis with $\psi$ replaced by $T\psi$, using that $[\square_g,T]=0$ by the Killing property of $T$. We moreover estimate
\begin{equation*}
\begin{split}
ND&\sum_{n=0}^N \chi O_{\infty}(r^{p-2})(rD^{-1}L)^n\psi L(rD^{-1}L)^{N}\psi+N(N-1)D\sum_{n=0}^{N-2} \chi O_{\infty}(r^{p-2})\slashed{\Delta}_{\s^2}(rD^{-1}L)^n\psi L(rD^{-1}L)^{N}\psi\\
\leq&\: \epsilon \chi r^{p-1}(L(rD^{-1}L)^{N}\psi)^2 +\epsilon \chi r^{p-3}|\snabla_{\s^2}(rD^{-1}L)^{N}\psi|^2+C\chi  \epsilon^{-1} \sum_{n=0}^N r^{p-3}((rD^{-1}L)^n\psi )^2\\
&+C \chi \sum_{n=0}^{N-2} \chi O_{\infty}(r^{p-1})|\snabla_{\s^2}(rD^{-1}L)^n\psi|^2+\textnormal{div}_{\s^2}(\ldots).
\end{split}
\end{equation*}
We can absorb the terms with a factor $\epsilon$ into $\chi D(p+2N) r^{p-1} (L(rD^{-1}L)^{N}\psi)^2$ and we apply the induction hypothesis together with \eqref{eq:hardy} to absorb also the terms in the sums above.

We are left with estimating the spacetime integral of $\frac{N}{2}D\chi r^{p-2}\slashed{\Delta}_{\s^2}(rD^{-1}L)^{N-1}\psi L(rD^{-1}L)^{N}\psi$. We first apply the Leibniz rule to obtain:
\begin{equation*}
\begin{split}
\frac{N}{2}&D\chi r^{p-2}\slashed{\Delta}_{\s^2}(rD^{-1}L)^{N-1}\psi L(rD^{-1}L)^{N}\psi=-L\left[\frac{N}{4}D\chi r^{p-2}\snabla_{\s^2}(rD^{-1}L)^{N-1}\psi\cdot \snabla_{\s^2}(rD^{-1}L)^{N}\right]\\
&+\frac{N}{4}D^2\chi r^{p-3}|\snabla_{\s^2}(rD^{-1}L)^{N}\psi|^2+O_{\infty}(r^{p-3})\snabla_{\s^2}(rD^{-1}L)^{N-1}\psi\cdot \snabla_{\s^2}(rD^{-1}L)^{N}\psi.
\end{split}
\end{equation*}
Note that the second term on the right-hand side has a good sign, whereas the third term can be estimated by applying Young's inequality together with the induction hypothesis.

The boundary integral terms along $\Sigma_{\tau_2}$ that are generated in the above estimates either have a good sign or can be estimated by using the induction hypothesis.
\end{proof}

\section{Decay-in-time estimates for time derivatives}
\label{sec:decaytimeder}
The goal of this section is to apply Proposition \ref{prop:horp} in order to obtain a hierarchy of integrated $r$-weighted energy estimates for time derivatives $T^k\phi$, with $k\in \N_0$, and convert this hierarchy into energy decay via a suitable use of the mean value theorem. We will show that larger values of $k$ correspond to longer hierarchies and faster decay in time.

In the lemma below we show that we can express energies involving $T$ derivatives in terms of energies involving $L$ and $\snabla_{\s^2}$ derivatives
\begin{lemma}
Let $R>3M$. For all $p\in \R$, we have that there exists a constant $C=C(R)>0$, such that in $\{r\geq R\}$:
\begin{equation}
\label{eq:convTintorX}
\begin{split}
\sum_{n_1+n_2\leq N} &\int_{\Sigma\cap\{r\geq R\}} r^{p+1}|\snabla_{\s^2}^{n_1} L T (rL)^{n_2}\psi|^2\,d\sigma dr\\
\leq  &\: C\sum_{\substack{ n_1+n_2\leq N+1}}\int_{\Sigma\cap\{r\geq R\}} r^{p-1}|\snabla_{\s^2}^{n_1} L (rL)^{n_2}\psi|^2+ r^{p-3}|\snabla_{\s^2}^{n_1+1} (rL)^{n_2}\psi|^2\,d\sigma dr.
\end{split}
\end{equation}
\end{lemma}
\begin{proof}
We use that $T=L+\underline{L}$ to estimate in $\{r\geq R\}$:
\begin{equation*}
\begin{split}
 r^{p+1}|\snabla_{\s^2}^{n_1} L T (rD^{-1}L)^{n_2}\psi|^2\leq &\: Cr^{p-1}|\snabla_{\s^2}^{n_1}  (r D^{-1}L)^{n_2+1}\psi|^2+Cr^{p-1}|\snabla_{\s^2}^{n_1}  (rD^{-1}L)^{n_2}\psi|^2\\
 &+ r^{p+1}|\snabla_{\s^2}^{n_1} \underline{L} L (rD^{-1}L)^{n_2}\psi|^2.
 \end{split}
\end{equation*}
Then we apply \eqref{eq:waveqcommrL} to estimate further:
\begin{equation*}
 r^{p+1}|\snabla_{\s^2}^{n_1} \underline{L} L (rD^{-1}L)^{n_2}\psi|^2\leq C \sum_{k=0}^{n_2} r^{p-3}|\snabla_{\s^2}^{n_1}\slashed{\Delta}_{\s^2} (rD^{-1}L)^{k}\psi|^2+C \sum_{k=0}^{n_2+1} r^{p-3}|\snabla_{\s^2}^{n_1} (rD^{-1}L)^{k}\psi|^2.
\end{equation*}
The inequality \eqref{eq:convTintorX} then follows by combining the above with \eqref{eq:sphere4}.
\end{proof}

\begin{proposition}
\label{prop:edecaytimeder}
Let $N\in \N$ and let $\delta>0$ be arbitrarily small. Then for $R>3M$ suitably large, there exists a constant $C=~C(M,h,\alpha,N,R,\delta)~>~0$ such that for all $\tau\geq 0$:
\begin{equation}
\label{eq:mainedcay}
\begin{split}
E_{\mathbf{N}}[T^N\phi](\tau)\leq &\:C(1+\tau)^{-1-2N+\delta} \biggg[\sum_{\substack{0\leq n_1+n_2+n_3\leq 2N\\ n_1+n_2\leq N}}\int_{\Sigma_{0}\cap \{r\geq R\}} r|\snabla_{\s^2}^{n_1}L(rL)^{n_2}T^{n_3}\psi|^2\,d\sigma dr\\
&+ \sum_{\substack{0\leq n_1+n_2\leq 3N\\ n_1\leq N}}E_{\mathbf{N}}[\snabla_{\s^2}^{n_1}T^{n_2}\phi](0)\biggg].
\end{split}
\end{equation}
and moreover,
\begin{equation}
\label{eq:mainrwedcay}
\begin{split}
\int_{\Sigma_{\tau}\cap \{r\geq R\}} r^2(LT^N\psi)\,d\sigma dr\leq &\: C(1+\tau)^{1-2N+\delta}\biggg[\sum_{\substack{0\leq n_1+n_2+n_3\leq 2N-1\\ n_1+n_2\leq N}}\int_{\Sigma_{0}\cap \{r\geq R\}} r|\snabla_{\s^2}^{n_1}L(rL)^{n_2}T^{n_3}\psi|^2\,d\sigma dr\\
&+ \sum_{\substack{0\leq n_1+n_2\leq 3N-1\\ n_1\leq N}}E_{\mathbf{N}}[\snabla_{\s^2}^{n_1}T^{n_2}\phi](0)\biggg].
\end{split}
\end{equation}
\end{proposition}
\begin{proof}
We first consider the case $N=0$. Let $\{\tau_i\}$ be a dyadic sequence with $\tau_i\to \infty$ as $i\to \infty$. Then we can apply the integrated energy estimate \eqref{eq:fulliednondeglossder} to conclude that
\begin{equation}
\label{eq:firstestrp}
\begin{split}
\int_{\Sigma_{\tau_{i+1}}\cap \{r\geq R\}}& r(L\psi)^2\,d\sigma dr+\int_{\tau_i}^{\tau_{i+1}} E_{\mathbf{N}}[\phi](\tau)\,d\tau\\
\leq&\: C \int_{\Sigma_{\tau_i}\cap \{r\geq R\}} r(L\psi)^2+r(LT\psi)^2\,d\sigma dr+CE_{\mathbf{N}}[\phi](\tau_i)+C E_T[T\phi](\tau_i).
\end{split}
\end{equation}

Hence, after applying the mean-value theorem, we conclude that there exists a sequence $\{\tau_i'\}$ with $\tau_i\leq ~\tau_i\leq \tau_{i+1}$ along which
\begin{equation*}
E_{\mathbf{N}}[\phi](\tau_i')\leq C (1+\tau_i')^{-1}\left[ \int_{\Sigma_{\tau_i}\cap \{r\geq R\}} r(L\psi)^2+ r(LT\psi)^2\,d\sigma dr+E_{\mathbf{N}}[\phi](\tau_i)+C E_T[T\phi](\tau_i)\right].
\end{equation*}
By applying \eqref{eq:nondegebound}, we then get that for all $\tau\geq 0$:
\begin{equation}
\label{eq:firstedecay}
E_{\mathbf{N}}[\phi](\tau)\leq C (1+\tau)^{-1}\left[ \int_{\Sigma_{0}\cap \{r\geq R\}} r(L\psi)^2+r(LT\psi)^2\,d\sigma dr+E_{\mathbf{N}}[\phi](0)+C E_T[T\phi](0)\right].
\end{equation}
Now let $N=1$ and replace $\psi$ with $T\psi$. We then proceed to estimate further: let $\{\tau_i\}$ be a dyadic sequence and $\gamma\geq 0$, then
\begin{equation}
\label{eq:gammasplit}
\begin{split}
\int_{\tau_i}^{\tau_{i+1}}\int_{\Sigma_{\tau}\cap \{r\geq R\}} r(LT\psi)^2\,d\sigma dr d\tau=&\:\int_{\tau_i}^{\tau_{i+1}}\int_{\Sigma_{\tau}\cap \{r\geq R\}\cap \{r\leq R+{\tau_i}^{\gamma}\}} r(LT\psi)^2\,d\sigma dr d\tau\\
&+\int_{\tau_i}^{\tau_{i+1}}\int_{\Sigma_{\tau}\cap \{r\geq R+{\tau_i}^{\gamma}\}} r(LT\psi)^2\,d\sigma dr d\tau\\
\leq &\:C(1+\tau_i)^{\gamma} \int_{\tau_i}^{\tau_{i+1}}\int_{\Sigma_{\tau}\cap \{r\geq R\}} (LT\psi)^2\,d\sigma dr d\tau\\
&+C(1+\tau_i)^{-\gamma} \int_{\tau_i}^{\tau_{i+1}}\int_{\Sigma_{\tau}\cap \{r\geq R\}} r^2(LT\psi)^2\,d\sigma dr d\tau.
\end{split}
\end{equation}

Subsequently, we apply \eqref{eq:convTintorX} with $N=1$ to estimate further
\begin{equation}
\label{eq:Tderrweight}
\begin{split}
 \int_{\Sigma_{\tau}\cap \{r\geq R\}} r^2(LT\psi)^2\,d\sigma dr\leq &\:C \int_{\Sigma_{\tau}\cap \{r\geq R\}} (L(rL)\psi)^2\,d\sigma dr\\
 &+C\sum_{k=0}^1\int_{\Sigma_{\tau}\cap \{r\geq R\}}  r^2|\snabla_{\s^2}^kX\phi|^2+|\snabla_{\s^2}^k\phi|^2+|\snabla_{\s^2}^{k+1}\phi|^2+ r^2h\tilde{h} |\snabla_{\s^2}^{k}T\phi|^2\,d\sigma dr.
 \end{split}
\end{equation}

We can therefore integrate \eqref{eq:Tderrweight} in $\tau$ and apply \eqref{eq:horp} with $N=1$ to obtain:
\begin{equation}
\label{eq:Tderrweightint}
\begin{split}
 \int_{\tau_i}^{\tau_{i+1}}&\int_{\Sigma_{\tau}\cap \{r\geq R\}} r^2(LT\psi)^2\,d\sigma dr d\tau\leq C\sum_{n_1+n_2\leq1}\int_{\Sigma_{\tau_i}\cap \{r\geq R\}} r(L(rL)^{n_1}T^{n_2}\psi)^2\,d\sigma dr+CE_T[\phi](\tau_i)+CE_T[T\phi](\tau_i)\\
 &+C\sum_{k=0}^1\int_{\tau_i}^{\tau_{i+1}}\int_{\Sigma_{\tau}\cap \{r\geq R\}}  r^2|\snabla_{\s^2}^kX\phi|^2+|\snabla_{\s^2}^k\phi|^2+|\snabla_{\s^2}^{k+1}\phi|^2+ r^2h\tilde{h} |\snabla_{\s^2}^{k}T\phi|^2\,d\sigma dr d\tau.
 \end{split}
\end{equation}
When $k=0$, we can estimate the second integral on the right-hand side of \eqref{eq:Tderrweightint} by applying \eqref{eq:fulliednondeg}. When $k>0$, we use the property $[\Omega_i,\square_{g_M}-V_{\alpha}]=0$ to directly apply \eqref{eq:fulliednondeg} to $\phi$ replaced with $\Omega_j\phi$, $j=1,2,3$. Then we can estimate additional angular derivatives in terms of $\Omega_j$-derivatives via \eqref{eq:sphere1b}.

Subsequently, we apply the mean-value theorem and choose $\gamma=0$ to obtain the existence of $\tau_i\leq \tau_i'\leq \tau_{i+1}$ such that
\begin{equation*}
\begin{split}
\int_{\Sigma_{\tau_i'}\cap \{r\geq R\}} r(LT\psi)^2\,d\sigma dr d\tau\leq&\: C(1+\tau_i)^{-1}\sum_{0\leq n_1+n_2+n_3\leq 1}\int_{\Sigma_{\tau_{i}}\cap \{r\geq R\}} r|\snabla_{\s^2}^{n_1} L(rL)^{n_2}T^{n_3}\psi|^2\,d\sigma dr\\
&+C(1+\tau_i)^{-1} \left[E_T[\phi](\tau_{i})+E_T[T\phi](\tau_{i})+E_T[\snabla_{\s^2}\phi](\tau_{i})\right].
 \end{split}
\end{equation*}
Applying \eqref{eq:horp} with $p=1$ and $N=0$, allows us to infer that in fact for all $\tau\geq 0$, we have that
\begin{equation}
\label{eq:rpwithgrowrweightb}
\begin{split}
\int_{\Sigma_{\tau}\cap \{r\geq R\}} r(LT\psi)^2\,d\sigma dr d\tau\leq&\: C(1+\tau)^{-1}\sum_{0\leq n_1+n_2+n_3\leq 1}\int_{\Sigma_{0}\cap \{r\geq R\}} r|\snabla_{\s^2}^{n_1} L(rL)^{n_2}T^{n_3}\psi|^2\,d\sigma dr\\
&+C(1+\tau)^{-1} \left[E_T[\phi](0)+E_T[T\phi](0)+E_T[\snabla_{\s^2}\phi](0)\right].
 \end{split}
\end{equation}

We can use \eqref{eq:rpwithgrowrweightb} to estimate the right-hand side of \eqref{eq:firstestrp} further and $\phi$ is replaced by $T\phi$. We then obtain:
\begin{equation}
\label{eq:edecaygamma0}
E_{\mathbf{N}}[T\phi](\tau)\leq C(1+\tau)^{-2}\left[\sum_{\substack{0\leq n_1+n_2+n_3\leq 2 \\ n_1+n_2\leq 1}}\int_{\Sigma_{0}\cap \{r\geq R\}} r|\snabla_{\s^2}^{n_1}L(rL)^{n_2}T^{n_3}\psi|^2\,d\sigma dr+ \sum_{\substack{0\leq n_1+n_2\leq 3\\n_1\leq 1}}E_{\mathbf{N}}[\snabla_{\s^2}^{n_1}T^{n_2}\phi](0)\right].
\end{equation}
Furthermore, we can split
\begin{equation*}
 \int_{\Sigma_{\tau}\cap \{r\geq R\}} r^2(LT\psi)^2\,d\sigma dr= \int_{\Sigma_{\tau}\cap \{R\leq r\leq R+\tau\}} r^2(LT\psi)^2\,d\sigma dr+ \int_{\Sigma_{\tau}\cap \{r\geq R+\tau\}} r^2(LT\psi)^2\,d\sigma dr
\end{equation*}
and estimate
\begin{equation*}
\int_{\Sigma_{\tau}\cap \{R\leq r\leq R+\tau\}} r^2(LT\psi)^2\,d\sigma dr\leq C(1+\tau) \int_{\Sigma_{\tau}\cap \{r\geq R\}} r(LT\psi)^2\,d\sigma dr
\end{equation*}
and, applying \eqref{eq:Tderrweight} with $R$ replaced by $R+\tau$, we can moreover estimate
\begin{equation*}
\int_{\Sigma_{\tau}\cap \{ r\geq R+\tau\}} r^2(LT\psi)^2\,d\sigma dr\leq C(1+\tau)^{-1} \int_{\Sigma_{\tau}}r (L(rL)\psi)^2\,d\sigma dr+CE_T[\phi](\tau)+CE_T[\snabla_{\s^2}\phi](\tau).
\end{equation*}
Hence,
\begin{equation}
\label{eq:rpwithgrowrweight}
\begin{split}
\int_{\Sigma_{\tau}\cap \{r\geq R\}} r^2(LT\psi)^2\,d\sigma dr\leq &\:C(1+\tau)^{-1}\Bigg[\sum_{\substack{0\leq n_1+n_2+n_3\leq 1 \\ n_1+n_2\leq 1}}\int_{\Sigma_{0}\cap \{r\geq R\}} r|\snabla_{\s^2}^{n_1}L(rL)^{n_2}T^{n_3}\psi|^2\,d\sigma dr\\
&+ \sum_{\substack{0\leq n_1+n_2\leq 2\\n_1\leq 1}}E_{\mathbf{N}}[\snabla_{\s^2}^{n_1}T^{n_2}\phi](0)\Bigg]+ C(1+\tau) \int_{\Sigma_{\tau}\cap \{r\geq R\}} r(LT\psi)^2\,d\sigma dr.
\end{split}
\end{equation}

Equipped with \eqref{eq:edecaygamma0} and \eqref{eq:rpwithgrowrweightb}, we can now make a better choice of $\gamma$ to improve the decay rate of $E_{\mathbf{N}}[T\phi](\tau)$: consider \eqref{eq:gammasplit} with $\gamma=\frac{1}{2}$ to obtain for all $\tau_i\leq \tau \leq \tau_{i+1}$:
\begin{equation*}
\begin{split}
\int_{\tau_i}^{\tau_{i+1}}\int_{\Sigma_{\tau}\cap \{r\geq R\}}& r(LT\psi)^2\,d\sigma dr d\tau\leq C(1+\tau_i)^{-1+\frac{1}{2}}\Bigg[\int_{\Sigma_{\tau_i}\cap \{r\geq R\}} r(LT\psi)^2\,d\sigma dr+CE_{\mathbf{N}}[T\phi](\tau_i)\Bigg]\\
&+C(1+\tau_i)^{-1-\frac{1}{2}}\sum_{n_1+n_2+n_3\leq 1}\int_{\Sigma_{\tau_{i}}\cap \{r\geq R\}} r|\snabla_{\s^2}^{n_1} L(rL)^{n_2}T^{n_3}\psi|^2\,d\sigma dr\\
&+C(1+\tau_i)^{-1-\frac{1}{2}}  \left[E_T[\phi](\tau_{i})+E_T[T\phi](\tau_{i})+E_T[\snabla_{\s^2}\phi](\tau_{i})\right]\\
\leq &\: C(1+\tau_i)^{-\frac{3}{2}}\left[\sum_{\substack{0\leq n_1+n_2+n_3\leq 1 \\ n_1+n_2\leq 1}}\int_{\Sigma_{0}\cap \{r\geq R\}} r|\snabla_{\s^2}^{n_1}L(rL)^{n_2}T^{n_3}\psi|^2\,d\sigma dr+ \sum_{\substack{0\leq n_1+n_2\leq 2\\n_1\leq 1}}E_{\mathbf{N}}[\snabla_{\s^2}^{n_1}T^{n_2}\phi](0)\right].
 \end{split}
\end{equation*}
and hence
\begin{equation*}
\begin{split}
E_{\mathbf{N}}[T\phi](\tau)\leq &\:C(1+\tau)^{-\frac{5}{2}}\left[\sum_{\substack{0\leq n_1+n_2+n_3\leq 2 \\ n_1+n_2\leq 1}}\int_{\Sigma_{\tau_i}\cap \{r\geq R\}} r|\snabla_{\s^2}^{n_1}L(rL)^{n_2}T^{n_3}\psi|^2\,d\sigma dr+C \sum_{\substack{0\leq n_1+n_2\leq 3\\n_1\leq 1}}E_{\mathbf{N}}[\snabla_{\s^2}^{n_1}T^{n_2}\phi](0)\right].
\end{split}
\end{equation*}
We can now continue this process iteratively by choosing $\gamma=2^{-k}$ and obtain:
\begin{align*}
\int_{\Sigma_{\tau}\cap \{r\geq R\}}& r(LT\psi)^2\,d\sigma dr d\tau\leq C(1+\tau)^{-2+\frac{1}{2^k}}\Bigg[\sum_{\substack{0\leq n_1+n_2+n_3\leq 1 \\ n_1+n_2\leq 1}}\int_{\Sigma_{\tau_i}\cap \{r\geq R\}} r|\snabla_{\s^2}^{n_1}L(rL)^{n_2}T^{n_3}\psi|^2\,d\sigma dr\\
&+C \sum_{\substack{0\leq n_1+n_2\leq 2\\n_1\leq 1}}E_{\mathbf{N}}[\snabla_{\s^2}^{n_1}T^{n_2}\phi](0)\Bigg],\\
E_{\mathbf{N}}[T\phi](\tau)\leq &\:C(1+\tau)^{-3+\frac{1}{2^k}}\left[\sum_{\substack{0\leq n_1+n_2+n_3\leq 2 \\ n_1+n_2\leq 1}}\int_{\Sigma_{\tau_i}\cap \{r\geq R\}} r|\snabla_{\s^2}^{n_1}L(rL)^{n_2}T^{n_3}\psi|^2\,d\sigma dr+C \sum_{\substack{0\leq n_1+n_2\leq 3\\n_1\leq 1}}E_{\mathbf{N}}[\snabla_{\s^2}^{n_1}T^{n_2}\phi](0)\right].
\end{align*}
Given $\delta>0$, we can therefore take $k$ suitably charge to obtain \eqref{eq:mainedcay} for $N=1$ and \eqref{eq:mainrwedcay} follows by additionally applying \eqref{eq:rpwithgrowrweight}. Then, we apply \eqref{eq:rpwithgrowrweight} to obtain moreover \eqref{eq:mainrwedcay} for $N=1$.

The general $N>1$ case proceeds inductively. We further estimate the first term on the right-hand side above after replacing $\phi$ with $T\phi$. In general, we split
\begin{equation}
\label{eq:gammasplitN}
\begin{split}
\sum_{n_1+n_2+n_3\leq N-1}&\int_{\tau_i}^{\tau_{i+1}}\int_{\Sigma_{\tau}\cap \{r\geq R\}} r|\snabla_{\s^2}^{n_1}L(rL)^{n_2}T^{n_3+1}\psi|^2\,d\sigma dr d\tau\\
\leq &\:C\sum_{1\leq n_1+n_2+n_3\leq N-1}(1+\tau_i)^{\gamma} \int_{\tau_i}^{\tau_{i+1}}\int_{\Sigma_{\tau}\cap \{r\geq R\}}  |\snabla_{\s^2}^{n_1}L(rL)^{n_2}T^{n_3+1}\psi|^2\,d\sigma dr d\tau\\
&+C\sum_{1\leq n_1+n_2+n_3\leq N}(1+\tau_i)^{-\gamma} \int_{\tau_i}^{\tau_{i+1}}\int_{\Sigma_{\tau}\cap \{r\geq R\}} r^2 |\snabla_{\s^2}^{n_1}L(rL)^{n_2}T^{n_3}\psi|^2\,d\sigma dr d\tau
\end{split}
\end{equation}
after applying \eqref{eq:convTintorX}. As we pass from $N$ to $N+1$, we apply the mean-value theorem an additional time and then optimize $\gamma$ as above in order to increase the decay rates in  \eqref{eq:mainedcay} and \eqref{eq:mainrwedcay} by 2.
\end{proof}

When $\alpha\geq 0$, we will moreover make use of the following stronger decay estimates:
\begin{proposition}
\label{prop:edecaytimeder2}
Let  $N\in \N$ and let $\alpha\geq 0$ or let $\alpha<0$ and $\phi=\phi_{\geq 1}$. Then for $R>3M$ suitably large, there exists a constant $C=C(M,h,V_{\alpha},N,R)>0$ such that for all $\tau\geq0$:
\begin{equation}
\label{eq:mainedcay2}
\begin{split}
E_{\mathbf{N}}[T^N\phi](\tau)\leq &\:C(1+\tau)^{-2-2N}\Bigg[\sum_{\substack{0\leq n_1+n_2+n_3\leq 2N+1\\ n_1+n_2\leq N}}\int_{\Sigma_{0}\cap \{r\geq R\}} r^2|\snabla_{\s^2}^{n_1}L(rL)^{n_2}T^{n_3}\psi|^2\,d\sigma dr\\
&+ \sum_{\substack{0\leq n_1+n_2\leq 3N+1\\ n_1\leq N}}E_{\mathbf{N}}[\snabla_{\s^2}^{n_1}T^{n_2}\phi](0)\Bigg].
\end{split}
\end{equation}
and moreover,
\begin{equation}
\label{eq:mainrwedcay2}
\begin{split}
\int_{\Sigma_{\tau}\cap \{r\geq R\}} r^2(LT^N\psi)\,d\sigma dr\leq &\: C(1+\tau)^{-2N}\Bigg[\sum_{\substack{0\leq n_1+n_2+n_3\leq 2N\\ n_1+n_2\leq N}}\int_{\Sigma_{0}\cap \{r\geq R\}} r^2|\snabla_{\s^2}^{n_1}L(rL)^{n_2}T^{n_3}\psi|^2\,d\sigma dr\\
&+ \sum_{\substack{0\leq n_1+n_2\leq 3N\\ n_1\leq N}}E_{\mathbf{N}}[\snabla_{\s^2}^{n_1}T^{n_2}\phi](0)\Bigg].
\end{split}
\end{equation}
\end{proposition}
\begin{proof}
As in the general $\alpha$ case, we have that \eqref{eq:firstedecay} holds. When $\alpha>0$ or $\phi=\phi_{\geq 1}$, we moreover have that \eqref{eq:horp} holds for $0< p\leq 2$. Taking $p=2$ and $N=0$, we can therefore further estimate:
\begin{equation*}
\int_{\tau_i}^{\tau_{i+1}} \int_{\Sigma_{\tau}\cap \{r\geq R\} }r(L\psi)^2\,d\sigma dr d\tau\leq C  \int_{\Sigma_{\tau_i}\cap \{r\geq R\}} r^2(L\psi)^2\,d\sigma dr+CE_{\mathbf{N}}[\phi](\tau_i).
\end{equation*}
Applying the mean-value theorem, we can therefore improve \eqref{eq:firstedecay} to obtain:
\begin{equation*}
E_{\mathbf{N}}[\phi](\tau)\leq C (1+\tau)^{-2}\left[ \int_{\Sigma_{0}\cap \{r\geq R\}} r^2(L\psi)^2+r(LT\psi)^2\,d\sigma dr+E_{\mathbf{N}}[\phi](0)+ E_T[T\phi](0)+ E_T[T^2\phi](0)\right],
\end{equation*}
which implies \eqref{eq:mainedcay2} with $N=0$. We moreover obtain from \eqref{eq:horp} with $p=2$ that
\begin{equation*}
\int_{\Sigma_{\tau}\cap \{r\geq R\}} r^2(L\psi)^2\,d\sigma dr\leq C \left[ \int_{\Sigma_{0}\cap \{r\geq R\}} r^2(L\psi)^2\,d\sigma dr+E_T[\phi](0)\right],
\end{equation*}
which implies \eqref{eq:mainrwedcay2} with $N=0$.

To obtain the estimates with $N\geq 1$, we then proceed as in the proof of Proposition \ref{prop:edecaytimeder}, by repeatedly applying \eqref{eq:convTintorX}. When $N=1$, for example, we use \eqref{eq:Tderrweightint} to apply the mean-value theorem once more and increase the decay rate of $E_{\mathbf{N}}[T\phi](\tau)$ from $-2$ to $-3$. Applying \eqref{eq:horp} with $N=1$ together with the mean-value theorem again then results in a decay rate of $-4$.

Similarly, \eqref{eq:Tderrweight}, together with \eqref{eq:horp} with $N=1$ and $p=1$ and $p=2$ allows us to obtain \eqref{eq:mainrwedcay2} with $N=1$. The general $N$ case then follows in a straightforward manner.
\end{proof}

\begin{remark}
In Propositions \ref{prop:edecaytimeder} and \ref{prop:edecaytimeder2} there is a loss of total number of derivatives on the right-hand side compared to the left-hand side. This is explained by the loss of derivatives in the integrated energy estimate in \eqref{eq:fulliedlossder} due the presence of trapped null geodesics. 
\end{remark}

\section{Time-inversion theory}
\label{sec:timinv}
In this section, we will show that, given a solution $\phi$ to \eqref{eq:waveeq} that arises from appropriate initial data, we can construct initial data for \eqref{eq:waveeq} that lead to a solution $\widetilde{\phi}$ satisfying the property that $T\widetilde{\phi}=\phi$. We will refer to $\widetilde{\phi}$ as the \emph{time integral} of $\phi$ (see the statement of Proposition \ref{prop:timeint} for a precise definition). Formally,
\begin{equation*}
\widetilde{\phi}(\tau,r,\theta,\varphi)=-\int_{\tau}^{\infty} \phi(\tau',r,\theta,\varphi)\,d\tau',
\end{equation*}
but since we have not yet established integrability of $\phi$ in $\tau$, we will instead construct $\widetilde{\phi}$ from initial data by considering an appropriate elliptic problem (and afterwards justify the above integral).

First, we develop the necessary elliptic theory in Section \ref{sec:elliptictheory}, which we then apply in Section \ref{sec:constrtimeint} to construct the time integral initial data.

\subsection{Elliptic theory}
\label{sec:elliptictheory}
We introduce the differential operator $\mathcal{L}$:
\begin{equation*}
\mathcal{L}u:=X(Dr^2Xu)+\slashed{\Delta}_{\s^2}u-r^2V_{\alpha}u.
\end{equation*}
\begin{proposition}
\label{prop:mainesttimeinv}
Let $u\in C^{\infty}_c(\Sigma)$. Let $-\beta_0<q<\beta_0$. Then there exists a constant $C=C(M,V_{\alpha},q,h)>0$, such that
\begin{equation}
\label{eq:mainellipticMpos}
\int_{\Sigma} r^{q}\left[r^2(Xu)^2+|\snabla_{\s^2}u|^2+u^2\right]\,d\sigma dr\leq C \int_{\Sigma} r^{q}(\mathcal{L}u)^2\,d\sigma dr.
\end{equation}
\end{proposition}
\begin{proof}

We split $u=u_{\leq \ell_0-1}+u_{\geq \ell_0}$, where $\ell_0\in \N$ will be chosen suitably large.\\
\\
\underline{$u_{\leq \ell_0-1}$}:\\
\\ 
Consider first $u_{\ell}$, with $0\leq \ell\leq \ell_0-1$. The constants appearing in the estimates for $u_{\ell}$ with $\ell\leq \ell_0-1$ below, will depend on $\ell_0$.
 
We consider the weight functions $w_{\ell}$ from Lemma \ref{lm:propertiesw} and define:
\begin{equation*}
\check{u}_{\ell}:=w_{\ell}^{-1}u_{\ell}.
\end{equation*}
Then
\begin{equation*}
\begin{split}
w_{\ell}\mathcal{L} u_{\ell}=w_{\ell}\mathcal{L}(w_{\ell}\check{u}_{\ell})=&\:w_{\ell}X(Dr^2 X(w_{\ell}\check{u}_{\ell}))-[\ell(\ell+1)+r^2V_{\alpha}]w_{\ell} ^2\check{u}_{\ell}\\
=&\: w_{\ell}^2 X(Dr^2 X \check{u}_{\ell})+2\frac{d w_{\ell}}{dr} X\check{u}_{\ell}+\check{u}_{\ell}w_{\ell}\left[\frac{d}{dr}\left(Dr^2\frac{dw_{\ell}}{dr}\right)-r^2V_{\alpha}w_{\ell} \right]\\
=&\: X(Dr^2 w_{\ell}^2 X\check{u}_{\ell}),
\end{split}
\end{equation*}
where we applied \eqref{eq:odew} to arrive at the last equality. 

Then it follows from multiple applications of the Leibniz rule that:
\begin{equation*}
\begin{split}
r^{-q}X\check{u}_{\ell}\cdot w_{\ell}\mathcal{L} u_{\ell}=&\:r^{-q}X\check{u}_{\ell} Dr^2 w_{\ell}^2X^2\check{u}_{\ell}+ (Dr^2w_{\ell}^2)' r^{-q} (X\check{u}_{\ell})^2\\
=&\: X\left(\frac{1}{2} Dr^{2-q}w_{\ell}^2  (X\check{u}_{\ell})^2\right)+\left[ r-M +\frac{q}{2}(r-2M) +w_{\ell}^{-1}\frac{dw_{\ell}}{dr}r(r-2M )\right]r^{-q}w_{\ell}^2(X\check{u}_{\ell})^2.
\end{split}
\end{equation*}
We can moreover apply \eqref{eq:asymptw} to estimate
\begin{equation*}
\left|w_{\ell}^{-1}\frac{dw_{\ell}}{dr}-\frac{1}{2}(\beta_{\ell}-1)r^{-1}\right|\leq C(r^{-\frac{1}{2}-\frac{1}{2}\beta_{\ell}}+r^{-\frac{3}{2}+\frac{1}{2}\beta_{\ell}}).
\end{equation*}

We apply a weighted Young's inequality to estimate
\begin{equation*}
\left|r^{-q}X\check{u}_{\ell}\cdot w_{\ell}\mathcal{L} u_{\ell}\right|\leq \epsilon r^{1-q}w_{\ell}^2(X\check{u}_{\ell})^2+\frac{1}{4\epsilon}r^{-1-q}(\mathcal{L}u_{\ell})^2.
\end{equation*}

We first take $q>q_{\ell_0}$, with $q_{\ell_0}$ suitably large, so that in $\{r\leq R\}$, with $3M<R<\infty$ arbitrarily large,
\begin{equation*}
r-M +\frac{q_{\ell_0}}{2}(r-2M) +w_{\ell}^{-1}\frac{dw_{\ell}}{dr}r(r-2M )>0.
\end{equation*}
Then we integrate over  $\Sigma\cap\{r\leq R\}$ to obtain
\begin{equation}
\label{eq:aneg1}
\int_{\s^2}(X\check{u}_{\ell})^2\Big|_{r=R}\,d\sigma+\int_{\Sigma\cap \{r\leq R\}} (X\check{u}_{\ell})^2\,d\sigma dr\leq C \int_{\Sigma\cap \{r\leq R\}} (\mathcal{L}u_{\ell})^2\,d\sigma dr.
\end{equation}

Let $\epsilon>0$ be arbitrarily small. We then consider the region $\{r\geq R\}$, with $R>3M$ suitably large so that
\begin{equation*}
r-M+\frac{q}{2}(r-2M)+w_{\ell}^{-1}\frac{dw_{\ell}}{dr}>\left(\frac{1}{2}+\frac{1}{2}q+\frac{1}{2}\beta_{\ell}-\epsilon\right)r.
\end{equation*}
Given any $q>-1-\beta_{\ell}$, we can therefore take $\epsilon>0$ suitably small, integrate over $\Sigma$ and apply \eqref{eq:aneg1} to obtain
\begin{equation}
\label{eq:keyeqellipticcheck}
\int_{\Sigma} r^{\beta_{\ell}-q}(X\check{u}_{\ell})^2\,d\sigma dr\leq C \int_{\Sigma} r^{-1-q}(\mathcal{L}u_{\ell})^2\,d\sigma dr.
\end{equation}
By restricting further to $-1-\beta_{\ell}<q<-1+\beta_{\ell}$ and applying \eqref{eq:hardy} together with compactness of the support of $u_{\ell}$, we obtain
\begin{equation*}
\int_{\Sigma} r^{\beta_{\ell}-q-2}\check{u}_{\ell}^2+r^{\beta_{\ell}-q}(X\check{u}_{\ell})^2\,d\sigma dr\leq C \int_{\Sigma} r^{-1-q}(\mathcal{L}u_{\ell})^2\,d\sigma dr
\end{equation*}
and hence,
\begin{equation*}
\int_{\Sigma} r^{-q-1}{u}_{\ell}^2+r^{-q+1}(X{u}_{\ell})^2\,d\sigma dr\leq C \int_{\Sigma} r^{-q-1}(\mathcal{L}u_{\ell})^2\,d\sigma dr.
\end{equation*}
\\
\underline{$u_{\geq \ell_0}$}:\\
\\

Consider now $u_{\geq \ell_0}$. We have that
\begin{equation*}
(\mathcal{L}u_{\geq \ell_0})^2= (X(Dr^2Xu_{\geq \ell_0}))^2+((\slashed{\Delta}_{\s^2}-r^2V_{\alpha})(u_{\geq \ell_0}))^2+2X(Dr^2Xu_{\geq \ell_0})(\slashed{\Delta}_{\s^2}-r^2V_{\alpha})(u_{\geq \ell_0}).
\end{equation*}
Let $-1-\beta_0<q<-1+\beta_0$. Then we integrate by parts to obtain:
\begin{equation*}
\begin{split}
\int_{\Sigma} 2r^{-q}X(Dr^2Xu_{\geq \ell_0})(\slashed{\Delta}_{\s^2}-r^2V_{\alpha})(u_{\geq \ell_0})\,d\sigma dr=&\:\int_{\Sigma} X(2r^{-q}Dr^2Xu_{\geq \ell_0}(\slashed{\Delta}_{\s^2}-r^2V_{\alpha})(u_{\geq \ell_0}))\,d\sigma dr\\
&+\int_{\Sigma} 2Dr^{2-q}(|\snabla_{\s^2}Xu_{\geq \ell_0}|^2+r^2V_{\alpha}(Xu_{\geq \ell_0})^2)\,d\sigma dr\\
&+\int_{\Sigma} 2qDr^{1-q}Xu_{\geq \ell_0}(\slashed{\Delta}_{\s^2}-r^2V_{\alpha})(u_{\geq \ell_0})\,d\sigma dr\\
&+\int_{\Sigma} 2Dr^{2-q}Xu_{\geq \ell_0}\frac{d}{dr}(r^2V_{\alpha})u_{\geq \ell_0}\,d\sigma dr.
\end{split}
\end{equation*}
We apply Young's inequality to estimate:
\begin{equation*}
\begin{split}
Xu_{\geq \ell_0}&\left[2qDr^{1-q}(\slashed{\Delta}_{\s^2}-r^2V_{\alpha})u_{\geq \ell_0}+2Dr^{2-q}\frac{d}{dr}(r^2V_{\alpha})u_{\geq \ell_0}\right]\\
\leq &\: 2(Cr^{-2}+q^2)D r^{2-q}(Xu_{\geq \ell_0})^2+\frac{1}{2}D r^{-q}((\slashed{\Delta}_{\s^2}-r^2V_{\alpha})(u_{\geq \ell_0}))^2+\frac{1}{2}Dr^{-q}u_{\geq \ell_0}^2,
\end{split}
\end{equation*}
for some suitably large $C>0$.

We apply \eqref{eq:sphere2} to infer that:
\begin{equation*}
\int_{\Sigma} 2(Cr^{-2}+q^2)D r^{2-q}(Xu_{\geq \ell_0})^2\,d\sigma dr\leq \int_{\Sigma} 2(\ell_0(\ell_0+1))^{-1}(Cr^{-2}+q^2)D r^{2-q}|\snabla_{\s^2}Xu_{\geq \ell_0}|^2\,d\sigma dr
\end{equation*}
and moreover
\begin{equation*}
\int_{\s^2}|\snabla_{\s^2}u_{\geq \ell_0}|^2+r^2V_{\alpha}u_{\geq \ell_0}^2\,d\sigma \geq \int_{\s^2}(\ell_0(\ell_0+1)+r^2V_{\alpha})u_{\geq \ell_0}^2\,d\sigma,
\end{equation*}
so for $\ell_0>0$ suitably large depending on $q$ and $V_{\alpha}$, we can absorb all unsigned terms in the integral of $r^{-q}(\mathcal{L}u_{\geq \ell_0})^2$ over $\Sigma$ to obtain: for all $-1-\beta_0<q<-1+\beta_0$,
\begin{equation}
\label{eq:ellipticauxgeqL}
\begin{split}
\int_{\Sigma}& r^{-q} (X(Dr^2Xu_{\geq \ell_0}))^2+r^{2-q}D |\snabla_{\s^2}Xu_{\geq \ell_0}|^2+r^{2-q}D(Xu_{\geq \ell_0})^2+r^{-q}(\slashed{\Delta}_{\s^2}u_{\geq \ell_0})^2+r^{-q}u_{\geq \ell_0}^2\,d\sigma d\sigma\\
\leq&\: C \int_{\Sigma}r^{-q}(\mathcal{L} u_{\geq \ell_0})^2\,d\sigma dr.
\end{split}
\end{equation}

In order to remover the degenerate factors $D$ that appear in front of the $(Xu_{\geq L})^2$ term, we additionally consider:
\begin{equation*}
\begin{split}
r \chi X u_{\geq \ell_0} \mathcal{L} u_{\geq \ell_0}=&\: X\left(\frac{1}{2}Dr^3 \chi (Xu_{\geq \ell_0})^2\right)+\frac{1}{2}r\left((Dr^2)'-Dr\right)\chi (Xu_{\geq \ell_0})^2-\frac{1}{2}X(r\chi |\snabla_{\s^2}u_{\geq \ell_0}|^2+\chi r^3V_{\alpha}u_{\geq \ell_0}^2)\\
&+\frac{1}{2}\chi (|\snabla_{\s^2}u_{\geq \ell_0}|^2+\frac{d}{dr}(r^3V_{\alpha})u_{\geq \ell_0}^2)-\frac{1}{2}Dr^3\frac{d\chi}{dr} (Xu_{\geq \ell_0})^2+r \frac{d\chi}{dr} ( |\snabla_{\s^2}u_{\geq \ell_0}|^2+ r^2V_{\alpha}u_{\geq \ell_0}^2)\\
&+\textnormal{div}_{\s^2}(\ldots),
\end{split}
\end{equation*}
with $\chi:[2M,\infty)\to\R_{>0}$ a smooth cut-off function, such that $\chi=1$ for $r\leq r_0$ and $\chi=0$ for $r>r_1$, where $2M<r_0<r_1<\infty$. Note first of all that $(Dr^2)'-Dr=r$ and
\begin{equation*}
\int_{\s^2} |\snabla_{\s^2}u_{\geq \ell_0}|^2+\frac{d}{dr}(r^3V_{\alpha})u_{\geq \ell_0}^2\,d\sigma \geq \frac{1}{2}\int_{\s^2} |\snabla_{\s^2}u_{\geq \ell_0}|^2+u_{\geq \ell_0}^2\,d\sigma
\end{equation*}
by \eqref{eq:sphere2} and $\ell_0$ suitably large. Furthermore, we can apply Young's inequality to estimate:
\begin{equation*}
|r \chi X u_{\geq \ell_0} \mathcal{L} u_{\geq \ell_0}|\leq r^2(X u_{\geq \ell_0})^2+(\mathcal{L} u_{\geq \ell_0})^2.
\end{equation*}
We now integrate $r \chi X u_{\geq \ell_0} \mathcal{L} u_{\geq \ell_0}$ over $\Sigma$ and apply the above estimates to obtain
\begin{equation*}
\begin{split}
\int_{\Sigma\cap\{r\leq r_0\}} r^2 (Xu_{\geq \ell_0})^2+|\snabla_{\s^2}u_{\geq\ell_0}|^2+u_{\geq \ell_0}^2\,d\sigma dr\leq &\: C\int_{\Sigma\cap\{r\leq r_1\}} (\mathcal{L} u_{\geq \ell_0})^2\,d\sigma dr\\
&+C\int_{\Sigma\cap\{r_0\leq r\leq r_1\}}  (Xu_{\geq\ell_0})^2+|\snabla_{\s^2}u_{\geq \ell_0}|^2\,d\sigma dr.
\end{split}
\end{equation*}
Finally, we combine the above equation with \eqref{eq:ellipticauxgeqL} to control (in particular):
\begin{equation*}
\int_{\Sigma} r^{2-q}(Xu_{\geq \ell_0})^2+r^{-q}u_{\geq \ell_0}^2\,d\sigma d\sigma\leq C \int_{\Sigma}r^{-q}(\mathcal{L} u_{\geq \ell_0})^2\,d\sigma dr.
\end{equation*}
By combining all the above estimates and replacing $q$ with $-(q-1)$ everywhere above, we conclude that \eqref{eq:mainellipticMpos} must hold.
\end{proof}

\begin{lemma}
\label{lm:commL}
The following identity holds:
\begin{equation}
\label{eq:commL}
\left[\mathcal{L}+ND'r^2X\right](rX)^Nu=(rX)^N\mathcal{L}u+N\sum_{n=0}^{N}O_{\infty}(r^{-1})(rX)^nu.
\end{equation}
\end{lemma}
\begin{proof}
We will prove \eqref{eq:commL} by induction. Note that \eqref{eq:commL} hold for $N=0$ by definition of $\mathcal{L}$. Now we suppose as an induction hypothesis that \eqref{eq:commL} holds for some $N\in \N_0$. We will show that \eqref{eq:commL} must then also hold with $N$ replaced by $N+1$.

One can verify first of all that
\begin{equation*}
\begin{split}
X(Dr^2X(rX)^{N+1}u)=&\:(rX+2)X(Dr^2X(rX)^Nu)-(Dr^2)'X (rX)^{N+1}u-(Dr^2)''(rX)^{N+1}u\\
=&\: rX(Dr^2X(rX)^Nu)-D'r^2 X(rX)^{N+1}u
\end{split}
\end{equation*}
Hence, for $M>0$:
\begin{equation*}
\begin{split}
\mathcal{L}(rX)^{N+1}u=&\:X(Dr^2X(rX)^{N+1}u)+\slashed{\Delta}_{\s^2}(rX)^{N+1}u-r^2V_{\alpha}(rX)^{N+1}u\\
=&\:rX\mathcal{L}(rX)^{N}u+rX(r^2V_{\alpha})(rX)^{N}u-D'r^2 X(rX)^{N+1}u\\
=&\:(rX)^{N+1}\mathcal{L}u-(N+1)D'r^2X (rX)^{N+1}u+\sum_{n=0}^{N+1}O_{\infty}(r^{-1})(rX)^nu,
\end{split}
\end{equation*}
where we applied the induction hypothesis to arrive at the last equality. We then conclude \eqref{eq:commL}  with $N$ replaced by $N+1$ by rearranging terms.
\end{proof}

We now obtain an analogue of Proposition \ref{prop:mainesttimeinv} with higher-order angular derivatives, as well as an analogue featuring the left-hand side of  \eqref{eq:commL}, the latter of which will be important when we consider higher-order derivatives with respect to $X$.

We introduce first the space:
\begin{equation*}
C_{c,*}^{\infty}(\Sigma)=\left\{u\in C_{c}^{\infty}(\Sigma)\,|\,u|_{r=2M}\equiv 0\right \}.
\end{equation*}

\begin{proposition}
\label{prop:keyhopropelliptic}
Let $u\in C^{\infty}_c(\Sigma)$. Let $-\beta_0<q<\beta_0$. Then there exists a constant $C=C(M,\alpha,q,h,N)>0$, such that
\begin{align}
\label{eq:mainellipticMposang}
\sum_{k=0}^K\int_{\Sigma} r^{q}\left[r^2|\snabla_{\s^2}^kXu|^2+|\snabla_{\s^2}^{k+1}u|^2+|\snabla_{\s^2}^ku|^2\right]\,d\sigma dr\leq&\: C \sum_{k=0}^K\int_{\Sigma} r^{q}|\snabla_{\s^2}^k\mathcal{L}u|^2\,d\sigma dr.
\end{align}

Let $u\in C^{\infty}_c(\Sigma)$ if $N\geq 0$  and $u\in C^{\infty}_{c,*}(\Sigma)$ if $N<0$. Let $-\beta_0<q<\beta_0$. Then there exists a constant $C=C(M,V_{\alpha},q,h,N)>0$, such that
\begin{equation}
\label{eq:mainellipticMposho}
\begin{split}
\int_{\Sigma}& r^{q}\left[D^2r^2(X(rXu))^2+r^2(Xu)^2+|\snabla_{\s^2}u|^2\right]\,d\sigma dr\\
\leq&\: C \int_{\Sigma} r^{q}((\mathcal{L}-ND'r^2X)u)^2+|N|r^q(u^2+|\snabla_{\s^2}u|^2+D^2r^2(Xu)^2)\,d\sigma dr.
\end{split}
\end{equation}
\end{proposition}
\begin{proof}
The estimate \eqref{eq:mainellipticMposang} follows immediately applying \eqref{eq:mainellipticMpos} with $\phi$ replaced by $\Omega^{I}\phi$, where $|I|\leq k$ and applying \eqref{eq:sphere1b}.

To obtain \eqref{eq:mainellipticMposho}, we first consider for $q\in \R$ and $j\in\{0,1\}$:
\begin{equation}
\label{eq:horedshift}
\begin{split}
N D^jr^{q+1} X u (\mathcal{L} u+ND'r^2Xu)=&\: X\left(\frac{N}{2}D^{j+1}r^{3+q} (Xu)^2\right)+\frac{1}{2}D^jr^{q+1}\left[N(2N+1-j)D'r-N(q-1)D\right] (Xu)^2\\
&-\frac{N}{2}X\left(r^{q+1} D^j|\snabla_{\s^2}u|^2+D^jr^{3+q}V_{\alpha}u^2\right)+\frac{N}{2}((q+1)D^j +jD'r)r^{q}|\snabla_{\s^2}u|^2\\
&+\frac{N}{2}\left(D^j\frac{d}{dr}(r^{q+3}V_{\alpha})+jD'r\right)u^2+\textnormal{div}_{\s^2}(\ldots).
\end{split}
\end{equation}
We have that $N(2N+1-j)>0$ for all $N\neq 0$. Furthermore, integrating \eqref{eq:horedshift} along $\Sigma$, we only obtain a boundary term at $r=2M$ with a bad sign if $N<0$ and $j=0$, in which case we appeal to the assumption that $u\in C^{\infty}_{c,*}(\Sigma)$ to conclude that this boundary term actually vanishes.

By taking $j=1$, we can therefore estimate:
\begin{equation*}
\int_{\Sigma} Dr^{q+2}(Xu)^2\,d\sigma dr\leq C \int_{\Sigma} r^{q}((\mathcal{L}-ND'r^2X)u)^2+|N|D^2r^{q+2}(Xu)^2+|N|r^q(u^2+|\snabla_{\s^2}u|^2)\,d\sigma dr.
\end{equation*}

Now, we take $j=0$ and we apply the above estimate to remove the $D$ factor on the left-hand side:
\begin{equation*}
\int_{\Sigma} r^{q+2}(Xu)^2\,d\sigma dr\leq C \int_{\Sigma} r^{q}((\mathcal{L}-ND'r^2X)u)^2+|N|D^2r^{q+2}(Xu)^2+|N|r^q(u^2+|\snabla_{\s^2}u|^2)\,d\sigma dr.
\end{equation*}
Finally, in order to control also the integral of $D^2r^{q+2}(X(rXu))^2$, we repeat the elliptic estimate for $u_{\geq \ell_0}$ in the proof of Proposition \ref{prop:mainesttimeinv} with $\mathcal{L} u+ND'r^2Xu$ replacing $\mathcal{L}$ and without the restriction to $\ell \geq \ell_0$, and we use that the additional term $ND'r^2Xu$ contributes as a lower-order term the elliptic estimate.
\end{proof}

In order to define in a clean way the relevant Hilbert spaces on which $\mathcal{L}$ is well-defined, we will modify $\mathcal{L}$ as follows: let $q\in \R$ and $N\in \Z$ then
\begin{align*}
\mathcal{L}_{q,N}v:=&\:r^{\frac{q}{2}}\mathcal{L}(r^{-\frac{q}{2}}v)+Nr^{\frac{q}{2}}D'r^2 X(r^{-\frac{q}{2}}v).
\end{align*}

Let $n\in \N_0$. We introduce the Hilbert spaces $\mathbf{H}_{n}$ as completions of $C_{c}^{\infty}(\Sigma)$ under the norms $||\cdot||_{n}$ defined as follows:
\begin{equation*}
||u||^2_{n}:= \sum_{0\leq n_1+n_2\leq n}\int_{\Sigma} |\snabla_{\s^2}^{n_1}(rX)^{n_2}u|^2\,d\sigma dr.
\end{equation*}
We denote with $D_{n}(\mathcal{L}_{q,N})$ the closure of $C_{c}^{\infty}(\Sigma)$ with respect to the norm $||\cdot||_{n}+||\mathcal{L}_{q,N}(\cdot)||_{n}$. Then
\begin{equation*}
\mathcal{L}_{q,N}: D_{n}(\mathcal{L}_{q,N})\to \mathbf{H}_{ n}
\end{equation*}
is a closed, densely-defined linear operator.

\begin{corollary}
\label{cor:mainellipticest}
Let $v\in C^{\infty}_c(\Sigma)$ and let $n_1,n_2\in \N_0$. Let $-\beta_0<q<\beta_0$. Then there exists a constant $C=C(M,V_{\alpha},q,n_1,n_2,h)>0$, such that:
\begin{equation}
\label{eq:mainellipticho}
||D\snabla_{\s^2}^{n_1}(rX)^{n_2} v||_{2}^2+||\snabla_{\s^2}^{n_1}(rX)^{n_2} v||_{1}^2\leq C ||\snabla_{\s^2}^{n_1}\mathcal{L}_{q,\pm n_2}(rX)^{n_2}v||^2_0+Cn_2||\snabla_{\s^2}^{n_1}(rX)^{n_2} v||_{0}^2+Cn_2||D\snabla_{\s^2}^{n_1}(rX)^{n_2} v||_{1}^2
\end{equation}
where to be able to take a minus sign in $\mathcal{L}_{q,\pm n_2}$, we need to assume $v\in C^{\infty}_{c,*}(\Sigma)$.
\end{corollary}
\begin{proof}
Let  $u= r^{-\frac{q}{2}}v$ when $M>0$. Then \eqref{eq:mainellipticho} follows immediately from Proposition \ref{prop:mainesttimeinv} and Proposition \ref{prop:keyhopropelliptic}.
\end{proof}

We will construct $\mathcal{L}_{q,N}^{-1}$ as an operator acting on $\mathbf{H}_{n}$, with arbitrarily large $n\in N_0$, by first considering the formal adjoint of $\mathcal{L}_{q,N}$ with respect to $\mathbf{H}_0$. 

Let
\begin{equation*}
\mathcal{L}_{q,N}^*: \mathcal{D}_0(\mathcal{L}^*_{q,N}) \to \mathbf{H}_0
\end{equation*}
be the adjoint operator of $\mathcal{L}_{q,N}$, with $w\in  \mathcal{D}_0(\mathcal{L}^*_{q,N})$ if and only if there exists a unique $f\in  \mathbf{H}_0$ such that for all $u\in C_c^{\infty}(\Sigma)$
\begin{equation*}
\la \mathcal{L}_{q,N} u, v\ra_{L^2([2M,\infty)\times \s^2)}=\la u, f\ra_{L^2([2M,\infty)\times \s^2)}.
\end{equation*}
Then $\mathcal{L}_{q,N}^* v=f$.

We introduce moreover the operator $\mathcal{L}^{\dag}_{q,N}:=\mathcal{L}_{-q,-N}$ with domain $\mathcal{D}_n(\mathcal{L}^{\dag}_{q,N})$ defined as the closure of the space $C_{*,c}^{\infty}(\Sigma)$ with respect to the norm $||\cdot||_n+||\mathcal{L}^{\dag}_{q,N}(\cdot)||_n$.

\begin{lemma}
\label{lm:reldagstar}
We have that $\mathcal{L}_{q,N}^*=\mathcal{L}_{q,N}^{\dag}$ with $ \mathcal{D}_0(\mathcal{L}^*_{q,N})= \mathcal{D}_0(\mathcal{L}_{q,N}^{\dag})$.
\end{lemma}
\begin{proof}
We will show that $(\mathcal{L}_{q,N}^{\dag})^*=\mathcal{L}_{q,N}$ and $\mathcal{D}_0((\mathcal{L}_{q,N}^{\dag})^*)= \mathcal{D}_0(\mathcal{L}_{q,N})$. It is easy to verify by definition of $\mathcal{L}^{\dag}_{q,N}$ and integration by parts that for all  $u\in C_{c,*}^{\infty}(\Sigma)$ and $v\in C_c^{\infty}(\Sigma)$:
\begin{equation}
\label{eq:Ldag}
\la \mathcal{L}_{q,N}^{\dag} u, v\ra_{L^2([2M,\infty)\times \s^2)}=\la u, \mathcal{L}_{q,N}v\ra_{L^2([2M,\infty)\times \s^2)}.
\end{equation}
Since $\mathcal{L}_{q,N}$ and $\mathcal{L}_{q,N}^{\dag}$ are both closed operators, the identity \eqref{eq:Ldag} holds more generally for $u\in \mathcal{D}_0(\mathcal{L}^{\dag}_{q,N})$ and $v\in \mathcal{D}_0(\mathcal{L}_{q,N})$. It follows that $\mathcal{L}_{q,N}=(\mathcal{L}_{q,N}^{\dag})^*$ as operators acting on $\mathcal{D}_0(\mathcal{L}_{q,N})$, with $\mathcal{D}_0(\mathcal{L}_{q,N}) \subseteq \mathcal{D}_0((\mathcal{L}_{q,N}^{\dag})^*)$.

It remains to show that $\mathcal{D}_0((\mathcal{L}_{q,N}^{\dag})^*) \subseteq \mathcal{D}_0(\mathcal{L}_{q,N})$. Observe that $\mathbf{H}_0=L^2([2M,\infty)_r\times \s^2)$ and the closure of $C_c^{\infty}(\Sigma)$ is $L^2([2M,\infty)_r\times \s^2)$ with respect to $||\cdot ||_0$.  Let $v\in \mathcal{D}_0((\mathcal{L}_{q,N}^{\dag})^*)$, then by $\eqref{eq:Ldag}$, $\mathcal{L}_{q,N}v \in \mathbf{H}_0$ exists in a weak sense. Hence, by the above observations, both $v$ and $\mathcal{L}_{q,N}v$ lie in the closure of $C_c^{\infty}(\Sigma)$ with respect to $||\cdot ||_0$ and therefore $v\in \mathcal{D}_0(\mathcal{L}_{q,N})$.
\end{proof}

\begin{proposition}
Let $N\in \N_0$ and $f\in \mathbf{H}_N$. Let $-\beta_0<q<\beta_0$. Then there exists a unique $v\in \mathbf{H}_{N+1}$ such that $\mathcal{L}_{q,0}v=f$.
\end{proposition}
\begin{proof}
Suppose $N=0$. Then it follows from Corollary \ref{cor:mainellipticest} with $n_1=n_2=0$ that $\ker \mathcal{L}_{q,0}=\{0\}$. We will establish bijectivity by showing that $\textnormal{Ran}\:  \mathcal{L}_{q,0}=\mathbf{H}_0$.

Let $w\in (\textnormal{Ran}\:  \mathcal{L}_{q,0})^{\perp}$ and $u\in C^{\infty}_{c,*}(\Sigma)$. Then we have that
\begin{equation*}
0=\la w,\mathcal{L}_{q,0}u\ra=\la \mathcal{L}^*_{q,0}w,u\ra,
\end{equation*}
so it follows that $w\in \mathcal{D}_0(\mathcal{L}_{q,0}^*)$ and $w\in \ker \mathcal{L}_{q,0}^*$. Using now that by Lemma \ref{lm:reldagstar}, $\mathcal{L}^*_{q,0}=\mathcal{L}^{\dag}_{q,0}=\mathcal{L}_{-q,0}$, we can conclude from Corollary \ref{cor:mainellipticest} with $n_1=n_2=0$ that $\mathcal{L}^*_{q,0}$ is injective, so $\ker \mathcal{L}_{q}^*=\{0\}$ and $w=0$. Hence, $(\textnormal{Ran}\:  \mathcal{L}_{q,0})^{\perp}=\{0\}$, so $\textnormal{Ran}\:  \mathcal{L}_{q,0}=\mathbf{H}_0$.

So, for $f\in  \mathbf{H}_0$ there exists a unique $v\in \mathcal{D}_0(\mathcal{L}_{q,0})$ such that $\mathcal{L}_{q,0}v=f$. By applying Corollary \ref{cor:mainellipticest} once more, we in fact have that $v\in \mathbf{H}_1$ and $Dv\in \mathbf{H}_2$.

Now suppose $N=1$. Then we need to show that in fact we have the following improved regularity statement: $v\in \mathbf{H}_2$. First of all, from the above argument, together with Corollary \ref{cor:mainellipticest} for $n_1=1$ and $n_2=0$, it follows that $\snabla_{\s^2}v\in \mathbf{H}_1$.

By \eqref{eq:commL}, we have that:
\begin{equation*}
\mathcal{L}_{q,1} (rX v)=rX \mathcal{L}_{q,0}v+ r^{\frac{q}{2}} \sum_{n=0}^1 O_{\infty}(r^{-1}) (rX)^n(r^{-\frac{q}{2}}v),
\end{equation*}
where the right-hand side is an element of $\mathbf{H}_0$ for $f=\mathcal{L}v\in \mathbf{H}_1$, since we already established that $v\in \mathbf{H}_1$.

We can therefore conclude that $rX v \in \mathbf{H}_1$, if we can establish bijectivity of $\mathcal{L}_{q,1}$ and apply Corollary \ref{cor:mainellipticest} with $n_1=0$ and $n_2=1$. Bijectivity in turn follows as above, by replacing $\mathcal{L}_{q,0}$ with $\mathcal{L}_{q,1}$ and $\mathcal{L}_{q,0}^*$ with $\mathcal{L}_{q,1}^*$.

The general $N>0$ case follows in the same manner via an induction argument, where we establish first that $v\in \mathbf{H}_{N}$ in order to conclude that $v\in \mathbf{H}_{N+1}$.
\end{proof}
From the above, we can immediately conclude the following properties for solutions to the equation $\mathcal{L}u=F$:
\begin{corollary}
\label{cor:Linv0}
Let $-\beta_0<q<\beta_0$, $N\in \N_0$. Let $F:\Sigma\to\R$ such that $r^{\frac{q}{2}}F\in \mathbf{H}_N$. Then there exists a unique $u: \Sigma\to \R$ with $r^{\frac{q}{2}}u\in \mathbf{H}_{N+1}$, such that
\begin{align*}
\mathcal{L}u=&\:F,\\
||r^{\frac{q}{2}} u||_{\mathbf{H}_{N+1}}\leq &\: C||r^{\frac{q}{2}}F||_{\mathbf{H}_{N+1}},
\end{align*}
with $C=C(V_{\alpha},M,N,q)>0$ a constant.
\end{corollary}

We obtain stronger estimates when restricting $\mathcal{L}$ to functions satisfiying $u=u_{\geq 1}$. In that case, it follows immediately that we can take $\beta_1$ to replace the role of $\beta_0$ in the proof of Proposition \ref{prop:mainesttimeinv} and obtain:
\begin{corollary}
\label{cor:Linv1}
Let $-\beta_1<q<\beta_1$, $N\in \N_0$. Let $F:\Sigma\to\R$ such that $r^{\frac{q}{2}}F_{\geq 1}\in \mathbf{H}_N$. Then there exists a unique $u: \Sigma\to \R$, with $u=u_{\geq 1}$ and with $r^{\frac{q}{2}}u_{\geq 1}\in \mathbf{H}_{N+1}$, such that
\begin{align*}
\mathcal{L}u=&\:F,\\
||r^{\frac{q}{2}} u_{\geq 1}||_{\mathbf{H}_{N+1}}\leq &\: C||r^{\frac{q}{2}}F_{\geq 1}||_{\mathbf{H}_{N+1}},
\end{align*}
with $C=C(V_{\alpha},M,N,q)>0$ a constant.
\end{corollary}

\subsection{Time integrals along $\Sigma_0$}
\label{sec:constrtimeint}
Let $\phi$ be a smooth solution to \eqref{eq:waveeq}. Then by \eqref{eq:waveeq1}, we have that
\begin{align*}
\mathcal{L}\phi=&\:F[T\phi],\quad \textnormal{with}\\
F[f]:=&\:2r(1-h)X(rf)-\frac{dh}{dr}r^2f+r^2h\tilde{h}Tf,\\
\end{align*}
for sufficiently regular functions $f$ on $\mathcal{R}$. We will denote the projection of $F[f]$ to spherical harmonics of degree $\ell$ by $F_{\ell}[f]$ (see \S \ref{sec:sphere}).

In the proposition below, we construct the time integral corresponding to solutions $\phi$ to \eqref{eq:waveeq} arising from appropriate initial data.
\begin{proposition}
\label{prop:timeint}
Let $-\beta_0<q<\beta_0$, $N\in \N_0$. Consider initial data $(\phi|_{\Sigma_0},T\phi|_{\Sigma_0})$ for \eqref{eq:waveeq}, such that:
\begin{align*}
r^{\frac{q}{2}+1}X(r\phi)|_{\Sigma_0}\in&\: \mathbf{H}_{N}\quad \textnormal{and}\quad r^{\frac{q}{2}}\phi|_{\Sigma_0}\in \mathbf{H}_{N} \quad \textnormal{and}\quad r^{\frac{q}{2}}T\phi|_{\Sigma_0}\in \mathbf{H}_{N}.
\end{align*}
Then there exists a unique solution $T^{-1}\phi$ to \eqref{eq:waveeq}, such that $T(T^{-1}\phi)=\phi$, with initial data
$(T^{-1}\phi|_{\Sigma_0},\phi|_{\Sigma_0})$, such that moreover:
\begin{align*}
r^{\frac{q}{2}}X (rT^{-1}\phi)|_{\Sigma_0}\in&\: \mathbf{H}_{N} \quad \textnormal{and}\quad r^{\frac{q}{2}}T^{-1}\phi|_{\Sigma_0}\in \mathbf{H}_{N+1} \quad \textnormal{and}\quad r^{\frac{q}{2}}T(T^{-1}\phi)|_{\Sigma_0}\in \mathbf{H}_{N+1},
\end{align*}
and there exists a constant $C=C(V_{\alpha},h,M,N,q)>0$, such that
\begin{equation*}
||r^{\frac{q}{2}}X (rT^{-1}\phi)|_{\Sigma_0}||_{\mathbf{H}_N}+||r^{\frac{q}{2}}T^{-1}\phi|_{\Sigma_0}||_{\mathbf{H}_N}\leq C ||r^{\frac{q}{2}+1}X (r\phi)|_{\Sigma_0}||_{\mathbf{H}_N}+C||r^{\frac{q}{2}}\phi|_{\Sigma_0}||_{\mathbf{H}_N}+C||r^{\frac{q}{2}}T\phi|_{\Sigma_0}||_{\mathbf{H}_N}.
\end{equation*}
We refer to $T^{-1}\phi$ as the \emph{time integral} of $\phi$.
\end{proposition}
\begin{proof}
Given the choice of initial data in the proposition, we can apply Corollary \ref{cor:Linv0} with the inhomogeneity $F=F[\phi]|_{\Sigma_0}$ to obtain a solution $u$ to $\mathcal{L}u=F$. We take as initial data for \eqref{eq:waveeq} $(u,\phi|_{\Sigma_0})$, and we denote the resulting solution as $\widetilde{\phi}$. By the choice of $F$ and by application of \eqref{eq:waveeq}, it follows then that
\begin{equation*}
(T\tilde{\phi}|_{\Sigma_0},T^2\tilde{\phi}|_{\Sigma_0})=(\phi|_{\Sigma_0},T\phi_{\Sigma_0}),
\end{equation*}
so by uniqueness of solutions to \eqref{eq:waveeq}, we must have that $T\widetilde{\phi}=\phi$.
\end{proof}

\begin{remark}
In Corollary \ref{cor:equivdeftimeint} below, we provide an alternative characterization of $T^{-1}\phi$ as an integral in time of $\phi$ along curves of constant $r,\theta,\varphi$, which motivates the nomenclature.
\end{remark}

In the corollary below, we show that for $\alpha\geq 0$ ($\beta_0\geq 1$), we can construct multiple time integrals.

\begin{corollary}
\label{cor:regmultipletimeinv}
Let $\beta_0\geq 1$. Let $N\in \N_0$ and $n_{\beta}=\lfloor \frac{\beta_0+1}{2} \rfloor \in \N_0$ and let $0<\delta<\beta_0-n_{\beta}$ be arbitrarily small. Consider initial data $(\phi|_{\Sigma},T\phi|_{\Sigma})$ for \eqref{eq:waveeq}, such that:
\begin{align*}
r^{\frac{\beta_0}{2}+1-\delta}X(r\phi)|_{\Sigma_0}\in&\: \mathbf{H}_{N}\quad \textnormal{and}\quad r^{\frac{\beta_0}{2}-\delta}\phi|_{\Sigma_0}\in \mathbf{H}_{N} \quad \textnormal{and}\quad r^{\frac{\beta_0}{2}-\delta}T\phi|_{\Sigma_0}\in \mathbf{H}_{N}.
\end{align*}
Then there exists a unique solution $T^{-1-n_{\beta}}\phi$ to \eqref{eq:waveeq}, such that $T^{1+n_{\beta}}(T^{-1-n_{\beta}}\phi)=\phi$, with initial data $(T^{-1-n_{\beta}}\phi|_{\Sigma},T^{-n_{\beta}}\phi|_{\Sigma})$, such that:
\begin{align*}
r^{\frac{\beta_0}{2}-n_{\beta}-\delta}X (rT^{-1-n_{\beta}}\phi)|_{\Sigma_0}\in&\: \mathbf{H}_{N} \quad \textnormal{and}\quad r^{\frac{\beta_0}{2}-n_{\beta}-\delta}T^{-1-n_{\beta}}\phi|_{\Sigma_0}\in \mathbf{H}_{N+1}\\
 \textnormal{and}&\quad r^{\frac{\beta_0}{2}-n_{\beta}-\delta}T(T^{-1-n_{\beta}}\phi)|_{\Sigma_0}\in \mathbf{H}_{N+1}.
\end{align*}
Furthermore, there exists a constant $C=C(V_{\alpha},h,M,N,\delta)>0$, such that:
\begin{equation*}
\begin{split}
||r^{\frac{\beta_0}{2}-n_{\beta}-\delta}X (rT^{-1-n_{\beta}}\phi)_{\Sigma_0}||_{\mathbf{H}_N}+||r^{\frac{\beta_0}{2}-n_{\beta}-\delta}T^{-1-n_{\beta}}\phi_{\Sigma_0}||_{\mathbf{H}_N}\leq &\:C ||r^{\frac{\beta_0}{2}+1-\delta}X (r\phi)_{\Sigma_0}||_{\mathbf{H}_N}+C||r^{\frac{\beta_0}{2}-\delta}\phi|_{\Sigma_0}||_{\mathbf{H}_N}\\
&+C||r^{\frac{\beta_0}{2}-\delta}T\phi|_{\Sigma_0}||_{\mathbf{H}_N}.
\end{split}
\end{equation*}

We refer to $T^{-k}\phi$, with $k\in \N$, as the \emph{$k$-th time integral} of $\phi$.
If we assume additionally that
\begin{align*}
r^{\frac{\beta_1}{2}+1-\delta}X(r\phi_{\geq 1})|_{\Sigma}\in&\: \mathbf{H}_{N}\quad \textnormal{and}\quad r^{\frac{\beta_1}{2}-\delta}\phi_{\geq 1}|_{\Sigma}\in \mathbf{H}_{N} \quad \textnormal{and}\quad r^{\frac{\beta_1}{2}-\delta}T\phi_{\geq 1}|_{\Sigma}\in \mathbf{H}_{N},
\end{align*}
then we obtain:
\begin{align*}
r^{\frac{\beta_1}{2}-n_{\beta}-\delta}X (rT^{-1-n_{\beta}}\phi_{\geq 1})|_{\Sigma}\in&\: \mathbf{H}_{N} \quad \textnormal{and}\quad r^{\frac{\beta_1}{2}-n_{\beta}-\delta}T^{-1-n_{\beta}}\phi_{\geq 1}|_{\Sigma}\in \mathbf{H}_{N+1}\\
 \textnormal{and}&\quad r^{\frac{\beta_1}{2}-n_{\beta}-\delta}T(T^{-1-n_{\beta}}\phi_{\geq 1})|_{\Sigma}\in \mathbf{H}_{N+1}.
\end{align*}
Furthermore, there exists a constant $C=C(V_{\alpha},h,M,N,\delta)>0$, such that:
\begin{equation*}
\begin{split}
||r^{\frac{\beta_1}{2}-n_{\beta}-\delta}X (rT^{-1-n_{\beta}}\phi)_{\Sigma_0}||_{\mathbf{H}_N}+||r^{\frac{\beta_1}{2}-n_{\beta}-\delta}T^{-1-n_{\beta}}\phi_{\Sigma_0}||_{\mathbf{H}_N}\leq &\:C ||r^{\frac{\beta_1}{2}+1-\delta}X (r\phi)_{\Sigma_0}||_{\mathbf{H}_N}+C||r^{\frac{\beta_1}{2}-\delta}\phi|_{\Sigma_0}||_{\mathbf{H}_N}\\
&+C||r^{\frac{\beta_1}{2}-\delta}T\phi|_{\Sigma_0}||_{\mathbf{H}_N}.
\end{split}
\end{equation*}

\end{corollary}
\begin{proof}
We proceed as in the proof of Proposition \ref{prop:timeint} with $q=\beta_0-\delta$ and $\delta>0$ sufficiently small to construct $T^{-1}\phi$. Then we repeat the argument of Proposition \ref{prop:timeint} with the role of $\phi$ replaced with $T^{-1}\phi$ and $q=\beta_0-1-\delta$ to construct $T^{-2}\phi$. We can repeat this procedure a total number $n_{\beta}$ times to obtain the desired estimates.

The improved estimates for $\phi_{\geq 1}$ follow as above, but we apply Corollary \ref{cor:Linv1} instead of Corollary \ref{cor:Linv0}.
\end{proof}

We will moreover need to determine quantitatively the precise asymptotic behaviour of $T^{-1-n_{\beta}}\phi_0|_{\Sigma}$ in $r$ as $r\to \infty$. \textbf{In the proposition below we will assume that $\beta_0\neq 2l+1$ for all $\ell\in \N_0$.} This is equivalent to the assumption that $\alpha\neq l(l+1)$ for all $l\in \N_0$. This assumption is made for the sake of convenience. It is still possible to obtain the precise large-$r$ behaviour when $\beta_0=2l+1$ for some $l\in \N$, but the proof of Proposition \ref{prop:preciselargertimeint} would have to be modified slightly due to cancellations appearing in the asymptotics; see Remark \ref{rm:beta2lplus1} and \S \ref{sec:proofideas}.
\begin{proposition}
\label{prop:preciselargertimeint}
Let $\beta_0\neq 2l+1$ for all $\ell\in \N_0$. Consider smooth initial data for \eqref{eq:waveeq} with corresponding solution $\phi$.
\begin{enumerate}[label=\emph{(\roman*)}]
\item Let $\beta_0<1$. If $\int_{2M}^{\infty}w_0F_0[\phi]\,dr$ is finite and $w_0F_0[\phi]$ is sufficiently regular at $r=\infty$ with respect to $rX$, then
\begin{equation*}
\mathfrak{I}_{\beta_0}[\phi]:=\lim_{r \to \infty} r^{\frac{1}{2}+\frac{1}{2}\beta_0}X(rT^{-1}\phi_0)=-\frac{1-\beta_0}{2\beta_0}\int_{2M}^{\infty}w_0F_0[\phi_0]\,dr'.
\end{equation*}
Furthermore, there exists a constant $C=C(M,\alpha,h)>0$ such that for all $N\in \N_0$:
\begin{align*}
&\left|(rX)^NX(rT^{-1}\phi_0)(0,r)-\mathfrak{I}_{\beta_0}[\phi]  (rX)^N(r^{-\frac{1}{2}-\frac{1}{2}\beta_0})\right|\\
\leq&\: Cr^{-\frac{1}{2}-\frac{3}{2}\beta_0}\sum_{n=0}^N\left(||(r^2(rX)^nX(r\phi_0)||_{L^{\infty}(\Sigma_0)}+||(rX)^n\phi_0||_{L^{\infty}(\Sigma_0)}+||r(rX)^nT\phi_0||_{L^{\infty}(\Sigma_0)}\right).
\end{align*}
\item Let $\beta_0>1$. Assume the initial data for \eqref{eq:waveeq} is compactly supported. Then, for $n_{\beta}=\lfloor \frac{\beta_0+1}{2}\rfloor$:
\begin{equation*}
\mathfrak{I}_{\beta_0}[\phi]:=\lim_{r \to \infty} r^{\frac{1}{2}+(\frac{1}{2}\beta_0-n_{\beta})}X(rT^{-1-n_{\beta}}\phi_0)=-\frac{(1-\beta_0)(3-\beta_0)\ldots (2n_{\beta}+1-\beta_0)}{2n_{\beta}!\beta_0(1-\beta_0)\ldots (n_{\beta}-\beta_0)}\int_{2M}^{\infty}w_0F_0[\phi_0]\,dr'
\end{equation*}
is finite and generically non-zero.  Furthermore, there exists a constant $C=C(M,\alpha,h)>0$ such that for all $N\in \N_0$:
\begin{align*}
&\left|(rX)^NX(rT^{-1-n_{\beta}}\phi_0)(0,r)-\mathfrak{I}_{\beta_0}[\phi](rX)^N(  r^{-\frac{1}{2}-\frac{1}{2}\beta_0+n_{\beta}})\right|\\
\leq&\: Cr^{-\frac{3}{2}-\frac{1}{2}\beta_0+n_{\beta}}\Biggl(\sum_{n=0}^{n_{\beta}+N} ||r^{\frac{3}{2}+\frac{\beta_0}{2}}(rX)^nX(r\phi_0)||_{L^{\infty}(\Sigma_0)}\\
&+||r^{\frac{1}{2}+\frac{\beta_0}{2}}(rX)^n\phi_0||_{L^{\infty}(\Sigma_0)}+||r^{\frac{3}{2}+\frac{\beta_0}{2}}(rX)^nT\phi_0||_{L^{\infty}(\Sigma_0)}\Biggr).
\end{align*}
\end{enumerate}
\end{proposition}
\begin{proof}
We can express:
\begin{equation*}
X(Dr^2 w_0^2 XT^{-1}\check{\phi}_0)=w_{0} F_0[\phi],
\end{equation*}
with $F_0=\pi_0F$.
By the regularity properties established in Proposition \ref{prop:timeint}, it follows that
\begin{align*}
Dr^2 w_0^2 XT^{-1}\check{\phi}_0(0,2M)=&\:0,\\
\lim_{r\to \infty} T^{-1}\check{\phi}_0(0,r)=&\:0.
\end{align*}
Hence, we can integrate in $r$ from $r=2M$ to obtain the explicit formula:
\begin{equation*}
Dr^2 w_0^2 XT^{-1}\check{\phi}_0(0,r)=\int_{2M}^r w_0F_0[\phi]\,dr,
\end{equation*}
from which it follows that
\begin{align*}
XT^{-1}\check{\phi}_0(0,r)=&\:(Dr^2)^{-1} w_0^{-2}\int_{2M}^rw_0F_0[\phi]\,dr',\\
T^{-1}\check{\phi}_0(0,r)=&-\int_{r}^{\infty}(Dr^2)^{-1} w_0^{-2}\int_{2M}^rw_0F_0[\phi]\,dr.
\end{align*}
By the assumptions we made on our initial data, we moreover have that
\begin{align*}
|d_0|<&\:\infty,\\
d_0:=&\:\int_{2M}^{\infty}w_0F_0[\phi]\,dr'.
\end{align*}
Suppose first that $\beta_0<1$. Then for sufficiently regular $w_0F_0[\phi]$ at $r=\infty$ with respect to $rX$:
\begin{align*}
\int_{2M}^{r}w_0F_0[\phi]\,dr'=&\:d_0D+DO_{\infty}(r^{ \frac{1}{2}\beta_0-\frac{1}{2}}).
\end{align*}
We moreover have that
\begin{align*}
w_0^{-2}=&\:r^{1-\beta_0}+O_{\infty}(r^{1-2\beta_0}).
\end{align*}
Hence,
\begin{align*}
XT^{-1}\check{\phi}_0(0,r)=&[d_0 r^{-2}+O_{\infty}(r^{ \frac{1}{2}\beta_0-\frac{5}{2}})][r^{1-\beta_0}+O_{\infty}(r^{1-2\beta_0})],\\
T^{-1}\check{\phi}_0(0,r)=&-d_0\beta_0^{-1}r^{-\beta_0}+O_{\infty}(r^{-\frac{1}{2}(1+\beta_0)})+O_{\infty}(r^{-2\beta_0}),
\end{align*}
so we obtain the following expansions:
\begin{align*}
rT^{-1}\phi_0(0,r)=&-d_0\beta_0^{-1}r^{\frac{1}{2}-\frac{1}{2}\beta_0}+O_{\infty}(1)+O_{\infty}(r^{\frac{1}{2}-\frac{3}{2}\beta_0}),\\
X(rT^{-1}\phi_0)(0,r)=&-\frac{1}{2\beta_0}(1-\beta_0)d_0r^{-\frac{1}{2}-\frac{1}{2}\beta_0}+O_{\infty}(r^{-\frac{1}{2}-\frac{3}{2}\beta_0}),
\end{align*}
So
\begin{equation*}
\mathfrak{I}_{\beta_0}[\phi]=\lim_{r \to \infty} r^{\frac{1}{2}+\frac{1}{2}\beta_0}X(rT^{-1}\phi_0)=-\frac{1}{2\beta_0}(1-\beta_0)d_0
\end{equation*}
is finite and generically non-zero. The inequalities in (i) follow by expressing the terms appearing in $O_{\infty}(...)$ via derivatives of $F_0[\phi]$. 

Now suppose $\beta_0> 1$. By the regularity properties established in Corollary \ref{cor:regmultipletimeinv}, it follows that
\begin{align*}
Dr^2 w_0^2 XT^{-n}\check{\phi}_0(0,2M)=&\:0,\\
\lim_{r\to \infty} T^{-n}\check{\phi}_0(0,r)=&\:0,
\end{align*}
for all $1\leq n\leq n_{\beta}$. By the compactness assumption on the initial data, we have in particular that
\begin{align*}
\int_{2M}^{r}w_0F_0[\phi]\,dr'=&\:d_0+DO_{\infty}(r^{-2}).
\end{align*}
We moreover have that
\begin{align*}
w_0^{-2}=&\:r^{1-\beta_0}+O_{\infty}(r^{-\beta_0}).
\end{align*}
Hence, using that $\beta_0\neq 1$, since $\beta_0\neq 2l+1$ for all $l\in \N_0$:
\begin{align*}
T^{-1}\check{\phi}_0(0,r)=&-d_0\beta_0^{-1}r^{-\beta_0}+O_{\infty}(r^{-\beta_0-1}),
\end{align*}
so we obtain:
\begin{align*}
rT^{-1}\phi_0(0,r)=&-d_0\beta_0^{-1}r^{\frac{1}{2}-\frac{1}{2}\beta_0}+O_{\infty}(r^{-\frac{1}{2}\beta_0-\frac{3}{2}}),
\end{align*}
and
\begin{align*}
w_0F_0[T^{-1}{\phi}]=&-\beta_0^{-1}(1-\beta_0)d_0 +O_{\infty}(r^{-1}).
\end{align*}

We now consider $T^{-2}\check{\phi}$. From the above expressions, it follows that
\begin{align*}
\int_{2M}^{r}w_0F_0[T^{-1}\phi]\,dr'=&\:-\beta_0^{-1}(1-\beta_0)d_0D r+D \log r O_{\infty}(r^{-1}),
\end{align*}
and hence,
\begin{align*}
T^{-2}\check{\phi}_0(0,r)=&-d_0\beta_0^{-1}r^{-\beta_0+1}+ O_{\infty}(r^{-\beta_0+\delta}),
\end{align*}
where we can take $\delta>0$ arbitrarily small.

We will establish analogous estimates for $T^{-k}\check{\phi}_0(0,r)$ with $2\leq k\leq 1+n_{\beta}$ by induction. We take as the inductive assumption:
\begin{equation*}
T^{-k}\check{\phi}_0(0,r)=-d_0\frac{(1-\beta_0)(3-\beta_0)\ldots (2k-3-\beta_0)}{\beta_0 (1-\beta_0)\ldots (k-1-\beta_0)}\frac{1}{(k-1)!} r^{k-1-\beta_0}+O_{\infty}(r^{k-2-\beta_0-\delta}),
\end{equation*}
with $2\leq k\leq n_{\beta}$.

Note that the above estimate holds for $k=2$, by above. With the induction assumption, we obtain:
\begin{align*}
rT^{-k}{\phi}_0(0,r)=&-d_0\frac{(1-\beta_0)(3-\beta_0)\ldots (2k-3-\beta_0)}{\beta_0 (1-\beta_0)\ldots (k-1-\beta_0)}\frac{1}{(k-1)!} r^{k-\frac{1}{2}-\frac{1}{2}\beta_0}+O_{\infty}(r^{k-\frac{3}{2}-\frac{1}{2}\beta_0-\delta}),\\
w_0 F[T^{-k}{\phi}_0]=&\:-d_0\frac{(1-\beta_0)(3-\beta_0)\ldots (2k-3-\beta_0)(2k-1-\beta_0)}{\beta_0 (1-\beta_0)\ldots (k-1-\beta_0)}\frac{1}{(k-1)!} r^{k-1}+O_{\infty}(r^{k-2+\delta}).
\end{align*}
Then it follows that
\begin{align*}
X(T^{-k-1}\check{\phi}_0(0,r))=&-d_0\frac{(1-\beta_0)(3-\beta_0)\ldots (2k-3-\beta_0)(2k-1-\beta_0)}{\beta_0 (1-\beta_0)\ldots (k-1-\beta_0)}\frac{1}{k(k-1)!} Dr^{k-1-\beta_0}+O_{\infty}(r^{k-2-\beta_0+\delta}),\\
T^{-k-1}\check{\phi}_0(0,r)=&-d_0\frac{(1-\beta_0)(3-\beta_0)\ldots (2k-3-\beta_0)(2k-1-\beta_0)}{\beta_0 (1-\beta_0)\ldots (k-1-\beta_0)(k-\beta_0)}\frac{1}{k(k-1)!} Dr^{k-\beta_0}+O_{\infty}(r^{k-1-\beta_0+\delta})
\end{align*}
Hence, the induction assumption holds also with $k$ replaced by $k+1$.

We arrive at the inequality in (ii) by estimating the terms in $O_{\infty}(\ldots)$ via $F_0[\phi]$ and its $X$-derivatives.\end{proof}

\begin{remark}
\label{rm:beta2lplus1}
If $\beta_0=2l+1$, for some $l\in \N_0$, it can be shown that, given smooth and compactly supported initial data for $\phi$, $T^{-l-2}\phi_0$ is well-defined and
\begin{equation*}
X(rT^{-l-2}\phi_0)|_{\Sigma}=D_{l}r^{-1}+O_{\infty}(r^{-2}).
\end{equation*}
Since we do not make use of this estimate, we will leave it to the reader to prove this statement and determine the precise form of $D_{l}$.
\end{remark}

\section{Late-time asymptotics and time-decay estimates}
\label{sec:latetimeasymp}
We will now apply Proposition \ref{prop:edecaytimeder} with $N=n_{\beta}+1$ to the time integral $T^{-n_{\beta}-1}\phi$, where $n_{\beta}=\lfloor \frac{\beta_0+1}{2}\rfloor$. We will see that the estimates in Corollary \ref{cor:regmultipletimeinv} do not provide enough regularity for $T^{-n_{\beta}-1}\phi$ to conclude that the weighted, higher-order energy norms appearing on the right-hand side of \eqref{eq:mainedcay} are finite. Hence, \textbf{the time-decay will be limited by the $r$-decay rates of $T^{-n_{\beta}-1}\phi$ along $\Sigma_0$}.

In order to go beyond (almost-)sharp time-decay estimates for $\phi$ to the precise late-time asymptotics of $\phi$ as $\tau\to \infty$, we will consider the difference
\begin{align*}
\widehat{\phi}=&\:\phi-\Phi,\\
\Phi(\tau,r)=&\: \Phi_0 w_0 (r)(\tau+1+2r)^{-\frac{1}{2}-\frac{\beta_0}{2}}(\tau+1)^{-\frac{1}{2}-\frac{\beta_0}{2}},\\
\Phi_0\in &\:\R\setminus\{0\}.
\end{align*}
The goal of this section is to show that the global leading-order behaviour in time of $\phi$ is actually determined by $\Phi$. We will show this as follows:
\begin{enumerate}
\item We will establish sufficiently strong decay estimates for $\widehat{\phi}$, which satisfies an inhomogeneous wave equation:
\begin{equation}
\label{eq:inhomwaveeq}
(\square_g-V_{\alpha})\widehat{\phi}=-(\square_g-V_{\alpha})\Phi.
\end{equation}
In Lemma \ref{lm:boxPsi} below, we show that the choice of the function $\Phi$ ensures that $(\square_g-V_{\alpha})\Phi$ has favourable decay properties in $\tau$ and $r$, which will allow us to apply the energy decay estimates of Proposition \ref{prop:edecaytimeder} to $\widehat{\phi}$.
\item We will choose the constant $\Phi_0$ such that the $(1+n_{\beta})$-th time integral of the difference $\widehat{\phi}$ has \emph{better} decay properties in $r$ along $\Sigma_0$ than the $(1+n_{\beta})$-th time integral of $\phi$. This allows us to use Proposition \ref{prop:edecaytimeder} to show that $\widehat{\phi}$ decays \emph{faster} than $\Phi$, so the late-time behaviour of $\phi$ is determined (globally) by $\Phi$.
\end{enumerate}
We define $\Phi^{(k)}$ as follows
\begin{equation*}
\Phi^{(k)}(\tau,r):=(-1)^k\int_{\tau}^{\infty}\int_{\tau_{1}}\ldots \int_{\tau_{k-1}}\Phi(\tau_k,r)\,d\tau_k d\tau_{k-1}\ldots d\tau_1.
\end{equation*}
Then $T^k\Phi^{(k)}=\Phi$. We will also denote $T^{-k}\Phi=\Phi^{(k)}$.

Furthermore, we can integrate by parts $k-1$ times in $\tau$ to express:
\begin{equation*}
\Phi^{(k)}(\tau,r)=\frac{(-1)^k}{(k-1)!} \int_{\tau}^{\infty} (s-\tau)^{k-1} \Phi(s,r)\,ds
\end{equation*}
for $k<\beta_0+1$.

We denote moreover
\begin{equation*}
T^{-k}\widehat{\phi}=T^{-k}\phi-\Phi^{(k)}.
\end{equation*}
\begin{lemma}
\label{lm:boxPsi}
\begin{enumerate}[label=\emph{(\roman*)}]
\item \label{item:Phiest1} For all $n_1,n_2\in \N_0$ and $k\in\N_0$, $k\leq \beta_0+1$, there exists a $C=C(M,V_{\alpha},n_1,n_2,k)>0$, such that
\begin{align}
\label{eq:boxPsiMpos}
\left|(rX)^{n_1}T^{n_2}(\square_g-V_{\alpha})\Phi^{(k)}\right|\leq &\: C|\Phi_0| r^{-\frac{3}{2}+\frac{1}{2}\beta_0-\min\{1,\beta_0\}}(\tau+1+2r)^{-\frac{1}{2}-\frac{1}{2}\beta_0}(\tau+1)^{-\frac{3}{2}-\frac{1}{2}\beta_0+k-n_{2}},
\end{align}
\item \label{item:Phiest2} Let $\beta_0\neq 2l+1$ for all $l\in \N_0$. For all $k\in \N_0$, such that $k<\frac{\beta_0}{2}+\frac{1}{2}$, we can expand
\begin{equation}
\label{eq:rasympPsi}
X(r\Phi^{(k+1)})|_{\Sigma_0}=\frac{(-1)^k2^{k-\beta_0}\pi \Gamma(\beta_0-k)\Phi_0}{k!\Gamma\left(\frac{\beta_0+1}{2}\right)\Gamma\left(\frac{\beta_0-1-2k}{2}\right)\cos \left(\frac{\beta_0-2k}{2\pi}\right)}r^{-\frac{1}{2}-\frac{1}{2}\beta_0+k}+O_{\infty}(r^{-\frac{3}{2}-\frac{1}{2}\beta_0+k})+O_{\infty}(r^{-2}).
\end{equation}
\end{enumerate}
\end{lemma}
\begin{proof}
We can write
\begin{equation*}
(\square_g-V_{\alpha})\Phi= r^{-2}X(Dr^2X\Phi)-2(1-h)XT\Phi+h\tilde{h}T^2\Phi+ \left[-2r+\left(r^2h\right)'\right]r^{-2}T\Phi-V_{\alpha}\Phi.
\end{equation*}
If we denote $f(\tau,r)=(\tau+1+2r)^{-\frac{1}{2}-\frac{1}{2}\beta_0}(\tau+1)^{-\frac{1}{2}-\frac{1}{2}\beta_0}$, and we use \eqref{eq:odew}, we can write:
\begin{equation*}
\begin{split}
\Phi_0^{-1}(\square_g-V_{\alpha})\Phi=&\: w_0D X^2f-2(1-h)w_0XTf+\left[(Dr^2)'+2Dr^2w_0^{-1}w_0'\right]r^{-2}w_0Xf\\
&+\left[-2(1-h)r^2w_0^{-1}w_0'-2r+\left(r^2h\right)'\right]r^{-2}w_0Tf+h\tilde{h}w_0T^2f.
\end{split}
\end{equation*}
We will need the following properties of the function $f$:
\begin{align*}
Xf=&\: -(\beta_0+1) (\tau+1+2r)^{-\frac{3}{2}-\frac{1}{2}\beta_0}(\tau+1)^{-\frac{1}{2}-\frac{1}{2}\beta_0},\\
Tf=&\: -\frac{1}{2}(\beta_0+1)( (\tau+1+2r)^{-1}+ (\tau+1)^{-1}) (\tau+1+2r)^{-\frac{1}{2}-\frac{1}{2}\beta_0}(\tau+1)^{-\frac{1}{2}-\frac{1}{2}\beta_0}\\
=&\: -(\beta_0+1)(\tau+1+2r)^{-\frac{3}{2}-\frac{1}{2}\beta_0}(\tau+1)^{-\frac{1}{2}-\frac{1}{2}\beta_0}-(\beta_0+1) r(\tau+1+2r)^{-\frac{3}{2}-\frac{1}{2}\beta_0}(\tau+1)^{-\frac{3}{2}-\frac{1}{2}\beta_0},\\
X^2f=&\:(\beta_0+1)(\beta_0+3) (\tau+1+2r)^{-\frac{5}{2}-\frac{1}{2}\beta_0}(\tau+1)^{-\frac{1}{2}-\frac{1}{2}\beta_0},\\
XTf=&\: \frac{1}{2}(\beta_0+1)(\beta_0+3)(\tau+1+2r)^{-\frac{5}{2}-\frac{1}{2}\beta_0}(\tau+1)^{-\frac{1}{2}-\frac{1}{2}\beta_0}+\frac{1}{2}(\beta_0+1)^2(\tau+1+2r)^{-\frac{3}{2}-\frac{1}{2}\beta_0}(\tau+1)^{-\frac{3}{2}-\frac{1}{2}\beta_0},\\
T^2f=&\: (\tau+1+2r)^{-\frac{1}{2}-\frac{1}{2}\beta_0} O_{\infty}(\tau+1)^{-\frac{5}{2}-\frac{1}{2}\beta_0}),\\
w_0=&\: r^{-\frac{1}{2}+\frac{1}{2}\beta_0}+O_{\infty}(r^{-\frac{1}{2}-\frac{1}{2}\beta_0})+O_{\infty}(r^{-\frac{3}{2}+\frac{1}{2}\beta_0}),\\
w_0'w_0^{-1}=&\: \frac{1}{2}(\beta_0-1) r^{-1}+O_{\infty}(r^{-1-\beta_0})+O_{\infty}(r^{-2}),
\end{align*}
to obtain
\begin{equation*}
\begin{split}
\Phi_0^{-1}\square_g\Phi=&\: (\beta_0+1)(\beta_0+3) [r^{-\frac{1}{2}+\frac{1}{2}\beta_0}+O_{\infty}(r^{-\frac{1}{2}-\frac{1}{2}\beta_0})+O_{\infty}(r^{-\frac{3}{2}+\frac{1}{2}\beta_0})](\tau+1+2r)^{-\frac{5}{2}-\frac{1}{2}\beta_0}(\tau+1)^{-\frac{1}{2}-\frac{1}{2}\beta_0}\\
&-(\beta_0+1)(\beta_0+3)[r^{-\frac{1}{2}+\frac{1}{2}\beta_0}+O_{\infty}(r^{-\frac{1}{2}-\frac{1}{2}\beta_0})+O_{\infty}(r^{-\frac{3}{2}+\frac{1}{2}\beta_0})](\tau+1+2r)^{-\frac{5}{2}-\frac{1}{2}\beta_0}(\tau+1)^{-\frac{1}{2}-\frac{1}{2}\beta_0}\\
&-(\beta_0+1)^2[r^{-\frac{1}{2}+\frac{1}{2}\beta_0}+O_{\infty}(r^{-\frac{1}{2}-\frac{1}{2}\beta_0})+O_{\infty}(r^{-\frac{3}{2}+\frac{1}{2}\beta_0})](\tau+1+2r)^{-\frac{3}{2}-\frac{1}{2}\beta_0}(\tau+1)^{-\frac{3}{2}-\frac{1}{2}\beta_0}\\
& -(\beta_0+1) [(\beta_0+1)r^{-\frac{3}{2}+\frac{1}{2}\beta_0}+O_{\infty}(r^{-\frac{3}{2}-\frac{1}{2}\beta_0})+O_{\infty}(r^{-\frac{5}{2}+\frac{1}{2}\beta_0})](\tau+1+2r)^{-\frac{3}{2}-\frac{1}{2}\beta_0}(\tau+1)^{-\frac{1}{2}-\frac{1}{2}\beta_0}\\
& +(\beta_0+1) [(\beta_0+1)r^{-\frac{3}{2}+\frac{1}{2}\beta_0}+O_{\infty}(r^{-\frac{3}{2}-\frac{1}{2}\beta_0})+O_{\infty}(r^{-\frac{5}{2}+\frac{1}{2}\beta_0})](\tau+1+2r)^{-\frac{3}{2}-\frac{1}{2}\beta_0}(\tau+1)^{-\frac{1}{2}-\frac{1}{2}\beta_0}\\
&+ (\beta_0+1) [(\beta_0+1)r^{-\frac{1}{2}+\frac{1}{2}\beta_0}+O_{\infty}(r^{-\frac{1}{2}-\frac{1}{2}\beta_0})+O_{\infty}(r^{-\frac{3}{2}+\frac{1}{2}\beta_0})](\tau+1+2r)^{-\frac{3}{2}-\frac{1}{2}\beta_0}(\tau+1)^{-\frac{3}{2}-\frac{1}{2}\beta_0}\\
&+ [r^{-\frac{5}{2}+\frac{1}{2}\beta_0}+O_{\infty}(r^{-\frac{5}{2}-\frac{1}{2}\beta_0})+O_{\infty}(r^{-\frac{7}{2}+\frac{1}{2}\beta_0})](\tau+1+2r)^{-\frac{1}{2}-\frac{1}{2}\beta_0} O_{\infty}((\tau+1)^{-\frac{5}{2}-\frac{1}{2}\beta_0}).
\end{split}
\end{equation*}
Note that most of the terms above cancel, so we are left with
\begin{equation*}
\Phi_0^{-1}(\square_g-V_{\alpha})\Phi=(\tau+1+2r)^{-\frac{1}{2}-\frac{1}{2}\beta_0}(\tau+1)^{-\frac{3}{2}-\frac{1}{2}\beta_0}O_{\infty}(r^{-\frac{3}{2}+\frac{1}{2}\beta_0-\min\{1,\beta_0\}}).
\end{equation*}
We conclude that the estimates in (i) hold.

We can verify by differentiating in time $\tau$ and appealing to standard properties of hypergeometric functions $\,_2F_1\left(a,b,c;z\right)$ (see for example \cite{olv74}) that:
\begin{equation}
\label{eq:hypgeomPsi}
\begin{split}
&X(r \cdot  r^{\frac{-1+\beta_0}{2}} w_0^{-1}\cdot\Phi^{(k+1)})(\tau,r)=\frac{(-1)^{k+1}}{(k+1)!}\frac{\pi(1+\beta_0)}{2 \cos\left(\frac{\pi}{2}(\beta_0-2k)\right)}r^{\frac{1}{2}(\beta_0-1)}(1+2r+\tau)^{-\beta_0}\\
\cdot &\Bigg[ -\frac{\Gamma(\beta_0-k)}{\Gamma\left(\frac{3+\beta_0}{2}\right)\Gamma\left(\frac{\beta_0-1-2k}{2}\right)}(2+2r+\tau)^k\,_2F_1\left(\frac{1}{2}(\beta_0-1),\beta_0-k,\frac{1}{2}(\beta_0-1-2k); \frac{\tau+1}{\tau+2r+1} \right)\\
&+\frac{\Gamma(1+k)}{\Gamma\left(\frac{\beta_0-1}{2}\right)\Gamma\left(\frac{5-\beta_0-1+k}{2}\right)}(2+2r+\tau)^{\frac{3-\beta_0+k}{2}}(1+\tau)^{\frac{3-\beta_0+2k}{2}} \,_2F_1\left(\frac{1}{2}(3+\beta_0),1+k,\frac{5-\beta_0+k}{2}; \frac{\tau+1}{\tau+2r+1}\right)\Bigg].
\end{split}
\end{equation}
To obtain (ii), we use that
\begin{equation*}
r^{\frac{-1+\beta_0}{2}}w_0^{-1}=1+O_{\infty}(r^{-\min \{\beta_0,1\}})
\end{equation*}
and we evaluate the expression on the right-hand side of \eqref{eq:hypgeomPsi} at $\tau=0$ and expand the hypergeometric functions in $\frac{1}{r}$.
\end{proof}
\begin{remark}
If we take $M=0$ in the Schwarzschild metric $g_M$ and we fix $V_{\alpha}=\alpha r^{-2}$, we in fact have that: $(\square_g-V_{\alpha})\Phi=0$.
\end{remark}
\begin{remark}
If $\beta_0=2l+1$ for some $l\in \N_0$, it can be shown that the following asymptotics hold:
\begin{equation*}
X(r\Phi^{(l+1)})(0,r)=(-1)^{l+1}\frac{l!}{2^{l+1}(2l+1)!} \Phi_0r^{-1}+O_{\infty}(r^{-2}).
\end{equation*}
\end{remark}

We now fix the constant $\Phi_0$:
\begin{equation}
\label{eq:choicePsi0}
\Phi_0=(-1)^{n_{\beta}}n_{\beta}!\frac{\Gamma\left(\frac{\beta_0+1}{2}\right)\Gamma\left(\frac{\beta_0-1-2n_{\beta}}{2}\right)\cos \left(\frac{\beta_0-2n_{\beta}}{2\pi}\right)}{2^{n_{\beta}-\beta_0}\pi \Gamma(\beta-n_{\beta})} \mathfrak{I}_{\beta}[\phi].
\end{equation}

With the choice \eqref{eq:choicePsi0}, it follows immediately that we can obtain improved decay properties in $r$ for $T^{-1-n_{\beta}}\widehat{\phi}_0$ compared to Proposition \ref{prop:preciselargertimeint}, since this choice of $\Phi_0$ ensures a cancellation between the leading-order terms in $r^{-1}$ of $X(rT^{-1-n_{\beta}}\widehat{\phi}_0)$ and  $X(r\Phi^{(1+n_{\beta})})$:
\begin{corollary}
\label{cor:estboxPsi}
Assume $\beta_0\neq 2l+1$ for all $l\in \N_0$. Let $\Phi_0$ satisfy \eqref{eq:choicePsi0} and assume smooth initial data for \eqref{eq:waveeq}.
\begin{enumerate}[label=\emph{(\roman*)}]
\item \label{item:diffdataest1} Let $\beta_0<1$. If $\lim_{r\to \infty}r^2X(r\phi)|_{\Sigma}<\infty$. Then for all $N\in \N_0$, there exists a constant $C=~C(M,\alpha,h,N)>~0$ such that:
\begin{align*}
&\left|(rX)^NX(rT^{-1}\widehat{\phi}_0)(0,r)\right|\\
\leq&\: Cr^{-\frac{1}{2}-\frac{3}{2}\beta_0}\sum_{n=0}^N\left(||(r^2(rX)^nX(r\phi_0)||_{L^{\infty}(\Sigma_0)}+||(rX)^n\phi_0||_{L^{\infty}(\Sigma_0)}+||r(rX)^nT\phi_0||_{L^{\infty}(\Sigma_0)}\right).
\end{align*}
\item \label{item:diffdataest2} Let $\beta_0>1$. Assume the initial data are compactly supported. Then, for $n_{\beta}=\lfloor \frac{\beta_0+1}{2}\rfloor$ and for all $N\in \N_0$, there exists a constant $C=C(M,\alpha,h,N)>0$, such that :
\begin{equation*}
\begin{split}
&\left|(rX)^NX(rT^{-1-n_{\beta}}\widehat{\phi}_0)(0,r)\right|\\
\leq&\: Cr^{-\frac{3}{2}-\frac{1}{2}\beta_0+n_{\beta}}\Biggl(\sum_{n=0}^{n_{\beta}+N} ||r^{\frac{3}{2}+\frac{\beta_0}{2}}(rX)^nX(r\phi_0)||_{L^{\infty}(\Sigma_0)}\\
&+||r^{\frac{1}{2}+\frac{\beta_0}{2}}(rX)^n\phi_0||_{L^{\infty}(\Sigma_0)}+||r^{\frac{3}{2}+\frac{\beta_0}{2}}(rX)^nT\phi_0||_{L^{\infty}(\Sigma_0)}\Biggr).
\end{split}
\end{equation*}
\end{enumerate}
\end{corollary}

\subsection{Energy decay estimates for $\widehat{\phi}=\phi-\Phi$}
\label{sec:edecay}
In this section and in Section \ref{sec:addedecay} below, we establish decay-in-time estimates for ($r$-weighted) energies of $\widehat{\phi}$ along $\Sigma_{\tau}$, which are necessary for obtaining the $L^{\infty}$ decay-in-time estimates in Section \ref{sec:pointdecay}.

We will first treat separately the cases $\alpha<0$ and $\alpha>0$.

\subsubsection{$\alpha<0$ }
In order to obtain time-decay for $T^{-1}\widehat{\phi}$, which for $\alpha$ such that $3\beta_0\leq 1$ decays too slowly in $r$ to apply directly the (inhomogeneous analogues of) the energy decay estimates of Proposition \ref{prop:edecaytimeder}, we first introduce $\widehat{\phi}_{\Sc,r_{\infty}}^{(1)}$, which we define to be the solution to
\begin{equation*}
(\square_g-V_{\alpha}) \widehat{\phi}_{\Sc,r_{\infty}}^{(1)}=-(\square_g-V_{\alpha})\Phi^{(1)},
\end{equation*}
arising from the cut-off initial data:
\begin{equation*}
(\widehat{\phi}_{\Sc,r_{\infty}}^{(1)}|_{\Sigma_0},T\widehat{\phi}_{\Sc,r_{\infty}}^{(1)}|_{\Sigma_0})= (\chi_{r_{\infty}}T^{-1}\widehat{\phi}|_{\Sigma_0}, \chi_{r_{\infty}}\widehat{\phi}|_{\Sigma_0})
\end{equation*}
with $r_{\infty}>3M$ and $\chi_{r_{\infty}}$ a smooth cut-off function such that $\chi_{r_{\infty}}=1$ for $r\leq r_{\infty}$ and  $\chi_{r_{\infty}}=0$ for $r\geq 2r_{\infty}$.

By the domain of dependence property of \eqref{eq:waveeq}, we have that
\begin{equation*}
T^{k-1}\widehat{\phi}=T^{k}\widehat{\phi}_{\Sc,r_{\infty}}^{(1)}\quad \textnormal{in}  \quad \left\{v\leq v|_{\tau=0}(r_{\infty}):=v_{r_{\infty}}\right \} 
\end{equation*}
for all $k\in \N_0$.

We first prove analogues of the non-degenerate energy boundedness estimate \eqref{eq:nondegebound} and the integrated energy estimates in Corollary \ref{cor:nondegiled} for $\widehat{\phi}_{\Sc,r_{\infty}}^{(1)}$, which take into account the inhomogeneity in the equation for $\widehat{\phi}$.

\begin{proposition}
\label{prop:inhomeboundied}
Let $\alpha<0$ and $3\beta_0\leq 1$, $N\in \N$ and let $\delta>0$ be arbitrarily small. Then there exists a constant $C=C(M,h,V_{\alpha},N,\delta)>0$ such that for all $0\leq \tau_1<\tau_2$
\begin{enumerate}[label=\emph{(\roman*)}]
\item \label{item:inhomest1}
\begin{equation}
\label{eq:inhomebound}
\begin{split}
E_{\mathbf{N}}[T^k\widehat{\phi}_{\Sc,r_{\infty}}^{(1)}](\tau_2)\leq &\:C E_{\mathbf{N}}[T^k\widehat{\phi}_{\Sc,r_{\infty}}^{(1)}](\tau_1)+C(1+\tau_1)^{-1} \int_{\Sigma_{\tau_1}\cap\{r\geq R\}} r(LT^k\widehat{\psi}_{\Sc,r_{\infty}}^{(1)})^2\,d\sigma dr\\ 
&+ C\Phi_0^2(1+\tau_1)^{-3\beta_0-2k+2\delta},
\end{split}
\end{equation}
\item \label{item:inhomest2}
\begin{equation}
\label{eq:inhomeied}
\begin{split}
E_{\mathbf{N}}[T^k\widehat{\phi}_{\Sc,r_{\infty}}^{(1)}](\tau_2)+& \int_{\Sigma_{\tau_2}\cap\{r\geq R\}} r(LT^k\widehat{\psi}_{\Sc,r_{\infty}}^{(1)})^2\,d\sigma dr+\int_{\tau_1}^{\tau_2}E_{\mathbf{N}}[T^k\widehat{\phi}_{\Sc,r_{\infty}}^{(1)}]\,d\sigma dr d\tau\\ 
 \leq &\: C E_{\mathbf{N}}[T^k\widehat{\phi}_{\Sc,r_{\infty}}^{(1)}](\tau_1)+CE_{\mathbf{N}}[T^{k+1}\widehat{\phi}_{\Sc,r_{\infty}}^{(1)}](\tau_1)+C \int_{\Sigma_{\tau_1}\cap\{r\geq R\}} r(LT^k\widehat{\psi}_{\Sc,r_{\infty}}^{(1)})^2\,d\sigma dr\\ 
 &+ C\Phi_0^2\max\{(1+\tau_1)^{1-3\beta_0-2k+2\delta},(1+\tau_2)^{1-3\beta_0-2k+2\delta}\},
 \end{split}
\end{equation}
\item \label{item:inhomest3}
\begin{equation}
\label{eq:inhomeied2}
\begin{split}
\int_{\tau_1}^{\tau_2}&\int_{\Sigma_{\tau}\cap\{R_0\leq r\leq R_1\}} \Bigg[\sum_{n_1+n_2+n_3\leq N}|\snabla_{\s^2}^{n_1}X^{n_2}T^{n_3+k}\widehat{\phi}_{\Sc,r_{\infty}}^{(1)}|^2+|\snabla_{\s^2}^{n_1+1}X^{n_2}T^{n_3+k}\widehat{\phi}_{\Sc,r_{\infty}}^{(1)}|^2\\ 
&+|\snabla_{\s^2}^{n_1}X^{n_2+2}T^{n_3+k}\widehat{\phi}_{\Sc,r_{\infty}}^{(1)}|^2+h_0|\snabla_{\s^2}^{n_1}X^{n_2}T^{n_3+1+k}\widehat{\phi}_{\Sc,r_{\infty}}^{(1)}|^2\Bigg]\,d\sigma dr d\tau\\ 
 \leq  &\: C \sum_{n=1}^{N+k}E_T[T^{n_2}\widehat{\phi}_{\Sc,r_{\infty}}^{(1)}](\tau_1)+C \int_{\Sigma_{\tau_1}\cap\{r\geq R\}} r(LT^{n}\widehat{\psi}_{\Sc,r_{\infty}}^{(1)})^2\,d\sigma dr\\ 
 &+ C\Phi_0^2\max\{(1+\tau_1)^{1-3\beta_0-2k+2\delta},(1+\tau_2)^{1-3\beta_0-2k+2\delta}\}.
 \end{split}
\end{equation}
\end{enumerate}
\end{proposition}
\begin{proof}
Let $0\leq \tau_1<\tau_2$. We will first prove \emph{\ref{item:inhomest1}} with $E_T$ instead of $E_{\mathbf{N}}$ on the right-hand side. Note that \eqref{eq:ebound} no longer holds due to the presence of an inhomogeneity in the equation for $\widehat{\phi}_{\Sc,r_{\infty}}^{(1)}$. In the estimate for \eqref{eq:ebound} we instead encounter the additional term:
\begin{equation*}
\int_{\tau_1}^{\tau_2}\int_{\Sigma_{\tau}} r^2 T\widehat{\phi}_{\Sc,r_{\infty}}^{(1)} \cdot (\square_g-V_{\alpha})\Phi^{(1)}\,d\sigma dr.
\end{equation*}

After integrating by parts in $\tau$ and applying Young's inequality with appropriate weights depending on $\tau$ and $r$ and $\epsilon>0$ and $\delta>0$ arbitrarily small, we can estimate the norm of the above integral by:
\begin{multline*}
\int_{\tau_1}^{\tau_2}\int_{\Sigma_{\tau}} \epsilon (1+\tau)^{-2}r^{1-\delta}(\widehat{\phi}_{\Sc,r_{\infty}}^{(1)})^2+  \frac{1}{4}\epsilon^{-1} r^{3+\delta}(1+\tau)^2((\square_g-V_{\alpha})T\Phi^{(1)})^2\,d\sigma dr d\tau\\
+\sum_{i=1}^2\int_{\Sigma_{\tau_i}} \epsilon (\widehat{\phi}_{\Sc,r_{\infty}}^{(1)})^2+\frac{1}{4}\epsilon^{-1}r^{2-2\delta}((\square_g-V_{\alpha})\Phi^{(1)})^2\,d\sigma dr.
\end{multline*}
We have that
\begin{equation*}
\sum_{i=1}^2\int_{\Sigma_{\tau_i}} \epsilon (\widehat{\phi}_{\Sc,r_{\infty}}^{(1)})^2\,d\sigma dr\leq \epsilon \left[E_T[\widehat{\phi}_{\Sc,r_{\infty}}^{(1)}](\tau_1)+E_T[\widehat{\phi}_{\Sc,r_{\infty}}^{(1)}](\tau_2)\right]
\end{equation*}
and by \eqref{eq:hardy}, it follows that
\begin{equation*}
\int_{\tau_1}^{\tau_2}\int_{\Sigma_{\tau}} \epsilon (1+\tau)^{-2}r^{1-\delta}(\widehat{\phi}_{\Sc,r_{\infty}}^{(1)})^2\,d\sigma dr d\tau\leq C \epsilon(1+\tau_1)^{-1} \sup_{\tau_1\leq \tau \leq \tau_2 }\int_{\Sigma_{\tau}} r^{1-\delta}(X\widehat{\psi}_{\Sc,r_{\infty}}^{(1)})^2\,d\sigma dr.
\end{equation*}
By applying moreover the estimates in \emph{\ref{item:Phiest1}} of Lemma \ref{lm:boxPsi}, we can also further estimate for $\delta>0$ suitably small:
\begin{align*}
\int_{\tau_1}^{\tau_2}\int_{\Sigma_{\tau}} \frac{1}{4}\epsilon^{-1} r^{3+\delta}(1+\tau)^2((\square_g-V_{\alpha})T\Phi^{(1)})^2\,d\sigma dr d\tau\leq &\: C\epsilon^{-1}  \Phi_0^2(1+\tau_1)^{-3\beta_0+\delta},\\
\sum_{i=1}^2\int_{\Sigma_{\tau_i}} \frac{1}{4}\epsilon^{-1}r^{2-2\delta}((\square_g-V_{\alpha})\Phi^{(1)})^2\,d\sigma dr\leq &\: C \epsilon^{-1}\Phi_0^2(1+\tau_1)^{-2-2\beta_0}.
\end{align*}
We can similarly, replace $\widehat{\phi}_{\Sc,r_{\infty}}^{(1)}$ by $T^k\widehat{\phi}_{\Sc,r_{\infty}}^{(1)}$ above and use that
\begin{align*}
\int_{\tau_1}^{\tau_2}\int_{\Sigma_{\tau}} \frac{1}{2}\epsilon^{-1} r^{3+\delta}(1+\tau)^2((\square_g-V_{\alpha})T^{k+1}\Phi^{(1)})^2\,d\sigma dr d\tau\leq &\:C\epsilon^{-1}  \Phi_0^2(1+\tau_1)^{-3\beta_0-2k+\delta},\\
\sum_{i=1}^2\int_{\Sigma_{\tau_i}} \frac{1}{4}\epsilon^{-1}r^{2-2\delta}((\square_g-V_{\alpha})T^k\Phi^{(1)})^2\,d\sigma dr\leq &\: C \epsilon^{-1}\Phi_0^2(1+\tau_1)^{-2-2\beta_0-2k}.
\end{align*}
Hence, for $\epsilon>0$ arbitrarily small, we can estimate
\begin{equation}
\label{eq:auxinhomeest}
\begin{split}
E_T[T^k\widehat{\phi}_{\Sc,r_{\infty}}^{(1)}](\tau_2)\leq &\: (1+\epsilon)E_T[T^k\widehat{\phi}_{\Sc,r_{\infty}}^{(1)}](\tau_1)+\epsilon (1+\tau_1)^{-1} \sup_{\tau_1\leq \tau \leq \tau_2 }\int_{\Sigma_{\tau}} r^{1-\delta}(LT^k\widehat{\psi}_{\Sc,r_{\infty}}^{(1)})^2\,d\sigma dr\\
&+C\epsilon^{-1}  \Phi_0^2(1+\tau_1)^{-3\beta_0-2k+\delta}.
\end{split}
\end{equation}
In order to replace $E_T$ on the left-hand side with $E_{\mathbf{N}}$ we proceed as in the derivation of \eqref{eq:redshiftest}, but due to the presence of an the inhomogeneity, we need to include the following additional term on the right-hand side:
\begin{equation*}
C\epsilon^{-1}\int_{\Sigma_{\tau_2}\cap\{r\leq (2+\delta)M\}}((\square_g-V_{\alpha})\Phi^{(1)})^2\,d\sigma dr d\tau\leq C\epsilon^{-1} \Phi_0^2 (1+\tau)^{-1-2\beta_0}.
\end{equation*}
Then, we apply the proof of Corollary \ref{cor:redshifebound} to conclude that \eqref{eq:auxinhomeest} holds with $E_T$ replaced with $E_{\mathbf{N}}$ on both sides of the equation.

In order to arrive at the integrated energy estimates in Corollary \ref{cor:nondegiled} with $\phi$ replaced by $\widehat{\phi}_{\Sc,r_{\infty}}^{(1)}$, we similarly have to estimate the following main additional term involving the extra inhomogeneity in the arguments leading to a proof of Corollary \ref{cor:nondegiled}:
\begin{multline*}
\int_{\tau_1}^{\tau_2}\int_{\Sigma_{\tau}} r  |L\widehat{\psi}_{\Sc,r_{\infty}}^{(1)}| \cdot r|(\square_g-V_{\alpha})\Phi^{(1)}|\,d\sigma dr\leq \epsilon \int_{\tau_1}^{\tau_2}\int_{\Sigma_{\tau}} (1+\tau)^{-1-\delta}r (L\widehat{\psi}_{\Sc,r_{\infty}}^{(1)})^2\,d\sigma dr d\tau\\
+\frac{1}{4}\epsilon^{-1} \int_{\tau_1}^{\tau_2}\int_{\Sigma_{\tau}} r^3(1+\tau)^{1+\delta}((\square_g-V_{\alpha})\Phi^{(1)})^2\,d\sigma dr.
\end{multline*}
We estimate the right-hand side as above via \emph{\ref{item:Phiest1}} of Lemma \ref{lm:boxPsi}, noting that
\begin{equation*}
\frac{1}{4}\epsilon^{-1} \int_{\tau_1}^{\tau_2}\int_{\Sigma_{\tau}} r^3(1+\tau)^{1+\delta}((\square_g-V_{\alpha})\Phi^{(1)})^2\,d\sigma dr \leq C\epsilon^{-1}\Phi_0^2\max\{(1+\tau_2)^{1-3\beta_0},(1+\tau_1)^{1-3\beta_0+2\delta}\}
\end{equation*}
and similarly, for $k\geq 1$:
\begin{equation*}
\frac{1}{4}\epsilon^{-1} \int_{\tau_1}^{\tau_2}\int_{\Sigma_{\tau}} r^3(1+\tau)^{1+\delta}((\square_g-V_{\alpha})T^k\Phi^{(1)})^2\,d\sigma dr \leq C\epsilon^{-1}\Phi_0^2(1+\tau_1)^{1-2k-3\beta_0+2\delta}.
\end{equation*}

By combining \eqref{eq:auxinhomeest} with the above integrated energy estimates, we can remove the supremum over $\tau$ appearing on the right hand side of \eqref{eq:auxinhomeest} to conclude that \emph{\ref{item:inhomest1}} holds, and we can also conclude that \emph{\ref{item:inhomest2}} holds.

Finally, \emph{\ref{item:inhomest3}} follows from \emph{\ref{item:inhomest2}}, without loss of $T$-derivatives on the right-hand side by applying elliptic estimates as in the proof of Corollary \ref{cor:nondegiled}, using that the additional terms involving the inhomogeneity can be bounded straightforwardly by applying \emph{\ref{item:inhomest1}} of Lemma \ref{lm:boxPsi} again.
\end{proof}

In the proposition below we obtain energy decay estimates for $\widehat{\phi}_{\Sc,r_{\infty}}^{(1)}$ in terms of initial energies for $\widehat{\phi}$.

\begin{proposition}
\label{prop:edecaytimedercutoff}
Let $\alpha<0$, $N\in \N$ and let $\delta>0$ be arbitrarily small. Then there exists a constant $C=C(M,h,V_{\alpha},N,\delta)>0$ such that for all $0\leq k\leq 1$:
\begin{equation}
\label{eq:mainedcaycutoff}
\begin{split}
&E_{\mathbf{N}}[T^{N+1-k}\widehat{\phi}_{\Sc,r_{\infty}}^{(1)}](\tau)\\
\leq &\: C(1+\tau)^{-3-2N+2k+\delta}r_{\infty}^{\max\{1-3\beta_0,0\}+\delta}\Bigg[\sum_{\substack{0\leq n_1+n_2+n_3\leq 2(N+1)\\ n_1+n_2\leq N+1}}\int_{\Sigma_{0}\cap \{r\geq R\}} r^{\min\{3\beta_0,1\}-\delta}|\snabla_{\s^2}^{n_1}L(rL)^{n_2}T^{n_3-1}\widehat{\psi}|^2\,d\sigma dr\\
&+ \sum_{\substack{0\leq n_1+n_2\leq 3(N+1)\\ n_1\leq N+1}}E_{\mathbf{N}}[\snabla_{\s^2}^{n_1}T^{n_2-1}\widehat{\phi}]\Bigg]+C|\Phi_0|^2 (1+\tau)^{-2-2N-3\beta_0+2k+2\delta}
\end{split}
\end{equation}
and moreover
\begin{equation}
\label{eq:mainedrwcaycutoff}
\begin{split}
&\int_{\Sigma_{\tau}\cap \{r\geq R\}} r^2(L T^{N+1-k}\widehat{\psi}_{\Sc,r_{\infty}}^{(1)})^2\,d\sigma dr \\
\leq &\: C(1+\tau)^{-1-2N+2k+\delta}r_{\infty}^{\max\{1-3\beta_0,0\}+\delta}\Bigg[\sum_{\substack{0\leq n_1+n_2+n_3\leq 2(N+1)\\ n_1+n_2\leq N+1}}\int_{\Sigma_{0}\cap \{r\geq R\}}  r^{\min\{3\beta_0,1\}-\delta}|\snabla_{\s^2}^{n_1}L(rL)^{n_2}T^{n_3-1}\widehat{\psi}|^2\,d\sigma dr\\
&+ \sum_{\substack{0\leq n_1+n_2\leq 3(N+1)\\ n_1\leq N+1}}E_{\mathbf{N}}[\snabla_{\s^2}^{n_1}T^{n_2-1}\widehat{\phi}]\Bigg]+C|\Phi_0|^2 (1+\tau)^{-2N-3\beta_0+2k+2\delta}.
\end{split}
\end{equation}
When $\phi=\phi_{\geq 1}$, we can take the $r$-weight inside the integrals on the right-hand sides above to be $r$ rather than $r^{\min\{3\beta_0,1\}-\delta}$ and we can omit the factor $r_{\infty}^{\max\{1-3\beta_0,0\}+\delta}$ and the term involving $|\Phi_0|^2$.
\end{proposition}
\begin{proof}
We proceed as in the proof of Proposition \ref{prop:edecaytimeder} and we estimate the additional terms present as a result of the inhomogeneity present in the equation for $\widehat{\phi}_{\Sc,r_{\infty}}^{(1)}$ as in the proof of Proposition \ref{prop:inhomeboundied}, via the application of the estimates in \emph{\ref{item:Phiest1}} of Lemma \ref{lm:boxPsi}, using moreover the energy boundedness and integrated energy estimates of Proposition \ref{prop:inhomeboundied}, in the case $3\beta_0\leq 1$.

Then we estimate the energies along $\Sigma_0$ for $T^k\widehat{\phi}_{\Sc,r_{\infty}}^{(1)}$ in terms of initial energies for $T^k(T^{-1}\widehat{\phi})$. The only non-trivial step here is the following, when $3\beta_0< 1$:
\begin{equation*}
\begin{split}
\int_{\Sigma_0} r(L (rL)^k T^l\widehat{\psi}_{\Sc,r_{\infty}}^{(1)})^2\,d\sigma dr=&\:\int_{\Sigma_0} r(L (rL)^k \chi_{r_{\infty}}T^{l-1}\widehat{\psi} )^2\,d\sigma dr\\
\leq &\: C\int_{\Sigma_0 \cap \{r\leq 2r_{\infty}\} } \chi_{r_{\infty}}^2 r(L (rL)^kT^{l-1}\widehat{\psi} )^2\,d\sigma dr\\
&+C\sum_{k_1+k_2=k}\int_{\Sigma_0 \cap \{r\leq 2r_{\infty}\} } r |(rX)^{k_1}X(\chi_{r_{\infty}})|^2( (rL)^{k_2}T^{l-1}\widehat{\psi} )^2\,d\sigma dr\\
\leq &\: Cr_{\infty}^{1-3\beta_0+\delta}\int_{\Sigma_0 \cap \{r\leq 2r_{\infty}\} } \chi_{r_{\infty}}^2  r^{3\beta_0-\delta}(L (rL)^k T^{l-1}\widehat{\psi} )^2\,d\sigma dr\\
&+Cr_{\infty}^{1-3\beta_0+\delta}\sum_{k_2\leq k}\int_{\Sigma_0 \cap \{R\leq r\leq 2r_{\infty}\} } r^{-2+3\beta_0-\delta}( (rL)^{k_2}T^{l-1}\widehat{\psi} )^2\,d\sigma dr\\
\leq &\: Cr_{\infty}^{1-3\beta_0+\delta}\left[\sum_{k_2\leq k}\int_{\Sigma_0}  r^{3\beta_0-\delta}(L (rL)^{k_2} T^{l-1}\widehat{\psi} )^2\,d\sigma dr+\sum_{k_1=0}^k E_T[T^{k_1+l-1} \widehat{\phi}]\right],
\end{split}
\end{equation*}
where we applied \eqref{eq:hardy} to estimate the $k_2=0$ term in the sum on the right-hand side of the second inequality and arrive at the last inequality, estimating the boundary terms at $r=R$ in terms of the $T$-energy by averaging appropriately around $r=R$.
\end{proof}

\subsubsection{$\alpha>0$ and $\beta_0\neq 2l+1$}
If we assume $\alpha>0$, we define $({\phi}_{\Sc,r_{\infty}}^{(1+n_{\beta})})_{\geq 1}$ to be the solution to \eqref{eq:waveeq} arising from initial data:
\begin{equation*}
((\phi_{\Sc,r_{\infty}}^{(1+n_{\beta})})_{\geq 1}|_{\Sigma_0},T(\phi_{\Sc,r_{\infty}}^{(1+n_{\beta})})_{\geq 1}|_{\Sigma_0})= (\chi_{r_{\infty}}T^{-1-n_{\beta}}{\phi}_{\geq 1}|_{\Sigma_0},\chi_{r_{\infty}}T^{-n_{\beta}}{\phi}_{\geq 1}|_{\Sigma_0}).
\end{equation*}
We emphasize that in contrast with ${\phi}_{\Sc,r_{\infty}}^{(1)}$ above, $({\phi}_{\Sc,r_{\infty}}^{(1+n_{\beta})})_{\geq 1}$ satisfies the \emph{homogeneous} equation \eqref{eq:waveeq}, so the energy boundedness and integrated energy decay estimates of Propositions \ref{prop:ebound} and \ref{cor:nondegiled} apply directly, without the need for an analogue of Proposition \ref{prop:inhomeboundied}. Note moreover that, in contrast with the $\alpha<0$ case, we have the additional, stronger, energy decay estimates in Proposition \ref{prop:edecaytimeder2} at our disposal.

\begin{proposition}
\label{prop:edecaytimebg1}
Let $\alpha>0$ and $\beta_0\neq 2l+1$ for all $l\in \N_0$. Let $N\in \N$, and let $\delta>0$ be arbitrarily small. Then there exists a constant $C=C(M,h,V_{\alpha},N,\delta)>0$ such that for all $0\leq k\leq 1+n_{\beta}$:
\begin{enumerate}[label=\emph{(\roman*)}]
\item \label{item:edecayag01}
\begin{equation}
\label{eq:mainedcaycutoffbg1}
\begin{split}
&E_{\mathbf{N}}[T^{N-k}\widehat{\phi}_0](\tau)\\
\leq &\: C(1+\tau)^{-3-2n_{\beta}-2N+2k+\delta}\Bigg[\sum_{\substack{0\leq n_1+n_2+n_3\leq 2(N+n_{\beta}+1)\\ n_1+n_2\leq N+1+n_{\beta}}}\int_{\Sigma_{0}\cap \{r\geq R\}}r|\snabla_{\s^2}^{n_1}L(rL)^{n_2}T^{n_3-1-n_{\beta}}\widehat{\psi}_0|^2\,d\sigma dr\\
&+ \sum_{\substack{0\leq n_1+n_2\leq 3(N+1)\\ n_1\leq N+n_{\beta}+1}}E_{\mathbf{N}}[\snabla_{\s^2}^{n_1}T^{n_2-1-n_{\beta}}\widehat{\phi}_0]\Bigg]+C|\Phi_0|^2 (1+\tau)^{-4-\beta_0-2N+2k+\delta}
\end{split}
\end{equation}
and moreover
\begin{equation}
\label{eq:mainedrwcaycutoffbg1}
\begin{split}
&\int_{\Sigma_{\tau}\cap \{r\geq R\}} r^2(L T^{N-k}\widehat{\psi_0})^2\,d\sigma dr \\
\leq &\: C(1+\tau)^{-1-2N+2k+\delta}\Bigg[\sum_{\substack{0\leq n_1+n_2+n_3\leq 2(N+n_{\beta}+1)\\ n_1+n_2\leq N+1}}\int_{\Sigma_{0}\cap \{r\geq R\}} r|\snabla_{\s^2}^{n_1}L(rL)^{n_2}T^{n_3-1-n_{\beta}}\widehat{\psi}_0|^2\,d\sigma dr\\
&+ \sum_{\substack{0\leq n_1+n_2\leq 3(N+1)\\ n_1\leq N+n_{\beta}+1}}E_{\mathbf{N}}[\snabla_{\s^2}^{n_1}T^{n_2-1-n_{\beta}}\widehat{\phi}_0]\Bigg]+C|\Phi_0|^2 (1+\tau)^{-4-\beta_0-2N+2k +2\delta}.
\end{split}
\end{equation}
\item  \label{item:edecayag02} if $\beta_1>2n_{\beta}$
\begin{equation}
\label{eq:mainedcaycutoffbg2}
\begin{split}
&E_{\mathbf{N}}[T^{N+n_{\beta}+1-k}({\phi}_{\Sc,r_{\infty}}^{(n_{\beta}+1)})_{\geq 1}](\tau)\\
\leq &\: C(1+\tau)^{-3-2n_{\beta}-2N+2k+\delta}r_{\infty}^{\max\{2n_{\beta}+1-\beta_1,0\}+\delta}\\
&\cdot \Bigg[\sum_{\substack{0\leq n_1+n_2+n_3\leq 2(N+n_{\beta}+1)\\ n_1+n_2\leq N+1+n_{\beta}}}\int_{\Sigma_{0}\cap \{r\geq R\}} r^{\min\{1,\beta_1-2n_{\beta}\}-\delta}|\snabla_{\s^2}^{n_1}L(rL)^{n_2}T^{n_3-1-n_{\beta}}{\psi}_{\geq 1}|^2\,d\sigma dr\\
&+ \sum_{\substack{0\leq n_1+n_2\leq 3(N+1)\\ n_1\leq N+n_{\beta}+1}}E_{\mathbf{N}}[\snabla_{\s^2}^{n_1}T^{n_2-1-n_{\beta}}{\phi}_{\geq 1}]\Bigg],
\end{split}
\end{equation}
and moreover
\begin{equation}
\label{eq:mainedrwcaycutoffbg2}
\begin{split}
&\int_{\Sigma_{\tau}\cap \{r\geq R\}} r^2(L T^{N+n_{\beta}+1-k}({\psi}_{\Sc,r_{\infty}}^{(n_{\beta}+1)})_{\geq 1})^2\,d\sigma dr \\
\leq &\: C(1+\tau)^{-1-2n_{\beta}-2N+2k+\delta}r_{\infty}^{\max\{2n_{\beta}+1-\beta_1,0\}+\delta}\\
\cdot &\Bigg[\sum_{\substack{0\leq n_1+n_2+n_3\leq 2(N+n_{\beta}+1)\\ n_1+n_2\leq N+n_{\beta}+1}}\int_{\Sigma_{0}\cap \{r\geq R\}}r^{\min\{1,\beta_1-2n_{\beta}\}-\delta}|\snabla_{\s^2}^{n_1}L(rL)^{n_2}T^{n_3-1-n_{\beta}}{\psi}_{\geq 1}|^2\,d\sigma dr\\
&+ \sum_{\substack{0\leq n_1+n_2\leq 3(N+1)\\ n_1\leq N+n_{\beta}+1}}E_{\mathbf{N}}[\snabla_{\s^2}^{n_1}T^{n_2-1-n_{\beta}}{\phi}_{\geq 1}]\Bigg],
\end{split}
\end{equation}
\item  \label{item:edecayag03} if $\beta_1\leq 2n_{\beta}$
\begin{equation}
\label{eq:mainedcaycutoffbg3}
\begin{split}
&E_{\mathbf{N}}[T^{N+n_{\beta}-k}({\phi}_{\Sc,r_{\infty}}^{(n_{\beta})})_{\geq 1}](\tau)\\
\leq &\: C(1+\tau)^{-2-2n_{\beta}-2N+2k}r_{\infty}^{2n_{\beta}-\beta_1+\delta}\\
&\cdot \Bigg[\sum_{\substack{0\leq n_1+n_2+n_3\leq 2(N+n_{\beta})+1\\ n_1+n_2\leq N+n_{\beta}}}\int_{\Sigma_{0}\cap \{r\geq R\}} r^{2+\beta_1-2n_{\beta}-\delta}|\snabla_{\s^2}^{n_1}L(rL)^{n_2}T^{n_3-n_{\beta}}{\psi}_{\geq 1}|^2\,d\sigma dr\\
&+ \sum_{\substack{0\leq n_1+n_2\leq 3(N+n_{\beta})\\ n_1\leq N+n_{\beta}}}E_{\mathbf{N}}[\snabla_{\s^2}^{n_1}T^{n_2-n_{\beta}}{\phi}_{\geq 1}]\Bigg],
\end{split}
\end{equation}
and moreover
\begin{equation}
\label{eq:mainedrwcaycutoffbg3}
\begin{split}
&\int_{\Sigma_{\tau}\cap \{r\geq R\}} r^2(L T^{N+n_{\beta}-k}({\psi}_{\Sc,r_{\infty}}^{(n_{\beta})})_{\geq 1})^2\,d\sigma dr \\
\leq &\:C(1+\tau)^{-2n_{\beta}-2N+2k}r_{\infty}^{2n_{\beta}-\beta_1+\delta}\\
&\cdot \Bigg[\sum_{\substack{0\leq n_1+n_2+n_3\leq 2(N+n_{\beta})+1\\ n_1+n_2\leq N+n_{\beta}}}\int_{\Sigma_{0}\cap \{r\geq R\}} r^{2+\beta_1-2n_{\beta}-\delta}|\snabla_{\s^2}^{n_1}L(rL)^{n_2}T^{n_3-n_{\beta}}{\psi}_{\geq 1}|^2\,d\sigma dr\\
&+ \sum_{\substack{0\leq n_1+n_2\leq 3(N+n_{\beta})\\ n_1\leq N+n_{\beta}}}E_{\mathbf{N}}[\snabla_{\s^2}^{n_1}T^{n_2-n_{\beta}}{\phi}_{\geq 1}]\Bigg].
\end{split}
\end{equation}
\end{enumerate}
\end{proposition}
\begin{proof}
Consider first the spherical harmonic mode $\phi_0$. We can repeat the steps in the proof of Proposition \ref{prop:edecaytimedercutoff}, using \eqref{eq:boxPsiMpos} to estimate the terms coming from the inhomogeneity in the equation for $\widehat{\phi}$. When considering $\phi_0$ in the case $\beta_0>1$, we carry out estimates directly for $\widehat{\phi}_0$, in contrast with Proposition \ref{prop:edecaytimedercutoff}.

Consider now the remaining spherical harmonic modes $\phi_{\geq 1}$. We will derive estimates for $({\phi}_{\Sc,r_{\infty}}^{(1+n_{\beta})})_{\geq 1}$. Since $(\square_g-V_{\alpha})({\phi}_{\Sc,r_{\infty}}^{(1+n_{\beta})})_{\geq 1}=0$, the energy boundedness estimate \eqref{eq:nondegebound} and the integrated energy estimates in Corollary \ref{cor:nondegiled} apply. 

Note that since $\beta_1>\beta_0$, we also have that
\begin{equation*}
\beta_1>2n_{\beta}-1.
\end{equation*}
Suppose that $\beta_1>2n_{\beta}$. Then we proceed again as in the proof of Proposition \ref{prop:edecaytimedercutoff}, but with $T^{1+n_{\beta}-k}({\phi}_{\Sc,r_{\infty}}^{(n_{\beta}+1)})_{\geq 1}$ replacing $\phi_0$, but there is no inhomogeneity present in the relevant wave equation. The only difference is the following estimate along $\Sigma_0$:
\begin{equation*}
\begin{split}
&\sum_{\substack{0\leq n_1+n_2+n_3\leq 2(N+n_{\beta}+1)\\ n_1+n_2\leq N+1}}\int_{\Sigma_{0}\cap \{r\geq R\}}r |\snabla_{\s^2}^{n_1}L(rL)^{n_2}T^{n_3}({\phi}_{\Sc,r_{\infty}}^{(1+n_{\beta})})_{\geq 1}|^2\,d\sigma dr\\
\leq &\:C r_{\infty}^{2n_{\beta}-\beta_1+\delta}\Bigg[\sum_{\substack{0\leq n_1+n_2+n_3\leq 2(N+n_{\beta}+1)\\ n_1+n_2\leq N+n_{\beta}+1}}\int_{\Sigma_{0}\cap \{r\geq R\}} r^{\min\{1,\beta_1-2n_{\beta}\}-\delta}|\snabla_{\s^2}^{n_1}L(rL)^{n_2}T^{n_3-1-n_{\beta}}{\psi}_{\geq 1}|^2\,d\sigma dr\\
&+E_T[\snabla_{\s^2}^{n_1}T^{n_2-1-n_{\beta}}{\phi}_{\geq 1}]\Bigg].
\end{split}
\end{equation*}

Now suppose that $\beta_1\leq 2n_{\beta}$. Then we apply the energy decay estimates in Proposition \ref{prop:edecaytimeder2} to $({\phi}_{\Sc,r_{\infty}}^{(n_{\beta})})_{\geq 1}$ instead of the energy decay estimates in Proposition \ref{prop:edecaytimeder}. In this case, we estimate the corresponding weighted energy along $\Sigma_0$ as follows:
\begin{equation*}
\begin{split}
&\sum_{\substack{0\leq n_1+n_2+n_3\leq 2(N+n_{\beta})+1\\ n_1+n_2\leq N+n_{\beta}}}\int_{\Sigma_{0}\cap \{r\geq R\}}r^2 |\snabla_{\s^2}^{n_1}L(rL)^{n_2}T^{n_3}({\phi}_{\Sc,r_{\infty}}^{(n_{\beta})})_{\geq 1}|^2\,d\sigma dr\\
\leq &\:C r_{\infty}^{\max\{2n_{\beta}+1-\beta_1,0\}}\Bigg[\sum_{\substack{0\leq n_1+n_2+n_3\leq 2(N+n_{\beta})+1\\ n_1+n_2\leq N+n_{\beta}}}\int_{\Sigma_{0}\cap \{r\geq R\}}r^{2+\beta_1-2n_{\beta}-\delta} |\snabla_{\s^2}^{n_1}L(rL)^{n_2}T^{n_3-n_{\beta}}{\psi}_{\geq 1}|^2\,d\sigma dr\\
&+E_T[\snabla_{\s^2}^{n_1}T^{n_2-n_{\beta}}{\phi}_{\geq 1}]\Bigg]. \qedhere
\end{split}
\end{equation*}
\end{proof}

\subsection{Additional $r$-weighted energy decay estimates}
\label{sec:addedecay}
In Section \ref{sec:edecay} we established energy decay estimates for (inhomogeneous) solutions to \eqref{eq:waveeq} arising from cut-off initial data with growing weights in the quantity $r_{\infty}$, which appears in the support of the cut-off. In this section we promote these estimates to genuine energy decay estimates for $\widehat{\phi}$ by interpolating between the decay estimates for the energies $E_{\mathbf{N}}$ and the energies with additional growing $r$-weights in the integrands. We moreover derive decay estimates for appropriate $r$-weighted energies for $\widehat{\phi}$, both for positive and for negative $r$-weights.

We first establish energy decay estimates in the $\alpha<0$ case:
\begin{proposition}
\label{prop:interpoledecay1}
Let $\alpha<0$, $N\in \N$, $p\geq 0$ and $p+\max\{1-3\beta_0,0\}<2$ and let $\delta>0$ be arbitrarily small. Then there exists a constant $C=C(M,h,V_{\alpha},N,\delta,p)>0$ such that for all $0\leq k\leq 1$
\begin{equation}
\label{eq:finaledecayal0}
\begin{split}
E_{\mathbf{N}}[T^{N-k}\widehat{\phi}]\\
\leq &\: C(1+\tau)^{\max\{1-3\beta_0,0\}-3-2N+2k+\delta}\Bigg[\sum_{\substack{0\leq n_1+n_2+n_3\leq 2(N+1)\\ n_1+n_2\leq N+1}}\int_{\Sigma_{0}\cap \{r\geq R\}}  r^{\min\{3\beta_0,1\}-\delta}|\snabla_{\s^2}^{n_1}L(rL)^{n_2}T^{n_3-1}\widehat{\psi}|^2\,d\sigma dr\\
&+ \sum_{\substack{0\leq n_1+n_2\leq 3(N+1)\\ n_1\leq N+1}}E_{\mathbf{N}}[\snabla_{\s^2}^{n_1}T^{n_2-1}\widehat{\phi}]+|\Phi_0|^2\Bigg]
\end{split}
\end{equation}
and
\begin{equation}
\label{eq:rpdecayhatphi}
\begin{split}
\int_{\Sigma_{\tau}\cap \{r\geq R\}}& r^p(LT^{N-k}\widehat{\psi})^2\,d\sigma dr\\
\leq &\: C(1+\tau)^{p+\max\{1-3\beta_0,0\}-3-2N+2k+\delta}\Biggl[\sum_{\substack{0\leq n_1+n_2+n_3\leq 2(N+1)\\ n_1+n_2\leq N+1}}\int_{\Sigma_{0}\cap \{r\geq R\}}  r^{\min\{3\beta_0,1\}-\delta}|\snabla_{\s^2}^{n_1}L(rL)^{n_2}T^{n_3-1}\widehat{\psi}|^2\,d\sigma dr\\
&+ \sum_{\substack{0\leq n_1+n_2\leq 3(N+1)\\ n_1\leq N+1}}E_{\mathbf{N}}[\snabla_{\s^2}^{n_1}T^{n_2-1}\widehat{\phi}]+|\Phi_0|^2\Biggr].
\end{split}
\end{equation}
When restricting to $\phi=\phi_{\geq 1}$, the above estimates hold for $0\leq p<2$ if we replace $\min\{3\beta_0,1\}$ and $\max\{1-3\beta_0,0\}$ with 1 and 0, respectively and we omit the term $|\Phi_0|^2$. 
\end{proposition}
\begin{proof}
First, recall that $T^{-1}\widehat{\phi}=\widehat{\phi}_{\Sc,r_{\infty}}^{(1)}$ in $v\leq v_{r_{\infty}}$. 

We split
\begin{equation*}
\begin{split}
\int_{\Sigma_{\tau}\cap\{r\geq R\}}&r^p(LT^{N-k}\widehat{\psi})^2\,d\sigma dr=  \int_{\Sigma_{\tau}\cap \{R\leq r\leq \tau+R\}} r^p(LT^{N-k}\widehat{\psi})^2\,d\sigma dr+\int_{\Sigma_{\tau}\cap \{r\geq \tau+R\}} r^p(LT^{N-k}\widehat{\psi})^2\,d\sigma dr.
\end{split}
\end{equation*}
First, choose $r_{\infty}=R+B \tau$, for some suitably large constant $B>0$. Then we can apply Proposition \ref{prop:edecaytimedercutoff} to obtain
\begin{equation*}
\begin{split}
&\int_{\Sigma_{\tau}\cap \{R\leq r\leq \tau+R\}} r^p(LT^{N-k}\widehat{\psi})^2\,d\sigma dr=\int_{\Sigma_{\tau}\cap \{v\leq v_{\infty}\} \cap \{R\leq r\leq \tau+R\}}r^p(LT^{N+1-k}\widehat{\psi}_{\Sc,r_{\infty}}^{(1)})^2\,d\sigma dr\\
\leq &\: C(1+\tau)^p E_{\mathbf{N}}[T^{N+1-k}\widehat{\phi}_{\Sc,r_{\infty}}^{(1)}](\tau)\\
\leq &\: C(1+\tau)^{p-3+\max\{1-3\beta_0,0\}-2N+2k+\delta}\Bigg[\sum_{\substack{0\leq n_1+n_2+n_3\leq 2(N+1)\\ n_1+n_2\leq N+1}}\int_{\Sigma_{0}\cap \{r\geq R\}}  r^{\min\{3\beta_0,1\}-\delta}|\snabla_{\s^2}^{n_1}L(rL)^{n_2}T^{n_3-1}\widehat{\psi}|^2\,d\sigma dr\\
&+ \sum_{\substack{0\leq n_1+n_2\leq 3(N+1)\\ n_1\leq N+1}}E_{\mathbf{N}}[\snabla_{\s^2}^{n_1}T^{n_2-1}\widehat{\phi}]+|\Phi_0|^2\Bigg].
\end{split}
\end{equation*}
Note that we can moreover estimate analogously:
\begin{equation*}
\begin{split}
\int_{\Sigma_{\tau}\cap \{r\leq \tau+R\}}& r^2(XT^{N-k}\widehat{\phi})^2+r^2\tilde{h}h(T^{N+1-k}\widehat{\phi})^2+|\snabla_{\s^2}T^{N-k}\widehat{\phi}|^2\,d\sigma dr\\
=&\:\int_{\Sigma_{\tau}\cap \{v\leq v_{\infty}\} \cap \{r\leq \tau+R\}} r^2(XT^{N+1-k}\widehat{\phi}_{\Sc,r_{\infty}}^{(1)})^2+r^2\tilde{h}h(T^{N+2-k}\widehat{\phi}_{\Sc,r_{\infty}}^{(1)})^2+|\snabla_{\s^2}T^{N+1-k}\widehat{\phi}_{\Sc,r_{\infty}}^{(1)}|^2\,d\sigma dr\\
\leq &\: E_{\mathbf{N}}[T^{N+1-k}\widehat{\phi}_{\Sc,r_{\infty}}^{(1)}](\tau)\\
\leq &\: C(1+\tau)^{-3+\max\{1-3\beta_0,0\}-2N+2k+\delta}\Bigg[\sum_{\substack{0\leq n_1+n_2+n_3\leq 2(N+1)\\ n_1+n_2\leq N+1}}\int_{\Sigma_{0}\cap \{r\geq R\}}  r^{\min\{3\beta_0,1\}-\delta}|\snabla_{\s^2}^{n_1}L(rL)^{n_2}T^{n_3-1}\widehat{\psi}|^2\,d\sigma dr\\
&+ \sum_{\substack{0\leq n_1+n_2\leq 3(N+1)\\ n_1\leq N+1}}E_{\mathbf{N}}[\snabla_{\s^2}^{n_1}T^{n_2-1}\widehat{\phi}]+|\Phi_0|^2\Bigg].
\end{split}
\end{equation*}

We now partition the interval $[R+\tau,\infty)$ into dyadic intervals $[r_{i},r_{i+1}]$, with $r_0=R+\tau$ and, taking $B>0$ suitably large, independently of $i$, we set $r_{\infty}=B r_{i+1}$. Then
 \begin{equation*}
\begin{split}
&\int_{\Sigma_{\tau}\cap \{r_i\leq r\leq r_{i+1}\}}  r^p(LT^{N-k}\widehat{\psi})^2\,d\sigma dr=\int_{\Sigma_{\tau}\cap \{v\leq v_{\infty}\} \cap\{r_i\leq r\leq r_{i+1}\}} r^p(LT^{N+1-k}\widehat{\psi}_{\Sc,r_{\infty}}^{(1)})^2\,d\sigma dr\\
\leq &\:Cr_i^{p-2}\int_{\Sigma_{\tau}\cap \{v\leq v_{\infty}\} \cap\{r_i\leq r\leq r_{i+1}\}} r^2(LT^{N+1-k}\widehat{\psi}_{\Sc,r_{\infty}}^{(1)})^2\,d\sigma dr\\
\leq &\: C(1+\tau)^{-1-2N+2k+\delta}r_{i}^{\max\{1-3\beta_0,0\}+p-2+\delta}\Bigg[\sum_{\substack{0\leq n_1+n_2+n_3\leq 2(N+1)\\ n_1+n_2\leq N+1}}\int_{\Sigma_{0}\cap \{r\geq R\}}  r^{\min\{3\beta_0,1\}-\delta}|\snabla_{\s^2}^{n_1}L(rL)^{n_2}T^{n_3-1}\widehat{\psi}|^2\,d\sigma dr\\
&+ \sum_{\substack{0\leq n_1+n_2\leq 3(N+1)\\ n_1\leq N+1}}E_{\mathbf{N}}[\snabla_{\s^2}^{n_1}T^{n_2-1}\widehat{\phi}]+|\Phi_0|^2\Bigg].
\end{split}
\end{equation*}
Using that $r_i\geq R+\tau$, for $p<2-\max\{1-3\beta_0,0\}$ we can find $\delta>0$ suitably small, such that
\begin{align*}
r_{i}^{\max\{1-3\beta_0,0\}+p-2+\delta}\leq C(1+\tau)^{-2+\max\{1-3\beta,0\}+p+2\delta}r_i^{-\delta}.
\end{align*}

After summing over $i$ and using the convergence of geometric series, we are then left with
\begin{equation*}
\begin{split}
&\int_{\Sigma_{\tau}\cap \{r\geq \tau+R\}} r^p(LT^{N-k}\widehat{\psi})^2\,d\sigma dr\\
\leq &\: C(1+\tau)^{p-3+\max\{1-3\beta_0,0\}-2N+2k+2\delta}\Bigg[\sum_{\substack{0\leq n_1+n_2+n_3\leq 2(N+1)\\ n_1+n_2\leq N+1}}\int_{\Sigma_{0}\cap \{r\geq R\}}  r^{\min\{3\beta_0,1\}-\delta}|\snabla_{\s^2}^{n_1}L(rL)^{n_2}T^{n_3-1}\widehat{\psi}|^2\,d\sigma dr\\
&+ \sum_{\substack{0\leq n_1+n_2\leq 3(N+1)\\ n_1\leq N+1}}E_{\mathbf{N}}[\snabla_{\s^2}^{n_1}T^{n_2-1}\widehat{\phi}]+|\Phi_0|^2\Bigg]. \qedhere
\end{split}
\end{equation*}
\end{proof}

We now establish  energy decay estimates in the $\alpha>0$ case:
\begin{proposition}
\label{prop:weightedecayag0}
Let $\alpha>0$, $\beta_0\neq 2l+1$ for all $l\in \N_0$. Let $N\in \N$, $0\leq p<2-\max\{2n_{\beta}+1-\beta_1,0\}$ and let $\delta>0$ be arbitrarily small. Then there exists a constant $C=C(M,h,V_{\alpha},N,\delta)>0$ such that for all $0\leq k\leq n_{\beta}+1$:
\begin{enumerate}[label=\emph{(\roman*)}]
\item \label{item:ag0rpeest1} for all $0\leq p\leq 2$,
\begin{equation}
\label{eq:rpdecayhatphibg1b}
\begin{split}
\int_{\Sigma_{\tau}\cap \{r\geq R\}}& r^p(LT^{N-k}\widehat{\psi}_0)^2\,d\sigma dr\\
\leq &\: C(1+\tau)^{p-3-2n_{\beta}-2N+2k+\delta}\Bigg[\sum_{\substack{0\leq n_1+n_2+n_3\leq 2(N+n_{\beta}+1)\\ n_1+n_2\leq N+n_{\beta}+1}}\int_{\Sigma_{0}\cap \{r\geq R\}} r|\snabla_{\s^2}^{n_1}L(rL)^{n_2}T^{n_3-1-n_{\beta}}\widehat{\psi}_0|^2\,d\sigma dr\\
&+ \sum_{\substack{0\leq n_1+n_2\leq 3(N+1)\\ n_1\leq N+n_{\beta}+1}}E_{\mathbf{N}}[\snabla_{\s^2}^{n_1}T^{n_2-1-n_{\beta}}\widehat{\phi}_0]+\Phi_0^2\Bigg].
\end{split}
\end{equation}
\item  \label{item:ag0rpeest2} if $\beta_1>2n_{\beta}$
\begin{equation}
\label{eq:mainedcaycutoffbg2b}
\begin{split}
&E_{\mathbf{N}}[T^{N-k} \phi_{\geq 1}](\tau)\\
\leq &\: C(1+\tau)^{-3-2n_{\beta}+2k+\max\{2n_{\beta}+1-\beta_1,0\}+\delta-2N}\\
&\cdot \Bigg[\sum_{\substack{0\leq n_1+n_2+n_3\leq 2(N+n_{\beta}+1)\\ n_1+n_2\leq N+1+n_{\beta}}}\int_{\Sigma_{0}\cap \{r\geq R\}} r^{\min\{1,\beta_1-2n_{\beta}\}-\delta}|\snabla_{\s^2}^{n_1}L(rL)^{n_2}T^{n_3-1-n_{\beta}}{\psi}_{\geq 1}|^2\,d\sigma dr\\
&+ \sum_{\substack{0\leq n_1+n_2\leq 3(N+1)\\ n_1\leq N+n_{\beta}+1}}E_{\mathbf{N}}[\snabla_{\s^2}^{n_1}T^{n_2-1-n_{\beta}}{\phi}_{\geq 1}]\Bigg],
\end{split}
\end{equation}
and moreover, for $0\leq p<2-\max\{2n_{\beta}+1-\beta_1,0\}$,
\begin{equation}
\label{eq:mainedrwcaycutoffbg2b}
\begin{split}
&\int_{\Sigma_{\tau}\cap \{r\geq R\}} r^p(L T^{N-k} \psi_{\geq 1})^2\,d\sigma dr \\
\leq &\: C(1+\tau)^{p-3-2n_{\beta}+2k+\max\{2n_{\beta}+1-\beta_1,0\}+\delta-2N}\\
\cdot &\Bigg[\sum_{\substack{0\leq n_1+n_2+n_3\leq 2(N+n_{\beta}+1)\\ n_1+n_2\leq N+n_{\beta}+1}}\int_{\Sigma_{0}\cap \{r\geq R\}}r^{\min\{1,\beta_1-2n_{\beta}\}-\delta}|\snabla_{\s^2}^{n_1}L(rL)^{n_2}T^{n_3-1-n_{\beta}}{\psi}_{\geq 1}|^2\,d\sigma dr\\
&+ \sum_{\substack{0\leq n_1+n_2\leq 3(N+1)\\ n_1\leq N+n_{\beta}+1}}E_{\mathbf{N}}[\snabla_{\s^2}^{n_1}T^{n_2-1-n_{\beta}}{\phi}_{\geq 1}]\Bigg],
\end{split}
\end{equation}
\item  \label{item:ag0rpeest3} if $\beta_1\leq 2n_{\beta}$
\begin{equation}
\label{eq:mainedcaycutoffbg3b}
\begin{split}
&E_{\mathbf{N}}[T^{N-k} \phi_{\geq 1}](\tau)\\
\leq &\: C(1+\tau)^{-2-\beta_1+2k+\delta-2N}\\
&\cdot \Bigg[\sum_{\substack{0\leq n_1+n_2+n_3\leq 2(N+n_{\beta})+1\\ n_1+n_2\leq N+n_{\beta}}}\int_{\Sigma_{0}\cap \{r\geq R\}} r^{2+\beta_1-2n_{\beta}-\delta}|\snabla_{\s^2}^{n_1}L(rL)^{n_2}T^{n_3-n_{\beta}}{\psi}_{\geq 1}|^2\,d\sigma dr\\
&+ \sum_{\substack{0\leq n_1+n_2\leq 3(N+n_{\beta})\\ n_1\leq N+n_{\beta}}}E_{\mathbf{N}}[\snabla_{\s^2}^{n_1}T^{n_2-n_{\beta}}{\phi}_{\geq 1}]\Bigg],
\end{split}
\end{equation}
and moreover, for $0\leq p< 2-2n_{\beta}+\beta_1$,
\begin{equation}
\label{eq:mainedrwcaycutoffbg3b}
\begin{split}
&\int_{\Sigma_{\tau}\cap \{r\geq R\}} r^p(L T^{N-k}\psi_{\geq 1})^2\,d\sigma dr \\
\leq &\:C(1+\tau)^{p-2-\beta_1+2k+\delta-2N}\\
&\cdot \Bigg[\sum_{\substack{0\leq n_1+n_2+n_3\leq 2(N+n_{\beta})+1\\ n_1+n_2\leq N+n_{\beta}}}\int_{\Sigma_{0}\cap \{r\geq R\}} r^{2+\beta_1-2n_{\beta}-\delta}|\snabla_{\s^2}^{n_1}L(rL)^{n_2}T^{n_3-n_{\beta}}{\psi}_{\geq 1}|^2\,d\sigma dr\\
&+ \sum_{\substack{0\leq n_1+n_2\leq 3(N+n_{\beta})\\ n_1\leq N+n_{\beta}}}E_{\mathbf{N}}[\snabla_{\s^2}^{n_1}T^{n_2-n_{\beta}}{\phi}_{\geq 1}]\Bigg].
\end{split}
\end{equation}
\end{enumerate}
\end{proposition}
\begin{proof}
We partition the interval $[R,\infty)$ and apply the same interpolation argument as in the proof of Proposition \ref{prop:interpoledecay1}, appealing to the decay estimates in Proposition \ref{prop:edecaytimebg1}. We omit the details.
\end{proof}
Now we can established improved decay rates for $r$-weighted energies with \emph{decaying} weights in $r$.

We first introduce the following weighted initial energies along $\Sigma_0$ involving time integrals of $\phi$:
\begin{align*}
\mathbf{D}_{\beta_0,N,\delta}^{(1)}[\phi]:=&\:\sum_{\substack{0\leq n_1+n_2+n_3\leq 2(N+2)\\ n_1+n_2\leq N+2}}\int_{\Sigma_{0}\cap \{r\geq R\}}  r^{\min\{3\beta_0,1\}-\delta}|L(rL)^{n_2}T^{n_3-1}\widehat{\psi}_0|^2+r^{1-\delta}|\snabla_{\s^2}^{n_1}L(rL)^{n_2}T^{n_3-1}{\psi}_{\geq 1}|^2\,d\sigma dr\\
&+ \sum_{\substack{0\leq n_1+n_2\leq 3(N+2)\\ n_1\leq N+2}}E_{\mathbf{N}}[\snabla_{\s^2}^{n_1}T^{n_2-1}\widehat{\phi}]+\mathfrak{I}_{\beta}^2[\phi]\quad \textnormal{when $\beta_0<1$},\\
\mathbf{D}_{\beta_0,N,\delta}^{(n_{\beta}+1)}[\phi]:=&\: \sum_{\substack{0\leq n_1+n_2+n_3\leq 2(N+2n_{\beta}+2)\\ n_1+n_2\leq N+2n_{\beta}+2}}\int_{\Sigma_{0}\cap \{r\geq R\}} r|\snabla_{\s^2}^{n_1}L(rL)^{n_2}T^{n_3-1-n_{\beta}}\widehat{\psi}_0|^2\\
&+ r^{\min\{\beta_1-2n_{\beta},1\}-\delta}|\snabla_{\s^2}^{n_1}L(rL)^{n_2}T^{n_3-1-n_{\beta}}{\psi}_{\geq 1}|^2+r^{2+\min\{\beta_1-2n_{\beta},0\}-\delta}|\snabla_{\s^2}^{n_1}L(rL)^{n_2}T^{n_3-n_{\beta}}{\psi}_{\geq 1}|^2\,d\sigma dr\\
&+ \sum_{\substack{0\leq n_1+n_2\leq 3(N+2n_{\beta}+2)\\ n_1\leq N+2n_{\beta}+2}}E_{\mathbf{N}}[\snabla_{\s^2}^{n_1}T^{n_2-1-n_{\beta}}\widehat{\phi}]+\mathfrak{I}_{\beta}^2[\phi] \quad \textnormal{when $\beta_0>1$}.
\end{align*}

The above energies will play a role in the energy decay estimates in Proposition \ref{prop:negredecay}.

\begin{proposition}
\label{prop:negredecay}
Let $\beta_0\notin 2\N_0+1$ and $\delta>0$ arbitrarily small. Then there exists a constant $C=C(M,h,V_{\alpha},N,\delta)>0$ such that for all $0\leq k \leq n_{\beta}+1$:
\begin{enumerate}[label=\emph{(\roman*)}]
\item \label{item:negredecay1} when $\alpha<0$ $(\beta_0<1)$:
\begin{equation}
\label{eq:negredecaybl1}
\int_{\Sigma_{\tau}} r^{-1+2\beta_0}(XT^{N-k}w_{0}^{-1}\widehat{\phi}_0)^2\,d\sigma dr\leq C(1+\tau)^{\max\{1-3\beta_0,0\}+\beta_0-5+2k-2N+\delta}\mathbf{D}_{\beta_0,N,\delta}^{(1)}[\phi],
\end{equation}
and
\begin{equation*}
\begin{split}
\int_{\Sigma_{\tau}}r^{-\beta_0+\delta}(T^{N-k}\widehat{\phi}_{\geq 1})^2+r^{2-\beta_0+\delta}(XT^{N-k}\widehat{\phi}_{\geq 1})^2\,d\sigma dr\leq &\: C(1+\tau)^{-3-\beta_0+2k-2N+2\delta}\mathbf{D}_{\beta_0,N,\delta}^{(1)}[\phi],
\end{split}
\end{equation*}
\item \label{item:negredecay2} and when $\alpha>0$ $(\beta_0>1)$:
\begin{equation}
\label{eq:negredecaybg1}
\begin{split}
\int_{\Sigma_{\tau}}r^{-\beta_0+\delta}(T^{N-k}\widehat{\phi})^2+r^{2-\beta_0+\delta}(XT^{N-k}\widehat{\phi})^2\,d\sigma dr\leq &\: C(1+\tau)^{-3-\beta_0-2n_{\beta}+2k+\max\{2n_{\beta}+1-\beta_1,0\}-2N+2\delta}\mathbf{D}_{\beta_0,N,\delta}^{(n_{\beta}+1)}[\phi].
\end{split}
\end{equation}
\end{enumerate}
\end{proposition}
\begin{proof}
We prove the $k=0$ case here. The $1\leq k\leq 1+n_{\beta}$ case proceeds entirely analogously. Note that $\widehat{\phi}$ satisfies the equation:
\begin{align*}
\mathcal{L}\widehat{\phi}|_{\Sigma_{\tau}}=&\:\widehat{F}[T{\phi}|_{\Sigma_{\tau}}],\\
F[T\widehat{\phi}|_{\Sigma_{\tau}}]:=&\:2r(1-h)X T\widehat{\psi}|_{\Sigma_{\tau}}+\frac{dh}{dr}r^2T\widehat{\phi}|_{\Sigma_{\tau}}+r^2h\tilde{h}T^2 \widehat{\phi}|_{\Sigma_{\tau}}+r^2\square_g\Phi |_{\Sigma_{\tau}}.
\end{align*}
Consider first $\beta_0>1$. By Corollary \ref{cor:Linv0} with $q=-\beta_0+\delta$, we have that
\begin{equation*}
\begin{split}
\int_{\Sigma_{\tau}} r^{-\beta_0+\delta}\widehat{\phi}^2+r^{2-\beta_0+\delta}(X\widehat{\phi})^2\,d\sigma dr\leq &\:C \int_{\Sigma_{\tau}} r^{-\beta_0+\delta}(\mathcal{L}\widehat{\phi})^2\,d\sigma dr\\
\leq &\:  C\int_{\Sigma_{\tau}} r^{2-\beta_0+\delta}(XT\widehat{\psi})^2+ r^{-2-\beta_0+\delta}(T\widehat{\phi})^2+r^{-\beta_0+\delta}(r^2\tilde{h}h)^2(T^2\widehat{\phi})^2\\
&+r^{4-\beta_0+\delta}(\square_g\Phi|_{\Sigma_{\tau}})^2\,d\sigma dr.
\end{split}
\end{equation*}
Note that
\begin{equation*}
\int_{\Sigma_{\tau}}r^{4-\beta_0+\delta}(\square_g\Phi|_{\Sigma_{\tau}})^2\,d\sigma dr \leq C\Phi_0^2(1+\tau)^{-4-2\beta_0+2\delta}.
\end{equation*}
We apply Corollary \ref{cor:Linv0} once more, with $q=-\beta_0+\delta+2$ to estimate further
\begin{equation*}
\begin{split}
\int_{\Sigma_{\tau}} &r^{-\beta_0+\delta}(L\hat{\psi})^2\,d\sigma dr\leq C \int_{\Sigma_{\tau}} r^{2-\beta_0+\delta}(XT\widehat{\psi})^2+ r^{-2-\beta_0+\delta}(T\widehat{\phi})^2+r^{-\beta_0+\delta}(r^2\tilde{h}h)^2(T^2\widehat{\phi})^2\,d\sigma dr\\
 \leq &\: C\int_{\Sigma_{\tau}}r^{4-\beta_0+\delta}(XT^2\widehat{\psi})^2+ r^{-\beta_0+\delta}(T^2\widehat{\phi})^2+r^{2-\beta_0+\delta}(r^2\tilde{h}h)^2(T^3\widehat{\phi})^2\,d\sigma dr\\
 &+C\int_{\Sigma_{\tau}}r^{6-\beta_0+\delta}(\square_gT\Phi|_{\Sigma_{\tau}})^2\,d\sigma dr,
 \end{split}
\end{equation*}
and we have that
\begin{equation*}
\int_{\Sigma_{\tau}}r^{6-\beta_0+\delta}(\square_gT \Phi|_{\Sigma_{\tau}})^2\,d\sigma dr \leq C\Phi_0^2(1+\tau)^{-4-2\beta_0+2\delta}.
\end{equation*}
By definition of $n_{\beta}$, we have that $\beta_0-1<2n_{\beta}<\beta_0+1$. If $2 n_{\beta}-\beta_0>0$, we apply Corollary \ref{cor:Linv0} a total of $n_{\beta}$ times to obtain:
\begin{equation*}
\begin{split}
\int_{\Sigma_{\tau}}&r^{-\beta_0+\delta}\widehat{\phi}^2+r^{2-\beta_0+\delta}(X\widehat{\phi})^2\,d\sigma dr\\
\leq &\: C\int_{\Sigma_{\tau}} r^{2n_{\beta}-\beta_0+\delta}(XT^{n_{\beta}}\widehat{\psi})^2+ r^{2n_{\beta}-4-\beta_0+\delta}(T^{n_{\beta}}\widehat{\phi})^2+r^{2n_{\beta}-2-\beta_0+\delta}(r^2\tilde{h}h)^2(T^{1+n_{\beta}}\widehat{\phi})^2\,d\sigma dr\\
&+C\Phi_0^2(1+\tau)^{-4-2\beta+2\delta}\\
\leq &\: C\int_{\Sigma_{\tau}} r^{2n_{\beta}-\beta_0+\delta}(XT^{n_{\beta}}\widehat{\psi})^2\,d\sigma dr+ CE_{\mathbf{N}}[T^{n_{\beta}}\widehat{\phi}](\tau)+C\Phi_0^2(1+\tau)^{-4-2\beta_0+2\delta}.
\end{split}
\end{equation*}
If $2 n_{\beta}-\beta_0<0$, we apply Corollary \ref{cor:Linv0} a total of $n_{\beta}+1$ times to obtain:
\begin{equation*}
\begin{split}
\int_{\Sigma_{\tau}}&r^{-\beta_0+\delta}\widehat{\phi}^2+r^{2-\beta_0+\delta}(X\widehat{\phi})^2\,d\sigma dr\\
\leq &\: C\int_{\Sigma_{\tau}} r^{2n_{\beta}+2-\beta_0+\delta}(XT^{n_{\beta}+1}\widehat{\psi})^2+ r^{2n_{\beta}-2-\beta_0+\delta}(T^{n_{\beta}+1}\widehat{\phi})^2+r^{2n_{\beta}-\beta_0+\delta}(r^2\tilde{h}h)^2(T^{2+n_{\beta}}\widehat{\phi})^2\,d\sigma dr\\
&+C\Phi_0^2(1+\tau)^{-4-2\beta+2\delta}\\
\leq &\: C\int_{\Sigma_{\tau}} r^{2n_{\beta}-\beta_0+2+\delta}(XT^{n_{\beta}+1}\widehat{\psi})^2\,d\sigma dr+ CE_{\mathbf{N}}[T^{n_{\beta}+1}\widehat{\phi}](\tau)+C\Phi_0^2(1+\tau)^{-4-2\beta_0+2\delta}.
\end{split}
\end{equation*}

We can therefore apply the $r$-weighted energy decay estimates in Proposition \ref{prop:weightedecayag0} and  \emph{\ref{item:edecayag01}} of Proposition \ref{prop:edecaytimebg1}, with $p=2n_{\beta}-\beta_0+\delta$ and $N=n_{\beta}$ (if $2n_{\beta}-\beta_0>0$) or $p=2n_{\beta}-\beta_0+2+\delta$ and $N=n_{\beta}+1$ (if $2n_{\beta}-\beta_0<0$).

We then obtain the following estimate for $\phi_0$ with $\beta_0>1$:
\begin{equation*}
\begin{split}
\int_{\Sigma_{\tau}}&r^{-\beta_0+\delta}(T^N\widehat{\phi}_0)^2+r^{2-\beta_0+\delta}(XT^N\widehat{\phi}_0)^2\,d\sigma dr\\
\leq &\: C(1+\tau)^{-3-\beta_0-2n_{\beta}+2\delta-2N}\Bigg[\sum_{\substack{0\leq n_1+n_2+n_3\leq 2(N+2n_{\beta}+2)\\ n_1+n_2\leq N+2n_{\beta}+2}}\int_{\Sigma_{0}\cap \{r\geq R\}} r|\snabla_{\s^2}^{n_1}L(rL)^{n_2}T^{n_3-1-n_{\beta}}\widehat{\psi}_0|^2\,d\sigma dr\\
&+ \sum_{\substack{0\leq n_1+n_2\leq 3(N+2n_{\beta}+2)\\ n_1\leq N+2n_{\beta}+2}}E_{\mathbf{N}}[\snabla_{\s^2}^{n_1}T^{n_2-1-n_{\beta}}\widehat{\phi}_0]+\Phi_0^2\Bigg]+C\Phi_0^2(1+\tau)^{-4-2\beta_0-2N+2\delta}.
\end{split}
\end{equation*}

To estimate the $\phi_{\geq 1}$ part of $\phi$, we consider separately the following three cases: 1) $\beta_0<2n_{\beta}$ and $\beta_1>2n_{\beta}$, 2) $\beta_0< 2n_{\beta}$ and $\beta_1\leq 2n_{\beta}$, 3) $\beta_0\geq 2n_{\beta}$ and $\beta_1>2n_{\beta}$. Since $\widehat{\phi}_{\geq 1}=\phi_{\geq 1}$, we can omit the terms above involving $\Phi_0^2$ when restricting to  $\phi_{\geq 1}$.\\
\\
\emph{ Case 1):}\\
\\
We apply the estimates in \emph{\ref{item:ag0rpeest2}} of Proposition \ref{prop:weightedecayag0} with $p=2n_{\beta}-\beta_0+\delta$, $N$ replaced by $N+n_{\beta}$ and $\delta>0$ appropriately small. Note that this choice of $p$ lies in the allowed range of $p$ in \emph{\ref{item:ag0rpeest2}} of Proposition \ref{prop:weightedecayag0}, since $\beta_0>2n_{\beta}-1$ and $\beta_1>2n_{\beta}-1$. We then obtain:
\begin{equation*}
\begin{split}
\int_{\Sigma_{\tau}}&r^{-\beta_0+\delta}(T^N\widehat{\phi}_{\geq 1})^2+r^{2-\beta_0+\delta}(XT^N\widehat{\phi}_{\geq 1})^2\,d\sigma dr\\
\leq &\: C(1+\tau)^{-3-\beta_0-2n_{\beta}+\max\{2n_{\beta}+1-\beta_1,0\}+2\delta-2N}\\
\cdot & \Bigg[\sum_{\substack{0\leq n_1+n_2+n_3\leq 2(N+2n_{\beta}+1)\\ n_1+n_2\leq N+2n_{\beta}+1}}\int_{\Sigma_{0}\cap \{r\geq R\}} r^{\min\{\beta_1-2n_{\beta},1\}-\delta}|\snabla_{\s^2}^{n_1}L(rL)^{n_2}T^{n_3-1-n_{\beta}}\widehat{\psi}_{\geq 1}|^2\,d\sigma dr\\
&+ \sum_{\substack{0\leq n_1+n_2\leq 3(N+2n_{\beta}+1)\\ n_1\leq 2n_{\beta}+1}}E_{\mathbf{N}}[\snabla_{\s^2}^{n_1}T^{n_2-1-n_{\beta}}\widehat{\phi}_{\geq 1}]\Bigg].
\end{split}
\end{equation*}

\emph{ Case 2):}\\
\\
We apply the estimates in \emph{\ref{item:ag0rpeest3}} of Proposition \ref{prop:weightedecayag0} with $p=2n_{\beta}-\beta_0+\delta$, $N$ replaced by $N+n_{\beta}$ and $\delta>0$ appropriately small. Note that this choice of $p$ lies in the allowed range of $p$ in \emph{\ref{item:ag0rpeest2}} of Proposition \ref{prop:weightedecayag0}, since $\beta_0>2n_{\beta}-1$ and $\beta_1>2n_{\beta}-1$ by assumption. We obtain:
\begin{equation*}
\begin{split}
\int_{\Sigma_{\tau}}&r^{-\beta_0+\delta}(T^N\widehat{\phi}_{\geq 1})^2+r^{2-\beta_0+\delta}(XT^N\widehat{\phi}_{\geq 1})^2\,d\sigma dr\\
\leq &\: C(1+\tau)^{-2-\beta_0-\beta_1+2\delta}\\
\cdot & \Bigg[\sum_{\substack{0\leq n_1+n_2+n_3\leq 2(N+2n_{\beta})\\ n_1+n_2\leq N+2n_{\beta}}}\int_{\Sigma_{0}\cap \{r\geq R\}} r^{2+\beta_1-2n_{\beta}-\delta}|\snabla_{\s^2}^{n_1}L(rL)^{n_2}T^{n_3-n_{\beta}}\widehat{\psi}_{\geq 1}|^2\,d\sigma dr\\
&+ \sum_{\substack{0\leq n_1+n_2\leq 3(N+2n_{\beta})\\ n_1\leq N+2n_{\beta}+1}}E_{\mathbf{N}}[\snabla_{\s^2}^{n_1}T^{n_2-n_{\beta}}\widehat{\phi}_{\geq 1}]\Bigg].
\end{split}
\end{equation*}

\emph{ Case 3):}\\
\\
We apply the estimates in \emph{\ref{item:ag0rpeest2}} of Proposition \ref{prop:weightedecayag0} with $p=2+2n_{\beta}-\beta_0+\delta$, $N$ replaced by $N+n_{\beta}+1$ and $\delta>0$ appropriately small. Note that this choice of $p$ lies in the allowed range of $p$ in \emph{\ref{item:ag0rpeest2}} of Proposition \ref{prop:weightedecayag0}, since $\beta_0>2n_{\beta}$ and $\beta_1>2n_{\beta}$ by assumption. We obtain:
\begin{equation*}
\begin{split}
\int_{\Sigma_{\tau}}&r^{-\beta_0+\delta}(T^N\widehat{\phi}_{\geq 1})^2+r^{2-\beta_0+\delta}(XT^N\widehat{\phi}_{\geq 1})^2\,d\sigma dr\\
\leq &\: C(1+\tau)^{-3-\beta_0-2n_{\beta}+\max\{2n_{\beta}+1-\beta_1,0\}+2\delta-2N}\\
\cdot & \Bigg[\sum_{\substack{0\leq n_1+n_2+n_3\leq 2(N+2n_{\beta}+2)\\ n_1+n_2\leq N+2n_{\beta}+2}}\int_{\Sigma_{0}\cap \{r\geq R\}} r^{\min\{\beta_1-2n_{\beta},1\}-\delta}|\snabla_{\s^2}^{n_1}L(rL)^{n_2}T^{n_3-1-n_{\beta}}\widehat{\psi}_{\geq 1}|^2\,d\sigma dr\\
&+ \sum_{\substack{0\leq n_1+n_2\leq 3(N+2n_{\beta}+2)\\ n_1\leq 2n_{\beta}+1}}E_{\mathbf{N}}[\snabla_{\s^2}^{n_1}T^{n_2-1-n_{\beta}}\widehat{\phi}_{\geq 1}]\Bigg].
\end{split}
\end{equation*}
We conclude that \emph{\ref{item:negredecay2}} must hold.

Consider now $\beta_0<1$ and restrict first to $\phi_0$. By \eqref{eq:keyeqellipticcheck} with $\ell=0$ and a suitably choice of $q$, it follows that for $\delta>0$ arbitrarily small
\begin{equation*}
\begin{split}
\int_{\Sigma_{\tau}} r^{-1+2\beta_0}(X(w_{0}^{-1} \widehat{\phi}_0))^2\,d\sigma dr \leq &\: C\int_{\Sigma_{\tau}} r^{-2+\beta_0}(\mathcal{L}\widehat{\phi}_0)^2\,d\sigma dr\\
\leq &\: C\int_{\Sigma_{\tau}} r^{\beta_0}(XT\widehat{\psi}_0)^2+ r^{-4+\beta_0}(T\widehat{\phi}_0)^2+r^{-2+\beta_0}(r^2\tilde{h}h)^2(T^2\widehat{\phi}_0)^2+r^{2+\beta_0}(\square_g\Phi|_{\Sigma_{\tau}})^2\,d\sigma dr\\
\leq &\: CE_{\mathbf{N}}[T\widehat{\phi}_0](\tau)+C\int_{\Sigma_{\tau} } r^{\beta_0}(XT\widehat{\psi}_0)^2+r^{2}(\square_g\Phi|_{\Sigma_{\tau}})^2\,d\sigma dr.
\end{split}
\end{equation*}
Hence, we can apply Proposition \ref{prop:interpoledecay1} with $p=\beta_0$ and $N$ replaced with $N+1$, together with the above estimate applied to $T^N\phi_0$ to obtain
\begin{equation*}
\begin{split}
\int_{\Sigma_{\tau}}& r^{-1+2\beta_0}(XT^Nw_{0}^{-1}\widehat{\phi}_0)^2\,d\sigma dr\\
 \leq &\: C(1+\tau)^{\max\{1-3\beta_0,0\}+\beta_0-5-2N+\delta}\Bigg[\sum_{\substack{0\leq n_1+n_2+n_3\leq 2(N+2)\\ n_1+n_2\leq N+2}}\int_{\Sigma_{0}\cap \{r\geq R\}}  r^{\min\{3\beta_0,1\}-\delta}|\snabla_{\s^2}^{n_1}L(rL)^{n_2}T^{n_3-1}\widehat{\psi}_0|^2\,d\sigma dr\\
&+ \sum_{\substack{0\leq n_1+n_2\leq 3(N+2)\\ n_1\leq N+2}}E_{\mathbf{N}}[\snabla_{\s^2}^{n_1}T^{n_2-1}\widehat{\phi}]+\Phi_0^2\Bigg]+C|\Phi_0|^2(1+\tau)^{-4-2\beta_0-2N+\delta}.
\end{split}
\end{equation*}
In order to obtain an estimate for $\phi_{\geq 1}$, we instead start as in the $\beta_0>1$ case, using that $\phi_{\geq 1}=\widehat{\phi}_{\geq 1}$, to obtain:
\begin{equation*}
\begin{split}
\int_{\Sigma_{\tau}}r^{-\beta_0+\delta}\widehat{\phi}_{\geq 1}^2+r^{2-\beta_0+\delta}(X\widehat{\phi}_{\geq 1})^2\,d\sigma dr\leq &\: C\int_{\Sigma_{\tau}} r^{2-\beta_0+\delta}(XT \widehat{\psi}_{\geq 1})^2+ CE_{\mathbf{N}}[T\widehat{\phi}_{\geq 1}](\tau).
\end{split}
\end{equation*}
Then we apply Proposition \ref{prop:interpoledecay1} with $p=2-\beta_0+\delta$ and $\phi_{\geq 1}$ replaced by $T^N\phi_{\geq 1}$ to obtain:
\begin{equation*}
\begin{split}
\int_{\Sigma_{\tau}}&r^{-\beta_0+\delta}(T^N\widehat{\phi}_{\geq 1})^2+r^{2-\beta_0+\delta}(XT^N\widehat{\phi}_{\geq 1})^2\,d\sigma dr\\
\leq &\: C(1+\tau)^{-3-\beta_0-2N+2\delta} \Bigg[\sum_{\substack{0\leq n_1+n_2+n_3\leq 2(N+2)\\ n_1+n_2\leq N+2}}\int_{\Sigma_{0}\cap \{r\geq R\}}   r |\snabla_{\s^2}^{n_1}L(rL)^{n_2}T^{n_3-1-n_{\beta}}\widehat{\psi}_{\geq 1}|^2\,d\sigma dr\\
&+ \sum_{\substack{0\leq n_1+n_2\leq 3(N+2)\\ n_1\leq N+2}}E_{\mathbf{N}}[\snabla_{\s^2}^{n_1}T^{n_2-1-n_{\beta}}\widehat{\phi}_{\geq 1}]\Bigg].
\end{split}
\end{equation*}
By combining the above estimates, we conclude \emph{\ref{item:negredecay1}}.
\end{proof}

Since $\mathbf{D}_{\beta_0,N,\delta}^{(1)}[\phi]$ and $\mathbf{D}_{\beta_0,N}^{(n_{\beta}+1)}[\phi]$ are expressed in terms of time integrals of $\phi$ along $\Sigma_0$, it remains to show that they can be bounded in terms of initial data for $\phi$ itself.

\begin{proposition}
\label{prop:initialdatanorms}
There exists a constant $C=C(M,V_{\alpha},h,N,\delta)>0$, such that
\begin{align*}
\mathbf{D}_{\beta_0,N,\delta}^{(1)}[\phi]\leq &\:C \mathbf{D}_{\beta_0,N,\delta}[\phi]\quad (\beta_0<1),\\
\mathbf{D}_{\beta_0,N,\delta}^{(n_{\beta}+1)}[\phi]\leq &\:C \mathbf{D}_{\beta_0,N,\delta}[\phi]\quad (\beta_0>1),
\end{align*}
with
\begin{align*}
\mathbf{D}_{\beta_0,N,\delta}[\phi]:= &\:\sum_{n\leq N+2}\left(||((r+1)^2(rX)^nX(r\phi_0)||_{L^{\infty}(\Sigma_0)}^2+||(rX)^n\phi_0||_{L^{\infty}(\Sigma_0)}^2+||r(rX)^nT\phi_0||_{L^{\infty}(\Sigma_0)}^2\right)\\
&+\sum_{n_1+n_2\leq N+2}\int_{\Sigma_0\cap\{r\geq R\}} r^{3-\delta} |\snabla_{\s^2}^{n_1}X(rX)^{n_2} \widehat{\psi}_{\geq 1}|^2+r^{-1-\delta} |\snabla_{\s^2}^{n_1}(rX)^{n_2} T\widehat{\psi}_{\geq 1}|^2\,d\sigma dr\\
&+ \sum_{\substack{0\leq n_1+n_2+n_3\leq 2N+3\\ n_1+n_2\leq N+2}}\int_{\Sigma_{0}\cap \{r\geq R\}}  r^{\min\{3\beta_0,1\}-\delta}|L(rL)^{n_2}T^{n_3}\widehat{\psi}_0|^2+r^{1-\delta}|\snabla_{\s^2}^{n_1}L(rL)^{n_2}T^{n_3-1}{\psi}_{\geq 1}|^2\,d\sigma dr\\
&+ \sum_{\substack{0\leq n_1+n_2\leq 3N+5\\ n_1\leq N+2}}E_{\mathbf{N}}[\snabla_{\s^2}^{n_1}T^{n_2}\widehat{\phi}]+\mathfrak{I}_{\beta}^2[\phi]\qquad (\beta_0<1),\\
\mathbf{D}_{\beta_0,N,\delta}[\phi]:=&\: \sum_{n\leq N+3n_{\beta}+2}\Biggl( ||(r+1)^{\frac{3}{2}+\frac{\beta_0}{2}}(rX)^nX(r\phi_0)||_{L^{\infty}(\Sigma_0)}^2\\
&+M^4||r^{\frac{1}{2}+\frac{\beta_0}{2}}(rX)^n\phi_0||_{L^{\infty}(\Sigma_0)}^2+M^4||r^{\frac{3}{2}+\frac{\beta_0}{2}}(rX)^nT\phi_0||_{L^{\infty}(\Sigma_0)}^2\Biggr)\\
&+\sum_{n_1+n_2\leq N+2n_{\beta}+2}\int_{\Sigma_{0}\cap \{r\geq R\}} r^{2+\min\{2n_{\beta}+1,\beta_1\}-\delta}|\snabla_{\s^2}^{n_1}X(rX)^{n_2}\widehat{\psi}_{\geq 1}|^2\\
 &+M^4r^{-2+\min\{2n_{\beta}+1,\beta_1\}-\delta}|\snabla_{\s^2}^{n_1}X(rX)^{n_2}T\widehat{\psi}_{\geq 1}|^2\,d\sigma dr\\
&+\sum_{\substack{0\leq n_1+n_2+n_3\leq 2N+3n_{\beta}+5\\ n_1+n_2\leq N+2n_{\beta}+2}}\int_{\Sigma_{0}\cap \{r\geq R\}} r|\snabla_{\s^2}^{n_1}L(rL)^{n_2}T^{n_3}\widehat{\psi}_0|^2\\
&+ r^{\min\{\beta_1-2n_{\beta},1\}-\delta}|\snabla_{\s^2}^{n_1}L(rL)^{n_2}T^{n_3}{\psi}_{\geq 1}|^2+r^{2+\min\{\beta_1-2n_{\beta},0\}-\delta}|\snabla_{\s^2}^{n_1}L(rL)^{n_2}T^{n_3}{\psi}_{\geq 1}|^2\,d\sigma dr\\
&+ \sum_{\substack{0\leq n_1+n_2\leq 3N+5n_{\beta}+5\\ n_1\leq N+2n_{\beta}+2}}E_{\mathbf{N}}[\snabla_{\s^2}^{n_1}T^{n_2}\widehat{\phi}]+\mathfrak{I}_{\beta}^2[\phi]\qquad (\beta_0>1).
\end{align*}
\end{proposition}
\begin{proof}
We first consider the case $\beta_0<1$ and the norm $\mathbf{D}_{\beta_0,N,\delta}^{(1)}[\phi]$. Note that all the terms in $\mathbf{D}_{\beta_0,N,\delta}^{(1)}[\phi]$ with $n_3\geq 1$ can directly be expressed in terms of $\phi$ along $\Sigma_0$. It therefore remains to estimate the terms with $n_3=0$. Furthermore, without loss of generality, we can replace $L$ derivatives with $X$ derivatives in $\mathbf{D}_{\beta_0,N,\delta}^{(1)}[\phi]$, using the relation \eqref{eq:LXT}.

By \emph{\ref{item:diffdataest1}} of Corollary \ref{cor:estboxPsi}, we can estimate
\begin{equation*}
\begin{split}
E_{\mathbf{N}}[T^{-1}\widehat{\phi}_0]+&\sum_{n\leq N+2}\int_{\Sigma_0\cap\{r\geq R\}} r^{\min\{3\beta_0,1\}-\delta} |X(rX)^n T^{-1}\widehat{\psi}_0|^2\,d\sigma dr\\\
\leq &\: C\sum_{n\leq N+2}\left(||((r+1)^2(rX)^nX(r\phi_0)||_{L^{\infty}(\Sigma_0)}^2+||(rX)^n\phi_0||_{L^{\infty}(\Sigma_0)}^2+||r(rX)^nT\phi_0||_{L^{\infty}(\Sigma_0)}^2\right).
\end{split}
\end{equation*}
Furthermore, we can apply Proposition \ref{prop:timeint} to estimate:
\begin{equation*}
\begin{split}
\sum_{n_1\leq N+2}E_{\mathbf{N}}[\snabla_{\s^2}^{n_1}T^{-1}\widehat{\phi}_{\geq 1}]+&\sum_{n_1+n_2\leq N+2}\int_{\Sigma_0\cap\{r\geq R\}} r^{1-\delta} |\snabla_{\s^2}^{n_1}X(rX)^{n_2} T^{-1}\widehat{\psi}_{\geq 1}|^2\,d\sigma dr\\\
\leq &\: C E_{\mathbf{N}}[\widehat{\phi}_{\geq 1}]+C E_{\mathbf{N}}[T\widehat{\phi}_{\geq 1}]\\
&+\sum_{n_1+n_2\leq N+2}\int_{\Sigma_0\cap\{r\geq R\}} r^{3-\delta} |\snabla_{\s^2}^{n_1}X(rX)^{n_2} \widehat{\psi}_{\geq 1}|^2+r^{-1-\delta} |\snabla_{\s^2}^{n_1}(rX)^{n_2} T\widehat{\psi}_{\geq 1}|^2\,d\sigma dr.
\end{split}
\end{equation*}

Now we consider the case $\beta_0>1$ and the norm $\mathbf{D}_{\beta_0,N,\delta}^{(n_{\beta}+1)}[\phi]$. Note that all the terms in $\mathbf{D}_{\beta_0,N,\delta}^{(n_{\beta}+1)}[\phi]$ with $n_3\geq 1+n_{\beta}$ can directly be expressed in terms of $\phi$ along $\Sigma_0$. It therefore remains to estimate the terms with $n_3\leq n_{\beta}$. Furthermore, without loss of generality, we can replace $L$ derivatives with $X$ derivatives in $\mathbf{D}_{\beta_0,N,\delta}^{(n_{\beta}+1)}[\phi]$, using the relation \eqref{eq:LXT}.

By \emph{\ref{item:diffdataest2}} of Corollary \ref{cor:estboxPsi}, we can estimate
\begin{equation*}
\begin{split}
 E_{\mathbf{N}}[T^{-1-n_{\beta}}\widehat{\phi}_0]+&\sum_{\substack{n_2\leq N+2n_{\beta}+2\\ n_3\leq n_{\beta}}}\int_{\Sigma_{0}\cap \{r\geq R\}} r|X(rX)^{n_2}T^{n_3-1-n_{\beta}}\widehat{\psi}_0|^2\,d\sigma dr\\
 \leq &\: C\sum_{n\leq N+3n_{\beta}+2}\Biggl( ||(r+1)^{\frac{3}{2}+\frac{\beta_0}{2}}(rX)^nX(r\phi_0)||_{L^{\infty}(\Sigma_0)}\\
&+||r^{\frac{1}{2}+\frac{\beta_0}{2}}(rX)^n\phi_0||_{L^{\infty}(\Sigma_0)}+||r^{\frac{3}{2}+\frac{\beta_0}{2}}(rX)^nT\phi_0||_{L^{\infty}(\Sigma_0)}\Biggr).
 \end{split}
 \end{equation*}
 In order to estimate the terms involving $\phi_{\geq 1}$, it is convenient to split the cases $\beta_1>2n_{\beta}$ and $\beta_0\leq 2n_{\beta}$. Suppose $\beta_1>2n_{\beta}$. Then we can control via Corollary \ref{cor:regmultipletimeinv}:
 \begin{equation*}
\begin{split}
 \sum_{n_1+n_2\leq N+2n_{\beta}+2}&E_{\mathbf{N}}[\snabla_{\s^2}^{n_1}T^{-1-n_{\beta}}\widehat{\phi}_{\geq 1}]+\sum_{\substack{n_1+n_2\leq N+2n_{\beta}+2\\ n_3\leq n_{\beta}}}\int_{\Sigma_{0}\cap \{r\geq R\}} r^{\min\{\beta_1-2n_{\beta},1\}-\delta}|\snabla_{\s^2}^{n_1}X(rX)^{n_2}T^{n_3-1-n_{\beta}}{\psi}_{\geq 1}|^2\\
 &+r^{2-\delta}|\snabla_{\s^2}^{n_1}X(rX)^{n_2}T^{n_3-n_{\beta}}{\psi}_{\geq 1}|^2\,d\sigma dr\\
 \leq &\: CE_{\mathbf{N}}[\widehat{\phi}_{\geq 1}]+CM E_{\mathbf{N}}[T\widehat{\phi}_{\geq 1}]+C\sum_{n_1+n_2\leq N+2n_{\beta}+2}\int_{\Sigma_{0}\cap \{r\geq R\}} r^{2+\min\{2n_{\beta}+1,\beta_1\}-\delta}|\snabla_{\s^2}^{n_1}X(rX)^{n_2}\widehat{\psi}_{\geq 1}|^2\\
 &+r^{-2+\min\{2n_{\beta}+1,\beta_1\}-\delta}|\snabla_{\s^2}^{n_1}X(rX)^{n_2}T\widehat{\psi}_{\geq 1}|^2\,d\sigma dr.
 \end{split}
 \end{equation*}
 Suppose $\beta_1\leq 2n_{\beta}$. Then we can similarly control via Corollary \ref{cor:regmultipletimeinv}:
 \begin{equation*}
\begin{split}
 E_{\mathbf{N}}[T^{-1-n_{\beta}}\widehat{\phi}_{\geq 1}]+&\sum_{\substack{n_1+n_2\leq N+2n_{\beta}+2\\ n_3\leq n_{\beta}}}\int_{\Sigma_{0}\cap \{r\geq R\}} r^{\beta_1-2n_{\beta}-\delta}|\snabla_{\s^2}^{n_1}X(rX)^{n_2}T^{n_3-1-n_{\beta}}{\psi}_{\geq 1}|^2\\
 &+r^{2+\beta_1-2n_{\beta}-\delta}|\snabla_{\s^2}^{n_1}X(rX)^{n_2}T^{n_3-n_{\beta}}{\psi}_{\geq 1}|^2\,d\sigma dr\\
 \leq &\: CE_{\mathbf{N}}[\widehat{\phi}_{\geq 1}]+CM E_{\mathbf{N}}[T\widehat{\phi}_{\geq 1}]+C\sum_{n_1+n_2\leq N+2n_{\beta}+2}\int_{\Sigma_{0}\cap \{r\geq R\}} r^{2+\beta_1-\delta}|\snabla_{\s^2}^{n_1}X(rX)^{n_2}\widehat{\psi}_{\geq 1}|^2\\
 &+r^{-2+\beta_1-\delta}|\snabla_{\s^2}^{n_1}X(rX)^{n_2}T\widehat{\psi}_{\geq 1}|^2\,d\sigma dr.
 \end{split}
 \end{equation*}
We conclude the estimates in the proposition by combining the above.
\end{proof}

\subsection{$L^{\infty}$ estimates}
\label{sec:pointdecay}
In this section, we use the ($r$-weighted) energy decay estimates of Sections \ref{sec:edecay} and \ref{sec:addedecay} to obtain $L^{\infty}$ estimates along each $\Sigma_{\tau}$ for $\widehat{\phi}$. We obtain pointwise estimates from energy estimates by applying the fundamental theorem of calculus in $r$ together with Sobolev estimates on $\s^2$.
\begin{proposition}
\label{prop:mainlinftyest}
Let $\beta_0\notin 2\N_0+1$ and $\delta>0$ arbitrarily small. Then there exists a constant $C=C(M,h,V_{\alpha},N,\delta)>0$ such that for all $0\leq k\leq n_{\beta}+1$:
\begin{align}
\label{eq:asympbl1}
&||T^{N-k}\phi- \Phi_0 w_0T^{N-k}((1+\tau+2r)^{-\frac{1}{2}-\frac{1}{2}\beta_0}(1+\tau)^{-\frac{1}{2}-\frac{1}{2}\beta_0})||_{L^{\infty}(S^2_{\tau,r})}\\ \nonumber
\leq&\: Cr^{-\frac{1}{2}+\frac{1}{2}\beta_0}(1+\tau+2r)^{-\frac{1}{2}-\frac{1}{2}\beta_0}(1+\tau)^{-1+\frac{1}{2}\max\{1-3\beta_0,0\}+k-N+\delta}\sqrt{\sum_{0\leq k\leq 2}\mathbf{D}_{\beta_0,N,\delta}[\snabla_{\s^2}^k\phi]},\quad (\beta_0<1)\\
\label{eq:asympbg1}
&||T^{N-k}\phi- \Phi_0 w_0T^{N-k}((1+\tau+2r)^{-\frac{1}{2}-\frac{1}{2}\beta_0}(1+\tau)^{-\frac{1}{2}-\frac{1}{2}\beta_0})||_{L^{\infty}(S^2_{\tau,r})}\\ \nonumber
\leq&\: Cr^{-\frac{1}{2}+\frac{1}{2}\beta_0}(1+\tau+2r)^{-\frac{1}{2}-\frac{1}{2}\beta_0}(1+\tau)^{-\frac{1}{2}-\frac{1}{2}\beta_0-\frac{1}{2}\min\{1-\beta_0+2n_{\beta},\beta_1-\beta_0\}-N+\delta}\sqrt{\sum_{0\leq k\leq 2}\mathbf{D}_{\beta_0,N,\delta}[\snabla_{\s^2}^k\phi]}. \quad (\beta_0>1)
\end{align}
\end{proposition}
\begin{proof} 
We consider first $\widehat{\phi}$. By integrating $X\widehat{\phi}$, using that $\phi(\tau,r,\theta,\varphi)\to0$ and $\Phi(\tau,r,\theta,\varphi)\to 0$ as $r\to \infty$, we can estimate for all $r_0\geq 2M$.
\begin{equation*}
\begin{split}
\int_{\Sigma_{\tau}\cap \{r=r_0\}}\widehat{\phi}^2\,d\sigma=&\:\left(\int_{\Sigma_{\tau}\cap \{r\geq r_0\}} X\widehat{\phi}\,d\sigma dr \right)^2\\
\leq &\: \int_{r_0}^{\infty} r^{-2}\,dr\cdot \int_{\Sigma_{\tau}\cap \{r\geq r_0\}} r^2 (X\phi)^2\,d\sigma dr\\
\leq  &\: r^{-\frac{1}{2}} E_{\mathbf{N}}[\phi](\tau).
\end{split}
\end{equation*}
We now integrate $X(\widehat{\psi}^2)$ in $R\leq r\leq r_0$ to obtain additionally:
\begin{equation*}
\begin{split}
\int_{\Sigma_{\tau}\cap \{r=r_0\}}\widehat{\psi}^2\,d\sigma=&\:\int_{\Sigma_{\tau}\cap \{R\leq r\geq r_0\}} X(\widehat{\psi}^2)\,d\sigma dr +\int_{\Sigma_{\tau}\cap \{r=R\}}\widehat{\psi}^2\,d\sigma\\
\leq &\: C\int_{\Sigma_{\tau}\cap \{R\leq r\geq r_0\}}  r^{-1+\delta}\widehat{\psi}^2+r^{1+\delta}(X\widehat{\psi})^2\,d\sigma dr+ CE_{\mathbf{N}}[\phi](\tau)\\
\leq &\: C\int_{\Sigma_{\tau}\cap \{R\leq r\geq r_0\}} r^{1+\delta}(L\widehat{\psi})^2\,d\sigma dr+ CE_{\mathbf{N}}[\phi](\tau),
\end{split}
\end{equation*}
where we applied \eqref{eq:hardy} to absorb the $r^{-1+\delta}\widehat{\psi}^2$ term in the last step.

Hence, we apply the energy estimates in Propositions \ref{prop:edecaytimebg1}, \ref{prop:interpoledecay1} and \ref{prop:weightedecayag0} together with Proposition \ref{prop:initialdatanorms} to obtain for all $r_0\geq 2M$:
\begin{align}
\label{eq:pointwisephi1}
\int_{\Sigma_{\tau}\cap \{r=r_0\}}\widehat{\psi}^2\,d\sigma\leq&\: C(1+\tau)^{-2+\max\{1-3\beta_0,0\}+2\delta}\mathbf{D}_{\beta_0,0,\delta}[\phi]\qquad (\beta_0<1),\\
\label{eq:pointwisephi2}
\int_{\Sigma_{\tau}\cap \{r=r_0\}}\widehat{\psi}^2\,d\sigma\leq&\: C(1+\tau)^{-2-2n_{\beta}+\max\{2n_{\beta}+1-\beta_1,0\}+2\delta}\mathbf{D}_{\beta_0,0,\delta}[\phi]\qquad (\beta_0>1).
\end{align}
Hence, by \eqref{eq:sphere5},
\begin{align*}
||\widehat{\psi}||_{L^{\infty}(\Sigma_{\tau})}\leq&\: C(1+\tau)^{-1+\frac{1}{2}\max\{1-3\beta_0,0\}+\delta}\sqrt{\sum_{0\leq k\leq 2}\mathbf{D}_{\beta_0,0,\delta}[\snabla_{\s^2}^k\phi]}\qquad (\beta_0<1),\\
||\widehat{\psi}||_{L^{\infty}(\Sigma_{\tau})}\leq&\: C(1+\tau)^{-1-n_{\beta}+\frac{1}{2}\max\{2n_{\beta}+1-\beta_1,0\}+\delta}\sqrt{\sum_{0\leq k\leq 2}\mathbf{D}_{\beta_0,0,\delta}[\snabla_{\s^2}^k\phi]}\qquad (\beta_0>1).
\end{align*}

When $\beta_0>1$, we integrate $X(r^{1-\beta_0+\delta}\widehat{\phi}^2)$ from $r=\infty$ and apply Proposition \ref{prop:negredecay} to obtain moreover:
\begin{equation*}
\begin{split}
\int_{\Sigma_{\tau}\cap \{r=r_0\}}r^{1-\beta_0+\delta}\widehat{\phi}^2\,d\sigma\leq &\: C\int_{\Sigma_{\tau}\cap \{r\geq r_0\}}  r^{-\beta_0+\delta}\widehat{\phi}^2+r^{2-\beta_0}(X\widehat{\phi})^2\,d\sigma dr\\
\leq &\: C(1+\tau)^{-3-\beta_0-2n_{\beta}+\max\{2n_{\beta}+1-\beta_1,0\}-2N+2\delta}\mathbf{D}_{\beta_0,0,\delta}[\phi].
\end{split}
\end{equation*}
Observe moreover that
\begin{equation*}
\int_{\Sigma_{\tau}\cap \{r=r_0\}} r^{\delta}(w_0^{-1}\widehat{\phi})^2\,d\sigma\leq \int_{\Sigma_{\tau}\cap \{r=r_0\}} r^{1-\beta_0+\delta}\widehat{\phi}^2\,d\sigma,
\end{equation*}
so
\begin{equation*}
||w_0^{-1}\widehat{\phi}||_{L^{\infty}(\Sigma_{\tau})}\leq C(1+\tau)^{-\frac{3}{2}-\frac{1}{2}\beta_0-n_{\beta}+\frac{1}{2}\max\{2n_{\beta}+1-\beta_1,0\}+\delta}\sqrt{\sum_{0\leq k\leq 2}\mathbf{D}_{\beta_0,0}[\snabla_{\s^2}^k\phi]}.
\end{equation*}
We similarly estimate when $\beta_0<1$, using Proposition \ref{prop:negredecay}:
\begin{equation*}
\begin{split}
\int_{\Sigma_{\tau}\cap \{r=r_0\}}r^{1-\beta_0+\delta}\widehat{\phi}_{\geq 1}^2\,d\sigma\leq &\: C\int_{\Sigma_{\tau}\cap \{r\geq r_0\}}  r^{-\beta_0+\delta}\widehat{\phi}_{\geq 1}^2+r^{2-\beta_0}(X\widehat{\phi}_{\geq 1})^2\,d\sigma dr\\
\leq &\: C(1+\tau)^{-3-\beta_0+2\delta}\mathbf{D}_{\beta_0,0}[\phi]
\end{split}
\end{equation*}
so
\begin{equation*}
||w_0^{-1}\widehat{\phi}_{\geq 1}||_{L^{\infty}(\Sigma_{\tau})}\leq C(1+\tau)^{-\frac{3}{2}-\frac{1}{2}\beta_0+\delta}\sqrt{\sum_{0\leq k\leq 2}\mathbf{D}_{\beta_0,0}[\snabla_{\s^2}^k\phi]}.
\end{equation*}

When $\beta_0<1$, we moreover integrate $X(w_0^{-1}\widehat{\phi}_0)$ from $r=r_0$ to $r=r_0+\tau$ to estimate:
\begin{equation*}
\begin{split}
\int_{\Sigma_{\tau}\cap \{r=r_0\}}w_0^{-2}\widehat{\phi}_0^2\,d\sigma\leq&\: C\left(\int_{\Sigma_{\tau}\cap \{r_0\leq r\leq r_0+\tau\}}X(w_0^{-1}\widehat{\phi}_0)\,d\sigma dr\right)^2+C\int_{\Sigma_{\tau}\cap \{r=r_0+\tau\}}w_0^{-2}\widehat{\phi}_0^2\,d\sigma\\
\leq&\: C\int_{r_0}^{r_0+\tau} r^{1-2\beta_0}\,dr\int_{\Sigma_{\tau}\cap \{r_0\leq r\leq r_0+\tau\}}r^{-1+2\beta_0}(X(w_0^{-1}\widehat{\phi}_0))^2\,d\sigma dr\\
&+C(r_0+\tau)^{-\beta_0}\int_{\Sigma_{\tau}\cap \{r=r_0+\tau\}}\widehat{\psi}_0^2\,d\sigma\\
\leq &\: C\tau^{2-2\beta_0}\int_{\Sigma_{\tau}\cap \{r_0\leq r\leq r_0+\tau\}}r^{-1+2\beta_0}(X(w_0^{-1}\widehat{\phi}_0))^2\,d\sigma dr+C(r_0+\tau)^{-\beta_0}\int_{\Sigma_{\tau}\cap \{r=r_0+\tau\}}\widehat{\psi}_0^2\,d\sigma.
\end{split}
\end{equation*}
By Proposition \ref{prop:negredecay} together with the above $L^{\infty}$ estimates for $\widehat{\psi}$ it then follows that:
\begin{equation*}
\int_{\Sigma_{\tau}\cap \{r=r_0\}}w_0^{-2}\widehat{\phi}_0^2\,d\sigma\leq C(1+\tau)^{\max\{1-3\beta_0,0\}-\beta_0-3+\delta}\mathbf{D}_{\beta_0,0}[\phi],
\end{equation*}
so
\begin{equation*}
||w_0^{-1}\widehat{\phi}_{0}||_{L^{\infty}(\Sigma_{\tau}\cap\{r\leq R+\tau\})}\leq C(1+\tau)^{\frac{1}{2}\max\{1-3\beta_0,0\}-\frac{1}{2}\beta_0-\frac{3}{2}+\frac{1}{2}\delta}\sqrt{\mathbf{D}_{\beta_0,0}[\phi]}.
\end{equation*}

To conclude that \eqref{eq:asympbl1} holds, we combine the $L^{\infty}$ estimates, replace$\phi$ with $T^{N-k}\phi$, modifying the corresponding energy decay estimates appropriately with an additional factor $(1+\tau)^{-2N}$ and observe that we can write:
\begin{align*}
-1+\frac{1}{2}\max\{1-3\beta_0,0\}=&\:-\frac{1}{2}-\frac{1}{2}\beta_0-\frac{1}{2}\min\{2\beta_0, 1-\beta_0+n_{\beta_0},\beta_1-\beta_0\}\quad (\beta_0<1)\\
-\frac{1}{2}-\frac{1}{2}\beta_0-\frac{1}{2}\min\{1-\beta_0+2n_{\beta},\beta_1-\beta_0\}-=&\: -\frac{1}{2}-\frac{1}{2}\beta_0-\frac{1}{2}\min\{2\beta_0, 1-\beta_0+n_{\beta_0},\beta_1-\beta_0\}\quad (\beta_0>1).
\end{align*}
\end{proof}

\begin{corollary}
\label{cor:equivdeftimeint}
The integrals
\begin{equation*}
\begin{split}
\phi^{(1+n)}(\tau,r,\theta,\varphi)=&(-1)^{1+n}\int_{\tau}^{\infty}\int_{\tau_1}^{\infty}\ldots \int_{\tau_{n}}^{\infty}\phi(s,r,\theta,\varphi)\,ds d\tau_{n}\ldots d\tau_1
\end{split}
\end{equation*}
are well-defined for $n\in \N_0$ with $0\leq n\leq n_{\beta}$ and satisfy:
\begin{align}
\label{eq:timeinid1}
\phi^{(1+n)}(\tau,r,\theta,\varphi) =&\:\frac{1}{n!}\int_{\tau}^{\infty} (\tau-s)^{n}\phi(s,r,\theta,\varphi)\,ds,\\
\label{eq:timeinid2}
\phi^{(1+n)}(\tau,r,\theta,\varphi) =&\: T^{-n}\phi(\tau,r,\theta,\varphi).
\end{align}
\end{corollary}
\begin{proof}
The integrals $\phi^{(1+n)}(\tau,r,\theta,\varphi)$ are well-defined for $0\leq n\leq n_{\beta}$ by the decay-in-time properties in Proposition \ref{prop:mainlinftyest}. To arrive at \eqref{eq:timeinid1}, we integrate by parts multiple times, making use of the decay-in-time of $\phi$ from Proposition \ref{prop:mainlinftyest}.

We moreover have that $T^k \phi^{(1+n_{\beta})}(\tau,r,\theta,\varphi)= \phi^{(1+n_{\beta}-k)}(\tau,r,\theta,\varphi)$ for $0\leq k\leq 1+n_{\beta}$. In particular, 
\begin{equation*}
T^{1+n_{\beta}}\phi^{(1+n_{\beta})}=\phi. 
\end{equation*}
By Proposition \ref{prop:timeint} and Corollary \ref{cor:regmultipletimeinv}, we also have that
\begin{equation*}
T^{1+n_{\beta}}(T^{-1-n_{\beta}}\phi)=\phi. 
\end{equation*}
Hence, we must have that there exists functions  $P_j(r,\theta,\varphi)$, with $0\leq j\leq n_{\beta}$ such that
\begin{equation*}
T^{-1-n_{\beta}}\phi=\phi^{(1+n_{\beta})}+\sum_{j=0}^{n_{\beta}} \tau^j P_j(r,\theta,\varphi).
\end{equation*}

By construction,  $\phi^{(1+n_{\beta})}(\tau,r,\theta,\varphi) \to 0$ as $\tau \to \infty$. From Proposition \ref{prop:mainlinftyest}, it follows moreover that $T^{-1-n_{\beta}}\phi(\tau,r,\theta,\varphi)\to 0$ as $\tau \to \infty$. Hence, $P_j\equiv 0$ for all $0\leq j\leq n_{\beta}$ and \eqref{eq:timeinid2} holds.
\end{proof}

\section{Late-time asymptotics for $\phi_{\geq \ell}$}
\label{sec:latetimeasymell}
The results in \S \ref{sec:timinv} and \S \ref{sec:latetimeasymp} (for $\beta_0>1$, or equivalently, $\alpha>0$) can immediately be generalized to obtain late-time asymptotics and decay estimates in time for $\phi_{\geq \ell}$. This follows from the identities:
\begin{align*}
(-\slashed{\Delta}_{\s^2}+V_{\alpha})\phi_{\ell}=&\:(\ell(\ell+1)+V_{\alpha})\phi_{\ell}=V_{\alpha+\ell(\ell+1)}\phi_{\ell},\quad \textnormal{with}\\
V_{\alpha+\ell(\ell+1)}=&\:V_{\alpha}+\ell(\ell+1)r^{-2},\\
\beta_{\ell ;\alpha}:=&\:\sqrt{1+4\alpha+4\ell (\ell+1)}=\beta_{0;\alpha+\ell(\ell+1)},\\
\beta_{\ell+1; \alpha}\leq &\: \beta_{1;\alpha+\ell(\ell+1)}.
\end{align*}
The above identities and inequality allow us to replace $w_0$ with $w_{\ell}$, $\beta_0$ with $\beta_{\ell}$ and let $\beta_{\ell+1}$ take on the role of $\beta_1$ in the main estimates of \S \ref{sec:timinv} and \S \ref{sec:latetimeasymp}. In particular, this means we need to define $n_{\beta_{\ell}}=\lfloor \frac{\beta_{\ell}+1}{2}\rfloor$ and:
\begin{equation*}
\mathfrak{I}_{\beta_{\ell}}[\phi]:=\lim_{r \to \infty} r^{\frac{1}{2}+(\frac{1}{2}\beta_{\ell}-n_{\beta_{\ell}})}X(rT^{-1-n_{\beta_{\ell}}}\phi_{\ell})=-\frac{(1-\beta_{\ell})(3-\beta_{\ell})\ldots (2n_{\beta_{\ell}}+1-\beta_{\ell})}{2n_{\beta_{\ell}}!\beta_{\ell}(1-\beta_{\ell})\ldots (n_{\beta_{\ell}}-\beta_{\ell})}\int_{2M}^{\infty}w_{\ell}F_{\ell}[\phi_{\ell}]\,dr'
\end{equation*}
and, furthermore,
\begin{align*}
\widehat{\phi}_{\geq \ell}=&\:\phi_{\geq \ell}-\Phi^{(\beta_{\ell})},\\
\Phi^{(\beta_{\ell})}(\tau,r,\theta,\varphi)=&\: \Phi_{\ell}(\theta,\varphi)w_{\ell} (r)(\tau+1+2r)^{-\frac{1}{2}-\frac{\beta_{\ell}}{2}}(\tau+1)^{-\frac{1}{2}-\frac{\beta_{\ell}}{2}},\\
\Phi_{\ell}(\theta,\varphi)=&\: \sum_{|m|\leq \ell}\Phi_{\ell m} Y_{\ell m}(\theta,\varphi),\\
 \Phi_{\ell m}\in &\:\R.
\end{align*}
Without loss of generality, we can fix $-\ell\leq m\leq \ell$ and assume $\Phi_{\ell m}\in \R\setminus\{0\}$. Furthermore, we choose
\begin{equation*}
\Phi_{\ell}=(-1)^{n_{\beta_{\ell}}}n_{\beta_{\ell}}!\frac{\Gamma\left(\frac{\beta_{\ell}+1}{2}\right)\Gamma\left(\frac{\beta_{\ell}-1-2n_{\beta_{\ell}}}{2}\right)\cos \left(\frac{\beta_{\ell}-2n_{\beta_{\ell}}}{2\pi}\right)}{2^{n_{\beta_{\ell}}-\beta_{\ell}}\pi \Gamma(\beta-n_{\beta_{\ell}})} \mathfrak{I}_{\beta_{\ell}}[\phi].
\end{equation*}

We then conclude by repeating the proof of Proposition \ref{prop:mainlinftyest} (in the case $\beta_0>1$) with the above considerations in mind:
\begin{proposition}
\label{prop:mainlinftyestell}
Let $\beta_{\ell}\neq 2l+1$ for all $l\in \N_0$, and let $\delta>0$ be arbitrarily small. Then there exists a constant $C=C(M,h,V_{\alpha},N,\delta)>0$ such that:
\begin{align}
\label{eq:asympbl1genl}
&||T^N\phi- \Phi_{\ell} w_{\ell}T^N((1+\tau+2r)^{-\frac{1}{2}-\frac{1}{2}\beta_{\ell}}(1+\tau)^{-\frac{1}{2}-\frac{1}{2}\beta_{\ell}})||_{L^{\infty}(S^2_{\tau,r})}\\ \nonumber
\leq&\: Cr^{-\frac{1}{2}+\frac{1}{2}\beta_{\ell}}(1+\tau+2r)^{-\frac{1}{2}-\frac{1}{2}\beta_{\ell}}(1+\tau)^{-\frac{1}{2}-\frac{1}{2}\beta_{\ell}-\frac{1}{2}\min\{1-\beta_{\ell}+2n_{\beta_{\ell}},\beta_{\ell+1}-\beta_{\ell}\}-N+\delta}\sqrt{\sum_{0\leq k\leq 2}\mathbf{D}_{\beta_{\ell},N,\delta}[\snabla_{\s^2}^k\phi_{\geq \ell}]},
\end{align}
with $\mathbf{D}_{\beta_{\ell},N,\delta}$ defined in \S \ref{sec:precise}.
\end{proposition}

Furthermore, we can repeat the proof of Corollary \ref{cor:equivdeftimeint} to obtain:
\begin{corollary}
\label{cor:equivdeftimeintell}
The integrals
\begin{equation*}
\begin{split}
\phi_{\geq \ell}^{(1+n)}(\tau,r,\theta,\varphi)=&(-1)^{1+n}\int_{\tau}^{\infty}\int_{\tau_1}^{\infty}\ldots \int_{\tau_{n}}^{\infty}\phi_{\geq \ell}(s,r,\theta,\varphi)\,ds d\tau_{n}\ldots d\tau_1
\end{split}
\end{equation*}
are well-defined for $n\in \N_0$ with $0\leq n\leq n_{\beta_{\ell}}$ and satisfy:
\begin{align}
\label{eq:timeinid1ell}
\phi_{\geq \ell}^{(1+n)}(\tau,r,\theta,\varphi) =&\:\frac{1}{n!}\int_{\tau}^{\infty} (\tau-s)^{n}\phi_{\geq \ell}(s,r,\theta,\varphi)\,ds,\\
\label{eq:timeinid2ell}
\phi_{\geq \ell}^{(1+n)}(\tau,r,\theta,\varphi) =&\: T^{-n}\phi_{\geq \ell}(\tau,r,\theta,\varphi).
\end{align}
\end{corollary}

\appendix
\section{Properties of the weight functions $w_{\ell}$}
\label{app:weightfunction}
We define the operator $\mathcal{L}_{w_{\ell}}$ as follows:
\begin{equation*}
\mathcal{L}_{w_{\ell}}u=\frac{d}{dr}\left(Dr^2 \frac{du}{dr}\right)-[r^2V_{\alpha}+\ell(\ell+1)]u,
\end{equation*}
with $V_{\alpha}$ satisfying Assumption \eqref{assm:A}.

Let $R>2M$ be arbitrarily large and let $H_{w_{\ell},p}^{N}$, $H_{w_{\ell},p;1}^{N}$  denote the completions of the function spaces:
\begin{equation*}
\{f\in C_0^{\infty}([2M,\infty)\,|\,f(R)=0\}
\end{equation*}
under the norm
\begin{align}
\label{eq:wenest}
||u||^2_{w_{\ell},N,p}:=&\:\sum_{k=0}^N\int_{2M}^{\infty}r^p\left(\left(r\frac{d}{dr}\right)^ku\right)^2\,dr,\\
\label{eq:wenest1}
||u||^2_{w_{\ell},N,p;1}:=&\:\sum_{k=0}^N\int_{2M}^{\infty}r^p\left(\left(r\frac{d}{dr}\right)^ku\right)^2\,dr+\int_{2M}^{\infty} Dr^p\left(\left(r\frac{d}{dr}\right)^{N+1}u\right)^2\,dr.
\end{align}
respectively, with $p\neq -1$ and
\begin{align*}
-\beta_{\ell}< &\:p<\beta_{\ell}\quad \textnormal{if $\beta_{\ell}\leq 1$},\\
-\beta_{\ell}<&\:p<2-\beta_{\ell}\quad \textnormal{if $\beta_{\ell}>1$}.
\end{align*}

\begin{proposition}
\label{prop:keyestw}
Let $u\in C_0^{\infty}([2M,R)\cup (R,\infty))$, then, for sufficiently large $R>2M$, there exists a $C=C(\alpha,p,N,R)>0$, such that
\begin{equation*}
||u||_{w_{\ell},N,p;1}\leq C||\mathcal{L}_{w_{\ell}}u||_{w_{\ell},N,p}.
\end{equation*}
\end{proposition}
\begin{proof}
First, consider the case $N=0$.  We first multiply $\mathcal{L}_{w_{\ell}}u$ with $-u$ and integrate on $[2M,R)$ to obtain:
\begin{equation*}
\begin{split}
\int_{2M}^{R}& Dr^2\left(\frac{du}{dr}\right)^2+\left[r^2V_{\alpha}+\ell(\ell+1)\right]u^2\,dr\\
= &\:\int_{2M}^{R}u \mathcal{L}_{w_{\ell}}u\,dr.
\end{split}
\end{equation*}
We apply Young's inequality to estimate:
\begin{equation*}
|u \mathcal{L}_{w_{\ell}}u|\leq \epsilon u^2+\frac{1}{4\epsilon} \left( \mathcal{L}_{w_{\ell}}u\right)^2.
\end{equation*}
Furthermore, by \eqref{eq:hardy}
\begin{equation*}
\int_{2M}^{R} u^2\,dr\leq 4\int_{2M}^{R} D^2r^{2}\left(\frac{du}{dr}\right)^2\,dr,
\end{equation*}
so, using that $\inf r^2V_{\alpha}>-\frac{1}{4}$, we obtain:
\begin{equation}
\label{eq:wenest1}
\begin{split}
\int_{2M}^{R}& Dr^2\left(\frac{du}{dr}\right)^2+u^2\,dr\leq C(\alpha)\int_{2M}^{R}(\mathcal{L}_{w_{\ell}}u)^2\,dr.
\end{split}
\end{equation}
Then we multiply by $-r^pu$ and integrate by parts on $[R,\infty)$ (using that $u$ vanishes at $r=R$) to obtain
\begin{equation}
\label{eq:wenest2}
\begin{split}
\int_{R}^{\infty}&  r^{2+p}\left(\frac{du}{dr}\right)^2+r^p\left[r^2V_{\alpha}+\ell(\ell+1)-\frac{1}{2}p(p+1)+O(r^{-1})\right]u^2\,dr\\
= &\:\int_{R}^{\infty}r^pu \mathcal{L}_{w_{\ell}}u\,dr.
\end{split}
\end{equation}
Furthermore, by the fact that $p\neq -1$, we can apply \eqref{eq:hardy} to obtain:
\begin{equation*}
\int_{R}^{\infty} r^p u^2\,dr\leq 4(p+1)^{-2}\int_R^{\infty} r^{p+2}\left(\frac{du}{dr}\right)^2\,dr.
\end{equation*} 
Since $r^2V_{\alpha}=\alpha + O_{\infty}(r^{-1})$, we therefore need additionally that:
\begin{equation*}
\alpha+\ell(\ell+1)-\frac{1}{2}p(p+1)>-\frac{1}{4}(p+1)^2,
\end{equation*}
or $-\beta_{\ell}<p<\beta_{\ell}$, and $R\gg M$ suitably large to absorb the $u^2$ term on the left-hand side of \eqref{eq:wenest2},  to obtain
\begin{equation}
\label{eq:wenest2b}
\begin{split}
\int_{R}^{\infty}& r^{2+p}\left(\frac{du}{dr}\right)^2+r^pu^2\,dr\leq C(\alpha,p)\int_{R}^{\infty}r^p(\mathcal{L}_{w_{\ell}}u)^2\,dr.
\end{split}
\end{equation}
We combine \eqref{eq:wenest1} and \eqref{eq:wenest2b} to obtain \eqref{eq:wenest} with $N=0$. The case $N>0$ then follows by induction, after integrating $-r^p(r\frac{d}{dr})^k(u) (r\frac{d}{dr})^k (\mathcal{L}_{w_{\ell}}u)$.
\end{proof}
\begin{corollary}
\label{eq:exuniqode}
Let $N\in \N_0$ and $\delta>0$ suitably small. Suppose $f\in H_{w_{\ell},p}^{N}$ with $p=\min\{\beta_{\ell},2-\beta_{\ell}\}-\delta$ and $p\neq -1$. Then there exists a unique $u\in H_{w_{\ell},p;1}^{N}$ such that $\mathcal{L}_{w_{\ell}}u=f$. In particular, if $f\in \cap_{k\in \N_0}H_{w_{\ell},p}^{k}$, then:
\begin{align*}
u\in&\: C^{\infty}([0,\infty)),\\
u(r)=&\:O_{\infty}(r^{-\frac{1}{2}-\frac{1}{2}\beta_{\ell}+\delta})\quad \textnormal{if}\quad  \alpha+\ell(\ell+1)\leq 0,\\
u(r)=&\: O_{\infty}(r^{-\frac{3}{2}+\frac{1}{2}\beta_{\ell}+\delta})\quad \textnormal{if}\quad \alpha+\ell(\ell+1)\geq 0.
\end{align*}
\end{corollary}
\begin{proof}
We apply Lax--Milgram with the bilinear map $B: H_{w_{\ell},p;1}^{0}\times H_{w_{\ell},p;1}^{0} \to \R$, defined by:
\begin{equation*}
\begin{split}
B(u,v)=&\:\int_{2M}^{R} Dr^2  \frac{du}{dr}(r)  \frac{dv}{dr}(r) + [r^2V_{\alpha}+\ell(\ell+1)]u(r) v(r)\,dr\\
&+\int_{R}^{\infty} Dr^{2+p} \frac{du}{dr}(r)  \frac{dv}{dr}(r) +r^{2+p}\left[r^2V_{\alpha}+\ell(\ell+1)-\frac{1}{2}p(p+1)\right]u(r) v(r)\,dr,
\end{split}
\end{equation*}
to obtain existence of solutions in $H_{w_{\ell},p}^{0} $, where the necessary coercivity properties of $B$ and the higher regularity properties of the solution follow from Proposition \ref{prop:keyestw}, together with standard interior regularity estimates.

Smoothness of $u$ follows from standard Sobolev estimates and the decay in $r$ follows by applying the fundamental theorem of calculus:
\begin{equation*}
\begin{split}
(r^{\frac{p}{2}+\frac{1}{2}}u)^2(r)=&\:(r^{\frac{p}{2}+\frac{1}{2}}u)^2(r_0)+\int_{r_0}^r(\frac{p}{2}+\frac{1}{2}) r^{p} u^2+2r^{p+1} u u'\,dr\\
\lesssim (r^{\frac{p}{2}+\frac{1}{2}}u)^2(r_0)+ \int_{r_0}^r  r^{p} u^2+r^{p+2}(u')^2\,dr,
\end{split}
\end{equation*}
together with Proposition \ref{prop:keyestw}, and analogous estimates for $u$ replaced by $(r\frac{d}{dr})^ku$, with $p=\beta_{\ell}-\delta$ if $\beta_{\ell}\leq 1$ and $p=2-\beta-\delta$ if $\beta_{\ell}>1$.
\end{proof}

\begin{proof}[Proof of Lemma \ref{lm:propertiesw}]
Define $v_{\ell}=r^{-\frac{1}{2}+\frac{1}{2}\beta_{\ell}}+b_0r^{-\frac{3}{2}+\frac{1}{2}\beta_{\ell}}$ and denote
\begin{equation*}
f_{\ell}(r)=(Dr^2 v_{\ell}')'(r)-(r^2V_{\alpha}+\ell(\ell+1))v_{\ell}(r).
\end{equation*}
Choose $b_0$ such that $f(R)=0$. Note that $f$ is smooth on $[0,\infty)$, and moreover $f_{\ell}=O_{\infty}(r^{-\frac{3}{2}+\frac{1}{2}\beta_{\ell}})$. Hence, $f\in H^N_{w_{\ell},p}$ for all $N\in \N$ with $p=\min\{\beta_{\ell},2-\beta_{\ell}\}-\delta$. We apply Corollary \ref{eq:exuniqode} to conclude that there exists a unique function $u\in C^{\infty}([0,\infty))$ such that $\mathcal{L}_{w_{\ell}}u=f_{\ell}$ and the properties listed in Corollary \ref{eq:exuniqode} hold.

Finally, we define $w_{\ell}=u-v_{\ell}$ to conclude that $\mathcal{L}_{w_{\ell}}w_{\ell}=0$ and $w_{\ell}$ satisfies the asymptotic properties stated in Lemma \ref{lm:propertiesw}. Note that strict positivity of $w_{\ell}$ can be shown to hold via a contradiction argument: suppose there exists $r_0\geq 2M$ such that $w_0(r_0)=0$. Then we can apply the estimates from the proof of Proposition \ref{prop:keyestw} with $R$ replaced by $r_0$ to obtain $L^2$-bounds for $w_{\ell}$ in $[r_0,\infty)$ with $p=0$, from which it follows that $w_{\ell}\equiv 0$ in $[r_0,\infty)$, which is a contradiction with the asymptotic properties of $w_{\ell}(r)$ as $r\to \infty$.
\end{proof}

\begin{remark}
From Proposition \ref{eq:exuniqode}, it follows that Assumption \ref{assm:A} implies  Assumption \ref{assm:B}.
\end{remark}

\small{\bibliographystyle{alpha} 
%\bibliography{../../references}

\end{document}